\documentclass[11pt]{article}

\usepackage{latexsym}
\usepackage{amsmath}
\usepackage{amssymb}
\usepackage{amsthm}
\usepackage{hyperref}
\usepackage{cite}
\usepackage{graphicx}
\usepackage{color}
\usepackage{enumitem}
\usepackage{bm}
\usepackage{comment}
\usepackage[left=1in,right=1in,top=1in,bottom=1in]{geometry}
\usepackage{xspace}
\setlist[description]{font=\normalfont\itshape\textbullet\space}
\usepackage[all]{xy}
\usepackage{doi, url}
\usepackage{framed}

\renewcommand{\paragraph}[1]{\vspace{6pt} \noindent \textbf{#1}\xspace}

\theoremstyle{plain}
\newtheorem{theorem}{Theorem}[section]

\newtheorem{corollary}[theorem]{Corollary}
\newtheorem{lemma}[theorem]{Lemma}

\newtheorem{proposition}[theorem]{Proposition}
\newtheorem{claim}[theorem]{Claim}
\newtheorem{fact}[theorem]{Fact}
\theoremstyle{definition}
\newtheorem{definition}[theorem]{Definition}
\newtheorem{example}[theorem]{Example}

\newcommand{\GL}{\mathrm{GL}}
\newcommand{\F}{\mathbb{F}}

\newcommand{\centre}{\mathrm{Z}}

\newcommand{\mtupletomatrix}{\mathtt{Matrix}}

\newcommand{\rowvectorspace}{\mathtt{RowVec}}
\newcommand{\rowtuplespace}{\mathtt{RowTuple}}

\newcommand{\poly}{\mathrm{poly}}

\newcommand{\rad}{\mathrm{rad}}

\newcommand{\M}{\mathrm{M}}

\newcommand{\charpoly}{\mathrm{ch}}

\newcommand{\vecz}{\mathbf{0}}
\newcommand{\tuple}[1]{\mathbf{#1}}

\newcommand{\spa}[1]{\mathcal{#1}}
\newcommand{\vspan}{\mathrm{span}}
\newcommand{\proj}{\mathrm{Proj}}
\newcommand{\rowvecproj}{\mathtt{RowVecProj}}
\newcommand{\EC}{\mathtt{EC}}

\newcommand{\IBCtuple}{IBC-tuple\ }
\newcommand{\IBCtuplenospace}{IBC-tuple}
\newcommand{\IBCtuples}{IBC-tuples\ }
\newcommand{\IBCtuplesnospace}{IBC-tuples}

\newcommand{\cA}{\spa{A}}

\newcommand{\cC}{\spa{C}}

\newcommand{\cG}{\spa{G}}

\newcommand{\vA}{\tuple{A}}
\newcommand{\vB}{\tuple{B}}
\newcommand{\vC}{\tuple{C}}
\newcommand{\vD}{\tuple{D}}
\newcommand{\vE}{\tuple{E}}
\newcommand{\vF}{\tuple{F}}
\newcommand{\vG}{\tuple{G}}
\newcommand{\vH}{\tuple{H}}

\newcommand{\vK}{\tuple{K}}

\newcommand{\vP}{\tuple{P}}
\newcommand{\vQ}{\tuple{Q}}

\newcommand{\va}{\tuple{a}}
\newcommand{\vb}{\tuple{b}}
\newcommand{\vc}{\tuple{c}}
\newcommand{\vd}{\tuple{d}}
\newcommand{\ve}{\tuple{e}}
\newcommand{\vf}{\tuple{f}}
\newcommand{\vg}{\tuple{g}}

\newcommand{\vq}{\tuple{q}}

\newcommand{\ibcspace}{V}
\newcommand{\ibcspacekernel}{K}
\newcommand{\ibcextensionspace}{E}
\newcommand{\ibcextensionspacekernel}{F}

\newcommand{\rowtupvs}[2]{{#1}^{{(#2)}}}

\DeclareMathOperator{\wc}{wc}

\DeclareMathOperator{\diag}{diag}

\DeclareMathOperator{\Env}{Env}

\newcommand{\too}%
{\xrightarrow{\text{\raisebox{-3pt}{$\sim$}}\,}}

\usepackage[T1]{fontenc}
\def\DJ{{\hbox{D\kern-.8em\raise.15ex\hbox{--}\kern.35em}}}

\usepackage{arydshln}
\hypersetup{
    pdftitle={Canonical forms for matrix tuples},
    pdfauthor={},
    bookmarksnumbered=true,     
    bookmarksopen=true,         
    bookmarksopenlevel=1,       
    colorlinks=true,            
    pdfstartview=FitH,           
    pdfpagemode=UseOutlines,    
    pdfpagelayout=OneColumn
}

\title{
Canonical forms for matrix tuples in polynomial time
}

\author{
Youming Qiao~\thanks{\texttt{youming.qiao@uts.edu.au}. University of Technology Sydney.}
\and 
Xiaorui Sun~\thanks{\texttt{xiaorui@uic.edu}. University of Illinois at Chicago.}
}

\date{}

\begin{document}

\maketitle

\begin{abstract}
Left-right and conjugation actions on matrix tuples have received considerable attention in theoretical computer science due to their connections with polynomial identity testing, group isomorphism, and tensor isomorphism.
In this paper, we present polynomial-time algorithms for computing canonical forms of matrix tuples over a finite field under these actions. 
Our algorithm builds upon new structural insights for matrix tuples, which can be viewed as a generalization of Schur's lemma for irreducible representations to general representations.

\end{abstract}

\thispagestyle{empty}

\newpage

\pagenumbering{arabic}  

\section{Introduction}
\vspace{-.2cm}Representing objects in a canonical and succinct way that can exhibit the underlying properties and structures of the objects is a fundamental problem in mathematics and computer science. 

A classic example is the Jordan normal form for matrices in linear algebra. It not only transforms the matrices into a canonical form under the similarity relation\footnote{Two $n\times n$ matrices $A$ and $B$ are similar if there exists an invertible matrix $T$ such that $A=TBT^{-1}$. Then $A$ and $B$ are similar if and only if their Jordan normal forms are the same.
}, but it also demonstrates important structural information such as characteristic polynomials, algebraic and geometric eigenvalue multiplicities, the structure of generalized eigenvectors, and invariant subspace decompositions. The Jordan normal form is an important archetype in some mathematical areas. For example, it leads to Jordan--Chevalley decompositions, a useful tool in the study of linear algebraic groups \cite{Borel}. It also implies that the classification problem for matrices under the conjugation action is ``tame'' which basically means it is classifiable in the representation theory of finite-dimensional algebras \cite{Ringel}. The Jordan normal form is also found useful in spectral graph theory \cite{DGT17,Her94}.

The general canonical form problem aims to transform 
combinatorial and algebraic objects, such as graphs~\cite{Bab19, babai2013faster, babai1983canonical,sun2015faster,schweitzer2019unifying,weisfeiler1968reduction}, tensors~\cite{NQT24,Wei84}, and groups~\cite{BEO02}, into a canonical representation such that for equivalent inputs, the output representation is the same. The study of these canonical forms leads to the development of a wide range of structural theories~\cite{babai2013faster, sun2015faster, schweitzer2019unifying, weisfeiler1968reduction}.

In this paper, we study the canonical representations of matrix tuples over finite fields. A matrix tuple is a sequence of matrices of the same size over the same finite field. 
We consider two actions for matrix tuples: the left-right action and the conjugation action. 
For a matrix tuple $(A_1, \dots, A_\ell)$ of size $n \times m$, 
the \emph{left-right action} by an $n\times n$ invertible matrix $L$ and an $m \times m$ invertible matrix $R$ transforms 
$(A_1, \dots, A_\ell)$ into another matrix tuple $(LA_1R^{-1}, \dots, LA_mR^{-1})$.
For the \emph{conjugation action}, it requires the matrix tuple to be square matrices, and the conjugation action by an invertible matrix $L$ sends $(A_1, \dots, A_\ell)$ to $(LA_1L^{-1}, \dots, LA_\ell L^{-1})$. 
Two matrix tuples are \emph{equivalent} if there exists a left-right action that transforms 
one matrix tuple into another, and two square matrix tuples are \emph{conjugate} if there is a conjugation action that transforms 
one matrix tuple into another.




A \emph{canonical form} algorithm for the matrix tuples in the equivalence case needs to satisfy the following conditions: given a matrix tuple $\vA=(A_1, \dots, A_\ell)$, the algorithm outputs $\vA^*=(A_1^*, \dots, A_\ell^*)$, such that $\vA$ and $\vA^*$ are equivalent, and for any matrix tuple $\vA'=(A_1', \dots, A_\ell')$ equivalent to $\vA$, the algorithm outputs the same $\vA^*$. In other words, $\vA^*$ serves as a representative in the set of matrix tuples equivalent to $\vA$. 
A canonical form algorithm in the matrix tuple conjugation case is defined in the same way by replacing ``equivalent'' with ``conjugate'' in the above.

As the main result of this article, we present polynomial-time canonical form algorithms for matrix tuples under equivalence or conjugation actions over finite fields.
In the following, we shall introduce motivations for studying this problem and then describe our results in more detail.

\vspace{-.2cm}\subsection{Motivations} 
Matrix tuples, which encode systems of linear transformations or bilinear forms, have been studied in various scenarios. 
We motivate the study of matrix tuples and their canonical form from both theoretical computer science and mathematics perspectives. 


\vspace{-.2cm}\subsubsection{Motivations from theoretical computer science}

\vspace{-.25cm}\paragraph{Orbit closure intersection problems.} 
Matrix tuples under the left-right action have received considerable attention in theoretical computer science \cite{AGL+18,DM20,GGOW16,HH21,IQ23,IQS17,IQS18}. 

One reason for interest in this action is the symbolic determinant identity testing (SDIT) problem. 
SDIT asks whether, for a given matrix tuple, the linear span of the matrices in this tuple contains a full-rank matrix. Derandomizing SDIT implies circuit lower bounds that seem beyond current techniques \cite{KI04}.
As the left-right action preserves matrix ranks, it is desirable to study the equivalence classes of matrix tuples under the left-right action, as pursued in several works mentioned above as well as in \cite{Mul17,MW21}.

In particular, the orbit closure intersection problems for the left-right and conjugation actions have surprising connections to many areas of mathematics \cite{GGOW16,IQS18}. By \cite{DM17}, they can be formulated as an instance of symbolic determinant identity testing. Recent advances \cite{AGL+18,DM20,IQ23} provide deterministic polynomial-time algorithms for these orbit closure intersection problems.

\paragraph{Matrix tuples from group isomorphism.} 
Testing the isomorphism of (finite) groups has been extensively studied since the 1970s. However, even after more than half a century, the best-known algorithm for group isomorphism remains a quasi-polynomial time algorithm dating back to the 1970s~\cite{FN70,Mil78}. 
Improving the running time for group isomorphism is of interest in both computer science and mathematics, as evidenced by Gowers' question \cite{Gow11}, which led to Wilson's work \cite{Wil19}.
On the other hand, due to the recent breakthrough in graph isomorphism by Babai \cite{Bab16}, group isomorphism has become a major bottleneck in making any further progress on graph isomorphism \cite{Bab16}. 

Very recently, Sun proposed an $n^{O((\log n)^{5/6})}$-time algorithm for testing isomorphism of $p$-groups of class $2$ and exponent $p$ \cite{Sun23}. 
This result removes a major barrier for an $n^{o(\log n)}$-time algorithm for group isomorphism. 
In Sun's work, the problem is reduced to understanding three matrix tuples under different actions, two of which are left-right actions. 
Therefore, it is conceivable that understanding the structure of matrix tuples can shed new light on further improvements to group isomorphism.

\paragraph{Towards understanding tensor isomorphism and canonical form.} 
Matrix tuples under equivalence actions also serve as an intermediate step to generalize our knowledge from matrices to tensors. 
Tensors have become increasingly important for computer science, not to mention their natural roles in statistics and quantum information. 

Equivalence relations of tensors are natural generalizations of equivalence relations of matrices, such as the similarity relation discussed in the context of Jordan normal forms.
One natural equivalence between tensors is as follows: Let $(A_1, \dots, A_n)$ and $(B_1, \dots, B_n)$ be two tuples of $n\times n$ matrices. 
We say that they are \emph{isomorphic as tensors} if there exist $n\times n$ invertible matrices $L, R, T=(t_{i,j})$, such that for every $i\in[n]$, $LA_iR^{-1}=\sum_{j\in[n]}t_{i,j}B_j$. 
The tensor isomorphism problem then asks to decide whether two given matrix tuples are isomorphic as tensors, and the tensor canonical form problem asks to compute a canonical form of the input matrix tuple that is invariant under the isomorphism as tensors.

The tensor isomorphism problem has been studied in a series of works~\cite{TI3,TI1,TI2,GQT22,TI4}, with applications found in quantum information~\cite{TI3,TI1}.
Current evidence suggests that tensor isomorphism is a hard problem. For $n\times n\times n$ tensors over a finite field $\F_q$, 
the best algorithm with worst-case analysis runs in time $q^{\tilde O(n^{1.5})}$ \cite{Frattini,Sun23}, with average-case analysis in time $q^{O(n)}$ \cite{BLQW,LQ17}, and with heuristic analysis in time $q^{\frac{1}{2}n}\cdot \poly(n, \log q)$ \cite{NQT24}. Indeed, because of these difficulties, digital signature schemes based on the assumed hardness of tensor isomorphism or equivalent problems have been proposed, including MEDS \cite{MEDS,MEDS_conf} and ALTEQ \cite{ALTEQ,ALTEQ_conf}, which are in submission to the NIST call for post-quantum digital signature schemes \cite{NIST_call}. Furthermore, recent research suggests that the tensor canonical form problem is also important to study in cryptography. For example, in~\cite{NQT24}, canonical forms for tensors are used as an isomorphism invariant for birthday paradox-based algorithms.

We believe that computing canonical forms for matrix tuples is an important intermediate step for the isomorphism and canonical form problems of tensors because matrix tuple equivalence is a more restricted form of tensor isomorphism.
Indeed, the best algorithms for tensor isomorphism~\cite{Frattini,NQT24,Sun23} are obtained by partially fixing the matrices in one direction of the tensor and then reducing the problem to certain equivalence problems for matrix tuples. 
Hence, we believe that understanding the structural and canonical form of matrix tuples is an important intermediate step toward understanding the structure and canonical form of tensors.

\vspace{-.25cm}\subsubsection{Motivations from mathematics}

\vspace{-.2cm}\paragraph{Matrix tuples as a wild classification problem.} In the representation theory of associative algebras, classifying representations of quivers is a central topic \cite{Ringel}, dating back to the work of Gelfand and Ponomarev\cite{gel-pon}. 
Roughly speaking, a classification problem is tame if it is classifiable (such as Jordan normal forms), and wild if it is not classifiable (defined as ``containing'' the problem of classifying pairs of matrices under simultaneous conjugation) (cf. \cite{Ben98}). 
The celebrated tame-wild dichotomy was proved by Drozd \cite{Drozd}. 
Classifying matrix tuples under equivalence (for no less than three matrices) or conjugation relations (for no less than two matrices) are well-known wild classification problems.

\vspace{-.15cm}\paragraph{Isomorphism and canonical form algorithms.} 
The wildness in classifying matrix tuples does not obstruct solutions to computational 
problems about them. 
Indeed, polynomial-time algorithms are known for testing whether two matrix tuples are conjugate or equivalent~\cite{BL08,IKS10,IQ17}. 
In particular, testing whether two matrix tuples are conjugate is a central problem in computer algebra \cite{BW15}, with practical algorithms implemented in computer algebra systems such as GAP \cite{GAP} and Magma \cite{Magma}. 
Interestingly, to the best of our knowledge, these algorithms have not led to a canonical form algorithm for matrix tuple conjugation or equivalence so far.

For canonical form problems, Belitskii and Sergeichuk presented algorithms for the canonical form problems for these actions over algebraically closed fields \cite{Bel00,bel-ser_compl,ser}. 
However, the complexity of the Belitskii--Sergeichuk algorithm seems missing in the literature\footnote{For example, there may be issues with the field extension degree required for the resulting canonical forms.}, and it is unclear to us whether their algorithm extends to the finite field setting. 
Nevertheless, their algorithms indicate that non-trivial algorithms can be designed for the canonical form problem, despite the wildness of the classification problem.

\subsection{Main results}  

We now state our main result. 
Let $\M(n\times m, \F_q)$ denote the linear space of $n\times m$ matrices over $\F_q$  and $\M(n \times m, \F_q)^\ell$ denote the linear space of matrix tuples of length $\ell$ with each matrix in $\M(n\times m, \F_q)$. 

\begin{theorem}\label{thm:main}
    There is a randomized Las-Vegas algorithm to compute a canonical form of a matrix tuple in $\M(n\times m, \F_q)^\ell$  under the equivalence relation in $\poly(n, m, \ell, \log q)$ time. 
\end{theorem}


Theorem~\ref{thm:main} relies on computing matrix algebra structures by Friedl, Ivanyos, and R\'onyai~\cite{FR85,Iva00,Ron90}. 
Since the algorithms for computing matrix algebra structures over finite fields utilize polynomial factoring~\cite{Ber67,CZ81}, our algorithm for Theorem~\ref{thm:main} is a Las-Vegas randomized algorithm.

It is well-known that a canonical form algorithm for matrix tuple equivalence implies a canonical form algorithm for matrix tuple conjugation.  
Therefore, Theorem~\ref{thm:main} also provides a canonical form algorithm for matrix tuple conjugation.

\begin{corollary}\label{cor:main}
There is a randomized Las-Vegas algorithm to compute a canonical form of a matrix tuple in $\M(n \times n, \F_q)^\ell$  under the conjugation relation in $\poly(n, \ell, \log q)$ time. 
\end{corollary}

The key to Theorem~\ref{thm:main} is a structural result for matrix tuples. To state our structural result, we introduce some definitions. 
Similar to a block-diagonal matrix, a block diagonal matrix tuple is a matrix tuple where a sequence of matrix tuple blocks lies along the diagonal, and all other entries in each matrix are zero.
A matrix tuple is decomposable if it is equivalent to a block-diagonal matrix tuple with at least two blocks; otherwise, it is indecomposable.


A row-submatrix tuple $\vB$ of a matrix tuple $\vA$ is a matrix tuple such that there exists a (not necessarily square) matrix $L$ such that $\vB =  L \vA$. 
A row-submatrix tuple $\vB$ of $\vA$ is an \emph{indecomposable-block-corresponding row-submatrix tuple}, or \emph{\IBCtuple}for short, if $\vA$ is equivalent to a block-diagonal matrix tuple $\vD$ such that $\vB$ corresponds to an indecomposable block of $\vD$.
In other words, there are invertible matrices $L$ and $R$ such that $\vD = \diag(\vD_1, \dots, \vD_d) = L \vA R^{-1}$ and $\vB R^{-1} = (B_1 R^{-1}, \dots, B_\ell R^{-1})$ equals $\begin{bmatrix}
    0 & \vD_i & 0
\end{bmatrix}$
as a row-submatrix tuple of $\vD$ for some block $\vD_i$. 

\begin{theorem}\label{thm:main_space_structure}
For a matrix tuple $\vA$ and an \IBCtuple $\vB$ of $\vA$, 
let $\ibcspace_\vB'$ denote the set of all the \IBCtuples $\vB'$ right-equivalent to $\vB$, i.e., there exists an invertible matrix $R$ such that $\vB' = \vB R^{-1}$, and let $\ibcspace_{\vB}$ be the linear span of $\ibcspace_\vB'$. 
Then there exists a subspace $K_\vB$ of $V_\vB$, such that $V_\vB'=V_\vB\setminus K_\vB$. 
\end{theorem}





To understand why Theorem~\ref{thm:main_space_structure} is interesting, note that the structure of $V_\vB'$, the set of \IBCtuples right-equivalent to $\vB$, could be highly nonlinear.
On the other hand, $V_\vB$ is a linear space of row-submatrix tuples. Theorem~\ref{thm:main_space_structure} then suggests that $V_\vB'$ can be viewed as a quotient space of $V_\vB$ over the subspace $\ibcspacekernel_{\vB}$, so it is close to being linear.

Theorem~\ref{thm:main_space_structure} can be seen as a generalization of the fundamental Schur's lemma in representation theory. Roughly speaking, a matrix tuple $\vA$ under the left-right action can be viewed as a representation of the so-called Kronecker quivers \cite{Ben98}. 
Schur's lemma states that if $\vA$ is not just indecomposable but also simple (also known as irreducible, a condition stronger than indecomposable, see \cite{Ben98}), then the endomorphism algebra of $\vA$ is a division algebra, which implies Theorem~\ref{thm:main_space_structure} for indecomposable and simple matrix tuples by specifying $K_\vA$ as the zero space.


Our result can be viewed as a generalization of Schur's lemma to general representations. We remark that such a generalization is nontrivial, even with the known characterization of endomorphism algebras of indecomposable modules as local algebras \cite[Section 2, Theorem 2]{Alp93}. In particular, 
Theorem~\ref{thm:main_space_structure} for indecomposable subrepresentations in the general situation is affected by the interactions between non-equivalent indecomposable subrepresentations. 
Furthermore, unlike Schur's lemma, which is concerned with endomorphism algebras or homomorphisms, the subject in Theorem~\ref{thm:main_space_structure} is indecomposable subrepresentations up to right equivalence, which seems not studied in the literature.
Our main contribution is to discover and utilize this structural result in the process of devising canonical form algorithms.

\section{Overview}\label{sec:overview}

In this section, we provide a high-level overview of the algorithm to compute canonical forms for matrix tuples under the left-right action in polynomial time.

For convenience, we assume that all the vectors are row vectors throughout the paper.
    For a matrix tuple $\vA = (A_1, \dots, A_\ell) \in \M(n\times m, \F_q)^\ell$, we define a tuple of row vectors $\va \in (\F_q^m)^\ell$ as a \emph{row tuple} of $\vA$ if there exists a row vector $v \in \F_q^n$ such that $\va = v \vA$, where $v \vA = (v A_1, \dots, v A_\ell)$. We refer to $v A_i$ as the \emph{$i$-th coordinate} of $\va$ for every $i \in [\ell]$. 
For example, 
    consider the matrix tuple $\vA = (A_1, A_2) \in \M(3 \times 4, \F_2)^2$:
    \[
A_1 = \begin{bmatrix}
    1 & 0 & 1 & 0 \\
    0 & 1 & 0 & 0 \\
    0 & 1 & 1 & 1 \\
\end{bmatrix},
A_2 = \begin{bmatrix}
    0 & 0 & 0 & 1 \\
    1 & 1 & 0 & 0 \\
    0 & 0 & 1 & 0 \\
\end{bmatrix}.
\]
Then $\va = ((1, 1, 0, 1), (0, 0, 1, 1))$ is a row tuple of $\vA$ because $\va = (1, 0, 1) \vA$, and   $(1,1,0,1)$ is the first coordinate of $\va$. 
On the other hand, 
$\va' = ((1, 1, 1, 1), (1, 1, 0, 1))$ is not because no row vector $v\in \F_2^3$ satisfies $\va' = v \vA$. \\

In this paper, we investigate \IBCtuples of a matrix tuple. 
Structurally, we prove Theorem~\ref{thm:main_space_structure}. 
Algorithmically, we give an algorithm to compute a representative \IBCtuple sequence, denoted as $\vB_1, \dots, \vB_k$, for an input matrix tuple $\vA$. 
This sequence consists of \IBCtuples that are mutually non-equivalent (treating \IBCtuples as matrix tuples), and every \IBCtuple of $\vA$ is equivalent to one of the \IBCtuples in the sequence. 
Moreover, the representative \IBCtuple sequence produced by our algorithm is canonical. That is, for two equivalent matrix tuples $\vA$ and $\vA'$, the representative \IBCtuple sequences $\vB_1, \dots, \vB_k$ for $\vA$ and $\vB_1', \dots, \vB_k'$ for $\vA'$ satisfy that $\vB_i$ and $\vB_i'$ are right-equivalent for every $i \in[k]$. 

Based on such a representative \IBCtuple sequence, we can compute the canonical form of the matrix tuple by selecting \IBCtuples in a certain way from the linear spaces spanned by \IBCtuples right-equivalent to each of $\vB_1, \dots, \vB_k$ according to Theorem~\ref{thm:main_space_structure} (Section~\ref{sec:canonical_form_from_rep}). 

In this overview, we focus on computing a representative \IBCtuple sequence for a matrix tuple canonically. 
We emphasize that this is a challenging task because an \IBCtuple may have a lot of equivalent \IBCtuples that are not right-equivalent to itself. 
Selecting a representative \IBCtuple among these equivalents canonically requires an in-depth analysis of the structure of equivalent \IBCtuples in the input matrix tuple.

Our algorithm systematically explores the structure of the matrix tuple by detecting non-trivial characteristic subspaces of row tuples.
It continues this process until the information gathered from these characteristic subspaces enables the construction of the desired \IBCtuples in a canonical way. 
Here, a linear subspace of row tuples (or row vectors)  for a matrix tuple is characteristic if, under any automorphism of the matrix tuple by left-right action, the subspace remains invariant.

The characteristic subspace naturally emerges when specific row tuples or row vectors are distinguished from others.
For instance, let us consider a matrix tuple where the first matrix is not full row rank. 
In this case, all the row tuples with zero vector as their first coordinate form a characteristic row tuple subspace within the linear space spanned by all the row tuples of the matrix tuple.
As another example, 
all the row vectors in the first matrix of a matrix tuple form a characteristic row vector subspace.
However, in some cases, identifying the characteristic row tuple subspace is not as straightforward as illustrated in previous examples, 
requiring a careful analysis of the matrix tuple structure.

\subsection{Matrix tuple without non-trivial characteristic row tuple subspace}\label{subsec:finite_field}
The first question is under which circumstances the matrix tuple does not have a non-trivial characteristic row tuple subspace and how to compute a representative  \IBCtuple sequence canonically in such cases. 
To address this question, let us consider the case in which the matrix tuple is square. Without loss of generality, we assume that every matrix in the matrix tuple is full rank. Otherwise, it would be possible to obtain some characteristic row tuple subspace using the aforementioned approach.

For a length-$\ell$ matrix tuple $\vA = (A_1, \dots, A_\ell)$, we investigate the induced matrix tuple $\vG = (A_1^{-1} A_2, \dots, A_1^{-1}A_\ell)$ of length $\ell - 1$ under the conjugation action. 
We show that $\vA$ has no nontrivial characteristic row tuple subspace if only if the matrix algebra generated by $\vG$ is isomorphic to a finite field. 
Furthermore, our algorithm verifies whether the matrix algebra generated by $\vG$ is a finite field.
If it is not, our algorithm proceeds to produce a characteristic row tuples subspace of $\vA$ based on the classification of matrix algebra by the Artin–Wedderburn–Mal’tsev theory.
If the matrix algebra generated by $\vG$ is isomorphic to a finite field, we show that all the \IBCtuples of the matrix tuple are equivalent, and thus, any representative \IBCtuple sequence contains a single \IBCtuplenospace. In addition, we give an algorithm to select a representative  \IBCtuple canonically.
That is, the \IBCtuple selected for the sequence among equivalent input matrix tuples are right-equivalent (Lemma~\ref{lem:conjugation_canonical_main}). 

Let us consider an example with $\vA = (A_1, A_2, A_3) \in \M(4 \times 4, \F_2)^3$ being the following matrix tuple 
\[
A_1 = \begin{bmatrix}
    1 & 0 & 1 & 0 \\
    0 & 1 & 0 & 0 \\
    0 & 1 & 1 & 1 \\
    1 & 1 & 1 & 1 \\
\end{bmatrix},
A_2 = \begin{bmatrix}
    0 & 0 & 0 & 1 \\
    1 & 1 & 0 & 0 \\
    0 & 0 & 1 & 0 \\
    0 & 1 & 1 & 0 \\
\end{bmatrix}, 
A_3 = \begin{bmatrix}
    1 & 0 & 1 & 1 \\
    1 & 0 & 0 & 0 \\
    0 & 1 & 0 & 1 \\
    1 & 0 & 0 & 1 \\
\end{bmatrix}.
\]
It can be verified that all the $A_i$ matrices are invertible. Let $\vG = (A_1^{-1} A_2, A_1^{-1} A_3)$. The matrix algebra generated by $\vG$ is isomorphic to $\F_4$. 

For the case that the matrix algebra generated by $\vG$ is isomorphic to a finite field, we observe that if an \IBCtuple $\vB$ contains a row tuple $\va$ (i.e., there is a row vector $v$ such that $\va = v \vB$), then all the row tuples containing a vector that is a linear combination of the coordinates of $\va$ must also be in $\vB$. 
For example, start with the row tuple $\va_ 1= e_1 \vA = ((1, 0, 1, 0), (0, 0, 0, 1), (1, 0, 1, 1))$, 
the row tuple $\va_2 = (1, 1, 0, 1)\vA = ((0, 0, 0, 1), (1, 0, 1, 1), (1, 0, 1, 0))$ is in any \IBCtuple containing $\va_1$
because the first coordinate of $\va_2$ is the same as the second coordinate of $\va_1$.

On the other hand, we show that in this case, any row-submatrix tuple $\vC$ can be decomposed into one or a few \IBCtuples if and only if any row tuple containing a coordinate being a row vector in $\vC$ is contained in $\vC$. 
As any row tuple containing a row vector that is a linear combination of coordinates of $\va_1$ and $\va_2$ is also a linear combination of $\va_1$ and $\va_2$, 
the row-submatrix tuple $\vB$ of two rows with $\va_1$ as the first row and $\va_2$ as the second row
is an \IBCtuple of $\vA$. 

Furthermore, let $\va_1'$ be an arbitrary row tuple of $\vA$ and $\va_2'$ be the row tuple of $\vA$ whose first coordinate equals the second coordinate of $\va_1'$. Then the row-submatrix tuple $\vB'$ of two rows with $\va_1'$ as the first row and $\va_2'$ as the second row is also an \IBCtuple of $\vA$, right isomorphic $\vB$. 
Hence, 
a representative \IBCtuple can be obtained by identifying row tuples that have to belong to the same \IBCtuple (as $\va_1$ and $\va_2$ in this example)
and 
defining the relations of these row tuples canonically (as the second coordinate of the first row of the \IBCtuple equals the first coordinate of the second row in this example). 

As another example, let $\vA' = (A_1', A_2', A_3')$ be the matrix tuple with
\[
A_1' = \begin{bmatrix}
    1 & 0 & 0 & 0 \\
    0 & 0 & 0 & 1 \\
    0 & 0 & 1 & 0 \\
    0 & 1 & 0 & 0 \\
\end{bmatrix},
A_2' = \begin{bmatrix}
    0 & 1 & 0 & 0 \\
    0 & 0 & 1 & 0 \\
    0 & 0 & 0 & 1 \\
    1 & 0 & 1 & 0 \\
\end{bmatrix}
A_3' = \begin{bmatrix}
    0 & 1 & 0 & 0 \\
    0 & 0 & 1 & 1 \\
    0 & 0 & 0 & 1 \\
    1 & 1 & 0 & 0 \\
\end{bmatrix}.
\]
Because the matrix algebra generated by $((A'_1)^{-1} A'_2, (A'_1)^{-1}A'_3)$ has a nontrivial radical (i.e., the largest nilpotent ideal of the matrix algebra) generated by the following matrices
\[\begin{bmatrix}
    0 & 0 & 1 & 0 \\
    0 & 0 & 0 & 0 \\
    0 & 0 & 0 & 0 \\
    0 & 0 & 0 & 0 \\
\end{bmatrix},
\begin{bmatrix}
    0 & 0 & 0 & 1 \\
    0 & 0 & 0 & 0 \\
    0 & 0 & 0 & 0 \\
    0 & 0 & 0 & 0 \\
\end{bmatrix},
\begin{bmatrix}
    0 & 0 & 0 & 0 \\
    0 & 0 & 1 & 0 \\
    0 & 0 & 0 & 0 \\
    0 & 0 & 0 & 0 \\
\end{bmatrix},
\begin{bmatrix}
    0 & 0 & 0 & 0 \\
    0 & 0 & 0 & 1 \\
    0 & 0 & 0 & 0 \\
    0 & 0 & 0 & 0 \\
\end{bmatrix},
\]
the linear span of the row vectors in the matrices of the radical, i.e., $\langle \{ (0, 0, 0, 1), (0, 0, 1, 0)\}\rangle$, is a characteristic row vector subspace of $\vA'$. Consequently, $\langle \{ e_2 \vA', e_3 \vA' \}\rangle$ is a characteristic row tuple subspace because the row tuples in this subspace have all coordinates in the characteristic row vector subspace. 
We will address this situation in Section~\ref{sec:overview_hier}. 

Next, we turn to the scenario where the matrix tuple is a rectangle with no obvious characteristic row tuple subspace. 
In this case, the rectangle must have more columns than rows because otherwise, there are characteristic row tuple subspaces, like the linear space spanned by row tuples with the $i$-th coordinate being zero.

On the other hand, 
although there is no characteristic row tuple subspace for such rectangle matrix tuples, there must be nontrivial characteristic row vector subspaces because there are more columns than rows. 
Let $U_i$ be the linear space spanned by row vectors of the $i$-th matrix in the matrix tuple. 
$U_i$ is a nontrivial characteristic row vector subspace for each $i$ if the $i$-th matrix is nonzero.

In our algorithm, we utilize the correspondence between these characteristic row vector subspaces to either identify some characteristic row tuple subspaces or reduce the rectangle matrix tuple problem to the square matrix tuple problem. 
Specifically, if $U_i$ and $U_j$ have some nontrivial intersection, then such an intersection allows us to find some nontrivial characteristic row tuple subspace.
Otherwise, the entire row vector space corresponds to a direct sum of some of the $U_i$ row vector subspaces. In this case, one can canonically define the correspondence between row vectors from different $U_i$ subspaces and then create a square matrix tuple corresponding to the rectangle matrix tuple in terms of the \IBCtuple structure.
 Hence, computing a representative \IBCtuple sequence for the rectangle matrix tuple is reduced to computing a representative \IBCtuple sequence for the constructed square matrix tuple.

\subsection{Matrix tuple with direct sum row tuple decomposition}\label{sec:direct_sum_overview}

In Section~\ref{subsec:finite_field}, we either obtain some characteristic row tuple subspace or have a complete characterization of the \IBCtuples and have an algorithm to construct a representative \IBCtuple sequence for such a matrix tuple. 
The next major question is how to canonically compute a representative \IBCtuple sequence given some nontrivial characteristic row tuple subspaces. 
In this section, we consider the base case that the entire row tuple space is the direct sum of some characteristic row tuple subspaces. We will study the general case in Section~\ref{sec:overview_hier}.

The solution for the direct sum of characteristic row tuple subspaces is to identify the correspondence between row tuples from different characteristic row tuple subspaces and construct a new matrix tuple without nontrivial characteristic row tuple subspaces by merging corresponding row tuples from different characteristic row tuple subspaces into a row tuple in the new matrix tuple (Section~\ref{sec:direct_sum_decomposition}). 
For example, consider the following matrix tuple $\vA = (A_1, \dots, A_4)\in \M(4 \times 4, \F_2)^4$. 
\[
A_1 = \begin{bmatrix}
    0 & 0 & 0 & 0 \\
    0 & 0 & 0 & 0 \\
    0 & 0 & 1 & 0 \\
    0 & 0 & 0 & 1 \\
\end{bmatrix},
A_2 = \begin{bmatrix}
    1 & 0 & 0 & 0 \\
    0 & 1 & 0 & 0 \\
    0 & 0 & 0 & 0 \\
    0 & 0 & 0 & 0 \\
\end{bmatrix},
A_3 = \begin{bmatrix}
    0 & 1 & 0 & 0 \\
    1 & 0 & 0 & 0 \\
    0 & 0 & 0 & 1 \\
    0 & 0 & 1 & 0 \\
\end{bmatrix},
A_4 = \begin{bmatrix}
    0 & 0 & 0 & 0 \\
    0 & 0 & 0 & 0 \\
    1 & 0 & 0 & 0 \\
    0 & 1 & 0 & 0 \\
\end{bmatrix}
\]
Naturally, $\vA$ has two characteristic row tuple subspaces: $T_1$, which contains all the row tuples of $\vA$ whose first coordinates are zero, i.e., $T_1 = \langle \{e_1 \vA, e_2 \vA\}\rangle $, and $T_2$ which contains all the row tuples of $\vA$ whose second coordinates are zero, i.e., $T_2 = \langle \{e_3 \vA, e_4 \vA\}\rangle $. 
Consequently, let $U_1$ be the linear space spanned by the second coordinates of the row tuples in $T_1$ (i.e., $U_1 = \langle \{e_1, e_2\}\rangle$) and $U_2$ be the linear space spanned by the first coordinates of the row tuples in $T_2$ (i.e., $U_2 = \langle \{e_3, e_4\}\rangle$). 
$U_1$ and $U_2$ are characteristic row vector subspaces, and the entire row vector space of $\vA$ is the direct sum of $U_1$ and $U_2$.

Since both the second coordinates of $T_1$ and the fourth coordinates of $T_2$ are row vectors in $U_2$, we can define an isomorphism $f : T_1 \rightarrow T_2$ such that $f(\va) = \vb$ for any $\va \in T_1, \vb \in T_2$ if and only if the second coordinate of $\va$ is equal to the fourth coordinate of $\vb$. 
Then, by choosing an arbitrary linear basis of the row tuples in $T_1$, we can construct a matrix tuple $\vC = (C_1, \dots, C_8) = \M(2 \times 4, \F_2)^8$ such that for any $\vc$ as a row tuple of $\vC$, the first four coordinates of $\vc$ corresponds to a row tuple $\va$ of $T_1$, and the last four coordinates of $\vc$ correspond to the $f(\va)$ of $T_2$. 
One example of $\vC$ is 
\vspace{-.1cm}\[
C_1 = \begin{bmatrix}
    0 & 0 & 0 & 0 \\
    0 & 0 & 0 & 0 \\
\end{bmatrix},
C_2 = \begin{bmatrix}
    1 & 0 & 0 & 0 \\
    0 & 1 & 0 & 0 \\
\end{bmatrix},
C_3 = \begin{bmatrix}
    0 & 1 & 0 & 0 \\
    1 & 0 & 0 & 0 \\
\end{bmatrix},
C_4 = \begin{bmatrix}
    0 & 0 & 0 & 0 \\
    0 & 0 & 0 & 0 \\
\end{bmatrix}
\]
\[
C_5 = \begin{bmatrix}
    0 & 0 & 1 & 0 \\
    0 & 0 & 0 & 1 \\
\end{bmatrix},
C_6 = \begin{bmatrix}
    0 & 0 & 0 & 0 \\
    0 & 0 & 0 & 0 \\
\end{bmatrix},
C_7 = \begin{bmatrix}
    0 & 0 & 0 & 1 \\
    0 & 0 & 1 & 0 \\
\end{bmatrix},
C_8 = \begin{bmatrix}
    1 & 0 & 0 & 0 \\
    0 & 1 & 0 & 0 \\
\end{bmatrix}.
\]

\vspace{-.1cm}By the correspondence between $\vC$ and $\vA$, 
we show that the \IBCtuples of $\vC$ and $\vA$ have a one-to-one correspondence. Therefore, once we obtain a canonical representative \IBCtuple sequence of $\vC$ according to Section~\ref{subsec:finite_field}, we can compute a representative \IBCtuple sequence of $\vA$ canonically. 

\vspace{-.2cm}\subsection{Matrix tuple with hierarchical row tuple decomposition}\label{sec:overview_hier}

\vspace{-.1cm}We discuss the algorithm for constructing the representative \IBCtuple sequence of a matrix tuple for general characteristic row tuple subspaces. 
The challenge arises from the potential impossibility of decomposing the entire row tuple space of the matrix tuple into a direct sum of a few characteristic row tuple subspaces, as demonstrated by $\vA'$ defined in Section~\ref{subsec:finite_field}.

\vspace{-.3cm}\subsubsection{Hierarchical row tuple
decomposition and quotient matrix tuple}

\vspace{-.1cm}We organize the characteristic row tuple subspaces hierarchically to unveil the direct sum property for characteristic row tuple subspaces level by level via maintaining a \emph{hierarchical row tuple decomposition} (Definition~\ref{def:row_tuple_decomposition}). 
For simplicity, in this overview, we assume the hierarchical row tuple decomposition contains only two levels.

A two-level hierarchical row tuple decomposition consists of a sequence of characteristic row tuple subspaces $T_1, \dots, T_\zeta$ and a parameter $1 < h \leq \zeta$.
The decomposition is hierarchical in the following sense: Let $W$ be the linear space spanned by the row vectors in the row tuples of $T_{h}, \dots, T_\zeta$.
Then, the following two conditions hold: 
\begin{enumerate}
\vspace{-.2cm}\item 
Let $S$ be the linear space spanned by row tuples with all coordinates in $W$. Then $T_{h}, \dots, T_\zeta$ is a direct sum decomposition of $S$.
\vspace{-.25cm}\item
    All the row tuples of the matrix tuple become the direct sum of $T_{1}, \dots, T_{h - 1}$ after shrinking all the row vectors in $W$ to zero for each row tuple of the matrix tuple. (The row tuples in $T_h, \dots, T_\zeta$ and the row tuples in $T_1, \dots, T_{h-1}$ that are also in $S$ shrink to zero row tuples.)
\end{enumerate}

\vspace{-.2cm}Based on a two-level hierarchical row tuple decomposition of a matrix tuple, 
we can construct a \emph{quotient matrix tuple} for the matrix tuple by shrinking all vectors in $W$ to zero. 
We ensure that the resulting matrix tuple has all the rows linearly independent by arbitrarily choosing a linear basis of the row tuples after shrinking. 
For example, let $\vA = (A_1, A_2) \in \M(8 \times 8, \F_2)^2$ with \begin{equation}\label{equ:overview_1}
A_1 = \begin{bmatrix}
    1 & 0 & 0 & 0 & 0 & 0 & 1 & 0\\
    0 & 1 & 0 & 0 & 0 & 0 & 0 & 0\\
    0 & 0 & 1 & 0 & 0 & 0 & 0 & 1\\
    0 & 0 & 0 & 1 & 0 & 0 & 0 & 0\\
    0 & 1 & 0 & 0 & 1 & 0 & 0 & 0\\
    0 & 1 & 0 & 0 & 0 & 1 & 0 & 0\\
    0 & 0 & 0 & 0 & 0 & 0 & 1 & 0\\
    0 & 0 & 0 & 1 & 0 & 0 & 0 & 1\\
\end{bmatrix},
A_2 = \begin{bmatrix}
    0 & 1 & 0 & 0 & 0 & 0 & 0 & 1\\
    1 & 1 & 1 & 0 & 0 & 0 & 0 & 0\\
    0 & 0 & 0 & 1 & 0 & 0 & 1 & 1\\
    0 & 0 & 1 & 1 & 0 & 0 & 0 & 0\\
    1 & 1 & 1 & 0 & 0 & 1 & 0 & 0\\
    1 & 1 & 1 & 0 & 1 & 1 & 1 & 0\\
    0 & 0 & 0 & 0 & 0 & 0 & 0 & 1\\
    0 & 0 & 1 & 1 & 0 & 0 & 1 & 1\\
\end{bmatrix}. 
\end{equation}
Let $T_1$ be the linear space spanned by all the row tuples of $\vA$. 
By Section~\ref{subsec:finite_field}, using the radical of the matrix algebra generated by $A_1^{-1}A_2$, $\vA$ has a characteristic row tuple subspace $T_2 = \{e_3 \vA, e_4 \vA, e_7 \vA, e_8 \vA\}$. 
$T_1$ and $T_2$ with $h=1$ form a two-level hierarchical row
tuple decomposition of $\vA$. 
Consequently, $W = \langle \{e_3, e_4, e_7, e_8\}\rangle$, and the matrix tuple $\vQ = (Q_1, Q_2)$ with
\begin{equation}\label{equ:overview_2}
Q_1 = \begin{bmatrix}
    1 & 0 & 0 & 0 \\
    0 & 1 & 0 & 0 \\
    0 & 1 & 1 & 0 \\
    0 & 1 & 0 & 1 \\
\end{bmatrix},
Q_2 = \begin{bmatrix}
    0 & 1 & 0 & 0 \\
    1 & 1 & 0 & 0 \\
    1 & 1 & 0 & 1 \\
    1 & 1 & 1 & 1 \\
\end{bmatrix}
\end{equation}
is a quotient matrix tuple of $\vA$ such that $e_1, e_2, e_3$ and  $e_4$ of $\vQ$ correspond to $e_1 + W, e_2 + W, e_5 + W$ and $e_6 + W$, where $v + W$ denotes $\{v + w : w \in W\}$.

\vspace{-.2cm}\subsubsection{\IBCtuple construction via quotient matrix tuple}
To obtain the \IBCtuples for an input matrix tuple, we investigate the relations between the \IBCtuples of the quotient matrix tuple and the input matrix tuple (Section~\ref{sec:quotient_matrix_tuple_main}).

As the base case, by the definition of a two-level hierarchical row
tuple decomposition, any quotient matrix tuple is associated with a direct sum row tuple decomposition induced by $T_1, \dots, T_{h-1}$. Hence, we can use the approach in Section~\ref{sec:direct_sum_overview} to obtain a representative \IBCtuple sequence for an arbitrary quotient matrix tuple. 
In the rest of this section, we give an overview of our approach to computing a representative \IBCtuple sequence for a matrix tuple $\vA$ based on a representative \IBCtuple sequence for a quotient matrix tuple $\vQ$ of $\vA$.




First, consider an arbitrary block-diagonal matrix tuple $\vD$ equivalent to $\vA$ such that all the blocks of $\vD$ are indecomposable. 
We observe that the row vectors from different blocks of $\vD$ that correspond to row vectors in $W$ (the row vectors shrunk to zero when computing the quotient matrix tuple) of $\vA$ via a left-right action span a row vector subspace of $\vD$ corresponding to $W$ of $\vA$.
If we shrink the row vectors in this row vector subspace to zero for $\vD$, similar to constructing a quotient matrix tuple for $\vA$, then we get a block-diagonal matrix tuple equivalent to $\vQ$. 
We show that this block-diagonal matrix tuple has every block corresponding to a few \IBCtuples of $\vQ$. 
So, we want to understand the following question: given an \IBCtuple $\vB$ for $\vQ$, if $\vA$ has an \IBCtuple $\vC$ containing a row-submatrix tuple corresponding to $\vB$ via the correspondence between $\vA$ and $\vQ$, what is such a row-submatrix tuple?

To answer this question, we extend the \IBCtuples for $\vQ$ to row-submatrix tuples of $\vA$. 
We say a row-submatrix tuple $\vC$ of $\vA$ is an \emph{extension} of an \IBCtuple $\vB$ of $\vQ$ in $\vA$ if 
$\vC$ corresponds to $\vB$ via the correspondence between $\vA$ and $\vQ$. For example, consider the matrix tuple $\vA$ defined by Equation (\ref{equ:overview_1}) and quotient matrix tuple $\vQ$ defined by Equation (\ref{equ:overview_2}). 
Using the result from Section~\ref{subsec:finite_field}, 
the following $\vB$ is an \IBCtuple of $\vQ$, 
\begin{equation}\label{equ:overview_4}
\vB = \begin{pmatrix} \begin{bmatrix}
    1 & 0 & 0 & 0 \\
    0 & 1 & 0 & 0 
\end{bmatrix}, 
\begin{bmatrix}
    0 & 1 & 0 & 0 \\
    1 & 1 & 0 & 0 
\end{bmatrix}\end{pmatrix}\end{equation}
and the following $\vC$ is an extension of $\vB$ in $\vA$. 
\begin{equation}\label{equ:overview_3}
\vC = 
\begin{pmatrix} \begin{bmatrix}
    1 & 0 & 0 & 0 & 0 & 0 & 1 & 0\\
    0 & 1 & 0 & 0 & 0 & 0 & 0 & 0
\end{bmatrix}, 
\begin{bmatrix}
    0 & 1 & 0 & 0 & 0 & 0 & 0 & 1\\
    1 & 1 & 1 & 0 & 0 & 0 & 0 & 0
\end{bmatrix}
\end{pmatrix}
\end{equation}

There are many feasible extensions for an \IBCtuple of $\vQ$ in $\vA$, but not all of them can be contained in an \IBCtuple of $\vA$. 
By investigating the relations between extensions and row vectors in $W$, we observe that for an extension to be contained in an \IBCtuple of $\vA$,
the linear span of row vectors in the extension must contain only the necessary row vectors from $W$. 
We refer to such extensions as ``essential extensions''.  

For example, the $\vC$ defined by Equation (\ref{equ:overview_3}) is not essential for $\vB$ as defined by Equation (\ref{equ:overview_4}) because there is an \IBCtuple of $\vA$ containing an extension of $\vB$ but not containing the row vectors $e_7$ and $e_8$, which are vectors in $W$ that are also in the linear space spanned by row vectors of $\vC$. 
The row-submatrix tuple $\vC'$ defined below is an essential extension of $\vB$ because all the \IBCtuples of $\vA$ containing an extension of $\vB$ always contain the row vector $e_3$, which is the only vector in $W$ and also in the linear space spanned by the row vectors of $\vC'$. 
\[
\vC' = 
\begin{pmatrix} \begin{bmatrix}
    1 & 0 & 0 & 0 & 0 & 0 & 0 & 0\\
    0 & 1 & 0 & 0 & 0 & 0 & 0 & 0
\end{bmatrix}, 
\begin{bmatrix}
    0 & 1 & 0 & 0 & 0 & 0 & 0 & 0\\
    1 & 1 & 1 & 0 & 0 & 0 & 0 & 0
\end{bmatrix}
\end{pmatrix}
\]

We show that such essential extensions exist, and if an \IBCtuple of $\vA$ contains a row-submatrix tuple corresponding to $\vB$, then the \IBCtuple of $\vA$ must contain an essential extension of $\vB$.  
We also give an algorithm to compute the canonical essential extensions.  
Furthermore, we show that if the \IBCtuples right-isomorphic to a given \IBCtuple of $\vQ$ satisfy Theorem~\ref{thm:main_space_structure}, 
then the essential extensions obtained for these \IBCtuples by our algorithm also have the linearity similar to Theorem~\ref{thm:main_space_structure}. 



Second, we explore the connection between essential extensions of \IBCtuples of $\vQ$ and row tuples in the linear span of $T_h, \dots, T_\zeta$ by constructing a new matrix tuple, called the \emph{compression matrix tuple}. 
Roughly speaking, to construct the compression matrix tuple, 
for each essential extension obtained, we use a new row tuple to canonically encode the intersection of $W$ and the linear space spanned by row vectors of the essential extension.
By the linearity of essential extensions for right-equivalent \IBCtuples of $\vQ$, the new row tuples constructed for extensions of right-equivalent \IBCtuples of $\vQ$ also span a linear space of row tuples.

The compression matrix tuple $\vE$ consists of row tuples with each coordinate in $W$. It contains two parts: one part corresponds to the row tuples of $\vA$ with all row vectors in $W$ (i.e., row tuples in $T_{h}, \dots, T_\zeta$), and another part corresponds to the newly constructed row tuples from essential extensions.




The construction of the compression matrix tuple naturally results in a direct sum decomposition of the characteristic row tuple subspaces. 
Therefore, we apply the algorithm described in Section~\ref{sec:direct_sum_overview} to find the \IBCtuples of the compression matrix tuple. 
If the algorithm in Section~\ref{sec:direct_sum_overview} 
returns a new characteristic row tuple subspace of $\vE$, then we can use this characteristic row tuple subspace to further refine the hierarchical row tuple decomposition we have for $\vA$.
We then restart the entire process with the refined hierarchical row tuple decomposition.

Finally, we study the consequence
of the algorithm in Section~\ref{sec:direct_sum_overview} returning a representative \IBCtuple sequence for the compression matrix tuple $\vE$.
We show that 
in this case,
for any block-diagonal matrix tuple $\vD_\vA$ equivalent to $\vA$, 
there exists a block-diagonal matrix tuple $\vD_\vE$ equivalent to $\vE$, such that there is a one-to-one correspondence between the blocks of $\vD_\vA$ 
and blocks of $\vD_\vE$, thereby
implying a correspondence between \IBCtuples of $\vA$ and \IBCtuples of $\vE$. 
By carefully analyzing this correspondence, we give an algorithm for constructing a representative \IBCtuple sequence for $\vA$ canonically guided by the \IBCtuples of $\vE$.

\paragraph{Paper organization }  In Section~\ref{sec:prelim}, we define the notations and provide the preliminaries. 
Section~\ref{sec:ibctuple} formally defines \IBCtuple and hierarchical row tuple decomposition and proves some useful properties regarding these definitions. 
Section~\ref{sec:conjugation} gives an algorithm for the case of square matrix tuple with all the matrices full rank. 
Section~\ref{sec:direct_sum_decomposition} gives an algorithm for the case that the matrix tuple is associated with a direct sum row tuple decomposition. 
Section~\ref{sec:quotient_matrix_tuple_main} presents an algorithm for the case of a general hierarchical row tuple decomposition.  
Section~\ref{sec:overall_algorithm} gives the overall canonical form algorithm and proves Theorem~\ref{thm:main} and Corollary~\ref{cor:main}.

\section{Preliminary}\label{sec:prelim}

In this section, we introduce notations and the basic facts we use throughout the paper. 

\paragraph{Vector spaces.} 
Let $V$ be a vector space. We use $U\leq V$ to denote that $U$ is a subspace of $V$, and $U < V$ to denote that $U$ is a subspace of $V$ and $U$ is not equal to $V$. For $S\subseteq V$, we use $\langle S\rangle$ to denote the linear space spanned by the elements in $S$.

We use 
$V=U_1\oplus \dots \oplus U_k$ to denote that $V$ is a direct sum of subspaces $U_1, \dots, U_k$. It should be noted that in this paper when we talk about a direct sum decomposition, we mean an \emph{ordered} tuple of vector spaces $(U_1, \dots, U_k)$ such that $V=U_1\oplus \dots\oplus U_k$.


\paragraph{Finite fields.} We use $\F_q$ to denote the finite field of order $q$. We will reserve $q$ as the field order and $p$ as the field characteristic. 

We use the following basic facts for finite fields.

   \begin{fact}\label{fact:eigenval}
        Let $f$ be an irreducible polynomial over $\F_q$ with degree $d$. All the roots of $f$ are over $\F_{q^d}$, and are distinct. Furthermore, if $\lambda$ is a root of $f$, then $\lambda^q, \lambda^{q^2}, \dots, \lambda^{q^{d-1}}$ are all the other roots. 
    \end{fact}

We need to factor an irreducible polynomial over $\F_q$ of degree $d$ in $\F_{q^d}$. In our algorithms, we use a representation of $\F_{q^d}$ by fixing an irreducible polynomial of degree $d$ over $\F_q$. 


    \begin{theorem}[\hspace{-.02cm}\cite{Ber67,CZ81}]\label{thm:factoring}
    Given a degree $d$ irreducible polynomial $f(x) \in \F_q[x]$, there is a randomized Las-Vegas algorithm with $\poly(d, \log q)$ running time to compute the roots of $f(x)$ over an extension field $\F_{q^d}$. 
\end{theorem}

\paragraph{Matrices and row vectors.} Let $\M(n \times m, \F)$ be the linear space of $n \times m$ matrices over a field $\F$. We use $\M(n, \F)$ to denote $\M(n\times n, \F)$ and $\GL(n, \F)$ to denote the general linear group of degree $n$ over $\F$. 
For a matrix $A\in\M(n, \F)$, the characteristic polynomial of $A$ is denoted by $\charpoly(A)$. 

For $i\in[d]$, let $D_i\in \M(n_i\times m_i, \F)$. Let $n=\sum_{i\in[d]}n_i$ and $m=\sum_{i\in[d]} m_i$. Then we use $\diag(D_1, \dots, D_d)$ to denote the block-diagonal $n\times m$ matrix $$\begin{bmatrix}
    D_1 & 0 & \dots & 0 \\
0 & D_2 & \dots & 0 \\
\vdots & \vdots & \ddots & \vdots \\
0 & 0 & \dots & D_d
\end{bmatrix}.$$

For a field $\F$, $\F^n$ is the linear space of length-$n$ \emph{row} vectors over $\F$. 
We use $e_i$ for some integer $i$ to denote the row vector whose elements are all zero except the $i$-th element equals 1. 

For a vector $v$ in a row vector space $V$ and $U_1, \dots, U_k$ such that $V = U_1 \oplus \dots \oplus U_k$, 
we use $\rowvecproj_{U_1, \dots, U_k}(v, U_i)$ to denote $v_i$, where $v = v_1 + \dots + v_k$ such that $v_i \in U_i$ for any $1 \leq i \leq k$. 
For convenience, we use $\rowvecproj(v, U_i)$ to denote $\rowvecproj_{U_1, \dots, U_k}(v, U_i)$ when there is no confusion. 

For a matrix $A \in \M(n \times m, \F)$ and a row vector $v \in \F^m$, we say $v$ is a row vector of $A$ if there is a $u \in \F^n$ such that $v = u A$.  
We use $\rowvectorspace(A)$ to denote the linear space spanned by all the row vectors of $A$, i.e., $\rowvectorspace(A) = \{v \in \F^m : \exists u \in \F^n \text{ s.t. } v = uA\}$.

For a matrix $A \in \M(n \times m, \F)$ and a row vector space $U \leq \F^n$, 
we use $U A$ to denote $\langle \{u A : u \in U\}\rangle$.

\paragraph{Matrix tuples and row tuples.} A matrix tuple $\vA=(A_1, \dots, A_\ell)\in\M(n\times m, \F)^\ell$ is a vector of matrices with the same size over the same field. 
Let $L$ and $R$ be matrices in $\M(n, \F)$ and $\M(m, \F)$ respectively, we denote $L \vA = (LA_1, \dots, L A_\ell)$ and $\vA R = (A_1 R, \dots, A_\ell R)$. 
A matrix tuple $\vC \in \M(s \times m, \F)^\ell$ for some integer $s$ is a row-submatrix of $\vA$ if there exists a $s \times n$ matrix $L$ such that $\vC = L \vA$. 

We use $\mtupletomatrix(\vA)$ to denote the $(n\cdot \ell) \times m$ matrix such that for any $1 \leq i \leq n, 1 \leq j \leq \ell$, the $((j - 1)\cdot n + i)$-th row of $\mtupletomatrix(\vA)$ equals the $i$-th row of $A_j$. In other words, \[\mtupletomatrix(\vA)=\begin{bmatrix}
    A_1 \\ A_2 \\ \vdots \\ A_\alpha
\end{bmatrix}.\]


Recall that $\F^n$ is the linear space of length-$n$ row vectors. For $u\in \F^n$, we call the length $\ell$ vector of row vectors $u \cdot \vA = (u A_1, \dots, u A_\ell)\in (\F^m)^\ell$ as a row tuple vector of $\vA$.
We use $\va, \vb,$ et al. to denote row tuples, i.e., $\va\in (\F^m)^\ell$ is a row tuple of $\vA$ if and only if there exists $u \in \F^n$ such that $\va = u \cdot \vA$. 

For a row tuple $\va$ of $\vA$ with $\va = u \vA$, we use $\va^{(i)}$ to denote the $i$-th coordinate of $\va$ for any $1 \leq i \leq \ell$, i.e., $\va^{(i)} = u A_i$.

\paragraph{Vector spaces associated with matrix tuples.} 
For a matrix tuple $\vA$, we use $\rowtuplespace(\vA)$ to denote the linear space spanned by all the row tuples of $\vA$, or the row tuple space of $\vA$, i.e., \[\rowtuplespace(\vA) = \{\va \in (\F^m)^\ell : \exists u \in \F^n \text{ s.t. } v = u \vA\}.\]
We also use $\rowvectorspace(\vA)$ to denote the linear space spanned by the row vectors for all the matrices of $\vA$, i.e., \[\rowvectorspace(\vA) = \langle \rowvectorspace(A_1) \cup \dots \cup \rowvectorspace(A_\ell)\rangle.\]

Let $V\leq \rowtuplespace(\vA)$ be a subspace of the row tuple space of $\vA$. We use $\rowtupvs{V}{i}$ to denote the linear space spanned by $i$-th coordinate of row tuples in $V$, i.e., $\rowtupvs{V}{i} = \{a_i : \va =  (a_1, \dots, a_\ell) \in V\}$.
We use $\rowvectorspace(V)$ to denote $\langle \rowtupvs{V}{1} \cup \dots \cup \rowtupvs{V}{\ell}\rangle$.


\paragraph{Multiplication with column vectors. }
Let $(v_1, \dots, v_m)$ be a vector such that linear combination of $v_i$ over $\F^m$ is well defined for every $b \in \F$ and all the $i \in [m]$. 
For a vector $u = (u_1, \dots, u_m) \in \F^m$, we use $u \cdot (v_1, \dots, v_m)^T$ to denote $u_1 \cdot v_1 + \dots + u_m \cdot v_m$. 
For a row tuple $\va = (a_1, \dots, a_\ell)$ with $a_i \in (\F^m)^\ell$, we use $\va \cdot (v_1, \dots, v_m)^T$ to denote $(a_1 \cdot (v_1, \dots, v_m)^T, \dots, a_\ell \cdot (v_1, \dots, v_m)^T)$. 
For a matrix tuple $\vA \in \M(n\times m, \vF)^\ell$ represented as
\[\vA = \begin{bmatrix}
    \va_1 \\ \vdots \\  \va_n
\end{bmatrix},\]
we use $\vA (v_1, \dots, v_m)^T$ to denote 
\[\begin{bmatrix}
    \va_1 \cdot (v_1, \dots, v_m)^T \\ \vdots \\ \va_n \cdot (v_1, \dots, v_m)^T
\end{bmatrix}.\]

\paragraph{Quotient Space.}    Let $V$ be a row vector space and $W$ be a subspace of $V$. 
    For each vector $v \in V$, the coset of $v$ with respect to $W$ is $v + W := \{v + w : w \in W\}$. 
    The quotient space of $V$ by $W$, denoted as $V / W$, is $\{v + W : v \in V\}$. 
    Let $\va = (a_1, \dots, a_\ell)$ be a row tuple with $a_i \in V$ for any $i$. 
    We define \[\va + W:= (a_1 + W, \dots, a_\ell + W).\]

    Let $T$ be a linear space of row tuples over vector space $V$. 
    We use $T / W$ to denote $\langle \{\va + W : \va \in T\} \rangle$. 
    Let $$\vA = \begin{bmatrix}
    \va_1 \\  \vdots \\ \va_\sigma
    \end{bmatrix}$$ be a matrix tuple such that $\va_i \in V^\ell$ for any $i \in [\sigma]$.
    We use  $\vA + W$ to denote 
    $$\begin{bmatrix}
    \va_1 + W \\ \vdots \\ \va_\sigma + W
\end{bmatrix}.$$    
    Let $Z$ be a linear space of matrix tuples such that for each $\vA \in Z$, $\rowtuplespace(\vA) \leq V^{\ell}$. 
    We use $Z / W$ to denote $\langle \{\vA  + W : \vA \in Z\} \rangle$.

\paragraph{Group actions on matrix tuples.} Let $\vA=(A_1, \dots, A_\ell)\in\M(n \times m, \F)^\ell$ be a matrix tuple. The left-right action of $(L, R)\in \GL(n, \F)\times \GL(m, \F)$ sends $\vA$ to $L\vA R^{-1}:=(LA_1R^{-1}, \dots, LA_\ell R^{-1})$. 
We say that $\vA, \vB\in\M(n \times m, \F)^\ell$ are \emph{equivalent}, if they are in the same orbit under this group action. 

Let $\vA=(A_1, \dots, A_\ell)\in\M(n, \F)^\ell$ be a square matrix tuple. The conjugation action of $T\in\GL(n, \F)$ sends $\vA$ to $T\vA T^{-1}$. We say that $\vA, \vB\in\M(n, \F)^\ell$ are \emph{conjugate}, if they are in the same orbit under this group action.

\paragraph{Matrix algebras.} Let $\F$ be a field. A matrix algebra $\cA$ over $\F$ is a linear space of matrices in $\M(n, \F)$ that is closed under matrix additions and multiplications, that is, for any $A, B\in \cA$, $AB\in \cA$. Given $\vA\in\M(n, \F)^m$, the matrix algebra generated by $\vA$, denoted by $\Env(\vA)$, is the smallest matrix algebra that contains every $A\in \vA$. 
\begin{proposition}[\hspace{-.02cm}{\cite{IR99}}]\label{prop:linear_basis}
    Given $\vA\in\M(n, \F_q)^m$, there is a Linear Basis Algorithm to compute a linear basis $(L_1, \dots, L_\ell)$ of $\Env(\vA)$ in time $\poly(n, m, \log q)$ satisfying the following condition: For two inputs $\vA$ and $\vA'$ that have an invertible matrix $P$ satisfying $\vA' = P \vA P^{-1}$, 
    then for any invertible matrix $P$ such that $\vA' = P \vA P^{-1}$, 
    $(L_1', \dots, L_\ell') = P (L_1, \dots, L_\ell)P^{-1}$, where $(L_1, \dots, L_\ell)$ and $(L_1', \dots, L_\ell')$ denote the outputs for $\vA$ and $\vA'$ respectively. 
\end{proposition}

\paragraph{Associative algebras.} 
The structure of finite-dimensional associative algebras over fields is a classical topic known as the Artin--Wedderburn--Mal'tsev theory. 

Given a matrix algebra $\cA$ over a field $\F$, an ideal of $\cA$ is a linear subspace of $\cA$ that is closed under left and right multiplications of elements from $\cA$. An algebra is simple if it has no non-trivial ideals. 

The \emph{Jacobson radical} $\rad(\cA)$ is the largest nilpotent ideal of $\cA$. An algebra $\cA$ is called semisimple, if $\rad(\cA)=0$. A semisimple $\cA$ is a direct sum of its minimal ideals $\cA=\cA_1\oplus \dots\oplus \cA_k$, and each $\cA_i$ is a simple algebra, called a simple component of $\cA$. By Wedderburn's theorem, a simple algebra $\cA$ is isomorphic to $\M(n, D)$ where $D$ is a division algebra over $\F$. When $\F=\F_q$ is a finite field, Wedderburn's little theorem shows that every division algebra over $\F$ is a field, so every simple algebra $\cA$ over a finite field $\F_q$ is isomorphic to $\M(n, \F_r)$ where $r=q^k$. 

For an algebra $\cA$, the quotient algebra $\cA/\rad(\cA)$ is semisimple, so it is a direct sum of simple algebras. In the case of matrix algebras over finite fields, there exists a subalgebra of $\cA$, $\wc(\cA)$, such that $\wc(\cA)\cap \rad(\cA)=0$ and $\wc(\cA)+\rad(\cA)=\cA$, and $\wc(\cA)\cong \cA/\rad(\cA)$. Such $\wc(\cA)$ is called a \emph{Wedderburn complement} of $\cA$, and it captures the semisimple structure of $\cA$. All Wedderburn complements are conjugate from one to another by units of the algebra.

\paragraph{Computing algebra structures over finite fields.} In several important works, Friedl, R\'onyai, and Ivanyos demonstrated how to compute the structure of a given associative algebra over several important fields, such as finite fields and algebraic number fields \cite{FR85,Iva00,Ron90}. 
\begin{theorem}[\hspace{-.02cm}\cite{FR85,Iva00,Ron90}]\label{thm:algebra_struct}
    Let $\cA\subseteq\M(n, \F_q)$ be a matrix algebra with linear basis $(A_1, \dots, A_\ell)$. There is a Las-Vegas randomized Matrix Algebra Structure Algorithm in $\poly(n, \log q)$ running time that computes one of the following outputs: \begin{enumerate}
    \item ($\cA$ is not semisimple) the non-trivial Jacobson radical $\rad(\cA)$,
    \item ($\cA$ is semisimple but not simple) a direct sum decomposition of $\cA=C_1\oplus \dots \oplus C_s$ where $C_i\subseteq \cA$ are the simple components, represented by linear bases, 
    \item ($\cA$ is simple) an isomorphism $f$ from $\cA$ to $\M(m, \F_{r})$ given by $f(A_1), \dots, f(A_\ell)$ for some $m = O(n), r = q^{O(n)}$ such that $\F_{r}$ is an extension field of $\F_q$.   
    \end{enumerate}
    such that for any two matrix algebras $\cA$ with linear basis $(A_1, \dots, A_\ell)$ and $\cA'$ with linear basis $(A_1', \dots, A_\ell')$ satisfying that there exists an invertible matrix $P$ such that \begin{equation}\label{equ:prelim}(A_1', \dots, A_\ell') = P (A_1, \dots, A_\ell) P^{-1},\end{equation} then for any invertible matrix $P$  such that Equation (\ref{equ:prelim}) holds, the outputs for $\cA$ and $\cA'$ satisfying the following conditions:
    \begin{enumerate}
        \item If the output for $\cA$ is $\rad(\cA)$, then the output for $\cA'$ is $\rad(\cA')$ such that $\rad(\cA') = P \cdot \rad(\cA) \cdot  P^{-1}$. 
        \item If the output for $\cA$ is a decomposition of $\wc(\cA)$ into more than one simple component, then the output for $\cA'$ is a decomposition of $\wc(\cA')$ with the same number of simple components. 
        \item If the output for $\cA$ the an isomorphism $f$ from $\cA$ to $\M(m, \F_{r})$, then the output for $\cA'$ is another isomorphism $f'$ from $\cA'$ to $\M(m, \F_{r})$ such that there is an invertible matrix $P' \in \GL(m, \F_{r})$ such that $f'(A_i') = P' f(A_i) (P')^{-1}$. 
    \end{enumerate}

\end{theorem}

\section{Indecomposable-block-corresponding row-submatrix tuple and hierarchical row tuple decomposition}\label{sec:ibctuple}

In this section, we introduce a few important definitions that we will use in the later sections and prove some basic properties of these definitions. 

In particular, this section defines the indecomposable-block-corresponding row-submatrix tuple and the hierarchical row tuple decomposition. 
Based on these definitions, in Section~\ref{sec:canonical_form_from_rep}, 
we provide an algorithm to construct a canonical form of the matrix tuple, assuming that indecomposable-block-corresponding row-submatrix tuples satisfying certain properties are given.

\subsection{Indecomposable-block-corresponding row-submatrix tuple}

We define the key definition indecomposable-block-corresponding row-submatrix tuple in this subsection. 

We start by extend the block-diagonal matrix to matrix tuple. 
As an analog of the block-diagonal matrix, 
for a sequence of $d$ non-zero matrix tuples $\vD_i\in \M(\nu_i\times \mu_i, \F)^\ell$, 
we use $\diag(\vD_1, \dots, \vD_d)$ to denote the block-diagonal $n\times m$ matrix tuple in $\M(n \times m, \F)^\ell$ $$\begin{bmatrix}
    \vD_1 & 0 & \dots & 0 \\
0 & \vD_2 & \dots & 0 \\
\vdots & \vdots & \ddots & \vdots \\
0 & 0 & \dots & \vD_d
\end{bmatrix},$$
where $n=\sum_{i\in[d]}\nu_i$ and $m=\sum_{i\in[d]} \mu_i$.
We call each of $\vD_1, \dots, \vD_d$ a block of $\vD$.

Let $\vA\in\M(n\times m, \F)^\ell$ be a matrix tuple. We say that $\vA$ is \emph{block-diagonalizable}, if it is equivalent to $\vD=\diag(\vD_1, \dots, \vD_d)\in \M(n\times m, \F)^\ell$ for some $d\geq 2$ and non-zero matrix tuples $\vD_1, \dots, \vD_d$, i.e., there exist $L \in \GL(n, \F)$ and $R\in \GL(m, \F)$ such that $\vD = L \vA R^{-1}$. 
Such a matrix tuple $\vD$ is referred to as a block-diagonalization of $\vA$.

We say that $\vA$ is \emph{indecomposable}, if it is not block-diagonalizable. 
We say that the block-diagonalization $\vD$ is minimal if every $\vD_i$ is indecomposable. 
We say that $\vD = \diag(\vD_1, \dots, \vD_d)$ is a minimum block-diagonalization of $\vA$ if there is no $L'$ and $R'$ such that $L' \vA (R')^{-1} = \diag(\vD'_1, \dots, \vD'_{d'})$ for some $d' > d$. 




The following is a consequence of the Krull--Schmidt theorem for quiver representations, by viewing a matrix tuple as a representation of the Kronecker quiver. 
\begin{lemma}[\hspace{-.02cm}{\cite[Theorem 1.11]{Kir16}}]
For a matrix tuple $\vA$, we have
\begin{enumerate}
\item Any minimal block-diagonalization of $\vA$ is a minimum block-diagonalization. 
\item Let $\vD = \diag(\vD_1, \dots, \vD_d)$ and $\vD' = \diag(\vD_1', \dots, \vD_d')$ be two minimum block-diagonalizations of $\vA$. 
Then there exists a bijection $f: [d] \rightarrow [d]$ such that for any $1 \leq i \leq d$, there exist invertible matrices $L_i$ and $R_i$ such that $\vD'_{f(i)} = L_i \vD_{i} (R_i)^{-1}$. 
\end{enumerate}
\end{lemma}

We investigate the relations between row-submatrix tuples of $\vA$ and blocks of a block signalization of $\vA$. 
We say a matrix tuple $\vC$ is a row-submatrix tuple of $\vA$ if every row tuple of $\vC$ is a row tuple of $\vA$. 
Let $L$ and $R$ be two invertible matrices such that $\vD = \diag(\vD_1, \dots, \vD_d) = L \vA R^{-1}$ with $\nu_j$ being the number of rows of $\vD_j$ for each $j\in[d]$. 
Denote 
\[\vA_{i, \vD, L} := \begin{bmatrix}
    e_{k + 1} L \vA \\ 
    \vdots \\
    e_{k + \nu_i} L \vA
\end{bmatrix},
\]
where $k = (\sum_{j = 1}^{i-1} \nu_j)$. 
Intuitively, $\vA_{i, \vD, L}$ is the row-submatrix tuple of $\vA$ such that $\vA_{i, \vD, L} \cdot R^{-1}$ corresponds to the block $\vD_i$ in $\vD$. 

By the correspondence between $\vA$ and its block-diagonalization $\vD$, the row tuple space of $\vA$ and the row vector space of $\vA$ are direct sum of row tuple spaces and row vector spaces of $\vA_{1, \vD, L}, \dots, \vA_{d, \vD, L}$ respectively.

\begin{fact}\label{fact:ibctuple_diagonalzation_most_basic}
Let $\vA$ be a matrix tuple in $\M(n\times m, \F_q)^\ell$. The following properties hold:
\begin{enumerate}
\item Let $\vD = \diag(\vD_1, \dots, \vD_d)$ be a block-diagonalization of $\vA$. Then for any invertible matrices $L$ and $R$ such that $\vD = L \vA R^{-1}$, we have 
\begin{equation}\label{equ:ibctuple_diag_basic_01}\rowtuplespace(\vA) = \rowtuplespace(\vA_{1, \vD, L}) \oplus \dots \oplus \rowtuplespace(\vA_{d, \vD, L}),\end{equation}
and \begin{equation}\label{equ:ibctuple_diag_basic_02}\rowvectorspace(\vA) = \rowvectorspace(\vA_{1, \vD, L}) \oplus \dots \oplus \rowvectorspace(\vA_{d, \vD, L}).\end{equation}
\item If there is a sequence of matrix tuples $\vA_1 \in \M(\sigma_1 \times m, \F)^\ell, \dots, \vA_k \in \M(\sigma_k \times m, \F)^\ell$ for some integers $\sigma_1, \dots, \sigma_k$ satisfying the following conditions:
\begin{enumerate}
    \item $n = \sigma_1 + \dots + \sigma_k$;
    \item $\rowtuplespace(\vA) = \rowtuplespace(\vA_1) \oplus \dots \oplus \rowtuplespace(\vA_k)$;
    \item $\rowvectorspace(\vA) = \rowvectorspace(\vA_1) \oplus \dots \oplus \rowvectorspace(\vA_k)$,
\end{enumerate}
then 
there is an invertible matrix $R$ such that 
\begin{equation}\label{equ:ibctuple_diag_basic_03}\begin{bmatrix}
    \vA_1 \\ \vdots \\ \vA_k
\end{bmatrix} R^{-1} = \diag(\vD_1, \dots, \vD_k)\end{equation}
is a block-diagonalization of $\vA$ with $\vA_{i, \vD, I_n} = \vA_i$ for any $1 \leq i \leq k$. 
\end{enumerate}

\end{fact}
\begin{proof}
We first prove the first property. Suppose $\vD_i$ is a $\nu_i \times \mu_i$ dimension matrix tuple for each $i \in [d]$. 
By the definition of block-diagonalization, we have 
$n = \sum_{i = 1}^d \nu_i$, $m = \sum_{i = 1}^d \mu_i$, 
\[\rowtuplespace(\vD) = \bigoplus_{i = 1}^d \rowtuplespace\left(\left \langle \left\{e_{(\sum_{j = 1}^{i - 1} \nu_j) + 1} \vD, \dots, e_{(\sum_{j = 1}^{i} \nu_j)} \vD\right\}\right\rangle\right),\]
and \[\rowvectorspace(\vD) = \bigoplus_{i = 1}^d \rowvectorspace\left(\left \langle\left\{e_{(\sum_{j = 1}^{i - 1} \nu_j) + 1} \vD, \dots, e_{(\sum_{j = 1}^{i} \nu_j)} \vD\right\}\right\rangle\right).\]
Since $\vD = L \vA R^{-1}$, 
by the definition of $\vA_{i, \vD, L}$, Equation (\ref{equ:ibctuple_diag_basic_01}) and Equation (\ref{equ:ibctuple_diag_basic_02}) hold.

For the second property, 
let $v_{i, 1}, \dots, v_{i, \dim(\rowvectorspace(\vA_i))}$ be an arbitrary linear basis of $\rowvectorspace(\vA_i)$. 
Let 
\[R = \begin{bmatrix}
    v_{1, 1} \\ \vdots \\ v_{1, \dim(\rowvectorspace(\vA_1))} \\ v_{2,1} \\ \vdots \\  v_{k, \dim(\rowvectorspace(\vA_k))}
\end{bmatrix}.\]
$R$ is an invertible matrix such that Equation (\ref{equ:ibctuple_diag_basic_03}) is satisfied. 
On the other hand, by the condition (b), there is an invertible matrix $L$ such that $L \vA = \begin{bmatrix}
    \vA_1 \\ \vdots \\ \vA_k
\end{bmatrix}$. Then the second property holds. 
\end{proof}

Now we are ready to define indecomposable-block-corresponding row-submatrix tuple. Roughly speaking, a indecomposable-block-corresponding row-submatrix tuple is a row-submatrix tuple that corresponds to a indecomposable block of a block diagonalization of the matrix tuple. 

\begin{definition}[Indecomposable-Block-Corresponding Row-Submatrix Tuple]\label{def:ibctuple}
    For a matrix tuple $\vA \in \M(n\times m, \F_q)^\ell$, an \emph{indecomposable-block-corresponding row-submatrix tuple} (or \emph{\IBCtuplenospace} for short) $\vB$ of $\vA$ is a row-submatrix tuple of $\vA$ satisfying the following condition: 
    There exists a minimum block-diagonalization $\vD = \diag(\vD_1, \dots, \vD_d)$ of $\vA$, invertible matrices $L \in \GL(n, \F_q)$ and $R \in \GL(m, \F_q)$ such that $\vD = L \vA R^{-1}$, and an integer $i \in [d]$ such that $\vB = \vA_{i, \vD, L}$. 
    
\end{definition}

We define the relations between \IBCtuples in the following definition.

\begin{definition}\label{def:ibctuple_rel}
    The relations between \IBCtuples are defined as follows:
    \begin{enumerate}
    \item
    Two \IBCtuples $\vB$ and  $\vB'$ (maybe for different matrix tuples) are \emph{equivalent} if they contain the same number of rows and columns, and there are invertible matrices $L$ and $R$ such that $\vB = L \vB' R^{-1}$.
    \item
    Two \IBCtuples $\vB$ and $\vB'$ (maybe for different matrix tuples) are \emph{right-equivalent} if they contain the same number of rows and columns, and there is an invertible matrix $R$ such that $\vB = \vB' R^{-1}$.
    \item
    Two \IBCtuples $\vB$ and $\vB'$ for the same matrix tuple are \emph{strongly correlated} if they are right-equivalent and $\vB - \vB'$ is not an \IBCtuplenospace. 
    \item
    Two \IBCtuples $\vB$ and $\vB'$ for the same matrix tuple are \emph{correlated} if they are right-equivalent and there are invertible matrices $L$ and $L'$ such that 
    $L\vB$ is strongly correlated to $\vB'$, and $L' \vB'$ is strongly correlated to $\vB$. 
    \end{enumerate}
\end{definition}

We also define the representative \IBCtuple sequence for a matrix tuple as follows. 

\begin{definition}[Representative \IBCtuple Sequence]\label{def:representative_ibctuple_sequence}
A sequence of \IBCtuples $\vB_1, \dots, \vB_k$ of matrix tuple $\vA$ is a \emph{representative relevant \IBCtuple sequence} of $\vA$ if the following two conditions hold:
\begin{enumerate}
    \item Any two different \IBCtuples in the sequence are not equivalent. 
    \item Any \IBCtuple of $\vA$ is equivalent to one of $\vB_1, \dots, \vB_k$. 
\end{enumerate}
\end{definition}

We remark that computing a desirable canonical representative \IBCtuple sequence is nontrivial because the algorithm needs to select the \IBCtuples that result in equivalent blocks in a minimum block-diagonalization of the input matrix tuple by fixing the invertible matrix multiplied on the left canonically, ensuring that the \IBCtuples selected for different input matrix tuples are right-equivalent.

An important goal of our algorithm for computing a canonical form of a given matrix tuple is to compute a ``canonical'' representative \IBCtuple sequence in the following sense: for two equivalent matrix tuples $\vA$ and $\vA'$ (i.e., there exist invertible matrices $L$ and $R$ such that $\vA' =  L \vA R^{-1}$), the representative \IBCtuple sequence $\vB_1, \dots, \vB_k$ computed for $\vA$ and $\vB_1', \dots, \vB_k'$ computed for $\vA'$ satisfy that $\vB_i$ and $\vB_i'$ are right-equivalent for all the $i\in[k]$.

We remark that computing a desirable canonical representative \IBCtuple sequence is nontrivial because the algorithm needs to choose the \IBCtuples that result in equivalent blocks in a minimum block-diagonalization of the input matrix tuple by fixing the invertible matrix multiplied on the left canonically, ensuring that the \IBCtuples selected for different input matrix tuples are right-equivalent.


\subsection{Hierarchical row tuple decomposition}

In this subsection, we introduce the hierarchical row tuple decomposition, which contains a sequence of row tuple subspaces that are both characteristic and block-compatible.
The characteristic property of a row tuple space or a row vector space has been briefly discussed in Section~\ref{sec:overview}. We provide a formal definition here. We also define the block-compatible property, which is an important definition to analyze block-diagonal matrix tuples.

\begin{definition}[Characteristic Subspace]\label{def:characteristic_subspace}
For a matrix tuple $\vA$, a subspace of row tuples $T \leq \rowtuplespace(\vA)$ is characteristic if for any invertible matrices $L$ and $R$ such that $\vA = L \vA R^{-1}$, 
\[T R^{-1} = \langle \{\va R^{-1} : \va \in T \} \rangle = T.\]

A subspace of row vectors $W \leq \rowvectorspace(\vA)$ is characteristic if for any invertible matrices $L$ and $R$ such that $\vA = L \vA R^{-1}$, 
\[ W R^{-1}= \langle \{v R^{-1} : v \in W \} \rangle = W.\]
\end{definition}

In addition to the characteristic property, we also require the row tuple subspaces in the hierarchical row tuple decomposition to satisfy the block-compatible property as defined in Definition~\ref{def:ibctuple_invariant_subspace}. This property enables us to decompose a row tuple subspace into the direct sum of its projections on blocks of any block-diagonalization of the matrix tuple.

\begin{definition}[Projection of Row-Submatrix Tuple on Blocks]\label{def:projection_matrix_tuple_diag}
For matrix tuple $\vA \in \M(n\times m, \F_q)^\ell$, and a sequence of row tuple subspaces $V_1, \dots, V_\ell \leq \rowtuplespace(\vA)$ such that $\rowtuplespace(\vA) = V_1 \oplus \dots \oplus V_\ell$. 
For an arbitrary row tuple $\va \in \rowtuplespace(\vA)$, let $\va_1, \dots, \va_\ell$ be the row tuples such that $\va = \sum_{i = 1}^\ell \va_i$, and $\va_i \in V_i$ for any $1 \leq i \leq \ell$. 
$\va_i$ is called the projection of $\va$ on $V_i$ with respect to $V_1, \dots, V_\ell$, denoted as $\proj_{V_1, \dots, V_\ell}(\va, i)$. 

For a block-diagonalization $\vD = \diag(\vD_1, \dots, \vD_d)$ of $\vA$ with invertible matrices $L$ and $R$ such that $\vD = L \vA R^{-1}$, 
the projection of row tuple $\va$ on $i$-th block of $\vD$, denoted as $\proj_{\vD, L}(\va, i)$, is 
\[\proj_{\vD, L}(\va, i) = \proj_{\rowtuplespace(\vA_{1, \vD, L}), \dots, \rowtuplespace(\vA_{d, \vD, L})}(\va, i).\]

Let $\vC$ be a $c \times m$ row-submatrix tuple of $\vA$, i.e., $\rowtuplespace(\vC) \leq \rowtuplespace(\vA)$.  
The projection of $\vC$ on $i$-th block of $\vD$, denoted as $\proj_{\vD, L}(\vC, i)$, is a $c\times m$ row-submatrix tuple of $\vA$ such that
\[\proj_{\vD, L}(\vC, i) := \begin{bmatrix}
    \proj_{\vD, L}(e_1 \vC, i) \\
    \vdots \\
    \proj_{\vD, L}(e_c \vC, i) 
\end{bmatrix}.\]
\end{definition}

The block-compatible property for a row tuple or row vector subspace is that the projection of each row tuple or row vector in the subspace on any block of an arbitrary block-diagonalization of the matrix tuple remains within the subspace.

\begin{definition}[Block-Compatible Subspace]\label{def:ibctuple_invariant_subspace}
For a matrix tuple $\vA$, a subspace of row tuples $T \leq \rowtuplespace(\vA)$ is block-compatible if for any block-diagonalization $\vD$ of $\vA$ and any invertible matrices $L$ and $R$ such that $\vD = L \vA R^{-1}$, 
$\proj_{\vD, L}(\va, i)$ is in $T$ for any $\va \in T$ and $i \in [d]$.

A subspace of row vectors $W \leq \rowvectorspace(\vA)$ is block-compatible if for any block-diagonalization $\vD$ of $\vA$ and any invertible matrices $L$ and $R$ such that $\vD = L \vA R^{-1}$, 
\[\rowvecproj_{\rowvectorspace(\vA_{1, \vD, L}), \dots, \rowvectorspace(\vA_{d, \vD, L})}(v, \rowvectorspace(\vA_{i, \vD, L}))\] is in $W$ for any $v \in W$ and $i \in [d]$.
\end{definition}

We remark that the characteristic property and block-compatible property are not equivalent. For an indecomposable matrix tuple, every row tuple subspace and every row vector subspace is block-compatible by definition, but it may not be characteristic.

\begin{fact}\label{fact:ibctuple_invariant_fund}
Let $\vA$ be a matrix tuple, $\vD = \diag(\vD_1, \dots, \vD_d)$ be a block-diagonalization of $\vA$, and $L, R$ be arbitrary invertible matrices such that $\vD = L \vA R^{-1}$.
Then for any block-compatible row tuple subspace $T \leq \rowtuplespace(\vA)$, 
    \[ T = (T \cap \rowtuplespace(\vA_{1, \vD, L})) \oplus \dots \oplus (T \cap \rowtuplespace(\vA_{d, \vD, L})),\]
    and for any block-compatible row vector subspace $W \leq \rowvectorspace(\vA)$, 
    \[W = (W \cap \rowvectorspace(\vA_{1, \vD, L})) \oplus \dots \oplus (W \cap \rowvectorspace(\vA_{d, \vD, L})).\]
\end{fact}
\begin{proof}
    Let $\va$ be an arbitrary row tuple in $T$, we have 
    \[\va = \sum_{i = 1}^d \proj_{\vD, L}(\va, i) \]
    by Definition~\ref{def:projection_matrix_tuple_diag}. Since $\proj_{\vD, L}(\va, i) \in T$ and $\proj_{\vD, L}(\va, i) \in \rowtuplespace(\vA_{i, \vD, L})$ for every $i \in [d]$,
    $\proj_{\vD, L}(\va, i) \in T \cap \rowtuplespace(\vA_{i, \vD, L})$ for any $i\in[d]$.
    Since $\rowtuplespace(\vA)$ is the direct sum of $\rowtuplespace(\vA_{1, \vD, L}), \dots, \rowtuplespace(\vA_{d, \vD, L})$, 
    $T$ is the direct sum of $T \cap \rowtuplespace(\vA_{i, \vD, L})$ for all the $i\in[d]$. 
    The proof for $W$ is similar. 
\end{proof}

In the following lemma, we provide several methods for constructing row tuple and row vector subspaces that are both characteristic and block-compatible. 
These methods serve as the basic operations for us to construct and refine the hierarchical row tuple decomposition throughout the rest of this paper.

\begin{fact}\label{fact:ibctuple_invariant_basic} For a matrix tuple $\vA \in \M(n\times m, \F_q)^\ell$, 
\begin{enumerate}
    \item Suppose $T_1$ and $T_2$ (or $W_1$ and $W_2$) are two characteristic block-compatible row tuple (or row vector) subspaces for $\vA$, then $T_1 \cap T_2$ and $\langle T_1 \cup T_2 \rangle$ (or $W_1 \cap W_2$ and $\langle W_1 \cup W_2 \rangle$) are characteristic block-compatible row tuple (or row vector) subspaces for $\vA$. 
    \item $\rowtuplespace(\vA)$ (or $\rowvectorspace(\vA)$) and the row tuple space that contains only the zero row tuple (or zero row vector) are characteristic block-compatible row tuple (or row vector) subspaces.
    \item For a characteristic block-compatible row tuple subspace $T$ and an arbitrary row vector $u \in \F_q^\ell$, $W = \langle \{u \cdot \mtupletomatrix(\va) : \va \in T\}\rangle $ is a characteristic block-compatible row vector space.  
    \item For a characteristic block-compatible row tuple subspace $T$ and a characteristic block-compatible row vector subspace $W$, $\langle \{\va \in T : \va^{(i)} \in W\} \rangle$ is a characteristic block-compatible row tuple subspace for any $i \in [\ell]$, where $\va^{(i)}$ denotes the $i$-th row vector of $\va$. 
    \item For a characteristic block-compatible row tuple subspace $T$, a sequence of characteristic block-compatible row vector subspaces $W_1, \dots, W_w$  such that $\rowvectorspace(W_1 \cup \dots \cup W_w) = W_1 \oplus \dots \oplus W_w$, for any $i \in[\ell]$ and $j \in [w]$, 
    \[W' = \langle \{\rowvecproj_{W_1, \dots, W_w}(\va^{(i)}, W_j) : \va \in T, \va^{(i)} \in \rowvectorspace(W_1 \cup \dots \cup W_w)\}\rangle\]
    is a characteristic block-compatible row vector subspace, and \[T' = \langle \{ \va \in T : \va^{(i)} \in \rowvectorspace(W_1 \cup \dots \cup W_w), \rowvecproj_{W_1, \dots, W_w}(\va^{(i)}, W_j) = 0\}\rangle \]
    is a characteristic block-compatible row tuple subspace.

    \item If $\vA \in \GL(n, \F_q)^\ell$, then for any matrix $A$ that is a polynomial of $A_1^{-1}A_2, \dots, A_1^{-1}A_\ell$, $\rowvectorspace(A)$ is a characteristic block-compatible row vector space 
\end{enumerate}
\end{fact}
\begin{proof}
Thorough this proof, let $\vD = \diag(\vD_1, \dots, \vD_d)$ be an arbitrary block-diagonalization of $\vA$, and $L, R$ be arbitrary invertible matrices such that $\vD = L \vA R^{-1}$. Let $L', R'$ be arbitrary invertible matrices such that $\vA = L' \vA (R')^{-1}$. 

For the first property, for any $\va \in T_1 \cap T_2$, by the definition of characteristic 
property, $\va (R')^{-1}$ and $\va R'$ are also in $T_1 \cap T_2$.  Hence, $T_1 \cap T_2$ is also characteristic. 
By the block-compatible property for $T_1$ and $T_2$, $\proj_{\vD, L}(\va, i)$ is in $T_1$ and also $T_2$ for any $i \in [d]$. Hence, $T_1 \cap T_2$ is block-compatible. 

For $\langle T_1 \cup T_2\rangle$, 
$\va_1 + \va_2$ is in $\langle T_1 \cup T_2\rangle$ for any $\va_1 \in T_1$ and $\va_2 \in T_2$.
On the other hand, 
any $\va \in T_1 \cup T_2$ has $\va_1 \in T_1$ and $\va_2 \in T_2$ such that $\va = \va_1 + \va_2$. 
Hence, we have $\va (R')^{-1} = \va_1 (R')^{-1} + \va_2 (R')^{-1} \in \langle T_1 \cup T_2\rangle$ and  $\va R' = \va_1 R' + \va_2 R' \in \langle T_1 \cup T_2\rangle$. Thus, $\langle T_1 \cup T_2\rangle$ is characteristic. 
For the block-compatible property, since $\proj_{\vD, L}(\va_1, i) \in T_1$ and $\proj_{\vD, L}(\va_2, i) \in T_2$, we have 
\begin{equation}\label{equ:ibctuple_diag_basic_04}\va = \sum_{i = 1}^d \left(\proj_{\vD, L}(\va_1, i) + \proj_{\vD, L}(\va_2, i)\right).\end{equation}
On the other hand, by Fact~\ref{fact:ibctuple_invariant_fund}, there are unique $\vc_i \in \rowtuplespace(\vA_{i, \vD, L})$ for each $i\in[d]$ such that $\va = \sum_{i = 1}^d \vc_i$, and $\proj_{\vD, L}(\va, i) = \vc_i$. 
By Equation (\ref{equ:ibctuple_diag_basic_04}), $\proj_{\vD, L}(\va, i) = \proj_{\vD, L}(\va_1, i) + \proj_{\vD, L}(\va_2, i)$ for any $i \in [d]$. Thus, $\langle T_1 \cup T_2\rangle$ is block-compatible. 
The case of row vector subspaces is similar. 

The second property is obtained by Definition~\ref{def:characteristic_subspace}, Definition~\ref{def:ibctuple_invariant_subspace}, and
Fact~\ref{fact:ibctuple_invariant_fund}.  

    For the third property, let $v$ be an arbitrary row vector in $W$. By the linearity, there exists a row tuple $\va_v \in T$ such that $u \cdot \mtupletomatrix(\va_v) = v$. 
    Since both $\va_v (R')^{-1}$ and $\va_v R'$ are in $T$,
    we have 
    \[v (R')^{-1} = u \cdot \mtupletomatrix(\va_v) (R')^{-1} = u \cdot  \mtupletomatrix(\va_v (R')^{-1}) \in W \]
    and     
    \[v R' = u \cdot \mtupletomatrix(\va_v) R' = u \cdot  \mtupletomatrix(\va_v R') \in W.\]
    Hence, $W$ is characteristic.

    By the definition of block-compatible row tuple subspace, $\va_v = \sum_{i=1}^d \proj_{\vD, L}(\va_v, i)$. 
    Since $\proj_{\vD, L}(\va_v, i)$  for every $i \in [d]$ is in $T$, $u \cdot \mtupletomatrix(\proj_{\vD, L}(\va, i))$ is in $W$ for every $i \in [d]$.
    Note that \[\rowvecproj_{\rowvectorspace(\vA_{1, \vD, L}), \dots, \rowvectorspace(\vA_{d, \vD, L})}(v, \rowvectorspace(\vA_{i, \vD, L})) = u \cdot \mtupletomatrix(\proj_{\vD, L}(\va, i)) \ in W\]
    for every $i \in [d]$. 
    Hence, $W$ is block-compatible, and thus 
    the third property holds. 

The fourth property is obtained in a way similar to the third property. 

    For the fifth property, we observe that \[T'' := \langle \{\va \in T : \va^{(i)} \in \rowvectorspace(W_1 \cup \dots \cup W_w)\}\rangle\] is a characteristic block-compatible row tuple subspace by the second and the sixth property. 
    Let $\va$ be an arbitrary vector in $T''$, and $a_k = \rowvecproj_{W_1, \dots, W_w}(\va^{(i)}, W_k)$ for any $k \in [w]$. 
    Since  
    \[\left(\va (R')^{-1}\right)^{(i)} = \va^{(i)} (R')^{-1} = \sum_{k = 1}^w \left(a_k (R')^{-1}\right),\]
    by the fact that $\rowvectorspace(W_1 \cup \dots \cup W_w) = W_1 \oplus \dots \oplus W_w$, and $W_1, \dots, W_w$ are characteristic, 
    we have 
    \[\rowvecproj_{W_1, \dots, W_w}((\va (R')^{-1})^{(i)}, W_j) = a_j (R')^{-1} = \rowvecproj_{W_1, \dots, W_w}(\va^{(i)}, W_j) (R')^{-1}.\]
    Hence, 
    \begin{align*}W' = & \langle \{\rowvecproj_{W_1, \dots, W_w}(\va^{(i)}, W_j) : \va \in T''\rangle \\ = & \langle \{\rowvecproj_{W_1, \dots, W_w}((\va (R')^{-1})^{(i)}, W_j) : \va \in T''\rangle  \\ = & W' R^{-1}\end{align*}
    and thus $W'$ is characteristic. 

    For the block-compatible property of $W'$, 
    let $v$ be an arbitrary vector in $W'$ and $\va_v$ be a row tuple in $T''$ such that $\rowvecproj_{W_1, \dots, W_w}(\va_v^{(i)}, W_j) = v$. 
    By the block-compatible property of $T''$, 
    $\proj_{\vD, L}(\va_v, d')$ is in $T''$ for any $d' \in [d]$. 
    By the block-compatible property of $W_j$ and the definition of $W'$, we have
    \[\rowvecproj_{W_1, \dots, W_w}((\proj_{\vD, L}(\va_v, d'))^{(i)}, W_j) \in W'.\]
    On the other hand, we have 
    \begin{align*}& \rowvectorspace_{W_1, \dots, W_w}(\rowvecproj_{\rowvectorspace(\vA_{1, \vD, L}), \dots, \vA_{d, \vD, L}}(v, \rowvectorspace(\vA_{d', \vD, L})), W_j) \\ =&  \rowvecproj_{W_1, \dots, W_w}((\proj_{\vD, L}(\va_v, d'))^{(i)}, W_j),\end{align*}
    and thus $W'$ is block-compatible. 
    The characteristic and block-compatible property for $T'$ is similar. 


    Now we prove the sixth property. Suppose $A = f(A_1^{-1}A_2, \dots, A_1^{-1}A_\ell)$ for a multivariate polynomial $f$. We first prove $\rowvectorspace(A)$ is characteristic. Recall that $L'$ and $R'$ are  two arbitrary invertible matrices such that $\vA = L' \vA (R')^{-1}$. 
    Thus, $A_i = L' A_i (R')^{-1}$ for any $i \in [\ell]$, and consequently, \[A_1^{-1} A_i = R' A_1^{-1} (L')^{-1} L' A_i (R')^{-1} = R' A_1^{-1} A_i (R')^{-1}\] for any $i \in \{2, \dots, \ell\}$. Hence, for any $C$ that is a monomial of $A_1^{-1}A_2, \dots, A_1^{-1}A_\ell$, $R' C (R')^{-1} = C$, and consequently, $R' A (R')^{-1} = A$. 
    Since $R'$ is an invertible matrix tuple, we have  \[\rowvectorspace(A) = \rowvectorspace(R' A(R')^{-1}) = \rowvectorspace(A R^{-1}).\] Thus, the $\rowvectorspace(A)$ is characteristic. 

    We prove $\rowvectorspace(A)$ is block-compatible. Recall that $L$ and $R$ are arbitrary invertible matrices such that $\vD = (D_1, \dots, D_\ell) = \diag(\vD_1, \dots, \vD_d) = L \vA R^{-1}$. 
    Suppose $\vD_i = (D_{i, 1}, \dots, D_{i, \ell})$. 
    Then $D_j = \diag(D_{1, j}, \dots, D_{d, j})$. 
    Since $A_1, \dots, A_\ell$ are all full rank, $\vD_1, \dots, \vD_d$ are all square blocks, and $D_{i, j}$ is invertible for any $i \in [d]$ and $j \in [\ell]$. 
    Suppose $\vD_i$ is an $n_i \times n_i$ block. Let $u_0 =0$ and $u_i = \sum_{k = 1}^{i} n_i$ for each $i \in [d]$.
    For any $A_1^{-1} A_j$, we have 
    \[D_1^{-1} D_j = \diag(D_{1, 1}^{-1}D_{1, j}, \dots, D_{d, 1}^{-1}D_{d, j}) = R A_1^{-1} L^{-1} L A_j R^{-1} = R A_1^{-1}  A_j R^{-1}.\]
    Hence, \begin{equation}\label{equ:polynomial_block_compatible}A = f(A_1^{-1}A_2, \dots, A_1^{-1}A_\ell) = R f(D_1^{-1}D_2, \dots, D_1^{-1}D_\ell) R^{-1}.\end{equation}
    Let $E = f(D_1^{-1}D_2, \dots, D_1^{-1} D_\ell)$. Since $f$ is a polynomial function, $E$ is also a block-diagonal matrix that can be represented as $\diag(E_1, \dots, E_d)$ for  $E_i$ of dimension $n_i \times n_i$ for each $i\in [d]$. 

    Let $S_i = \langle e_{u_{i-1} + 1} E R^{-1}, \dots, e_{u_i} E R^{-1}\rangle$ for any $i \in [d]$. 
    By the block-diagonal structure of $E$,
    $S_i$ is a subspace of $\langle e_{u_{i-1} + 1} R^{-1}, \dots, e_{u_i} R^{-1}\rangle$. 
    By Equation (\ref{equ:polynomial_block_compatible}), we have 
    \[\rowvectorspace(A) = \rowvectorspace(R A) = S_1 \oplus \dots \oplus S_d.\]
    On the other hand, by the definition of $\vA_{i, \vD, L}$ for any $i \in [d]$, we have 
    \[\rowvectorspace(\vA_{i, \vD, L}) = \langle \{e_{u_{i - 1} + 1} R^{-1}, \dots, e_{u_i} R^{-1}\}\rangle.\] Hence, $S_i$ is a subspace of $\rowvectorspace(\vA_{i, \vD, L})$. 
    Thus, $\rowvectorspace(A)$ is a direct sum of row vector subspaces of $\rowvectorspace(\vA_{i, \vD, L})$ for each $i\in [d]$, and then $\rowvectorspace(A)$ is block-compatible.
\end{proof} 

Based on the characteristic and block-compatible properties, we define the hierarchical row tuple decomposition as follows. 

\begin{definition}[Hierarchical Row Tuple Decomposition]\label{def:row_tuple_decomposition}
    For a matrix tuple $\vA \in \M(n \times m, \F_q)^\ell$, 
    a sequence of characteristic block-compatible row tuple subspaces $T_1, \dots, T_\zeta \leq \rowtuplespace(\vA)$ with parameters $h_0, \dots, h_\beta$ for some integer $\beta \geq 0$ is a \emph{depth-$\beta$ hierarchical row tuple decomposition} of $\vA$ if the following conditions hold with 
    $T_{\zeta + 1}:=\vecz$, 
    $h_{\beta+1}: = \zeta + 1$, 
     $W_i:=\langle \rowvectorspace(T_{h_i}) \cup \dots \cup \rowvectorspace(T_\zeta)\rangle$ for $0 \leq i \leq \beta$, and $W_{\beta + 1} := \vecz$.  
    \begin{enumerate}
        \item $1 = h_0 < h_1 < h_2 < \dots < h_\beta \leq \zeta$.  
        \item $\rowtuplespace(\vA) = \langle T_1 \cup  \dots \cup T_\zeta \rangle$.
        \item For any $0 \leq i \leq \beta$, $\rowvectorspace(T_{h_i}) / W_{i + 1} = W_i  / W_{i+1}$, and \[\langle \{\va \in \rowtuplespace(\vA) :\forall j\in[\ell],  \va^{(j)} \in W_i\}\rangle = \langle T_{h_i} \cup \dots \cup T_\zeta \rangle.\]
        \item For any $0 \leq i \leq \beta$, 
        \[\langle \{v \vA : v \in \F^n \text{ and }  \forall j \in [\ell],  v A_{j}  \in   W_{i} \} \rangle / W_{i+1} = (T_{h_i} / W_{i+1}) \oplus \dots \oplus (T_{h_{i + 1} - 1} / W_{i+1}).\]

    \end{enumerate}
    In addition, for each $1 \leq i \leq \zeta$, we use $K_i$ to denote $T_i \cap \langle T_{h_\gamma} \cup \dots \cup T_{\zeta} \rangle$, where $\gamma$ is the integer such that $h_{\gamma - 1} \leq i < h_\gamma$.
 
\end{definition}

\subsection{Properties of \IBCtuplesnospace}

In this subsection, we define four properties of \IBCtuples and prove some basic observations regarding these properties. 

In the end, we will show that these four properties hold for all the \IBCtuples for a matrix tuple, and thus Theorem~\ref{thm:main_space_structure} holds. 
However, such a proof is not straightforward. We will prove these properties for the matrix tuple by alternatively proving them for some induced submatrix tuples of the matrix tuple and investigating the implications of these properties for the induced submatrix tuples to the matrix tuple.

\begin{definition}\label{def:four_prop}
    
For an indecomposable-block-corresponding row-submatrix tuple $\vB\in\M(\sigma\times m, \F)^\ell$ of $\vA$, we define properties for $\vB$:
\begin{enumerate}
    \item (Space property) Let $\ibcspace_\vB$ be the linear space spanned by all the \IBCtuples that are right-equivalent to $\vB$, i.e., \[\ibcspace_\vB =\left\langle \left\{\vB' \mid \vB'\text{ is an \IBCtuple of }\vA\text{, and }\exists Q\in\GL(m, \F), \vB'=\vB Q\right\}\right\rangle.\] $\vB$ satisfies the space property if there is a subspace $\ibcspacekernel_\vB \leq \ibcspace_\vB$ such that
    every element of $\ibcspacekernel_\vB$ is not an \IBCtuple of $\vA$ right-equivalent to $\vB$, and every element in $\ibcspace_\vB$ but not in $\ibcspacekernel_\vB$ is an \IBCtuple of $\vA$ right-equivalent to $\vB$.
    If an \IBCtuple $\vB$ of $\vA$ satisfies the space property, then we call $\ibcspace_{\vB}$ the \IBCtuple space of $\vB$, and $\ibcspacekernel_{\vB}$ the \IBCtuple space kernel of $\vB$. 
    \item (Row tuple space property) $\vB$ satisfies the row tuple space property with respect to a depth-$\beta$ hierarchical row tuple decomposition $T_1, \dots, T_\zeta$ with parameters $h_0, \dots, h_\beta$ if $\vB$ satisfies the space property, and 
    there are row tuple space parameters $1 \leq t_{\vB, 1}, \dots, t_{\vB, \sigma} \leq \zeta$ such that 
    for every $1 \leq i \leq \sigma$, 
    $\langle \{e_i \vB' \mid \vB' \in \ibcspace_{\vB}\}\rangle = T_{t_{\vB, i}}$, and for every $\vB' \in \ibcspace_\vB$, $\vB' \in \ibcspacekernel_\vB$ if and only if $e_i \vB' \in K_{t_{\vB, i}}$. 
    \item (Dimension property) $\vB$ satisfies the dimension property if for any \IBCtuple $\vB'$ that is right-equivalent to $\vB$, $\dim(\rowvectorspace(\vB')) = \dim(\rowvectorspace(\vB))$. Furthermore, for two \IBCtuples $\vB'$ and $\vB''$ right-equivalent to $\vB$, $\vB'$ and $\vB''$ are strongly correlated iff $\dim(\rowvectorspace(\vB' - \vB'')) < \dim(\rowvectorspace(\vB'))$. 
    \item (Block-compatible property)
    For any block-diagonalization $\vD = \diag(\vD_1, \dots, \vD_d)$ of $\vA$ with $\vD = L \vA R^{-1}$, 
    the projection of any $\vB' \in \ibcspace_\vB$ on any block of $\vD$ is in $\ibcspace_\vB$,
    and the projection of $\vB' \in \ibcspacekernel_\vB$ on any block of $\vD$ is in $\ibcspacekernel_\vB$, i.e., for any $1 \leq i\leq d$, $\proj_{\vD, L}(\vB', i) \in \ibcspace_\vB$ for any $\vB' \in \ibcspace_\vB$, and $\proj_{\vD, L}(\vB'', i) \in \ibcspacekernel_\vB$ for any $\vB'' \in \ibcspacekernel_\vB$.
\end{enumerate}
\end{definition}

In the rest of this section, we prove a few useful observations regarding the four properties of the \IBCtuples of a matrix tuple. 
First, we provide a subroutine to extract the row vectors of an \IBCtuplenospace. 

\begin{framed}
\noindent \textbf{Row Vector Extraction Algorithm}

\noindent \textbf{Input:} A matrix tuple $\vB \in \M(n\times m, \F_q)^\ell$. 

\noindent \textbf{Output:} Two matrices $Y_\vB$ and $Z_\vB$. 

\noindent \textbf{Algorithm:}

\begin{enumerate}

    \item Let $s = 0$.
    \item For every $1 \leq i \leq \ell\cdot n$, if $e_i \cdot \mtupletomatrix(\vB)$ is not in $\langle \{x_1, \dots, x_s \}\rangle$, then let $s = s + 1$,  $x_{s} = e_i \cdot \mtupletomatrix(\vB)$, and $a_s = i$. 
    \item 
    Let $Y_\vA$ be a $s \times (\ell \cdot  n)$ matrix such that $i$-th row of $Y_\vB$ is $e_{a_i}$. 
    Let $Z_\vB$ be a $(\ell \cdot n) \times s$ matrix such that $Z_\vB Y_\vB \mtupletomatrix(\vB) = \mtupletomatrix(\vB)$. 
\end{enumerate}
\vspace{-.4cm}
\end{framed}

\begin{fact}\label{fact:ibctuple_char_matrix}
For a matrix $\vB \in \M(\sigma \times m, \F_q)^\ell$ for a matrix tuple $\vA$, we have the following observations for the Row Vector Extraction Algorithm:
\begin{enumerate}
    \item The running time of the algorithm is $\poly(\sigma, m, \ell, \log q)$. 
    \item $(I - Z_\vB Y_\vB) \mtupletomatrix(\vB)$ is a zero matrix.
    \item  Let $\vB'$ be a matrix tuple such that there is an invertible matrix $R$ such that $\vB' = \vB R^{-1}$. The following properties hold.
    \begin{enumerate}
        \item $Y_{\vB} = Y_{\vB'}$ and $Z_{\vB} = Z_{\vB'}$.  
        \item If $\vB$ is an \IBCtuple satisfying the dimension property, then $\vB$ and $\vB'$ are strongly correlated iff $Y_{\vB}(\mtupletomatrix(\vB) - \mtupletomatrix(\vB'))$ is not full row rank. 
        \item $(I - Z_\vB Y_\vB) \vB'$ is a zero matrix. 
    \end{enumerate}
\end{enumerate}

\end{fact}

\begin{proof}
The running time of the algorithm is obtained by the definition of the algorithm. The second property is obtained by the definition of $Y_\vB$ and $Z_\vB$. 

    For the property 3(a), 
    by the definition of $\mtupletomatrix(\vB)$ and $\mtupletomatrix(\vB')$, we have $\mtupletomatrix(\vB) = \mtupletomatrix(\vB') R^{-1}$. 
    Hence, for any $v \in \F^{\ell\cdot n}$, $v \cdot \mtupletomatrix(\vB) = v \cdot \mtupletomatrix(\vB') R^{-1}$. 
    
    Let $s_{\vB, i}$ and $s_{\vB', i}$ denote the value of $s$ for $\vB$ and $\vB'$ respectively at the end of $i$-th iteration of Step 2 of the algorithm, and .
    By induction on $i$ of Step 2, $s_{\vB, i} = s_{\vB', i}$ for each $i$ of Step 2. Thus $s_\vB = s_{\vB'}$, and $a_{\vB, j} = a_{\vB', j}$ for any $1 \leq j \leq s_{\vB}$ after Step 2.
    Hence, we have $Y_{\vB} = Y_{\vB'}$ and \[Y_{\vB} \mtupletomatrix(\vB) = Y_{\vB} \mtupletomatrix(\vB') R^{-1} = Y_{\vB'}  \mtupletomatrix(\vB') R^{-1}.\]
    Then, we have
    \begin{align*}Z_{\vB'} Y_{\vB'} \mtupletomatrix(\vB') R^{-1}  =  & \mtupletomatrix(\vB') R^{-1} \\ = &\mtupletomatrix(\vB) \\ = & Z_\vB Y_\vB \mtupletomatrix(\vB) \\ = & Z_\vB Y_{\vB'} \mtupletomatrix(\vB')R^{-1}.\end{align*} 
    On the other hand, since $R$ is full rank, and $Y_{\vB} \mtupletomatrix(\vB)$ and $Y_{\vB'} \mtupletomatrix(\vB')$ are full row ranks, $Z_{\vB}$ and $Z_{\vB'}$ are unique. Hence, $Z_\vB = Z_{\vB'}$. 

    For the property 3(b), 
    by the definition, 
    for any matrix tuple $\vC$, $\rowvectorspace(\vC)$ equals the linear space spanned by the row vectors of $\mtupletomatrix(\vC)$.  
    On the other hand, by the Row Vector Extraction Algorithm, the linear space spanned by the row vectors of $\mtupletomatrix(\vC)$ equals the linear space spanned by the row vectors $Y_\vC \mtupletomatrix(\vC)$. 
    Hence, $Y_{\vB}(\mtupletomatrix(\vB) - \mtupletomatrix(\vB')) = Y_{\vB}\mtupletomatrix(\vB - \vB')$ is not full row rank if and only if $\dim(\rowvectorspace(\vB' - \vB)) < \dim(\rowvectorspace(\vB'))$, which, by the dimension property, if and only if $\vB'$ is not strongly correlated to $\vB$. 

    The property 3(c) is obtained by the first property and the fact that $Z_{\vB'} Y_{\vB'} \mtupletomatrix(\vB') = \mtupletomatrix(\vB')$ from the algorithm. 
\end{proof}

\subsection{Representative \IBCtuple sequence based matrix tuple canonical form}\label{sec:canonical_form_from_rep}


We give an algorithm to compute a minimum block-diagonalization of the matrix tuple based on a given representative \IBCtuple sequence, along with the associated \IBCtuple space, \IBCtuple space kernel, and parameters for each \IBCtuple in the sequence. 
Additionally, the algorithm yields a canonical form of the matrix tuple under the condition that equivalent input matrix tuples have representative \IBCtuple sequences satisfying the property that equivalent \IBCtuples are right-equivalent (as described in the second property of Lemma~\ref{lem:canonical_with_ibctuples}).

Consequently, computing a canonical form of a matrix tuple reduces to computing a representative \IBCtuple sequence such that equivalent \IBCtuples from the representative sequences for equivalent input matrix tuples are right-equivalent. We address this question in Section~\ref{sec:conjugation}, Section~\ref{sec:direct_sum_decomposition}, and Section~\ref{sec:quotient_matrix_tuple_main}.

Formally, we prove the following lemma in the rest of this section.

\begin{lemma} \label{lem:canonical_with_ibctuples}
There is an \IBCtuple Selection Algorithm that satisfies the following conditions.
\begin{enumerate}
    \item For a matrix tuple $\vA \in \M(n \times m, \F_q)^\ell$, a hierarchical row tuple decomposition $T_1, \dots, T_\zeta$ with parameters $h_0, \dots, h_\beta$, and a representative \IBCtuple sequence $\vB_1, \dots, \vB_k$ of $\vA$ such that each of $\vB_i$ satisfies the four properties of Definition~\ref{def:four_prop} with \IBCtuple space $\ibcspace_{\vB_i}$, \IBCtuple space kernel $\ibcspacekernel_{\vB_i}$, and parameters for each of $\vB_1, \dots, \vB_k$, in $\poly(n, m, \ell, \log q)$ running time, the \IBCtuple Selection Algorithm outputs a minimum block-diagonalization of $\vA$ such that for different executions of the algorithm, the output is fixed. 

    \item The algorithm is canonical in the following sense: For two inputs of $\vA, T_1, \dots, T_\zeta, \vB_1, \dots, \vB_k$ and $\vA', T_1', \dots, T_{\zeta'}', \vB_1', \dots, \vB_{k'}'$ of the algorithm such that there exist invertible matrices $L$ and $R$ satisfying the following conditions:
    \begin{enumerate}
        \item $\vA' = L \vA R^{-1}$. 
        \item $\zeta' = \zeta$, and $T_i' = T_i R^{-1}$ for any $i \in [\zeta]$. 
        \item The parameters of the two hierarchical row tuple decompositions are the same. 
        \item $k' = k$. For any $i\in[k]$,  $\vB_i'$ is right-equivalent to $\vB_i$,
        $\ibcspace_{\vB_i'} = \langle \{\vE R^{-1} : \vE \in \ibcspace_{\vB_i} \} \rangle$, and $\ibcspacekernel_{\vB_i'} = \langle \{\vE R^{-1} : \vE \in \ibcspacekernel_{\vB_i} \} \rangle$.
        \item For any $i\in[k]$,  and $t_{\vB_i', j} = t_{\vB_i, j}$ for any $j \in [\sigma_i]$, where $\sigma_i$ is the number of rows of $\vB_i$. 
    \end{enumerate}
    the outputs of the algorithm for $\vA$ and $\vA'$ are the same. 
\end{enumerate}
\end{lemma}

Before proving Lemma~\ref{lem:canonical_with_ibctuples}, we first prove a few lemmas for the \IBCtuples satisfying the properties defined in Definition~\ref{def:four_prop}. 
The following fact gives a sufficient and necessary condition for two \IBCtuples to be correlated. 

\begin{fact}\label{fact:ibctuple_correlated_necessary_condition}
    For an \IBCtuple $\vB$ of a matrix tuple $\vA$, if $\vB$ satisfies the space property and dimension property, then for another \IBCtuple $\vB'$ that is right-equivalent to $\vB$, 
    $\vB$ and $\vB'$ are correlated iff
    there is an invertible matrix $L$ such that 
    $L\vB$ is right-equivalent to $\vB$, and
    $L\vB$ and $\vB'$ are strongly correlated.
\end{fact}
\begin{proof}
    If $\vB$ and $\vB'$ are correlated, then the existence of $L$ is by Definition~\ref{def:ibctuple}. In the rest of this proof, we prove that the existence of $L$ implies that $\vB$ and $\vB'$ are correlated by showing that the existence of $L$ implies that there is an invertible matrix $L'$ such that $L'\vB'$ is right-equivalent to $\vB$, and $\vB$ and $L'\vB'$ are strongly correlated. 

    We first prove the following two properties:
    \begin{enumerate}
    \item For any $\vB_1$ right-equivalent to $\vB$, $L \vB_1$ is also an \IBCtuple right-equivalent to $\vB$. 
    \item For any $\vB_1 \in \ibcspacekernel_\vB$, $L \vB_1$ is also in $\ibcspacekernel_\vB$. 
    \end{enumerate}

    For the first property, because $L \vB$ is right-equivalent to $\vB$, then there is an invertible matrix $R$ such that $L\vB = \vB (R)^{-1}$. 
    Since $\vB_1$ is right-equivalent to $\vB$, so $\vB_1 = \vB (R_1)^{-1}$ for some invertible $R_1$. 
    Hence, $L \vB_1 = L\vB (R_1)^{-1} = \vB (R)^{-1} (R_1)^{-1}$, which means $L\vB_1$ is right-equivalent to $\vB$ by Definition~\ref{def:ibctuple}.

    To prove the second property, 
    For an arbitrary $\vB_2$ that is in $\ibcspace_\vB$ but not in $\ibcspacekernel_\vB$ and an arbitrary $\vB_1 \in \ibcspacekernel_\vB$, $\vB_2$ and $\vB_2 + \vB_1$ are strongly correlated by the space property. 
    Hence, $L \vB_2$ and $L(\vB_2 + \vB_1)$ are both right-equivalent to $\vB$ by the first property. By the space property, both $L \vB_2$ and $L(\vB_2 + \vB_1)$ are in $\ibcspace_\vB$ but not in $\ibcspacekernel_\vB$. Because $\ibcspace_\vB$ is a linear space, $L \vB_1$ is also in $\ibcspace_\vB$. 
    
    On the other hand, by the dimension property, since $\vB_1$ is in $\ibcspacekernel_\vB$ and $\vB_2$ is in $\ibcspace_\vB$ but not in $\ibcspacekernel_\vB$, $\dim(\rowvectorspace(\vB_1)) < \dim(\rowvectorspace(\vB_2))$, and thus $\dim(\rowvectorspace(L\vB_1)) < \dim(\rowvectorspace(\vB_2))$. Since $L \vB_1$ is in $\ibcspace_\vB$, by the dimension property, $L \vB_1$ is in $\ibcspacekernel_\vB$. Hence, the second property holds. 

    To prove the lemma, since $L$ is an invertible matrix over a finite field, there is an integer $c$ such that $P^c = I$.
    Since both $\vB$ and $L^{c - 1}\vB'$ are in $\ibcspace_\vB$ by the first property, 
    $\vB - L^{c-1}\vB' = L^{c - 1} (L \vB - \vB')$ is in $\ibcspacekernel_\vB$, because $L\vB - \vB'$ is in $\ibcspacekernel_\vB$. 

    Let $L' = L^{c - 1}$. Both $\vB$ and $L'
    \vB'$ are \IBCtuples right-equivalent to $\vB$ by the first property, and thus they are strongly correlated.
\end{proof}

The following fact gives a characterization of the linear space spanned by all the \IBCtuples correlated to an \IBCtuple $\vB$.

\begin{fact}
Suppose an \IBCtuple $\vB$ of $\vA$ satisfies the space property and dimension property. 
Let $U$ be the linear space spanned by all the \IBCtuples that are correlated to $\vB$, then $\ibcspacekernel_\vB \leq U$, and for any $\vB' \in U$ but not in $\ibcspacekernel_\vB$, $\vB'$ and $\vB$ are correlated. 
\end{fact}
\begin{proof}
    We show that if two \IBCtuples $\vB'$ and $\vB''$ both are correlated to $\vB$ but not strongly correlated, then $\vB' - \vB''$ and $\vB$ are correlated. 

    Because $\vB$ and $\vB'$ are correlated, there is an invertible matrix $L'$ such that $\vB' = L' \vB + \vC'$ for some $\vC' \in \ibcspacekernel_\vB$. 
    Similarly, there is an invertible matrix $L''$ such that $\vB'' = L'' \vB + \vC''$ for some $\vC'' \in \ibcspacekernel_\vB$.
    Hence, we have \[\vB' - \vB'' = L' \vB + \vC' - L'' \vB - \vC'' = (L' - L'') \vB + (\vC' - \vC'').\]

    Since $\vB' - \vB''$ is in $\ibcspace_\vB$ by the space property, $\vB' - \vB''$ is either an \IBCtuple right-equivalent to $\vB$, or in $\ibcspacekernel_\vB$ and thus not an \IBCtuple right-equivalent to $\vB$. 
    
    If $\vB' - \vB''$ is an \IBCtuple right-equivalent to $\vB$, then $(L' - L'') \vB$ is an \IBCtuple right-equivalent to $\vB$ by the space property, and thus equivalent to $\vB' - \vB''$. 
    Furthermore, $(L' - L'') \vB$ is strongly correlated to $\vB' - \vB''$. By Fact~\ref{fact:ibctuple_correlated_necessary_condition}, $\vB' - \vB''$ and $\vB$ are correlated.

    Hence, any matrix tuple in $U$ is either an \IBCtuple correlated to $\vB$ or is not an \IBCtuplenospace. 
    Since $\ibcspacekernel_\vB$ is a subspace of $U$ because $\vB + \vC$ is in $\ibcspace_\vB$ for any $\vC \in \ibcspacekernel_\vB$,  every matrix tuple in $U$ but not in $\ibcspacekernel_\vB$ is an \IBCtuple correlated to $\vB$.     
\end{proof}

We also characterize the projection of an \IBCtuple on blocks of a block-diagonalization of the matrix tuple.  

\begin{claim}\label{claim:ibctuple_projection_iso}
    Let $\vA$ be a matrix tuple such that every \IBCtuple satisfies the space property and the block-compatible property, $\vD = \diag(\vD_1, \dots, \vD_d)$ be a minimum block-diagonalization of $\vA$, and $L, R$ be invertible matrices such that $\vD = L \vA (R)^{-1}$. 
    For an \IBCtuple $\vB$ of $\vA$ and an integer $i$, 
    if $\proj_{\vD, L}(\vB, i)$ is not in $\ibcspacekernel_\vB$, 
    then $\vA_{i, \vD, L}$ and $\vB$ are equivalent. 
\end{claim}
\begin{proof}
    We prove this by contradiction. Suppose $\proj_{\vD, L}(\vB, i)$ is not in $\ibcspacekernel_\vB$, and $\vA_{i, \vD, L}$ is not equivalent to $\vB$. 

    By the block-compatible property, $\proj_{\vD, L}(\vB, i)$ is in $\ibcspace_\vB$. 
    Since $\proj_{\vD, L}(\vB, i)$ is not in $\ibcspacekernel_\vB$, $\proj_{\vD, L}(\vB, i)$ is an \IBCtuple right-equivalent to $\vB$ by the space property. The number of rows in $\vD_i$ cannot be smaller than the number of rows in $\proj_{\vD, L}(\vB, i)$, because every row tuple of $\proj_{\vD, L}(\vB, i)$ is also a row tuple of $\vD$. The number of rows in $\vD_i$ also cannot be equal to the number of rows in $\proj_{\vD, L}(\vB, i)$, because otherwise $\proj_{\vD, L}(\vB, i)$ and $\vA_{i, \vD, L}$ are equivalent. So the number of rows in $\vD_i$ must be larger than the number of rows in $\proj_{\vD, L}(\vB, i)$. 

    Next, we show that $\vD_i$ is not indecomposable, and thus $\vD$ is not a minimum block-diagonalization. 
    By Definition~\ref{def:projection_matrix_tuple_diag}, $\rowtuplespace(\proj_{\vD, L}(\vB, i))$ is a subpsace of $\rowtuplespace(\vA_{i, \vD, L})$. 
    Since $\vD_i$ has more rows than $\proj_{\vD, L}(\vB, i)$, \[\dim(\proj_{\vD, L}(\vB, i)) < \dim(\rowtuplespace(\vA_{i, \vD, L})).\] 

    On the other hand,    
    because $\proj_{\vD, L}(\vB, i)$ is an \IBCtuplenospace, so there exist $L'$ and $R'$ such that $L' \vA (R)^{-1} = \vE =  \diag(\vE_1, \vE_2)$ such that 
    $\vA_{1, \vE, L'} = \proj_{\vD, L}(\vB, i)$. 
    Both $\proj_{\vE, L}(\vA_{i, \vD, L}, 1)$ and $\proj_{\vE, L}(\vA_{i, \vD, L}, 2)$ are non-zero matrix tuples. 
    By the definition of block-diagonalization, we have \[\vA_{i, \vD, L} = \proj_{\vE, L}(\vA_{i, \vD, L}, 1) \oplus \proj_{\vE, L}(\vA_{i, \vD, L}, 2),\] and \[\rowtuplespace(\vA_{i, \vD, L}) = \rowtuplespace(\proj_{\vE, L}(\vA_{i, \vD, L}, 1)) \oplus \rowtuplespace(\proj_{\vE, L}(\vA_{i, \vD, L}, 2)),\]  
    and thus $\vD_i$ is not indecomposable by Fact~\ref{fact:ibctuple_diagonalzation_most_basic}. 
\end{proof}

Consider the following scenario: Suppose we want to construct a minimum block-diagonalization of a matrix tuple by selecting \IBCtuples corresponding to the blocks of the minimum block-diagonalization. 
Suppose a few block corresponding \IBCtuples have already been selected. 
The following claim provides a sufficient condition for selecting another \IBCtuplenospace.

\begin{claim}\label{claim:sequential_ibctuple_simultaneous}
Let $\vB_1, \dots, \vB_{d_0}$ for some integer $d_0$ be a sequence of \IBCtuples of $\vA$ that are not necessarily equivalent such that 
all of $\vB_1, \dots, \vB_{d_0}$ satisfies the four properties of Definition~\ref{def:four_prop}, and 
there is a minimum block-diagonalization $\vD = \diag(\vD_1, \dots, \vD_d)$ of $\vA$ for some integer $d > d_0$ with $\vD = L \vA R^{-1}$ such that $\vB_i= \vA_{i, \vD, L}$ for all $1 \leq i \leq d_0$. 
Let $\vB$ be an \IBCtuple of $\vA$ satisfying the four properties of Definition~\ref{def:four_prop},  $S = \{i \in [d_0] : \vB_i \text{ equivalent to } \vB\}$, and $L_i$ for $i \in S$ be an invertible matrix such that $L_i \vB_{i} \in \ibcspace_\vB$. 
Let $U = \langle \{\vB' \in \ibcspace_\vB : \exists i\in S \text{ s.t. } \vB'  \text{ is correlated to } L_i \vB_i\} \rangle$.

If $\vB$ is not in $U$, then 
there is a minimum block-diagonalization $\vD' = \diag(\vD_1', \dots, \vD_d)$ of $\vA$ such that $\vB_i= \vA_{i, \vD', L}$ for all $1 \leq i \leq d_0$, and $\vB = \vA_{d_0 + 1, \vD', L}$. 
\end{claim}
\begin{proof}
We turn $\vD$ into $\vD'$ as required by multiplying invertible matrices from the left and right.
We investigate the projections of $\vB$ on blocks of $\vD$. 
By the block-compatible property, $\proj_{\vD, L}(\vB, i)$ is in $\ibcspace_\vB$. 
By Claim~\ref{claim:ibctuple_projection_iso} and the fact that $\vB$ is not in $U$, there is a $j > d_0$ such that $\proj_{\vD, L}(\vB, j)$ is in $\ibcspace_\vB$ but not in $\ibcspacekernel_\vB$. 
So there exist invertible matrices $L''$ and $R''$ such that 
\[\vD'' = L'' \vD (R'')^{-1} = \diag(\vD_1, \dots, \vD_{d_0}, \vD_j, \vD_{d_0 + 1}, \dots, \vD_{j-1}, \vD_{j+1}, \dots, \vD_d) = L'' L \vD (R)^{-1} (R'')^{-1}.\]
The new matrix tuple is still a minimum block-diagonalization of $\vA$, and we have $\vB_i = \vA_{i, \vD'', L'' L}$ for all $1 \leq i \leq d_0$.


By the definition of block-diagonalization, we have 
\[\rowtuplespace(\vA) = \rowtuplespace(\vA_{1, \vD'', L'' L}) \oplus  \dots \oplus \rowtuplespace(\vA_{d, \vD'', L'' L}),\]
and \[\rowvectorspace(\vA) = \rowvectorspace(\vA_{1, \vD'', L'' L}) \oplus  \dots \oplus \rowvectorspace(\vA_{d, \vD'', L'' L}).\]
Let $T = \oplus_{i \in [d] \setminus \{d_0 + 1\}}\rowtuplespace(\vA_{i, \vD'', L'' L})$ and $S = \oplus_{i \in [d] \setminus \{d_0 + 1\}}\rowvectorspace(\vA_{i, \vD'', L'' L})$, we have 
\[\rowtuplespace(\vA) = \rowtuplespace(\vA_{d_0 + 1, \vD'', L'' L}) \oplus T \text{ and }
\rowvectorspace(\vA) = \rowvectorspace(\vA_{d_0 + 1, \vD'', L'' L}) \oplus S.\]

In the rest of this proof, we show that \begin{equation}\label{equ:sequnetial_ibctuple_selection}\rowtuplespace(\vA) = \rowtuplespace(\vB) \oplus T \text{ and }
\rowvectorspace(\vA) = \rowvectorspace(\vB) \oplus S,\end{equation}
and then the claim follows by the definition of block-diagonalization. 

Since $\proj_{\vD'', L'' L}(\vB, d_0 + 1)$ is in $\ibcspace_\vB$ but not in $\ibcspacekernel_\vB$, 
by Claim~\ref{claim:ibctuple_projection_iso}, $\vA_{d_0 + 1, \vD'', L'' L}$ is equivalent to $\vB$. 
Hence, we have $\dim(\rowtuplespace(\vB)) = \dim(\rowtuplespace(\vA_{d_0 + 1, \vD'', L'' L}))$ and $\dim(\rowvectorspace(\vB)) = \dim(\rowvectorspace(\vA_{d_0 + 1, \vD'', L'' L}))$.

Moreover, since $\vA_{d_0 + 1, \vD'', L'' L}$ is equivalent to $\vB$, every non-zero row tuple in $\rowtuplespace(\vB)$ is not in $T$, and  every non-zero row vector in $\rowvectorspace(\vB)$ is not in $S$. 
Thus,  $\rowtuplespace(\vB) \cap T$ contains only the zero row tuple, and $\rowvectorspace(\vB) \cap V$ contains only the zero vector.

We show that $\rowtuplespace(\vB) \cap T$ contains only the zero row tuple, and $\rowvectorspace(\vB) \cap V$ contains only the zero vector. 
Then by Claim~\ref{claim:ibctuple_projection_iso}, we have $\dim(\rowtuplespace(\vB)) = \dim(\rowtuplespace(\vA_{d_0 + 1, \vD'', L'' L}))$ and $\dim(\rowvectorspace(\vB)) = \dim(\rowvectorspace(\vA_{d_0 + 1, \vD'', L'' L}))$. Thus, we have Equation (\ref{equ:sequnetial_ibctuple_selection}). 
\end{proof}

Now we give the \IBCtuple Selection Algorithm and prove Lemma~\ref{lem:canonical_with_ibctuples}.

\begin{framed}
\noindent \textbf{\IBCtuple Selection Algorithm}

\noindent \textbf{Input:} 
\begin{enumerate}
    \item Matrix tuple $\vA \in \M(n\times m, \F_q)^\ell$ with hierarchical row tuple decomposition $T_1, \dots, T_\zeta$ and parameters $h_0, \dots, h_
    \beta$. 
    \item Representative \IBCtuple sequence $\vB_1, \dots, \vB_k$ for $\vA$ with \IBCtuple space $\ibcspace_{\vB_i}$, \IBCtuple space kernel $\ibcspacekernel_{\vB_i}$, and parameters $t_{\vB_i, 1}, \dots, t_{\vB_i, \sigma_i}$ for each $i \in [k]$.
\end{enumerate}

\noindent \textbf{Output:} A matrix tuple. 

\noindent \textbf{Algorithm:} 
\begin{enumerate}
\item Let $\delta = 0$. 
\item For $i = 1, \dots, k$, let $s_i = \min_{j \in [\sigma_i]} t_{\vB_i, j}$, $u_i = \min \{j \in [\sigma_i] : t_{\vB_i, j} = s_i\}$, $H = K_{s_j}$, and repeat the following procedure until $H = T_{s_i}$:
\begin{enumerate}
    \item Choose an arbitrary $\vB^\dagger \in \ibcspace_{\vB_i}$ such that $e_{s_i} \vB^\dagger$ is not in $H$. 
    \item Let $\delta = \delta + 1$ and $\vC_\delta = \vB^\dagger$. 
    \item For $j = 1$ to $\sigma_i$, if $t_{\vB_i, j} = s_i$, then let $H = \langle H \cup \{e_j \vB^\dagger\}\rangle$. 
\end{enumerate}
\item Let 
\[\vC = (C_1, \dots, C_\ell) = \begin{bmatrix}
     \vC_1 \\ \vdots \\ \vC_\delta
\end{bmatrix},\]
and $\sigma$ be the number of rows of $\vC$, $w = 0$, $S$ be a row vector space initially contains only the zero vector. 
\item  For $i = 1, \dots, \ell$, for $j = 1, \dots, \sigma$, if $e_j C_i$ is not in $S$, then $w = w + 1$, $v_i = e_j C_i$, and $S = \langle S \cup \{e_j C_i\}\rangle$.
\item Let $P \in \M(s \times m, \F_q)$ be the matrix such that $e_i P = v_i$ for any $i \in [s]$. Return $\vC P^{-1}$.  
\end{enumerate}
\vspace{-.3cm}
\end{framed}

\vspace{.2cm}\begin{proof}[Proof of Lemma~\ref{lem:canonical_with_ibctuples}]
By the space property and row tuple space property of $\vB_1, \dots, \vB_k$, 
for any $k\in[\delta]$
$\vC_k$ selected in Step 2 is an \IBCtuple not in the linear space spanned by \IBCtuples in $\vC_1, \dots, \vC_{k-1}$ that are right-equivalent to $\vC_k$ and their correlated \IBCtuplesnospace. 
By the definition of \IBCtuple and Claim~\ref{claim:sequential_ibctuple_simultaneous}, the output of the algorithm is a minimum block-diagonalization of the input matrix. 

In addition, since the output is a minimum block-diagonalization of $\vA$, we have \[\rowtuplespace(\vA) = \rowtuplespace(\vC_1) \oplus \dots \oplus \rowtuplespace(\vC_\delta)\] and \[\rowvectorspace(\vA) = \rowvectorspace(\vC_1) \oplus \dots \oplus \rowvectorspace(\vC_\delta).\]
Hence, for different choices of the matrix tuples $\vC_1, \dots, \vC_\delta$ in Step 2 of the algorithm, the output is the same. Hence, the first property holds.

For the second property, 
suppose we first run the algorithm for $\vA$,
By the correspondence between the two inputs, the numbers of matrix tuples selected in Step 2 of the algorithm for $\vA$ and $\vB$ are the same. 
 Let $\vC_1, \dots, \vC_\delta$ and $\vC_1', \dots, \vC_{\delta}'$ denote the matrix tuples selected in Step 2 of the algorithm for $\vA$ and $\vB$ respectively. 
By the correspondence between the two inputs, 
for any $j \in [\delta]$, if $\vC_k' = \vC_k R^{-1}$ for all the $k\in[j-1]$, 
when selecting $\vC_j'$ for $\vA'$, $\vC_j R^{-1}$ is a valid choice for $\vC_j'$. 
Hence, by induction on each iteration of Step 2 for $\vA'$, $\vC_1 R^{-1}, \dots, \vC_\delta R^{-1}$ is a sequence of matrix tuples selected by one execution of the algorithm for $\vA'$. 
By Step 5 of the algorithm, the outputs for $\vA$ and $\vA'$ are the same. 
\end{proof}






\section{\IBCtuples for full rank square matrix tuples}\label{sec:conjugation}

In this section, we show that for a square matrix tuple in which each matrix is full rank, there exists an algorithm to either compute a nontrivial characteristic block-compatible row tuple subspace if it exists, or compute a representative \IBCtuple sequence containing a single \IBCtuple $\vB$ (implying that all the \IBCtuples of $\vA$ are equivalent to $\vB$). More formally, we prove the following lemma in this section. 

\begin{lemma}\label{lem:conjugation_canonical_main}
There is a Full Rank \IBCtuple Space Algorithm satisfying the following properties:
\begin{enumerate}
\item
    Given a matrix tuple 
    $\vA \in \M(n, \F_q)^\ell$ such that every $A_i \in \vA$ is full rank, in $\text{poly}(n, \ell, \log q)$ time, the algorithm gives one of the following outputs:
    \begin{enumerate}
        \item A representative \IBCtuple sequence $\vB$ of $\vA$ with $\ibcspace_{\vB}$,  $\ibcspacekernel_\vB$, and $t_{\vB, 1}, \dots, t_{\vB, \eta}$ where $\eta$ is the number of rows of $\vB$ satisfying the following conditions:
            \begin{enumerate}
                \item Every \IBCtuple of $\vA$ is equivalent to $\vB$. 
                \item $\vB$ satisfies the space property with \IBCtuple space $\ibcspace_\vB$ and \IBCtuple space kernel $\ibcspacekernel_{\vB}$ such that $\ibcspacekernel_{\vB}$ contains only the zero matrix tuple.
                \item $\vB$ satisfies row tuple space property with parameters $t_{\vB, 1} = \dots =  t_{\vB, \eta} = 1$ such that $T_1 = \rowtuplespace(\vA)$ and $K_1$ contains only the zero row tuple of $T_1$.
                \item $\vB$ satisfies the dimension property and block-compatible property.  
            \end{enumerate}
        
        \item A representative \IBCtuple sequence of $\vA$ containing a single \IBCtuple $\vB$ with \IBCtuple space $\ibcspace_{\vB}$ and  \IBCtuple space kernel $\ibcspacekernel_{\vB}$ such that  $\ibcspacekernel_{\vB} < \ibcspace_{\vB}$ contains only the zero matrix tuple;
        \item A non-trivial characteristic block-compatible row tuple subspace $S < \rowtuplespace(\vA)$. 
    \end{enumerate}
    \item The algorithm is canonical in the following sense: 
    Let $\vA' \in \M(n, \F_q)^\ell$ be another matrix tuple such that there exist invertible matrices $L$ and $R$ satisfying $\vA' = L \vA R^{-1}$. The outputs for $\vA$ and $\vA'$ satisfy the following conditions:
    (Let $L$ and $R$ be arbitrary invertible matrices such that such that $\vA' = L \vA R^{-1}$.)
    \begin{enumerate}
        \item If the output for $\vA$ is a representative \IBCtuple sequence of $\vA$ containing a single \IBCtuple $\vB$ with $\ibcspace_{\vB}$ and $\ibcspacekernel_{\vB}$, then the output for $\vA'$ is a representative \IBCtuple sequence of $\vA'$ containing a single \IBCtuple $\vB'$ right-equivalent to $\vB$ with $\ibcspace_{\vB'} = \{\vC R^{-1} : \vC \in \ibcspace_{\vB}\}$ and $\ibcspacekernel_{\vB'} = \{\vC R^{-1} : \vC \in \ibcspace_{\vB}\}$. 
        \item If the output for $\vA$ is a subspace $S < \rowtuplespace(\vA)$, then the output for $\vA'$ is $S R^{-1}$. 
    \end{enumerate}
\end{enumerate}

\end{lemma}

We provide an overview of our approach to prove Lemma~\ref{lem:conjugation_canonical_main}.
Let the input matrix tuple be denoted by $\vA \in \GL(n, \F_q)^\ell$.
The high-level idea of our algorithm is to try to find a non-zero singular matrix (a non-zero matrix that is not full rank) as a linear combination of matrices in $\vA$. 

If such a singular matrix exists, the algorithm will canonically find a desirable singular matrix and use the row vectors in this matrix to compute a characteristic block-compatible row tuple subspace. 
Otherwise, after confirming the non-existence of a desirable singular matrix, the algorithm canonically constructs \IBCtuples of the matrix by analyzing the consequences of this non-existence.

We study the matrix tuple $\vG = (G_1, \dots, G_{\ell - 1})$ such that $G_i = A_1^{-1} A_{i+1}$ for all $i\in[\ell - 1]$ under the conjugation action over $\F_q$. For this action, it is natural to use the matrix algebra generated by $\vG$, i.e., the \emph{enveloping} algebra $\cG:=\Env(\vG)$.
We show that if we can find a non-zero singular matrix in $\cG$, then the row vectors in the singular matrix form a non-trivial characteristic block-compatible row vector subspace, and then we can construct a corresponding characteristic block-compatible row tuple subspace based on the characteristic block-compatible row vector subspace. 

According to the Artin–Wedderburn–Mal’tsev theory, the matrix algebra $\cG$ either contains a non-zero singular matrix or the matrix algebra is isomorphic to a finite field.

As our main algorithmic contribution, we give an algorithm to canonically find a non-zero singular matrix in an arbitrary finite matrix algebra that is not isomorphic to a finite field based on a given canonical basis of the matrix algebra. 
We remark that this is not an easy task, because the total number of singular matrices in the matrix algebra might be a very tiny fraction of all the matrices in the algebra, and it is very inefficient to enumerate a desirable matrix. 
Hence, to achieve a polynomial running time, 
we carefully analyze the structure of the matrix algebra and explore the representation of the matrices in the extension field to exhibit the structure of the singular matrices.

In addition, we explore the finite field condition to directly construct \IBCtuples for $\vA$. Specifically, we show that all the \IBCtuples are equivalent and give an algorithm to canonically select a representative \IBCtuple by exploring the relations between different matrices in $\vA$.

\subsection{Conjugation action for full rank matrix tuples}\label{sec:conjugation_sub}

We give an algorithm to either compute a canonical singular matrix in the matrix algebra or certify the matrix algebra is isomorphic to a finite field. More specifically, we prove the following lemma in this subsection.

\begin{lemma}\label{lemma:conjugate_algebra_main}
There is a Conjugation Matrix Algebra Algorithm satisfying the following properties:
\begin{enumerate}   
    \item Let $\vA = (A_1, \dots, A_{\ell+1}) \in \M(n, \F_q)^{\ell+1}$ be a matrix tuple such that $A_i$ is full rank for each $i\in[\ell + 1]$, and $\vG  = (G_1, \dots, G_\ell)\in \GL(n, \F_q)^\ell$ be the matrix tuple such that $G_i = A_1^{-1} A_{i+1}$ for each $i \in [\ell]$. With $\poly(n, \ell, \log q)$ running time, the Conjugation Matrix Algebra Algorithm with $\vG$ as the input outputs a matrix $C \in  \Env(\vG)$ such that 
    \begin{enumerate}
        \item If $C$ is a zero matrix, then $\Env(\vG)$ is a finite field. 
        \item If $C$ is a non-zero matrix, then $C$ is in $\Env(\vG)$ and $C$ is not full rank. 
    \end{enumerate}
    \item Let $\vG, \vG' \in \GL(n, \F_q)^\ell$ be two inputs of the Conjugation Matrix Algebra Algorithm such that there exists an invertible matrix $P$ such that $\vG' = P \vG P^{-1}$. 
    Let $C$ and $C'$ denote the output for $\vG$ and $\vG'$ respectively. 
    Then for any invertible matrix $P$ such that $\vG' = P \vG P^{-1}$,  $C' = P C P^{-1}$. 
\end{enumerate}
\end{lemma}

We start by giving an algorithm to compute a canonical linear basis of a given subspace for a linear space given by a canonical linear basis. 
Here the canonical basis means that the basis are invariant under the conjugation operation.

\begin{framed}
\noindent \textbf{Subspace Linear Basis Algorithm}

\noindent \textbf{Input:} Matrix space $S \leq \M(n \times m, \F_q)$ and matrix tuple $\vA = (A_1, \dots, A_\ell) \in \M(n \times m, \F_q)^\ell$ such that $A_1, \dots, A_\ell$ are linearly independent and $S \leq \langle \{A_1, \dots, A_\ell\}\rangle$.

\noindent \textbf{Output:} Linear basis $B_1, \dots,  B_{\dim(S)}$  for $S$. 

\noindent \textbf{Algorithm}
\begin{enumerate}
\item Let $t = 0$. 
\item For every $i = 1, \dots, \ell$, if $\left( A_i + \langle \{A_{i+1}, \dots, A_\ell \} \rangle\right) \cap S \neq \emptyset$, then $t = t + 1$ and  $d_t = i$. 
\item For $i = 1, \dots, t$, let $B_i$ be the matrix in $S$ such that $B_i = \sum_{j = 1}^\ell b_{i, j} A_j$ for $b_{i, j} \in \F_q$ and $b_{i, d_i} = 1$ and $b_{i, d_j} = 0$ for every $j \in \{1, \dots, i - 1\} \cup \{i + 1, \dots, t\}$.
\item Return $B_1, \dots, B_t$. 
\end{enumerate}
\vspace{-.3cm}
\end{framed}

\begin{claim}\label{claim:subspace_linear_basis}
    The following properties for the Subspace Linear Basis Algorithm hold:
    \begin{enumerate}
    \item For an input $S \leq \M(n\times m, \F_q)$ and $\vA \in \M(n \times m, \F_q)^\ell$, in time $\poly(n, m, \ell, \log q)$, the output is a linear basis of $S$. 
    \item The algorithm is canonical in the following sense: for two inputs $S, \vA$ and $S', \vA'$ such that there exist invertible matrices $L$ and $R$ satisfying the following conditions:
    \begin{enumerate}
        \item $\dim(S) = \dim(S')$;
        \item $S' = \langle \{L A R^{-1} : A \in S\} \rangle$;
        \item $\vA' = L \vA R^{-1}$,        
    \end{enumerate}
    then 
    $B_i' = L B_i R^{-1}$ for each $i\in[\dim(S)]$, where $B_1, \dots, B_{\dim(S)}$ and $B_1', \dots, B_{\dim(S)}'$ are the outputs for $S, \vA$ and $S', \vA'$ respectively. 
    
    \end{enumerate}
\end{claim}
\begin{proof}
Since $A_1, \dots, A_\ell$ are linearly independent, 
for any different integers $i, i' \in [\ell]$, 
    \[\left( A_i + \langle \{A_{i+1}, \dots, A_\ell \} \rangle\right) \cap \left( A_{i'} + \langle \{A_{i'+1}, \dots, A_\ell \} \rangle\right)\] is an empty set.
    Hence, the output of the algorithm is a linear basis of $S$. 
    On the other hand, \[\left( A_i + \langle \{A_{i+1}, \dots, A_\ell \} \rangle\right)\] for any $i$ can be computed in polynomial time by taking an arbitrary linear basis of $S$, writing every matrix in the linear basis as a linear combination of $A_1, \dots, A_\ell$, and solving a linear equation to determine if there is a matrix in $S$ is in $ A_i + \langle \{A_{i+1}, \dots, A_\ell \}$. 
    The first property is obtained. 

    For the second property, 
    let $t, d_1, \dots, d_t$ and $t', d_1', \dots, d_t'$ be the variables for $\vA$ and $\vA'$ respectively. 
    By induction on Step 2 of the algorithm, we have $t = t'$ and $d_i = d_i'$ for every $i\in[t]$. Hence, the second property holds. 
\end{proof}

Now we give an algorithm to deal with the case that the matrix algebra $\cG$ is isomorphic to full matrix algebra $\M(m, \F_r)$ for some $m$ and $r$, supposing an isomorphism from $\cG$ to $\M(m, \F_r)$ is given. 
The goal of the algorithm is to find a non-zero singular matrix canonically in the sense that the output matrix is invariant under conjugation operation.

\begin{framed}
\noindent \textbf{Full Matrix Algebra Algorithm}

\noindent \textbf{Input:} Matrix Algebra $\cG$ with linear basis $G_1, \dots, G_\ell \in \GL(n, \F_q)$ and isomorphism $f: \cG \rightarrow \M(m, \F_{r})$.

\noindent \textbf{Output:} Matrix $G \in \cG$.

\begin{enumerate}
\item 
Let $\alpha$ be the smallest integer such that $f(G_\alpha)$ is not $c \cdot I$ for some $c \in \F_r$. 
\item Compute the characteristic polynomial of $\charpoly(f(G_\alpha ))$, and an arbitrary root $\lambda \in \F_{r^m}$ of $\charpoly(f(G_\alpha ))$. Denote $\lambda_1 = \lambda, \lambda_2 = \lambda^{r}, \lambda_3 = \lambda^{r^2} \dots,\lambda_{m} = \lambda^{\left(r^{m - 1}\right)}$. 
\item For $i = 1, \dots, m$, 
\begin{enumerate}
\item Compute a matrix $P_i \in \GL(m, \F_{r^m})$ such that 
\[P_i f(G_\alpha)  P_i^{-1} = \diag(\lambda_i, \lambda_i^{r}, \dots, \lambda_i^{r^{m - 1}}) \in \M(m, \F_{r^{m}}).\] 
\item For $j = 1, \dots, m$, 
compute $S_{i, j} = \langle \{G \in \cG: e_1 P_i f(G) P_i^{-1} = 0\}\rangle$. If $\dim(S_{i, j}) > 0$, 
return the first matrix in the output of the Subspace Linear Basis Algorithm for $S_{i, j}$ with $G_1, \dots, G_\ell$.
\item Let $V_{i, j} = \langle \{G \in \cG :  e_1 P_i f(G) P_i^{-1} = b \cdot e_j \text{ for some } b \in \F_{r^m}\}\rangle$.
\item Let $B_i$ be the first matrix of the output of the Subspace Linear Basis Algorithm for $V_{i, 2}$ with $G_1, \dots, G_\ell$. 
If $\charpoly(f(B_i))$ is an irreducible polynomial in $\F_r[x]$, then 
\begin{enumerate}
\item Compute all the $m$ matrices $A_{i, 1}, \dots, A_{i, m} \in V_{i, 1}$ such that $f(A_{i, j})$ has the same characteristic polynomial with $f(B_i)$ for each $j \in [m]$.
\item Compute $C_{i, j} \in \cG$ be the rank 1 matrices such that $C_{i, j} = A_{i, j} + B_i + Z$ for some $Z \in \langle V_{i, 3} \cup \dots \cup V_{i, m} \rangle$. 
\end{enumerate}
\end{enumerate}
\item If there exists a $B_i$ such that $\charpoly(f(B_i))$ is not irreducible, return the lexically smallest $B_i$ with respect to $G_1, \dots, G_\ell$ among all the $B_i$ with a reducible $\charpoly(f(B_i))$.  
\item Otherwise, return the lexically smallest matrix $C_{i, j}$ for all $i, j \in [m]$ with respect to linear basis $G_1, \dots, G_\ell$.  
\end{enumerate}
\vspace{-.4cm}
\end{framed}

We first prove some basic properties of the Full Matrix Algebra Algorithm. 

\begin{lemma}\label{lem:diagonal_new}
Suppose we run the Full Matrix Algebra Algorithm for a matrix algebra $\cG$ with a linear basis $G_1, \dots, G_\ell \in \GL(n, \F_q)$ and an isomorphism $f: \cG \rightarrow \M(m, \F_{r})$ for some $m > 1$ such that for each $k \in [\ell]$, $f(G_k)$ is either $c \cdot I$ for some $c\in \F_r$ or a matrix with an irreducible characteristic polynomial. 
We have the following observations:
\begin{enumerate}
    \item Let $P_i$ and $P_i^\dagger$ be two valid matrices obtained in Step 3(a) of the algorithm for any $i \in [m]$, i.e., \[P_i f(G_\alpha) P_i^{-1} = P_i^\dagger f(G_\alpha) (P_i^\dagger)^{-1} = \diag(\lambda_i, \lambda_i^r, \dots, \lambda_i^{r^{\ell-1}}). \]
        Then there exists a \emph{diagonal} matrix $Y \in \GL(m, F_{r^m})$ such that for every $j \in [\ell]$ \[Y P_i f(G_j) P_i^{-1} Y^{-1} = P_i^\dagger f(G_j) (P_i^\dagger)^{-1}.\]
        \item If for some $i \in [m]$,  $\dim_{\F_r}(S_{i, 1}) = \dots = \dim_{\F_r}(S_{i, m})= 0$ for $S_{i, 1}, \dots, S_{i, m}$ obtained in Step 3(b)(i), then the following properties hold:
        \begin{enumerate}
            \item 
        $\dim_{\F_r}(V_{i, 1}) = \dots = \dim_{\F_r}(V_{i, m}) = m$. 
    \item Every non-zero matrix $H \in V_{i, 1}$, $P_i f(H) P_i^{-1}$ is a diagonal matrix such that every diagonal entry is non-zero.   
    \item For every $j \in \{2, \dots, m\}$, every non-zero matrix $H\in V_{i, j}$ satisfies that, for any $k, k'\in[m]$, the $k$-th row  $k'$-th column of $P_i f(H) P_i^{-1}$ is non-zero if and only if \[k'\equiv k+j-1 (\mathrm{mod} \  m).\] 
    \item For any $H \in \vG$, $H$ can be uniquely written as $\sum_{j = 1}^m U_{H, j}$ for $U_{H, j} \in V_{i, j}$ for each $j \in [m]$. 
    \item  For any $G\in \cC$ such that $G$ is of rank $1$, then $P_i f(G) P_i^{-1}$ is also of rank $1$, and every entry in $P_i f(G) P_i^{-1}$ is non-zero.
    \item Let $H_1\in V_{i, 1}$ and $H_2\in V_{i, 2}$ such that the product of non-zero entries of $P_i f(H_1) P_i^{-1}$ is equal to the product of non-zero entries of $P_i f(H_2) P_i^{-1}$. Then there exists a unique rank-$1$ matrix $E$ in $\cG$ such that $U_{E, 1} = H_1$ and $U_{E, 2} = H_2$ for $U_{E, 1}, \dots, U_{E, m}$ defined as the property 2(d).
    Furthermore, such $E$ can be computed in  $\poly(n, \log q)$ time. 
    \end{enumerate}
        \end{enumerate}
    \end{lemma}
    \begin{proof}
    Since $m > 1$, not all the matrices in $\M(m, \F_r)$ is $c \cdot I$ for some $c \in \F_r$. 
Hence, $\alpha$ in Step 1 exists, and $G_\alpha$ has an irreducible characteristic polynomial. 
Since $G_\alpha$ has an irreducible polynomial, by Fact~\ref{fact:eigenval}, $\charpoly(f(G_\alpha))$ has distinct roots over $\F_{r^m}$, and $\lambda_1, \dots, \lambda_m$ defined by Step 2 of the algorithm are the distinct roots of $\charpoly(f(G_\alpha))$. 

Now we focus on the algorithm for one $i \in [m]$. Since the characteristic polynomial of $G_\alpha$ has distinct roots over $\F_{r^m}$, 
the matrix $P_i \in \M(m, \F_{r^m})$ always exists.
So, for the first property, we have 
    \[P_i^\dagger P_i^{-1}\diag(\lambda_i, \lambda_i^r, \dots, \lambda_i^{r^{\ell-1}})P_i (P_i^\dagger)^{-1} = P_i^\dagger f(G_\alpha) (P_i^\dagger)^{-1} = \diag(\lambda_i, \lambda_i^r, \dots, \lambda_i^{r^{\ell-1}}),\]
    and thus $P_i^\dagger P_i^{-1}$ is invertible and must be a diagonal matrix. 
    Let $Y = P_i^\dagger P_i^{-1}$. The first property follows.

    For the property 2(a), note that for every nonzero $c\in \F_{r^m}$ and any $j\in [m]$, there exists at most one matrix $D\in V_{i, j}$, such that the first row and $j$-th column of $P_i f(D)P_i^{-1}$ equals $c$. As otherwise, there exists a non-zero matrix $G \in \cG$ such that $P_i f(G) P_i^{-1}$ has the first row being all-zero, a contradiction to the condition of $\dim_{\F_r}(S_{i, 1}) = 0$. It follows that $\dim_{\F_r}(S_{i, j})\leq m$ for every $j\in m$. 

On the other hand, let $U$ be the linear span (with coefficients in $\F_r$) of the first rows of $P_i f(G) P_i^{-1}$ for all the $G\in \cG$. By $\dim_{\F_r}(S_{i, 1})=0$, we have that the $\dim_{\F_r}(U) = m^2$ as $\dim_{\F_r}(\M(m, \F_r)) = m^2$. It follows for every $j \in [m]$ and $c\in \F_{r^m}$, there exists $H \in V_{i, j}$ such that the first row $j$-th column of $P_i f(H) P_i^{-1}$ equals $c$. And thus we have for all $j \in[m]$, $\dim_{\F_r}(V_{i, j})=m$, and the property 2(a) follows.

For the property 2(b), let $g(x)\in\F_r[x]$ be an arbitrary polynomial with coefficients in $\F_r$. Then $g(P_i f(G_\alpha) P_i^{-1})$ is a diagonal matrix over $\F_{r^m}$, so $G_\alpha$ is in $V_{i, 1}$. As $\charpoly(f(G_\alpha))$ is irreducible over $\F_r$, $\langle \{g( P_i f(G_\alpha) P_i^{-1})\mid g(x)\in\F_r[x]\}\rangle$ is a linear space of dimension $m$ over $\F_r$. 
Since $g( P_i f(G_\alpha) P_i^{-1}) = P_i g(f(G_\alpha)) P_i^{-1}$ for any polynomial $g \in \F_r[x]$, 
by the property 2(a),  
\[V_{i, 1} = \langle\{G \in \cG: \exists g(x)\in\F_r[x], P_i f(G) P_i^{-1}= g( P_i f(G_\alpha) P_i^{-1})\}\rangle.\] 
Furthermore, every diagonal entry of $P_i f(G) P_i^{-1}$ is non-zero for any non-zero $G \in V_{i, 1}$ is non-zero, because otherwise $\dim_{\F_r}(S_{i, k}) > 0$ for some $k \in [m]$, contradicting to the that $\dim_{\F_r}(S_{i, k}) = 0$ for every $k\in [m]$. 
The property 2(b) then follows.


    For the property 2(c), let $j$ be an arbitrary number in $\{2, \dots, m\}$. 
    Let $H_1$ be a matrix in $V_{i, 1}$ and $H_j$ be a matrix in $V_{i, j}$. As $P_i f(H_1) P_i^{-1}$ is a diagonal matrix by the property 2(b) and $P_i f(H_j) P_i^{-1}$ has only $j$-th column non-zero in the first row, we have $H_1 H_j$ and $H_j H_1$ are in $V_{i, j}$ because 
    \[P_i f(H_1 H_j) P_i^{-1} = P_i f(H_1)f(H_j) P_i^{-1} =  P_i f(H_1) P_i^{-1} P_i f(H_j) P_i^{-1}\] 
    and \[P_i f( H_j H_1) P_i^{-1} = P_i f(H_j)f(H_1) P_i^{-1} =  P_i f(H_j) P_i^{-1} P_i f(H_1) P_i^{-1}.\] 
    Furthermore, we have for $\tilde H_1\neq H_1\in V_{i, j}$, $\tilde H_1'H_j\neq H_1 H_j$, and $H_j\tilde H_1\neq H_j H_1$. Similarly, for $\tilde H_j \neq H_j \in V_{i, j}$, $H_1\tilde H_j\neq H_1H_j$ and $\tilde H_j'H_1\neq H_jH_1$. It follows that for any $H_j, \tilde H_j \in V_{i, j}$, there exist $H_1, \tilde H_1\in V_{i, 1}$, such that $\tilde H_j=H_1H_j=H_j\tilde H_1$. 

    Now take $T \in V_{i, j}$, and let $t_{k, k'}$ be the $k$-th row $k'$-th column of $P_i f(T) P_i^{-1}$. 
    Let $H_1 \in V_{i, 1}$ be the matrix such that $P_i f(H_1) P_i^{-1}=\diag(\lambda_i, \dots, \lambda_i^{r^{m-1}})$. 
    Then 
    $$P_i f(H_1T) P_i^{-1} =P_i f(H_1) P_i^{-1} P_i f(T) P_i^{-1} =\begin{bmatrix}
        0 & \dots & \lambda_i t_{1, j} & \dots & 0 \\
        \lambda_i^r t_{2, 1} & \dots & \lambda_i^r t_{2, j} & \dots & \lambda_i^r t_{2, m}\\
        \dots\\
        \lambda_i^{r^{m-1}} t_{m, 1} & \dots & \lambda_i^{r^{m-1}} t_{m, j} & \dots & \lambda_i^{r^{m-1}} t_{m, m}\\
        \end{bmatrix}.$$
    Take $\tilde H_1 \in V_{i, 1}$ such that $H_1T=T\tilde H_1$. Suppose $P_i f(\tilde H_1)P_i^{-1}=\diag(\mu_1, \dots, \mu_m)$. By $H_1T=T\tilde H_1$, $\mu_j=\lambda_i$. By the property 2(b), we have \[\tilde P_i f(H_1) P_i^{-1}=\diag(\lambda_i^{m - j + 1}, \dots, \lambda_i^{m - 1}, \lambda_i, \lambda_i^r, \dots, \lambda_i^{m - j}).\] 
    By $H_1T=T\tilde H_1$, for any $k\in [m]$, the $k$-th row $k'$-th column of $P_i f(T) P_i^{-1}$ is zero for all the $k'\not\equiv k+j-1 (\mathrm{mod} \  m)$. 

    Furthermore, for any $k, k' \in [m]$ such that $k' \equiv k+j-1 (\mathrm{mod} \  m)$, the $k$-th row $k'$-th column of $P_i f(T) P_i^{-1}$ must be non-zero, because otherwise $\dim_{\F_r}(S_{i, k}) > 0$, contradicting to the assumption that $\dim_{\F_r}(S_{i, k}) = 0$ for every $k\in [m]$. 
    The property 2(c) then follows. 
    
    The property 2(d) is obtained by property 2(a), 2(b), and 2(c). 

    For the property 2(e), let $G$ be an arbitrary matrix in $\cG$ such that $G$ is a rank 1 matrix. 
    Since $f$ is an isomorphism between $\cG$ and $\M(m, \F_r)$, $P_i f(G)P_i^{-1}$ is also a rank one matrix. 
    Hence, $P_i f(G) P_i^{-1}= u^T v$, where $u, v\in \F_{r^m}^m$. As $P_i f(G) P_i^{-1}$ is rank-$1$, $P_i f(G) P_i^{-1}$ has a non-zero entry. Suppose $k$-th row $k'$-th column of $P_i f(G) P_i^{-1}$ is non-zero. By the property 2(b), 2(c) and 2(d), letting 
    \[h = \left\{\begin{array}{ll} (k' - k + 1 \pmod{m}) & \text{ if } (k' - k + 1 \pmod{m}) \neq 0\\
    m & \text{ if } (k' - k + 1 \pmod{m}) = 0
    \end{array}\right.,
    \]
    the $k''$-th row $k'''$-th column of $P_i f(G) P_i^{-1}$ is non-zero for any $k'', k'''\in[m]$ satisfying $k''' \equiv k'+h-1 (\mathrm{mod} \  m)$. 
    This implies that every entry of $u$ and $v$ is non-zero, and it follows that every entry in $P_i f(G) P_i^{-1}$ is non-zero, and thus the property 2(e) holds.

    Now we prove the property 2(f). 
    We first show that every rank 1 matrix $H$ in $\cG$ is determined by $U_{H, 1}$ and $U_{H, 2}$, and satisfies the condition that the product of non-zero entries of $P_i f(U_{H, 1})P_i^{-1}$ is equal to the product of non-zero entries of $P_i f(U_{H, 2})P_i^{-1}$, where $U_{H, 1}, \dots, U_{H, m}$ is matrices defined by the property 2(d). 
    Let $P_i f(U_{H, 1})P_i^{-1} = \diag(\theta_1, \dots, \theta_m)$ and         
        $$P_i f(U_{H,2})P_i^{-1} = \begin{bmatrix}
        0 & \kappa_1 & 0 & \dots  & 0 \\
        0 & 0 & \kappa_2 & \dots & 0 \\
        \vdots & \vdots & \ddots & \ddots & \vdots \\
        0 & 0 & \dots & 0 & \kappa_{m-1}\\
        \kappa_m & 0 & \dots & 0 & 0 \\
        \end{bmatrix}.$$ 

        We show that $H$ is fully determined by $U_{H, 1}$ and $U_{H, 2}$. 
        For $P_i f(H)P_i^{-1}$ to be a rank-$1$ matrix, there must exist $u=(u_1, \dots, u_m)$, $v=(v_1, \dots, v_m)\in \F_{r^m}^m$, such that $P_i f(H)P_i^{-1}=u^T v$. Therefore, $\theta_k=u_k v_k$, and $\kappa_k=u_k v_{(k \mod m) + 1}$. Then the ratios $u_1/u_m, u_2/u_1, \dots, u_m/u_{m - 1}$, and $v_1/v_2, v_2/v_3, \dots, v_m/v_1$, are all determined by $\theta_1, \dots, \theta_m$ and $\kappa_1, \dots, \kappa_m$. 
        Note that 
        \[\frac{\kappa_1}{\theta_1}\cdot \ldots \cdot \frac{\kappa_m}{\theta_m} =\frac{v_2}{v_1}\cdot \ldots \cdot \frac{v_1}{v_m} = 1.\]
        We have the product of non-zero entries of $P_i f(U_{H, 1})P_i^{-1}$ is equal to the product of non-zero entries of $P_i f(U_{H, 2})P_i^{-1}$.

        For any $k$ and $k'$ such that $k'>k+1$, the $k$ -th row $k'$-th column of $P_i f(H)P_i^{-1}$ is 
        \[u_kv_{k'}=u_kv_{k+1}\cdot (v_{k+2}/v_{k+1})\cdot\ldots\cdot(v_{k'}/v_{k'-1})=\kappa_k\cdot (v_{k+2}/v_{k+1})\cdot\ldots\cdot(v_{k'}/v_{k'-1}).\] For $k'<k$, the $k$-th row $k'$-th column of $P_i f(H)P_i^{-1}$ is 
        \[u_k v_{k'}=u_kv_k\cdot (v_{k-1}/v_k)\cdot\ldots\cdot (v_{k'}/v_{k'+1})=\theta_k\cdot (v_{k-1}/v_k)\cdot\ldots\cdot (v_{k'}/v_{k'+1}).\] Therefore, all other entries of $P_i f(H)P_i^{-1}$ are determined by $\theta_k$ and $\kappa_k$. This shows that such $H$ is determined by $U_{H, 1}$ and $U_{H, 2}$. 

        On the other hand, note that the number of $(H_1, H_2)\in V_{i, 1}\times V_{i, 2}$ satisfying products of non-zero entries of $P_i f(H_1)P_i^{-1}$ and $P_i f(H_2)P_i^{-1}$ being the same is $(r^m-1)\cdot(r^m-1)/(r-1)$, which matches the number of rank-$1$ matrices in $\M(m, \F_r)$. This can be seen by the property 2(a) and the fact that the product of non-zero entries of $P_i f(H)P_i^{-1}$
        for any $H\in V_{i, j}$ is in $\F_r$ for every $j \in [m]$. 
        
        Finally, the running time to compute a matrix $E$ based on given $H_1$ and $H_2$ can be done in $\poly(n, \log q)$ time by $q^\ell = r^{m^2}$ using the fact that $f$ is an isomorphism.  Then the property 2(f) holds.
\end{proof}

Now we are ready to show the Full Matrix Algebra Algorithm finds a singular matrix canonically.

\begin{claim}\label{claim:full_matrix_algebra_main}
    The following properties for the Full Matrix Algebra Algorithm hold:
    \begin{enumerate}
    \item For a matrix algebra $\cG$ with a linear basis $G_1, \dots, G_\ell \in \GL(n, \F_q)$ and an isomorphism $f: \cG \rightarrow \M(m, \F_{r})$ for some $m > 1$ such that for each $i \in [\ell]$, $f(G_i)$ is 
    either $c \cdot I$ for some $c\in \F_r$ or a matrix with an irreducible characteristic polynomial,
    in $\poly(n, \log q)$ time, the output is a matrix $G \in \cG$ that is either not full rank or has a reducible characteristic polynomial over $\F_r$. 
    \item For different executions of the algorithm with the same input, the output is fixed. 
    \item The algorithm is canonical in the following sense: 
    Let $\cG$ with linear basis $G_1, \dots, G_\ell$ and isomorphism $f: \cG \rightarrow \M(m, \F_r)$ be one input, and $\cG'$ with linear basis $G_1', \dots, G_\ell'$ and isomorphism $f': \cG' \rightarrow \M(m', \F_{r'})$ be another input to the Full Matrix Algebra Algorithm such that there exist an invertible matrix $L$ such that $G_i' = L G_i L^{-1}$ for each $i \in [\ell]$. 
    For any invertible matrix $L$ such that $G_i' = L G_i L^{-1}$ for each $i \in [\ell]$, 
    $G' = L G L^{-1}$, where $G$ and $G'$ are the outputs by the algorithm for $G_1, \dots, G_\ell, f$ and $G_1', \dots, G_\ell', f'$ respectively. 
    
    \end{enumerate}
\end{claim}

\begin{proof}
Since $f$ is an isomorphism between $\cG$ and $\M(m, \F_r)$, $\cG$ has the same number of matrices as $\M(m, \F_r)$, 
we have $\ell \leq n^2$, and $q^\ell = r^{m^2}$, and thus $m = O(\poly(n, \log q))$ and $\log r = O(\poly(n, \log q))$.

    Since $m > 1$, not all the matrices in $\M(m, \F_r)$ is $c \cdot I$ for some $c \in \F_r$. 
Hence, $\alpha$ in Step 1 exists, and $G_\alpha$ has an irreducible characteristic polynomial. 
Since $G_\alpha$ has an irreducible polynomial, by Fact~\ref{fact:eigenval}, $\charpoly(f(G_\alpha))$ has distinct roots over $\F_{r^m}$, and $\lambda_1, \dots, \lambda_m$ defined by Step 2 of the algorithm are the distinct roots of $\charpoly(f(G_\alpha))$. By Theorem~\ref{thm:factoring}, Step 2 can be done in $\poly(n, \log q)$ time. 

For each $i$ of Step 3, Step 3(a), 3(b), and 3(c) take $\poly(n, \log q)$ time. 
For Step 3(d), by Theorem~\ref{thm:factoring}, it takes $\poly(n, \log q)$ time to determine whether the characteristic polynomial is an irreducible polynomial. 
To compute $A_{i, 1}, \dots, A_{i, m}$, we only need to compute the roots of the characteristic polynomial of $B_i$ over $\F_{r^m}$ and find all the matrices $A\in V_{i, 1}$ such that 
the first row first column of $P_i f(A) P_i^{-1}$ equal to a root of the characteristic polynomial of $B_i$ for any $k\in[m]$. By Lemma~\ref{lem:diagonal_new}, there are $m$ such matrices and they can be computed in $\poly(n, \log q)$ time by Theorem~\ref{thm:factoring}.
By Lemma~\ref{lem:diagonal_new}, $D_{i, j}$ for each $j \in [m]$ is unique for fixed $A_{i, j}$ and $B_i$, and can be computed in $\poly(n, \log q)$ time. 

By Claim~\ref{claim:subspace_linear_basis} and Lemma~\ref{lem:diagonal_new}, the output of the algorithm is either not full rank or has a reducible characteristic polynomial over $\F_r$, and the running time of the algorithm is $\poly(n, \log q)$. Thus the first property holds. 

Now we prove the second property. 
We first show that for any $i \in [m]$, the different choices of $P_i$
does not change the output. 
Let $Y$ be an arbitrary invertible diagonal matrix in $\GL(m, \F_{r^m})$, we have
\[\langle \{G \in \cG: e_1 P_i f(G) P_i^{-1} = 0\}\rangle = \langle \{G \in \cG: e_1 Y P_i f(G) P_i^{-1} Y^{-1} = 0\}\rangle\]
and 
\[\begin{split}
& \langle \{G \in \cG :  e_1 P_i f(G) P_i^{-1} = b \cdot e_j \text{ for some } b \in \F_{r^m}\}\rangle \\
= & \langle \{G \in \cG :  e_1 Y P_i f(G) P_i^{-1} Y^{-1} = b \cdot e_j \text{ for some } b \in \F_{r^m}\}\rangle\\
\end{split}\]
Hence $S_{i, j}$ and $V_{i, j}$ do not change for all the $j \in [m]$, and thus $B_i, A_{i, 1}, \dots, A_{i, m}, C_{i, 1}, \dots, C_{i, m}$ do not change by different choices of $P_i$ by Claim~\ref{claim:subspace_linear_basis}. 
Together with Fact~\ref{fact:eigenval}, the second property holds. 

Now we prove the third property. Since $f$ is an isomorphism  from $\cG$ to $\M(m, \F_r)$, $f'$ is an isomorphism from $\cG'$ to $\M(m', \F_r)$, and there exists an invertible matrix $L$ such that $G_i' = L G_i L^{-1}$ for any $i\in [\ell]$. 
Then $g : \M(m, \F_r) \rightarrow \M(m', \F_{r'})$ such that $g(A) = f'(L^{-1} f^{-1}(A) L)$ for any $A \in \M(m, \F_r)$ is an isomorphism from $\M(m, \F_r)$ to $\M(m',\F_{r'})$. Hence, $m = m'$, $r = r'$, and $g$ is an automorphism of $\M(m, \F_r)$ up to an automorphism of $\F_r$.

Let $h$ denote an isomorphism from the finite field $\F_r$ for $f$ to the finite field $\F_{r'}$ for $f'$.

For Step 1 of the algorithm, for any matrix $A$, 
$A = c I$ for some $c \in \F_r$ if and only if $g(A)$ is $c' I$ for some $c' \in \F_{r'}$. 
Hence, the integer $\alpha$ selected in Step 1 of the algorithm is the same for $\cG$ and $\cG'$.

For Step 2 of the algorithm,  $\charpoly(f(G_\alpha))$ and $\charpoly(f'(G_\alpha'))$ are same up to the isomorphism $h$.
Let $\lambda$ and $\lambda'$ be the roots selected for $\cG$ and $cG'$ respectively. 
By Fact~\ref{fact:eigenval}, there is an integer $s \in \{0, 1, \dots, m-1\}$ such that $\lambda' = (h(\lambda))^{r^{s}}$. 

For Step 3, let $i$ and $i'$ be for loop variables for $\cG$ and $\cG'$ respectively such that $\lambda_{i'}$ for $\cG'$ equals $h(\lambda_i)$ with $\lambda_i$  for $\cG'$. 
Thus, we have $S_{i', j}' = L S_{i, j} L^{-1}$ and $V_{i', j}' = L V_{i, j} L^{-1}$ for any $j\in [m]$. And consequently, $B_i' = L B_iL^{-1}$, and there is a bijection $\rho : [m]\rightarrow [m]$ such that $A_{i', j}' = L A_{i, \rho(j)} L^{-1}$ and  $C_{i', j}' = L C_{i, \rho(j)} L^{-1}$ for any $j\in[m]$. 
By Claim~\ref{claim:subspace_linear_basis}, the third property holds.
\end{proof}

Now we give the algorithm for Lemma~\ref{lemma:conjugate_algebra_main}. 
Roughly speaking, we consider three cases of the algorithm: (1) if $\cG$ is not semisimple, then $\cG$ has a nontrivial radical, and we find a canonical singular matrix in the radical; (2) if the algorithm is a direct sum of at least two components, then the algorithm find a canonical singular matrix by comparing the identities of all the components; (3) if $\cG$ is isomorphic to a full matrix algebra $\M(m, \F_r)$, then we use the Full Matrix Algebra to handle the case of $m > 2$, or certify $\cG$ is isomorphic to a finite field.

\begin{framed}
\noindent \textbf{Conjugation Matrix Algebra Algorithm}

\noindent \textbf{Input:} Matrix tuple $\vG = (G_1, \dots, G_\ell) \in \GL(n, \F_q)^\ell$.

\noindent \textbf{Output:} 
Matrix $C \in \Env(\vG)$.

\begin{enumerate}
\item Run the Linear Basis Algorithm to compute a linear basis of $\cG = \Env(\vG)$, and denote the output as a matrix tuple $\vG^\dagger = (G_1^\dagger, \dots, G_\tau^\dagger) \in \M(n, \F_q)^\tau$ for some integer $\tau$. 
\item For $i = 1, \dots, \tau$, if $G_i^\dagger$ is not full rank, then return $G_i^\dagger$.

\item Run the  Matrix Algebra Structure Algorithm for $\cG$.
\begin{enumerate}
\item If the output is the nontrivial radical $\rad(\cG)$, then 
return the row vector space of the first matrix in the output of 
the Subspace Linear Basis Algorithm for $\rad(\cG)$ with $G_1^\dagger, \dots, G_\tau^\dagger$. 

\item 
If the output is a direct sum decomposition of $\cG = \cC_1 \oplus \dots \oplus \cC_s$ for some $s > 1$, then 
compute the identity $I_i$ of $\cC_i$ for each $i\in [s]$ and return the lexically smallest matrix among $I_1, \dots, I_s$ with respect to basis $G_1^\dagger, \dots, G_\tau^\dagger$. 

\item If the output is an isomorphism $f : \cG \rightarrow \M(m, \F_r)$ for some $m$ and $r$, then 
\begin{enumerate}
    \item For $i = 1, \dots, \tau$, if $1 < \dim(\langle \{ G \in \cG : G G_i^\dagger = G_i^\dagger G\}\rangle) < \tau$, then return the output of the Conjugation Matrix Algebra Algorithm with input being the output of the Subspace Linear Basis Algorithm for $\langle \{ G \in \cG : G G_i^\dagger = G_i^\dagger G\}\rangle$ with $G_1^\dagger, \dots, G_\tau^\dagger$.  
    \item If $m > 1$, then 
run the Full Matrix Algebra Algorithm with $\cG$ using $G_1^\dagger, \dots, G_\tau^\dagger$ as a linear basis and $f$, and denote the output as $C$. If $C$ is not full rank, then return $C$, otherwise, return the output of the Conjuration Matrix Algebra Algorithm with input being the output of the Subspace Linear Basis Algorithm for $\langle \{ G \in \cG: G C = C G\} \rangle$ with $G_1^\dagger, \dots, G_\tau^\dagger$.

\item Otherwise, return the zero matrix in $\cG$. 
\end{enumerate}
\end{enumerate}
\end{enumerate}
\vspace{-.3cm}
\end{framed}

\begin{proof}[Proof of Lemma~\ref{lemma:conjugate_algebra_main}]
We first show that the algorithm either outputs a desirable matrix in $\cG$ or recurses with a subalgebra of $\cG$ with a dimension smaller than $\cG$ such that the subalgebra contains some matrix that is not full rank.
Then, the first property is followed by an induction on the dimension of $\cG$. 

By Proposition~\ref{prop:linear_basis}, $\vG^\dagger$ is a linear basis of $\cG = \Env(\vG)$, and thus every matrix in the matrix tuple $\vG^\dagger$ is a linear combination of matrices in $\vG$. 

If the algorithm outputs a matrix in Step 2, then the output is a non-zero matrix that is no full rank.

For Step 3, by Theorem~\ref{thm:algebra_struct}, one of Step 3(a), 3(b), and 3(c) applies. 
If the output of the Matrix Algebra Structure Algorithm is a nontrivial radical $\rad(\cG)$, then by the definition of radical for a matrix algebra, $\rad(\cG)$ is a linear subspace of $\langle \{G_1^\dagger, \dots, G_\tau^\dagger  \} \rangle$ such that every matrix in the subspace is not full rank. 
Hence, the output obtained in Step 3(a) is a desirable non-zero matrix that is not full rank.

If the output of the Matrix Algebra Structure Algorithm is a direct sum decomposition of $\cG = C_1 \oplus \dots C_s$ for some $s > 1$,
then the identity for each component in the direct sum decomposition is a matrix that is not full rank. 
Hence, the  lexically smallest identity among the identities of each component in the direct sum decomposition is a desirable matrix.

Consider the case that $\cG$ is isomorphic to $\M(m, \F_r)$ for some $m$ and $r$ with an isomorphism mapping $f : \cG \rightarrow \M(m, \F_r)$. 
If a matrix $G$ in $\cG$ is full rank and has a nontrivial centralizing algebra, i.e., $\centre_\cG(G)=\{G' \in \cG\mid G'G=GG'\}$, then by the following fact, 
$\centre_\cG(G)$ is a nontrivial subalgebra of $\cG$ containing some non-zero matrices not full rank because $\centre_{\M(m, \F_r)}(f(G))$ has a reducible characteristic polynomial by the following fact, and thus, we can recurse with $\centre_\cG(G)$. 
\begin{fact} Let $C$ be a matrix in the matrix algebra $\cC=\M(m, \F_r)$ for some $m > 1$ and $r$. Then 
\begin{enumerate}
    \item $\centre_\cC(C)=\cC$ if and only if $C=c I$ for some $c\in \F_r$. 
    \item $\centre_\cC(C)$ is isomorphic to a field if and only if $\charpoly(C)$ is irreducible.
\end{enumerate}
\end{fact}

For the case of 3(c)ii, by Claim~\ref{claim:full_matrix_algebra_main}, either we directly get a desirable matrix, or we can recurse with a non-trivial subalgebra of $\cG$ containing some non-zero matrix that is not full rank. 

For the case of 3(c)iii, the $\cG$ is isomorphic to a finite field. So, we output a zero matrix as desired. 

Hence, the algorithm either outputs a desirable matrix or recurses with a subalgebra of $\cG$ with a dimension smaller than $\cG$ such that the subalgebra contains some matrix that is not full rank.
By an induction on the dimension of $\cG$, the first property follows. 

The second property of the lemma is obtained by Claim~\ref{claim:subspace_linear_basis}, Claim~\ref{claim:full_matrix_algebra_main}, and an induction on each step of the algorithm. 
\end{proof}

\subsection{\IBCtuple construction for finite field matrix tuples}\label{sec:field}

We give an algorithm to compute an \IBCtuple (see Definition~\ref{def:ibctuple}) with respect to an arbitrary row tuple of a matrix tuple $\vA=(A_1,  \dots, A_\ell)\in\GL(n, \F_q)^{\ell}$ if the matrix algebra generated by $A_1^{-1}A_2, A_1^{-1}A_3, \dots, A_1^{-1}A_\ell$ is isomorphic to a finite field. 




\begin{framed}
\noindent \textbf{Finite Field Single \IBCtuple Algorithm}

\noindent \textbf{Input:} $\vA=(A_1, A_1, \dots, A_\ell)\in\GL(n, \F_q)^\ell$ and a row tuple $\va \in \rowtuplespace(\vA)$. 

\noindent \textbf{Output:} \IBCtuple $\vB_\va$ for $\va$.

\noindent \textbf{Algorithm:}
\begin{enumerate}
    \item Set $t = 1$, $\va_1 = \va$, and $H$ be the linear space spanned by $\va_1^{(1)}$. 
    \item For $i=1$ to $t$,
    \begin{enumerate}
        \item for $j = 1$ to $\ell$, if $\va_i^{(j)} \notin H$
        \begin{enumerate}
            \item Set $t = t + 1$, and let $\va_{t}$ be the row tuple of $\vA$ such that $\va_{t}^{(1)} = \va_{i}^{(j)}$. 
            
    \item Let $H$ be $\langle H \cup \{\va_{t}^{(1)}\}\rangle$. 
    \end{enumerate}
    \end{enumerate}
    \item Return \[\vB_{\va} = \begin{bmatrix}
        \va_1 \\ \vdots \\ \va_{t}
    \end{bmatrix}.\] 
\end{enumerate}
\vspace{-.3cm}
\end{framed}

We first prove some useful observations for the algorithm to compute a single \IBCtuplenospace. 
\begin{claim}\label{claim:left-right-to-conjugation-first}

We have the following observations for the Finite Field Single \IBCtuple Algorithm for a matrix tuple for a matrix tuple $\vA=(A_1,  \dots, A_\ell)\in\GL(n, \F_q)^{\ell}$ if the matrix algebra generated by $A_1^{-1}A_2, A_1^{-1}A_3, \dots, A_1^{-1}A_\ell$ is isomorphic to a finite field: 
\begin{enumerate}
    \item For any input $\va \in \rowtuplespace(\vA)$ and any $i, j \in [\ell]$, 
    \[\left\langle \left\{ \va_{k}^{(i)} : k\in [t]\right\}\right\rangle = \left\langle \left\{ \va_{k}^{(j)} : k\in [t]\right\}\right\rangle.\]
    \item $t = \dim(\cG)$, where $\cG$ is the matrix algebra generated by $A_1^{-1} A_2, \dots, A_1^{-1}A_\ell$. 
    \item For any $\va \in \rowtuplespace(\vA)$ and any non-zero row tuple $\vc \in \rowtuplespace(\vB_{\va})$, $\rowtuplespace(\vB_{\vc}) = \rowtuplespace(\vB_\va)$. 
    \item     For any two non-zero row tuples $\va, \va' \in \rowtuplespace(\vA)$. 
    Let $\va_1, \dots, \va_{t}$ and $\va_1', \dots, \va'_{t'}$ be the row tuples selected by the algorithm for $\va$ and $\va'$ respectively. Then the following conditions hold:
    \begin{enumerate}
    \item $t = t'$. 
    \item 
    For any $2 \leq k \leq t$, 
    if $\va_k$ is obtained for $i$ and $j$ in Step 2 of the algorithm for $\vA$, then $\va_k'$ is obtained for the same $i$ and $j$ in Step 2 of the algorithm for $\vA'$. 
    
    \item Let $R$ be an arbitrary matrix such that $(\va'_k)^{(1)} = \va_k^{(i)} R$ for any $k \in [t]$, then $\va_k' = \va_k R$ for any $k \in [t]$. 
    \item 
    The output of the algorithm with input $\va + \va'$ is $\vB_{\va} + \vB_{\va'}$. 
    \end{enumerate}
\end{enumerate}
\end{claim}
\begin{proof}
For the first property, let $H_s$ be the linear span of $\{\va_1^{(s)},\dots, \va_t^{(s)}\}$ for any $s \in [\ell]$, and $U$ be the linear span of $H_1 \cup \dots \cup H_\ell$. 
By the algorithm, $H_1 = U$. On the other hand, $H_s\subseteq U$ for any $s \in [\ell]$. We then observe that $\dim(H_s)=t$ for any $s \in [\ell]$, as $A_1, \dots, A_\ell$ are invertible. It follows that $H_s = U$ for every $s \in [\ell]$. 
Thus, the first property holds. 

To prove the second property, let $B_s = A_{1}^{-1}A_{s + 1}$ for any $s \in [\ell- 1]$.
We observe that for any non-zero row tuple $\va \in \rowtuplespace(\vA)$, \[\va^{(1)} B_s = \va^{(s + 1)}\] for any $s \in [\ell - 1]$. 
By the algorithm, 
for each $2 \leq u \leq t$, 
there is a matrix $Q_u$, which is a monomial of $B_1, \dots, B_s$, such that $\va_u^{(1)} = \va_1^{(1)} Q_u$. 
On the other hand, every monomial $Q$ of $B_1, \dots, B_s$ must be a linear combination of $I, Q_2, \dots, Q_u$ because otherwise there exist some $i \in [t], j\in [\ell]$ such that $\va_i^{(j)}$ is not a linear combination of $\va_1^{(1)}, \dots, \va_t^{(1)}$, contradicting to the algorithm. 
Thus, $t$ equals the dimension of the matrix algebra generated by $B_1, \dots, B_s$, and the second property holds. 

For the third property, we observe that $\rowvectorspace(\vB_\vc)$ is a subspace of $\rowvectorspace(\vB_\va)$ by the algorithm. Consequently, $\rowtuplespace(\vB_\vc)$ is a subspace of $\rowtuplespace(\vB_\va)$. 
On the other hand, the second property implies that 
\[\dim(\rowtuplespace(\vB_\va)) = \dim(\rowtuplespace(\vB_\vc)).\] 
Hence, we have \[\rowtuplespace(\vB_{\vc}) = \rowtuplespace(\vB_\va).\] 

Now we prove the last property. 
The property 4(a) is obtained by the second property.
The property 4(b) is obtained by induction on each $i$ and $j$ for Step 2 of the algorithm with respect to $\va$ and $\va'$. The property 4(c) is obtained by property 4(b) and the fact that $\vc^{(1)} B_s = \vc^{(s + 1)}$ for any $s \in [\ell - 1]$ for any $\vc \in \rowtuplespace(\vA)$. The property 4(d) is also obtained by 4(b). 
\end{proof}

We prove the output of the algorithm is an \IBCtuple of the input matrix tuple and has the input row tuple as the first row of the output \IBCtuplenospace. Furthermore, we also show that the output is canonical under the left-right action of the input matrix tuple.

\begin{lemma}\label{lem:single_ibctuple_finite_field}
The following properties hold for the Finite Field Single \IBCtuple Algorithm:

    \begin{enumerate}
        \item 
        For an input of the Finite Field Single \IBCtuple Algorithm, which contains a matrix tuple $\vA = (A_0, A_1, \dots, A_\ell)\in\GL(n, q)^{\ell}$ such that $A_1^{-1}A_2, \dots, A_1^{-1}A_\ell$ generate a matrix algebra that is isomorphic to a finite field, and a row tuple $\va \in \rowtuplespace(\vA)$, 
        the algorithm outputs an \IBCtuple $\vB_\va$ of $\vA$ such that $e_1 \vB_\va = \va$ in $\poly(n, \log q)$ time.

        \item The algorithm is canonical in the following sense: Let $\vA', \va'$ be another input of the algorithm such that there exist invertible matrices $L$ and $R$ such that $\vA' = L \vA R^{-1}$ and $\va' = \va R^{-1}$, then 
        for any invertible matrices $L$ and $R$ such that $\vA' = L \vA R^{-1}$ and $\va' = \va R^{-1}$, 
        Then $\vB_{\va'} = \vB_{\va} R^{-1}$, where $\vB_{\va}$ and $\vB_{\va'}$ are the outputs for $\vA, \va$ and $\vA', \va'$ respectively. 
    \end{enumerate}
\end{lemma}

\begin{proof}
    The running time of the algorithm is clearly $\poly(n, \log q)$. So to prove the first property, we only need to show the output is an \IBCtuple for $\vA$. 
    Let $\sigma$ denote the number of rows in $\vB_\va$.
    Consider the following subroutine to selected row tuples:
    \begin{enumerate}
        \item $i = 1$, $\vc_1 = \va$, and compute $\vB_{\vc_1}$ by the  Finite Field Single IBC-Tuple Algorithm, $W_1 = \rowvectorspace(\vB_{\vc_1})$. 
        \item Repeat the following procedure until $W_i = \rowvectorspace(\vA)$
        \begin{enumerate}
            \item $i = i + 1$. Arbitrarily select a row tuple $\vc_i \in \rowtuplespace(\vA)$ such that $\vc_i^{(1)}$ is not in $W_{i-1}$.
            \item Compute $\vB_{\vc_i}$ by the  Finite Field Single IBC-Tuple Algorithm with $\vc_i$. 
            \item Let $W_i = \langle W_{i-1} \cup \rowvectorspace(\vB_{\vc_i})\rangle$.
        \end{enumerate}
    \end{enumerate}
    We observe that for the above subroutine, $\vB_{\vc_j}$ satisfies the condition that $\rowvectorspace(\vB_{\vc_j}) \cap W_{j - 1}$ contains only the zero row vector for any $2 \leq j \leq i$, because if there is a non-zero vector in $\rowvectorspace(\vB_{\vc_j}) \cap W_{j - 1}$, then every row vector in $\rowvectorspace(\vB_{\vc_j})$ is in $W_{j - 1}$ by Claim~\ref{claim:left-right-to-conjugation-first}, contradicting to the choice of $\vc_j$. 
    By the subroutine above, we have \[\rowvectorspace(\vA) = \rowvectorspace(\vB_{\vc_1}) \oplus \dots \oplus \rowvectorspace(\vB_{\vc_i}).\]
    By Claim~\ref{claim:left-right-to-conjugation-first}, we have $i = n / \sigma$, and 
    \[\rowtuplespace(\vA) = \rowtuplespace(\vB_{\vc_1}) \oplus \dots \oplus \rowtuplespace(\vB_{\vc_{i}}).\]
    By Fact~\ref{fact:ibctuple_diagonalzation_most_basic}, there is a characteristic block-diagonalization of $\vA$ with $n / \sigma$ blocks such that every $\vB_{\vc_j}$ corresponds to a block in the block-diagonalization. 
    On the other hand, by Claim~\ref{claim:left-right-to-conjugation-first}, every block in the block-diagonalization is indecomposable, so the first property holds.  

    For the second property, 
    we observe that the outputs for $\vA, \va$ and $L \vA, \va$ are the same for any invertible matrix $L$. 
    And for any invertible matrix $R$, the output for $\vA R^{-1}, \va R^{-1}$ is equal to the output for $\vA, \va$ multiplying $R^{-1}$. Hence, the second property holds.
\end{proof}
\subsection{Algorithm for full rank square matrix tuples}

We give the algorithm for Lemma~\ref{lem:conjugation_canonical_main}. 
For an input $\vA = (A_1, \dots, A_\ell)$, 
Our algorithm first runs the Conjugation Matrix Algebra Algorithm for the matrix algebra generated by $A_1^{-1} A_2, \dots, A_1^{-1}A_\ell$. By Section~\ref{sec:conjugation_sub}, if the matrix algebra is not isomorphic to a finite field, then the output of our algorithm is a non-zero singular matrix in the matrix algebra, and the row vector space of this singular matrix is a characteristic block-compatible row vector subspace of $\vA$. We further construct a characteristic block-compatible row tuple subspace using the row vector subspace. 
Otherwise, the matrix algebra is isomorphic to a finite field. Then we use the algorithm given in Section~\ref{sec:field} to compute an \IBCtuple of $\vA$, and compute its \IBCtuple space and \IBCtuple space kernel of the input matrix tuple.

\begin{framed}
\noindent \textbf{Full Rank \IBCtuple Space Algorithm}

\noindent \textbf{Input:} Matrix tuple $\vA = (A_1, \dots, A_\ell) \in \GL(n, \F_q)^\ell$ such that $A_i$ is full rank for every $i\in[\ell]$.

\noindent \textbf{Output:} 
\IBCtuple $\vB$ of $\vA$ with \IBCtuple space $\ibcspace_\vB$, \IBCtuple space kernel $\ibcspacekernel_\vB$, row tuple subspace $T_1 = \rowtuplespace(\vA)$, and row tuple space parameters $t_{\vB, 1}, \dots, t_{\vB, \sigma}$, or 
row vector subspace $S < \rowvectorspace(\vG)$.

\noindent \textbf{Algorithm:} 

\begin{enumerate}
\item Run the Conjugation Matrix Algebra Algorithm with $(A_1^{-1} A_2, \dots, A_1^{-1}A_\ell)$. If the output is a non-zero matrix $C$, then return $\langle \{ \va \in \rowtuplespace(\va) : \va^{(1)} \in \rowvectorspace(C)\}\rangle$.
\item Arbitrarily select a linear basis $\va_1, \dots, \va_n$ of $\rowtuplespace(\vA)$, and run the Finite Field Single \IBCtuple Algorithm for each $\va_1, \dots, \va_n$. 
\item Let $\vB = \vB_{\va_1}$, $\ibcspace_\vB = \langle \{\vB_{\va_1}, \dots, \vB_{\va_n}\}\rangle$, $\ibcspacekernel_\vB$ be the subspace of $\ibcspace_\vB$ that contains only the zero matrix tuple. 
\item Let $T_1 = \rowtuplespace(\vA)$, $\sigma$ be the number of rows of $\vB$, and $t_{\vB,1} = \dots = t_{\vB_\sigma} = 1$. 
\item Return $\vB, \ibcspace_\vB, \ibcspacekernel_\vB, T_1, t_{\vB, 1}, \dots, t_{\vB_\sigma}$.
\end{enumerate}
\vspace{-.3cm}
\end{framed}

\begin{proof}[Proof of Lemma~\ref{lem:conjugation_canonical_main}]

For the first property of the lemma, 
the correctness of the algorithm is obtained by Fact~\ref{fact:ibctuple_invariant_basic}, Lemma~\ref{lemma:conjugate_algebra_main}, Claim~\ref{claim:left-right-to-conjugation-first}, and Lemma~\ref{lem:single_ibctuple_finite_field}. 
The running time of the algorithm is obtained by Lemma~\ref{lemma:conjugate_algebra_main} and Lemma~\ref{lem:single_ibctuple_finite_field}. 
The space property, row tuple space property, dimension property, and block-compatible property of the \IBCtuples are obtained by Claim~\ref{claim:left-right-to-conjugation-first}. 

For the second property of the lemma, if the output of the algorithm is a row tuple subspace, then the second property is obtained by Lemma~\ref{lemma:conjugate_algebra_main}. 
Otherwise, let $\va_1, \dots, \va_n$ and $\va_1', \dots, \va_n'$ are the row tuples selected in Step 2 of the algorithm for $\vA$ and $\vA'$ respectively. 
Otherwise, by Lemma~\ref{lem:single_ibctuple_finite_field}, the \IBCtuples output by the algorithm for $\vA$ and $\vA'$ are right-equivalent. 
Observe that the \IBCtuple space output by the algorithm is independent of the choice of the row tuples selected in Step 2 of the algorithm. 
Let $\vB$ and $\vB'$ be the \IBCtuples output by the algorithm for $\vA$ and $\vA'$ respectively. 
We have 
\[\ibcspace_{\vB'} = \langle \{\vB_{\va_1'}, \dots, \vB_{\va_n'}\}\rangle 
= \langle \{\vB_{\va_1 R^{-1}}, \dots, \vB_{\va_n R^{-1}}\}\rangle = \{ \vC R^{-1} : \vC \in \ibcspace_{\vB}\}.\]
Since both $K_\vB$ and $K_{\vB'}$ contain only the zero matrix tuple by the algorithm, we also have $\ibcspacekernel_{\vB'} = \{\vC R^{-1} : \vC \in \ibcspacekernel_{\vB}\}$. 
Hence, the second property of the lemma holds. 
\end{proof}

\section{\IBCtuples from direct sum row tuple decomposition}\label{sec:direct_sum_decomposition}

In this section, we extend the result obtained in Section~\ref{sec:conjugation} to the case where the matrix tuple, which is not necessarily square and where every matrix in the tuple is full rank, is associated with a hierarchical row tuple decomposition such that the decomposition is a direct sum of row tuple space.
Specifically, we prove the following lemma in this section.

\begin{lemma}\label{lem:direct_sum_ibctuple_algo}
    There is a Direct Sum Decomposition Algorithm that satisfies the following conditions:
    \begin{enumerate}
        \item 
            Given a matrix tuple $\vA$, a depth-0 hierarchical row tuple decomposition $T_1, \dots, T_\zeta$ of $\vA$ such that $\rowtuplespace(\vA) = T_1 \oplus \dots \oplus T_\zeta$,  $\rowvectorspace(T_1) = \rowvectorspace(\vA)$, and $\dim(T_1) = \dots = \dim(T_\zeta)$, 
        the output of the algorithm satisfies one of the following conditions:
        \begin{enumerate}
            \item The output is a nontrivial characteristic block-compatible row tuple subspace $S < T_1$. 
            \item The output is an \IBCtuple $\vB$ of $\vA$ with $\ibcspace_{\vB}$,  $\ibcspacekernel_\vB$, and $t_{\vB, 1}, \dots, t_{\vB, \eta}$, where $\eta$ is the number of rows of $\vB$, satisfying the following conditions:
            \begin{enumerate}
                \item Every \IBCtuple of $\vA$ is equivalent to $\vB$. 
                \item $\vB$ satisfies the space property with \IBCtuple space $\ibcspace_\vB$ and \IBCtuple space kernel $\ibcspacekernel_{\vB}$ such that $\ibcspacekernel_{\vB}$ contains only the zero matrix tuple.
                \item $\vB$ satisfies row tuple space property with parameters $t_{\vB, 1},\dots, t_{\vB_\eta}$ such that each of $K_1, \dots, K_\zeta$ contains only the zero row tuple.
                \item $\vB$ satisfies the dimension property and block-compatible property.  
            \end{enumerate}
        \end{enumerate}
        \item The algorithm is canonical in the following sense: Let $\vA' = (\vA_1', \dots, \vA_\ell')$ be a matrix tuple and $T_1', \dots, T_\zeta'$ be a depth-0 hierarchical row tuple decomposition $T_1, \dots, T_\zeta$ of $\vA'$ such that there exist invertible matrices $L$ and $R$ satisfying $\vA' = L \vA R^{-1}$ and $T_i' = T_i R^{-1}$ for any $1 \leq i \leq \zeta$. 
        Then, the output for $\vA$ and $\vA'$ satisfies the following conditions: (Let $L$ and $R$ be arbitrary invertible matrices such that $\vA' = L \vA R^{-1}$ and $T_i' = T_i R^{-1}$ for any $i \in [\zeta]$.)
        \begin{enumerate}
            \item If the output for $\vA$ is a subspace $S < T_i$, then the output for $\vA'$ is $S R^{-1}$. 
            \item  If the output for $\vA$ is an \IBCtuple $\vB$ with $\ibcspace_\vB$ and $\ibcspacekernel_{\vB}$, then
        the output for $\vA'$ is an \IBCtuple $\vB'$ with $\ibcspace_{\vB'}$ and $\ibcspacekernel_{\vB'}$ satisfying the following conditions:
        \begin{enumerate}
        \item $\vB'$ is right-equivalent to $\vB$. \item 
        $\ibcspace_{\vB'} = \{\vC \in \ibcspace_\vB: \vC R^{-1}\}$.
        \item 
        $\ibcspacekernel_{\vB'} = \{\vC \in \ibcspacekernel_\vB: \vC R^{-1}\}.$
        \end{enumerate}
        \end{enumerate}

    \end{enumerate}
\end{lemma}

The high-level idea for our algorithm is to construct a new matrix tuple satisfying the condition of Lemma~\ref{lem:conjugation_canonical_main} in Section~\ref{sec:conjugation} 
so that the \IBCtuples for the input matrix tuple and the \IBCtuples for the new input matrix tuple have one-to-one correspondence.
If, during the process, any of the input row tuple subspaces can be refined into a smaller characteristic block-compatible row tuple subspace, then we output the refined row tuple subspace. Otherwise, we can use the algorithm obtained in  Section~\ref{sec:conjugation} to compute the \IBCtuples for the new input matrix tuple and then transform them into the \IBCtuples for the input matrix tuple by the correspondence between \IBCtuples of the two matrix tuples.

There are two steps to construct the new matrix tuple. 
In the first step, we transform the input matrix tuple with a direct sum row tuple decomposition into an intermediate matrix tuple over the same row vector space with a single row tuple decomposition by matching row tuples from different characteristic block-compatible row tuple subspaces together.

In the second step, we transform the intermediate matrix tuple into a desirable square matrix tuple.

Both steps rely on decomposing the row vector space of the input matrix tuple as a direct sum of characteristic block-compatible row vector subspaces and investigating the projection of row tuples in the input matrix tuple and intermediate matrix tuple on these characteristic block-compatible row vector subspaces.

Before we give the details of our algorithm, we first review the projection of a row vector subspace on a direct sum decomposition of the entire row vector space. 
Let $W$ be a row vector space, and $W_1, \dots, W_w \leq W$ be row vector subspaces such that $W=W_1\oplus \dots\oplus W_w$. 
Recall that for a row vector $v \in W$, $v$ can be uniquely written as $v_1 + \dots + v_d$ such that $v_j \in W_j$ for any $j \in [w]$. The row vector $v_j$ is the \emph{projection} of $v$ on $W_j$ with respect to $W_1, \dots, W_w$, denoted as $\rowvecproj_{W_1, \dots, W_w}(v, W_j)$.
Consequently, the projection of a row vector subspace of $V\leq W$ on $W_j$ (with respect to $W_1, \dots, W_w$) is the subspace spanned by the projections of row vectors in $V$ on $W_j$ with respect to $W_1, \dots, W_w$.

We say a row vector subspace $V \leq W$ has a \emph{zero projection} on $W_j$ (with respect to $W_1, \dots, W_w$), if every row vector in $V$ has a zero projection on $W_j$ with respect to $w_1, \dots, W_w$. 
We say a row vector subspace $V \leq W$ has a \emph{matching
projection} on $W_j$ (with respect to $W_1, \dots, W_w$), if $\dim(V) = \dim(W_j)$ and every non-zero row vector in $V$ has a non-zero projection on $W_j$ with respect to $W_1, \dots, W_w$.  \\

We use the following algorithm to decompose the row vector space spanned by the row vectors in the row tuples of a row tuple subspace into a direct sum of row vector subspaces. 



\begin{framed}
\noindent \textbf{Row Vector Space Direct Sum Algorithm}

\noindent \textbf{Input:} Characteristic block-compatible row tuple subspace $T$ with row tuple length $\ell$.

\noindent \textbf{Output:} Characteristic block-compatible row vector subspaces $W_1, \dots, W_w$ for some integer $w$ such that $\rowvectorspace(T) = W_1 \oplus \dots \oplus W_w$, or non-trivial characteristic block-compatible row tuple subspace $S < T$.

\noindent \textbf{Algorithm:}

\begin{enumerate}
\item For $i = 1, \dots, \ell$, if $0 < \dim\left(T^{(i)}\right) < \dim(T)$, then return $\left\langle \left\{\va \in T: \va^{(i)} = 0\right \}\right\rangle$.  
\item Let $w = 0$. 
\item For $i = 1, \dots, \ell$, 
\begin{enumerate}
    \item If $\dim\left(T^{(i)} \cap \left(W_1 \oplus \dots \oplus W_w\right)\right) = 0$, then $w = w + 1$, $W_w = T^{(i)}$, $u_w = i$, and continue.
    \item If $0 < \dim\left(T^{(i)} \cap (W_1 \oplus \dots \oplus W_w)\right) < \dim\left(T^{(i)}\right)$, return \[\langle \va \in T: \va^{(i)}  \in T^{(i)} \cap(W_1 \oplus \dots \oplus W_w) \rangle.\]
    \item 
    For each $1 \leq j \leq w$,  
    if $0 < \dim(W_{i, j}) < \dim(W_j)$, then  return \[\langle \{\va \in T: \va^{(u_j)} \in W_{i, j}\}\rangle,\] where $W_{i, j}$ is the linear space spanned by the projection of $T^{(i)}$ on $W_j$ with respect to $W_1, \dots, W_w$. 
\end{enumerate}
\item Return $W_1, \dots, W_w$. 
\end{enumerate}
\vspace{-.3cm}
\end{framed}

\begin{claim}\label{claim:row_vector_space_direct_sum}
The Row Vector Space Direct Sum Algorithm has the following properties:
\begin{enumerate}
\item
    Given a characteristic block-compatible row tuple subspace $T$ such that every row tuple is in $(\F_q^m)^\ell$, in $\text{poly}(m, \ell, \dim(T), \log q)$ time, the Row Vector Space Direct Sum Algorithm gives one of the following outputs:
    \begin{enumerate}
        \item A sequence of characteristic block-compatible row vector subspaces $W_1, \dots, W_w$ such that \[\rowvectorspace(T) = W_1 \oplus \dots \oplus W_w,\] and for every $1 \leq \ell' \leq \ell$, the projection $T^{(\ell')}$ on $W_j$ with respect to $W_1, \dots, W_w$ is either a zero projection or a matching projection. 
        \item A non-trivial characteristic block-compatible row tuple subspace $S$ such that $S < T$. 

    \end{enumerate}
    \item The algorithm is canonical in the following sense: 
    Let $T'$ be another row tuple space such that every row tuple of $T'$ is also in $(\F_q^m)^\ell$ and there exists an invertible matrix $R$ satisfying $T' = T R^{-1}$. The outputs for $T$ and $T'$ satisfy the following conditions:
    (Let $R$ be an arbitrary matrix such that such that $T' = T R^{-1}$.)
    \begin{enumerate}
        \item If the output for $T$ is a sequence of row vector spaces $W_1, \dots, W_w$, then the output for $T'$ is $W_1 R^{-1}, \dots, W_w R^{-1}$. 
        \item If the output for $T$ is a subspace $S < T$, then the output for $T'$ is $S R^{-1}$. 
    \end{enumerate}
\end{enumerate}
\end{claim}
\begin{proof}
    For the first property, we observe that if the algorithm does not terminate at Step 3 for some $i$, 
    then it is either the case of Step 3(a) or the conditions of none of Step 3(a), 3(b), and 3(c) applies. 
    
    If Step 3(a) is executed for $i$, then the projection of $T^{(i)}$ has a full projection on $W_w$. 
    If none of Step 3(a), 3(b), and 3(c) was executed for $i$, $T^{(i)}$ has either a matching projection or a zero projection on $W_j$ for every $1 \leq j \leq w$.
    
    For either case, $T^{(i)}$ has either a matching projection or a zero projection on $W_j$ for every $1 \leq j \leq w$. 
    By Fact~\ref{fact:ibctuple_invariant_basic}, the output is characteristic block-compatible. 
    The first property is then obtained by the algorithm. 

    For the second property, by the condition of $T' = T R^{-1}$, we have \[(T')^{(i)} = T^{(i)} R^{-1}\] for every $1 \leq i \leq \ell$. Hence, if the algorithm for $T$ terminates for some $i$ in Step 1, then the algorithm for $T'$ terminates for the same $i$ in Step 1, and then the second property holds. For Step 3, by induction on $i$, in the $i$-th iteration of Step 3, the algorithm for $T$ and the algorithm for $T'$ take the same case, and after the $i$-th iteration of Step 3, the values of $w$ are the same for $T$ and $T'$, and $W_i' = W_i R^{-1}$ for any $1 
    \leq i \leq w$ suppose $W_1, \dots, W_w$ denote the row vector spaces for $T$ and $W_1', \dots, W_w'$ denote the row vector spaces for $T'$. Thus, the second property holds. 
\end{proof}

The following algorithm considers the case where the input matrix tuple is associated with a decomposition of the row vector space. 
The algorithm either finds a nontrivial characteristic block-compatible row tuple subspace or constructs a square matrix tuple with every matrix in the tuple being full rank. 
For the latter case, the algorithm uses the result in Section~\ref{sec:conjugation} to compute an \IBCtuple of the square matrix tuple together with the \IBCtuple space and \IBCtuple space kernel. Then the algorithm constructs an \IBCtuple (and its \IBCtuple space and \IBCtuple space kernel) of the input matrix tuple corresponding to the \IBCtuple of the square matrix tuple.

\begin{framed}
\noindent \textbf{Single Row Tuple Subspace Algorithm}

\noindent \textbf{Input:} Matrix tuple $\vA \in \M(n \times m, \F_q)^\ell$, and characteristic block-compatible row vector spaces $W_1, \dots, W_w$ satisfying the following conditions:
\begin{enumerate}
    \item $\rowvectorspace(\vA) = W_1 \oplus \dots \oplus W_w$.
    \item For every $1 \leq j \leq w$, there is a $k \in [\ell]$ such that $W_j = \rowvectorspace(\vA^{(k)})$. 
    \item For every $1 \leq k \leq \ell, 1 \leq j \leq w$, the projection of $(\rowtuplespace(\vA))^{(k)}$ on $W_j$ with respect to $W_1, \dots, W_w$ is either a matching projection or a zero projection.
\end{enumerate}

\noindent \textbf{Output:} A representative \IBCtuple sequence for $\vA$ that contains only a single \IBCtuple $\vB$, and $\ibcspace_{\vB}, \ibcspacekernel_\vB$, or a non-trivial subspace $S < \rowtuplespace(\vA)$.

\noindent \textbf{Algorithm:}
\begin{enumerate}
    \item Construct a matrix tuple $\vP = (P_1, \dots, P_{\ell \cdot w}) \in \M(n \times m, \F_q)^{\ell \cdot w}$ as follows: for any $1 \leq k \leq \ell$ and $1 \leq j \leq w$, $P_{((k-1)\cdot w + j)}$ is the matrix such that $e_i P_{((k-1)\cdot w + j)}$ equals the projection of $e_i A_k$ on $W_j$. 
    \item For every $1 \leq j \leq w$, let $a_j$ be the smallest integer such that $a_j = (k-1) \cdot w + j$ for some $k \in [\ell]$ and the $P_{a_j}$ is a non-zero matrix. 
    If $a_j$ exists, then let $Q_{j}$ be an arbitrary $m \times n$ matrix such that $P_ {a_j} Q_{j} = I_n$, 
    otherwise, let $Q_{j}$ be the zero $m\times n$ matrix. 
    \item Construct a matrix tuple $\vC= (C_1, \dots, C_{\ell \cdot w}) \in \M(n, \F_q)^{\ell\cdot w}$ as follows: For each $1 \leq k \leq \ell$ and $1 \leq j \leq w$, $C_{(k-1)\cdot w + j}= P_{(k-1)\cdot w + j} \cdot Q_{j}$. 
    \item Run the Full Rank \IBCtuple Space Algorithm on $\vC$. 
    \begin{enumerate}
        \item If the output is a row tuple subspace $S_0 < \rowtuplespace(\vC)$, then return $\langle \{\va \in \rowtuplespace(\vA): \exists v \in \F_q^n, \vc \in S,  \text{ s.t. } \vc = v \vC, \va = v \vA \}\rangle$. 
        \item If the output is a representative \IBCtuple sequence containing a single \IBCtuple $\vE$ for $\vC$ with $\ibcspace_{\vE}$ and $\ibcspacekernel_{\vE}$, then output a representative \IBCtuple sequence for $\vA$ containing a single \IBCtuple $\vB$ with $\ibcspace_\vB$ and $\ibcspacekernel_{\vB}$ as follows: Let $\sigma$ be the number of rows in $\vE$. 
        Let $\vB$ be the matrix tuple with $\sigma$ rows such that for any $i\in[\sigma]$, $e_i \vB = v_i \vA$, where $v_i$ is the row vector such that $e_i \vE = v_i \vC$, 
        $\ibcspace_\vB = \langle \{\vB^\dagger \in \M(\sigma \times m, \F_q)^\ell: \exists \vE^\dagger \in \ibcspace_{\vE} \text{ s.t. } \forall i \in [\sigma], \exists v_i \in\F_q^^n, e_i \vB^\dagger = v_i \vA, e_i \vE^\dagger = v_i \vC\} \rangle$, and $\ibcspacekernel_{\vB}$ be the linear space that contains only the zero matrix tuple of $\ibcspace_\vB$. 
    \end{enumerate}
\end{enumerate}
\vspace{-.3cm}
\end{framed}

\begin{claim}\label{claim:single_tuple_multiple_spbspace}
        There is a Single Row Tuple Subspace Algorithm that achieves the following:
    \begin{enumerate}
        \item For a matrix tuple $\vA \in \M(n \times m, \F_q)^\ell$ and a sequence of row vector spaces $W_1, \dots, W_w$  such that $\rowvectorspace(\vA) = W_1 \oplus \dots \oplus W_w$, for every $1 \leq j \leq w$ there is a $k \in [\ell]$ such that $W_j = (\rowvectorspace(\vA))^{(k)}$, and for every $1 \leq k \leq \ell, 1 \leq j \leq w$, the projection of $(\rowtuplespace(\vA))^{(k)}$ on $W_j$ with respect to $W_1, \dots, W_w$ is either a matching projection or a zero projection,
        the output of the algorithm satisfies one of the following conditions:
        \begin{enumerate} 
            \item The output is an \IBCtuple $\vB$ of $\vA$ with $\ibcspace_{\vB}$ and $\ibcspacekernel_\vB$ satisfying the following conditions:
            \begin{enumerate}
                \item Every \IBCtuple of $\vA$ is equivalent to $\vB$. 
                \item $\vB$ satisfies the space property with \IBCtuple space $\ibcspace_\vB$ and \IBCtuple space kernel $\ibcspacekernel_{\vB}$ such that $\ibcspacekernel_{\vB}$ contains only the zero matrix tuple.
                \item Let $T_1 = \rowtuplespace(\vA)$ and $K_1$ be the linear space contains only zero row tuple of $T_1$. $\vB$ satisfies row tuple space property with $t_{\vB, 1} = \dots, = t_{\vB, \sigma} = 1$, where $\sigma$ is the number of rows of $\vB$. 
                \item $\vB$ satisfies the dimension property and block-compatible property.  
            \end{enumerate}
            \item The output is a nontrivial characteristic block-compatible row tuple subspace $S$, which is a subspace of $\rowtuplespace(\vA)$.
        \end{enumerate}
        \item The algorithm is canonical in the following sense: Let $\vA'\in \M(n \times m, \F_q)^\ell$ be another matrix tuple and $W_1', \dots, W_w'$ be a sequence of row vector spaces such that there exist invertible matrices $L$ and $R$ satisfying $\vA' = L \vA R^{-1}$ and for any invertible matrices $L$ and $R$ such that $\vA' = L \vA R^{-1}$, $W_i' = W_i R^{-1}$ for any $1 \leq i \leq w$. 
        Then, the output for $\vA$ and $\vA'$ satisfies the following conditions: (Let $L$ and $R$ be arbitrary invertible matrices such that $\vA' = L \vA R^{-1}$.)
        \begin{enumerate}
            \item If the output for $\vA$ is a subspace $S < \rowtuplespace(\vA)$, then the output for $\vA'$ is $S R^{-1}$. 
            \item  If the output for $\vA$ is an \IBCtuple $\vB$ with $\ibcspace_\vB$ and $\ibcspacekernel_{\vB}$, then
        the output for $\vA'$ is an \IBCtuple $\vB'$ of $\vA$ with $\ibcspace_{\vB'}$ and $\ibcspacekernel_{\vB'}$ satisfying the following conditions:
        \begin{enumerate}
        \item $\vB'$ is right-equivalent to $\vB$. \item 
        $\ibcspace_{\vB'} = \{\vC \in \ibcspace_\vB: \vC R^{-1}\}$.
        \item 
        $\ibcspacekernel_{\vB'} = \{\vC \in \ibcspacekernel_\vB: \vC R^{-1}\}.$
        \end{enumerate}
        \end{enumerate}

    \end{enumerate}

\end{claim}

\begin{proof}
For the first property, we first show that if Step 4(a) of the algorithm is taken, then the output of the algorithm satisfies condition (a) of the first property. 

Let $\vC'$ be an arbitrary matrix tuple such that $\rowtuplespace(\vC') \leq \rowtuplespace(\vC)$. 
Denote the number of rows for $\vC'$ by $\sigma_{\vC'}$. 
Let $f(\vC')$ be the matrix tuple with a number of rows the same as $\vC'$ such that $\rowtuplespace(f(\vC')) \leq \rowtuplespace(\vA)$, and for every $1 \leq i \leq \sigma_{\vC'}$, $e_i f(\vC') = v \vA$, where $v$ is the row vector such that $e_i \vC' = v \vC$. 
Conversely, for a matrix tuple $\vA'$ with $\rowtuplespace(\vA') \leq \rowvectorspace(\vA)$, 
we use $f^{-1}(\vA')$ to denote the matrix tuple $\vC'$ with $\rowtuplespace(\vC') \leq \rowvectorspace(\vC)$ such that $f(\vC') = \vA'$.
By the construction of $\vC$ as Step 3 of the algorithm, $f(\vC')$ and $f^{-1}(\vA')$ exist and unique for any $\vC'$ and $\vA'$ such that $\rowtuplespace(\vC') \leq \rowtuplespace(\vC)$ and $\rowtuplespace(\vA') \leq \rowtuplespace(\vA)$. 

We show that for any minimum block-diagonalization $\vG = \diag(\vG_1, \dots, \vG_g)$ for $\vC$ such that $\vG = L_{\vG} \vC (R_{\vG})^{-1}$ for some invertible matrices $L_\vG$ and $R_\vG$, there is a block-diagonalization $\vD = \diag(\vD_1, \dots, \vD_g)$ of $\vA$ with $\vD = L_\vG \vA (R')^{-1}$ for some invertible $R'$ such that $\vA_{i, \vD, L_\vG} = f(\vC_{i, \vG, L_\vG})$ for any $1 \leq i \leq g$. 
By Definition~\ref{def:ibctuple} and Fact~\ref{fact:ibctuple_diagonalzation_most_basic}, $\vC_{1, \vG, L_\vG}, \dots, \vC_{k, \vG, L_\vG}$ are \IBCtuples of $\vC$ such that \[\rowtuplespace(\vC) = \rowtuplespace(\vC_{1, \vG, L_\vG})\oplus \dots \oplus \rowtuplespace(\vC_{k, \vG, L_\vG}),\] and \[\rowvectorspace(\vC) = \rowvectorspace(\vC_{1, \vG, L_\vG})\oplus \dots \oplus \rowvectorspace(\vC_{k, \vG, L_\vG}).\]
Let 
\[W_{j, i} = \left\langle\left\{v \in W_j : v Q_j \in \rowvectorspace\left(\vC_{i, \vG, L_\vG}^{(a_j)}\right)\right\}\right\rangle\]
for any $1 \leq j \leq w$ and $1 \leq i \leq g$.
By the construction of $\vC$, we have \[W_j = W_{j, 1} \oplus \dots \oplus W_{j, g}\] for every $1 \leq j \leq w$. On the other hand, by the definition of $f$ function, we have \[\rowvectorspace(f(\vC_{i, \vG, L_\vG})) = W_{1, i} \oplus \dots \oplus W_{w, i}\] for every $1 \leq i \leq g$.
By the construction of $\vC$, we have \[\rowtuplespace(\vA) = \rowtuplespace(f(\vC_{1, \vG, L_\vG})) \oplus \dots \oplus \rowtuplespace(f(\vC_{g, \vG, L_\vG}))\] 
and $\rowvectorspace(\vA) = \rowvectorspace(\vC_{1, \vG, L_\vG}) \oplus \dots \oplus \rowvectorspace(\vC_{g, \vG, L_\vG})$. 
By Fact~\ref{fact:ibctuple_diagonalzation_most_basic}, the block-diagonalization $\vD$ exists. 

Next we show that if $\vA$ has a minimum block-diagonalization $\vD = \diag(\vD_1, \dots, \vD_d)$ with $\vD = L_\vD \vA (R_\vD)^{-1}$ for invertible matrices $L_\vD$ and $R_\vD$, then $\vC$ has a block-diagonalization $\vG = \diag(\vG_1, \dots, \vG_d)$ with $\vC_{i, \vG, L_\vD} = f^{-1}(\vA_{i, \vD, L_\vD})$. 
By Fact~\ref{fact:ibctuple_diagonalzation_most_basic}, $\rowvectorspace(\vA) = \rowvectorspace(\vA_{1, \vD, L_\vD}) \oplus \dots \oplus \rowvectorspace(\vA_{d, \vD, L_\vD})$. 
Since \[\rowvectorspace(\vA) = W_1 \oplus \dots \oplus W_w\] and 
$W_j = \rowvectorspace(\vA^{(a_j)})$ for every $1 \leq j \leq w$, 
we have 
\[W_j = \rowvectorspace((\vA_{1, \vD, L_\vD})^{(a_j)}) \oplus \dots \oplus \rowvectorspace((\vA_{d, \vD, L_\vD})^{(a_j)})\] for every $1 \leq j \leq w$. 
On the other hand, 
by the construction of $\vC$, we have $v \cdot \vC^{(a_j)} = v \cdot \vC^{(a_{j'})}$ for any $j, j' \in [w]$, which implies that $v P_{a_j} Q_j = v P_{a_{j'}} Q_{j'}$ for any $j, j' \in [w]$. 
Hence, 
\[\rowvectorspace((\vA_{i, \vD, L_\vD})^{(a_j)}) Q_j =((\vA_{i, \vD, L_\vD})^{(a_{j'})}) Q_{j'}\] for any $j, j'\in[w]$ and $1 \leq i \leq [d]$. 
Since $\vD$ is a block-diagonalization of $\vA$, we have \[\rowvectorspace(\vC) = \rowvectorspace(f^{-1}(\vA_{1, \vD, L_\vD})) \oplus \dots \oplus \rowvectorspace(f^{-1}(\vC_{d, \vD, L_\vD})).\]
In addition, by the definition of $\vC$, we 
have \[\rowtuplespace(\vC) = \rowtuplespace(f^{-1}(\vA_{1, \vD, L_\vD})) \oplus \dots \oplus \rowtuplespace(f^{-1}(\vC_{d, \vD, L_\vD})).\]
By Fact~\ref{fact:ibctuple_diagonalzation_most_basic}, $\vC$ has a block-diagonalization $\vG$ as required. 

By the above two properties and Claim~\ref{claim:sequential_ibctuple_simultaneous}, for every \IBCtuple $\vE$ of $\vC$, $f(\vE)$ is an \IBCtuple of $\vA$.  By the correspondence between $\vA$ and $\vC$, and  Lemma~\ref{lem:conjugation_canonical_main}, if the output is a row tuple subspace, the output is characteristic block-compatible. Thus, the first property of the claim holds.

For the second property, 
by the induction on each step of the algorithm,
if the algorithm on $\vA$ takes Step 4(a), then by the induction on steps of the algorithm and Lemma~\ref{lem:conjugation_canonical_main}, the algorithm for $\vA'$ also takes Step 4(a), and then the condition (a) for the second property holds. 
Otherwise, the algorithm for both $\vA$ and $\vA'$ takes Step 4(b). 
By the first property of this claim, the output for $\vA$ contains an \IBCtuple $\vB$ of $\vA$, and the \IBCtuple space and the \IBCtuple space kernel of $\vB$.
And the output for $\vA'$ contains an \IBCtuple $\vB'$ of $\vA'$, and the \IBCtuple space and the \IBCtuple space kernel of $\vB'$. 
Since $\vA$ is equivalent to $\vA'$, and all the \IBCtuples of $\vA$ are equivalent by Lemma~\ref{lem:conjugation_canonical_main}, 
the condition (b) of the second property holds. 
\end{proof}

Now we are ready to give the algorithm for Lemma~\ref{lem:direct_sum_ibctuple_algo}. 
The algorithm first computes a sequence of characteristic block-compatible row vector subspaces such that the row vector space of the input matrix tuple is the direct sum of these subspaces. 
Second, the algorithm constructs a new matrix tuple such that every row tuple in the new matrix tuple corresponds to a row tuple in each input row tuple subspace. This step utilizes the projection of row vectors in each given row tuple subspaces on the row vector subspaces obtained in the first step. 
After running the Single Row Tuple Subspace Algorithm on the new matrix tuple,
the final solution is obtained by transforming the solution of the Single Row Tuple Subspace Algorithm to the solution of the input matrix tuple.

\begin{framed}
\noindent \textbf{Direct Sum Decomposition Algorithm}

\noindent \textbf{Input:} 
Matrix tuple $\vA \in \M(n\times m, \F_q)^\ell$, and row tuple subspaces $T_1, \dots, T_\zeta \leq \rowtuplespace(\vA)$ such that $\rowtuplespace(\vA) = T_1 \oplus \dots \oplus T_\zeta$,  $\rowvectorspace(T_1) = \rowvectorspace(\vA)$, and $\dim(T_1) = \dots = \dim(T_\zeta)$.

\noindent \textbf{Output:} Representative \IBCtuple sequence for $\vA$ that contains only a single \IBCtuple $\vB$ with \IBCtuple space $\ibcspace_{\vB}$ and \IBCtuple space kernel $\ibcspacekernel_\vB$, or nontrivial characteristic block-compatible row tuple subspace $S < T_1$.

\noindent \textbf{Algorithm:}
    \begin{enumerate}
        \item Run the Row Vector Space Direct Sum Algorithm for $T_1$. If the output is a subspace of $T_1$, then return the subspace; otherwise, denote the output as $W_1, \dots, W_w$. 
        \item For all the $i\in[\zeta], j\in[w], k \in[\ell]$, 
        if the projection of $(T_i)^{(k)}$ on $W_j$ with respect to $W_1, \dots, W_w$ is neither a zero projection nor a matching projection, then 
        \begin{enumerate}
            \item Let $W = \langle\{v \in W_j : \exists \va \in T_i, \text{ the projection of } \va^{(k)} \text{ on } W_j \text{ is } v \}\rangle$.
            \item Let $k'$ be the smallest integer such that the projection of $T_1^{(k')}$ on $W_j$ is a matching projection. Return $\langle \va \in T_1 : \text{ the projection of } \va^{(k')} \text{ on } W_j \text{ is in } W\rangle$.
        \end{enumerate}
        \item Let $c=\dim(T_1)$. Construct a  matrix tuple $\vC = (C_1, \dots, C_{\ell \cdot \zeta}) \in \M(c \times m, \F_q)^{\ell \cdot \zeta}$ as follows:
        \begin{enumerate}
            \item Let $\va_1, \dots, \va_c$ be an arbitrary linear basis of $T_1$, and $(C_1, \dots, C_\ell)$ be the matrix tuple such that $e_i (C_1, \dots, C_\ell) = \va_i$ for every $i\in[c]$.
            \item For $2 \leq i \leq \zeta$, 
            let $k_0$ be the smallest integer such that $\dim((T_i)^{(k_0)}) > 0$, 
            $j$ be the smallest integer such that the projection of $(T_i)^{(k_0)}$ on $W_j$ is non-zero,
            and $k_1$ be the smallest integer such that the projection of $(T_1)^{(k_1)}$ on $W_j$ is non-zero. Let  
            $(C_{(i - 1)\cdot \ell + 1}, \dots, C_{i\cdot \ell})$ be that matrix tuple such that 
            the projection of $(e_r (C_{(i - 1)\cdot \ell + 1}, \dots, C_{i\cdot \ell}))^{(k_0)}$ on $W_j$ equals the projection of $\va_{r}^{(k_1)}$ on $W_j$ for all $r \in [c]$.

                
            
         \end{enumerate}
        
        \item Run the Single Row Tuple Subspace Algorithm on $\rowtuplespace(\vC)$. 
        \begin{enumerate}
            \item If the output is a row tuple space $S < \rowvectorspace(\vC)$, then return $\langle \{\va \in T_1 : \exists \vc \in S, \va^{(1)} = \vc^{(1)} \}\rangle$.
            \item Otherwise, let $\vE$, $\ibcspace_\vE$, and $\ibcspacekernel_\vE$ denote the output. Let $f$ be a function maps a matrix tuple $\vE = (E_1, \dots, E_{\ell \cdot \zeta})$ with row tuples in $\rowtuplespace(\vC)$ to matrix tuple $\vF = (F_1, \dots, F_\zeta)$ with row tuples in $\rowtuplespace(T_1 \cup \dots \cup T_\zeta)$ as follows: let $\sigma_\vE$ denote the number of rows of $\vE$, let $\vF_i = (E_{(i - 1)\cdot \ell + 1}, \dots, E_{i \cdot \ell})$ for any $i \in [\zeta]$, and
            $\vF$ be the matrix tuple of $\sigma_\vE \cdot \zeta$ rows such that the $((i - 1)\cdot \sigma_\vE + 1)$ to $(i\cdot \sigma_\vE)$-th row of $\vF$ is equal to $\vF_i$. 
            \item Return $\vB = f(\vE), \ibcspace_\vB = \langle f(\vE') : \vE' \in \ibcspace_{\vE} \rangle$, $\ibcspacekernel_\vB$ contains only the zero matrix tuple in $\ibcspace_\vB$, and $t_{\vB, 1}, \dots, t_{\vB, \sigma_\vE \cdot \zeta}$ such that $t_{\vB, (i-1) \cdot \sigma_\vE + i'} = i$ for any $i \in [\zeta], i' \in [\sigma_\vE]$.
        \end{enumerate}
    \end{enumerate}
    \vspace{-.3cm}
\end{framed}

\begin{proof}[Proof of Lemma~\ref{lem:direct_sum_ibctuple_algo}]
By the definition of $\vC$ and the function $f$, 
for an arbitrary minimum block-diagonalization $\vG = \diag(\vG_1, \dots, \vG_g)$ for $\vC$ such that $\vG = L_{\vG} \vC (R_{\vG})^{-1}$ for some invertible matrices $L_\vG$ and $R_\vG$, there is a block-diagonalization $\vD = \diag(\vD_1, \dots, \vD_g)$ of $\vA$ with $\vD =L_\vD  \vA (R_\vD)^{-1}$ for some invertible matrices $L_\vD$ and $R_\vD$ such that $\vA_{i, \vD, L_\vD} = f(\vC_{i, \vG, L_\vG})$ for any $1 \leq i \leq g$. 

Conversely, we show that if $\vA$ has a minimum block-diagonalization $\vD = \diag(\vD_1, \dots, \vD_d)$ with $\vD = L_\vD \vA (R_\vD)^{-1}$ for invertible matrices $L_\vD$ and $R_\vD$, then $\vC$ has a block-diagonalization $\vG = \diag(\vG_1, \dots, \vG_d)$ satisfying $\vG = L_{\vG} \vC (R_{\vG})^{-1}$ for some invertible matrices $L_\vG$ and $R_\vG$ such that for any $1 \leq i \leq d$, 
$\vC_{i, \vG, L_\vD} = f^{-1}(L_{\vA, i}\vA_{i, \vD, L_\vD})$ with some invertible matrices $L_{\vA, 1}, \dots, L_{\vA, d}$, where $f^{-1}(\vF)$ for some $\vF$ satisfying $\rowtuplespace(\vF) \leq \rowtuplespace(\vA)$ is $\vE$ such that $f(\vE) = \vF$. 
By Fact~\ref{fact:ibctuple_diagonalzation_most_basic}, we have \[\rowtuplespace(\vA) = \rowtuplespace(\vA_{1, \vD, L_\vD}) \oplus  \dots \oplus \rowtuplespace(\vA_{d, \vD, L_\vD})\] and \[\rowvectorspace(\vA) = \rowvectorspace(\vA_{1, \vD, L_\vD}) \oplus  \dots \oplus \rowvectorspace(\vA_{d, \vD, L_\vD}).\] 
By Claim~\ref{claim:row_vector_space_direct_sum}, 
for each $W_j$, there exists a $k$ such that $\rowvectorspace((T_1)^{(k)}) = W_j$. 
Hence, Let $W_{j, i} = W_j \cap \rowvectorspace(\vA_{i, \vD, L_\vD})$ for any $j \in [w], i \in [d]$, 
we have 
\[W_j = W_{j, 1} \oplus \dots \oplus W_{j, k}\]
for any $j \in [w]$
and 
\[\rowvectorspace(\vA_{i, \vD, L_\vD}) = W_{1, i} \oplus \dots \oplus W_{w, i}\]
for any $i \in [d]$. 
Thus, for any $\va \in T_r$ for some $r \in [\zeta]$, if there exists a $j \in [w]$ and a $k \in [\ell]$ such that the projection of $\va^{(k)}$ on $W_j$ is in $W_{j, i}$ for some $i \in [d]$, then $\va \in \rowtuplespace(\vA_{i, \vD, L_\vD})$. 
By the construction of $\vC$, letting \[\vC_{i} = \langle \{\va \in \rowtuplespace(\vC): \rowvectorspace(\va) \leq \rowvectorspace(\vA_{i, \vD, L_\vD})\}\rangle.\] 
We have 
\[\rowtuplespace(\vC) = \rowtuplespace(\vC_1) \oplus \dots \oplus \rowtuplespace(\vC_d)\]
and 
\[\rowvectorspace(\vC) = \rowvectorspace(\vC_1) \oplus \dots \oplus \rowvectorspace(\vC_d).\]
Hence, $\vC$ has a required block-diagonalization. 

By the above two properties and Claim~\ref{claim:sequential_ibctuple_simultaneous}, for every \IBCtuple $\vE$ of $\vC$, $f(\vE)$ is an \IBCtuple of $\vA$. On the other hand, if the output is obtained in Step 1 or Step 4(a), the output is a subspace of $T_1$, and by Claim~\ref{claim:row_vector_space_direct_sum} and Claim~\ref{claim:single_tuple_multiple_spbspace}, the row tuple subspace is characteristic block-compatible. 
Thus, the first property of the claim holds.

For the second property, 
by the induction on each step of the algorithm,
if the algorithm on $\vA$ takes Step 4(a), then by the induction on steps of the algorithm and Lemma~\ref{lem:conjugation_canonical_main}, the algorithm for $\vA'$ also takes Step 4(a), and then the condition (a) for the second property holds. 
Otherwise, the algorithm for both $\vA$ and $\vA'$ takes Step 4(b). 
By the first property of this claim, the output for $\vA$ contains an \IBCtuple $\vB$ of $\vA$, and the \IBCtuple space and the \IBCtuple space kernel of $\vB$.
And the output for $\vA'$ contains an \IBCtuple $\vB'$ of $\vA'$, and the \IBCtuple space and the \IBCtuple space kernel of $\vB'$. 
Since $\vA$ is equivalent to $\vA'$, and all the \IBCtuples of $\vA$ are equivalent by Lemma~\ref{lem:conjugation_canonical_main}, by the correspondence between $T_1, \dots, T_\zeta$ and $T_1', \dots, T_\zeta'$
the condition (b) of the second property holds.
\end{proof}

\section{\IBCtuples from hierarchical row tuple decomposition}\label{sec:quotient_matrix_tuple_main}

In this section, we give an algorithm to compute a representative \IBCtuple sequence for a matrix tuple with a general hierarchical row tuple decomposition. 

In our approach, we first investigate the \IBCtuples of the quotient matrix tuples regarding the hierarchical row tuple decomposition defined as follows. 

\begin{definition}[Quotient matrix tuple]\label{def:quotient_matrix_tuple}
    Let $\vA\in\M(n\times m, \F_q)^\ell$, and $T_1, \dots, T_\zeta$ be a depth-$\beta$ hierarchical row tuple decomposition of $\vA$ with parameters $h_0, \dots, h_\beta$ for some integer $\beta$. For any  $i\in[\beta]$, let $W_i:=\langle \rowvectorspace(T_{h_i}) \cup \dots \cup \rowvectorspace(T_\zeta)\rangle$, $n_i := n - \dim(\langle \{T_{h_i} \cup \dots \cup T_\zeta\} \rangle )$ and $m_i:=\dim(\rowvectorspace(\vA)/ W_\gamma)$. 

    For any $\gamma \in [\beta]$,
    a matrix tuple $\vQ \in \M(n_\gamma \times m_\gamma, \F_q)^\ell$ with formatting vector $(v_1 + W_\gamma, \dots, v_{m_\gamma} + W_\gamma)$ is a \emph{$\gamma$-quotient matrix tuple} of $\vA$ with respect to $T_1, \dots, T_\zeta$ and parameters $h_0, \dots, h_\beta$ if the following conditions hold:
    \begin{enumerate}

        \item $v_1, \dots, v_{m_\gamma}  \in \rowvectorspace(\vA$) such that $v_1 + W_\gamma, \dots, v_{m_\gamma} + W_\gamma$ is a linear basis of $\rowvectorspace(\vA)/ W_\gamma$. 

        \item For any $v \in \F_q^n$, there is a $x \in \F_q^{n_\gamma}$ such that $x \vQ \cdot (v_1 + W_\gamma, \dots, v_{m_\gamma} + W_\gamma)^T = v \vA + W_\gamma$. 
    \end{enumerate}
    \end{definition}

The high-level intuition of a $\gamma$-quotient matrix tuple is a matrix tuple that shrinks all the row vectors in $W_\gamma$ into zero row vectors and then removes the zero row tuples in the remaining matrix tuple.  

The purpose of the quotient matrix tuple is that the quotient matrix tuple has a simpler structure than the original input matrix tuple so that we can first compute the \IBCtuples of the quotient matrix tuple and then compute the \IBCtuples of the input matrix tuple with the guidance of the \IBCtuples of the quotient matrix tuple. 

The major result of this section is an algorithm to canonically compute a representative \IBCtuple sequence for a matrix tuple with a depth-$\beta$ hierarchical row tuple decomposition based on a representative \IBCtuple sequence for a $\beta$-quotient matrix tuple (Lemma~\ref{lemma:extension_main} and Lemma~\ref{lemma:extension_main_canonical}). 

With this algorithm, we can compute a representative \IBCtuple sequence for the matrix tuple with a depth-$\beta$ hierarchical row tuple decomposition as follows: We first use the algorithm in Section~\ref{sec:direct_sum_decomposition} to compute a representative \IBCtuple sequence for a $1$-quotient matrix tuple of the input matrix tuple. 
Then by induction, for any $1 \leq \gamma \leq \beta$, using the algorithm in this section, we can compute a representative \IBCtuple sequence for a $(\gamma + 1)$-quotient matrix tuple of the input matrix tuple by a representative \IBCtuple sequence for a $\gamma$-quotient matrix tuple based on the fact that any $\gamma$-quotient matrix tuple of a matrix tuple is also a $\gamma$-quotient matrix tuple of any $(\gamma+1)$-quotient matrix tuple of the matrix tuple.
Finally, we can compute a representative \IBCtuple sequence for the input matrix tuple by the representative \IBCtuple sequence for the $\beta$-quotient matrix tuple.

In this section, we prove the following two lemmas.

\begin{lemma}\label{lemma:extension_main}
    There is an \IBCtuple Space Algorithm with the following input 
    \begin{enumerate}
        \item a matrix tuple $\vA \in \M(n \times m, \F_q)^\ell$ with a depth-$\beta$ hierarchical row tuple decomposition $T_1, \dots, T_\zeta$ with parameters $h_0, \dots, h_\beta$ for $\vA$,
        \item a $\beta$-quotient matrix tuple $\vQ$ of $\vA$ with formatting vector $(v_1 + W_\beta, \dots, v_{m_\beta} + W_\beta)$ with respect to $T_1, \dots, T_\zeta$,
        \item a representative \IBCtuple sequence $\vB_1, \dots, \vB_k$ of $\vQ$ with \IBCtuple space $\ibcspace_{\vB_i}$, \IBCtuple space kernel $\ibcspacekernel_{\vB_i}$, and row tuple space parameters $t_{\vB_i, 1}, \dots, t_{\vB_i, \sigma_i}$ for every $i \in [k]$ such that every $\vB_i$ for $i \in [k]$ satisfies the space property, row tuple space property, dimension property, and block-compatible property,
    \end{enumerate}
    such that with $\poly(n, m, \ell, \log q)$ running time, the output of the algorithm satisfies one of the following conditions:
        \begin{enumerate}
            \item The output is a nontrivial subspace $S < T_i$ for some $i\in [\zeta]$ such that $\dim(S / W_\gamma) > 0$. 
            \item A representative relevant \IBCtuple sequence $\vG_1, \dots, \vG_{g}$ of $\vA$ with \IBCtuple space $\ibcspace_{\vG_i}$, \IBCtuple space kernel $\ibcspacekernel_{\vG_i}$, and row tuple space parameters $t_{\vG_i, 1}, \dots, t_{\vB_i, \sigma_{\vG_i}}$ for every $i \in [g]$ such that every $\vG_i$ for $i \in [g]$ satisfies the space property, row tuple space property, dimension property, and block-compatible property.
        \end{enumerate}
    \end{lemma}
    \begin{lemma}\label{lemma:extension_main_canonical}
    The \IBCtuple Space Algorithm is canonical in the following sense: 
        For two inputs $\vA, T_1, \dots, T_\zeta, h_0, \dots, h_\beta, \vQ, \vB_1, \dots, \vB_k$ and $\vA', T_1', \dots, T_\zeta', h_0', \dots, h_\beta, \vQ', \vB_1', \dots, \vB_k'$ such that $\vA, \vA' \in \M(n\times m, \F_q)^\ell$ and there exist invertible matrices $L \in \GL(n, \F_q), R\in \GL(m, \F_q)$ satisfying the following conditions:
        \begin{enumerate}
            \item[i.] $\vA' = L \vA R^{-1}$;
            \item[ii.] $T_i' = T_i R^{-1}$ for every $i \in [\zeta]$;
            \item[iii.] $h_i = h_i'$ for any $i \in \{0, 1, \dots, \beta\}$;
            \item[iv.] $\vB_{i}$ is right-equivalent to $\vB_{i}'$ for any $i \in [k]$;
            \item[v.] $\ibcspace_{\vB_i'} = \{\vE (R')^{-1} : \vE \in \ibcspace_{\vB_i}\}$ and $\ibcspacekernel_{\vB_i'} = \{\vE (R')^{-1} : \vE \in \ibcspacekernel_{\vB_i}\}$ for every $i \in [k]$, where $R' \in \GL(m_\beta, \F_q)$ is a matrix such that \[(R')^{-1} (v_1' + W_\beta', \dots, v_{m_\beta}' + W_\beta')^T  = (v_1 R^{-1}  + W_\beta R^{-1}, \dots, v_{m_\beta} R^{-1} + W_\beta R^{-1})^T,\]
        \end{enumerate}
        The output for $\vA$ and $\vA'$ satisfies the following conditions: (Let $L$ and $R$ be arbitrary invertible matrices satisfying (i) - (v))
        \begin{enumerate}
            \item If the output for $\vA$ is a subspace $S < T_i$ for some $i\in [\zeta]$, then the output for $\vA'$ is $S R^{-1}$. 
            \item  If the output for $\vA$ is a representative \IBCtuple sequence $\vG_1, \dots, \vG_g$ with $\ibcspace_{\vG_i}$ and $\ibcspacekernel_{\vG_i}$ for every $i \in [g]$, then
        the output for $\vA'$ is a representative \IBCtuple sequence $\vG_1',\dots, \vG_g'$ with $\ibcspace_{\vG_1}, \dots, \ibcspace_{\vG_g'}$ and $\ibcspacekernel_{\vG_1}, \dots, \ibcspacekernel_{\vG_g'}$ satisfying the following conditions for every $i\in[g]$
        \begin{enumerate}
        \item $\vG_i'$ is right-equivalent to $\vG_i$; \item 
        $\ibcspace_{\vG_i'} = \{\vC \in \ibcspace_{\vG_i}: \vC R^{-1}\}$;
        \item 
        $\ibcspacekernel_{\vG_i'} = \{\vC \in \ibcspacekernel_{\vG_i}: \vC R^{-1}\}.$
        \end{enumerate}
        \end{enumerate}
\end{lemma}

In Section~\ref{sec:quotient_matrix_tuple}, we give some additional definitions and properties regarding the quotient matrix tuple. 
Section~\ref{sec:ibc_extension} defines the extension of \IBCtuples of quotient matrix tuples, the major tool to understand the \IBCtuples of quotient matrix tuples and the \IBCtuples of the input matrix tuple. 
Section~\ref{sec:irrelevant_ibctuples} and Section~\ref{sec:relevant_ibctuples} construct \IBCtuples of the input matrix tuple for two different cases. 
Section~\ref{sec:quotient_matrix_tuple_ibc_final} proves Lemma~\ref{lemma:extension_main} and Lemma~\ref{lemma:extension_main_canonical}.

\subsection{Quotient matrix tuple}\label{sec:quotient_matrix_tuple}

We give some additional definitions regarding the quotient matrix tuples and prove some basic properties of the quotient matrix tuples.

For a matrix tuple with a depth-$\beta$ hierarchical row tuple decomposition $T_1, \dots, T_\zeta$ of $\vA$
with parameters $h_0, \dots, h_\beta$ for some integer $\beta$, for any $\vQ$ that is a $\gamma$-quotient matrix tuple of $\vA$ with formatting vector $(v_1 + W_\gamma, \dots, v_{m_\gamma} + W_\gamma)$ for some $\gamma \in [\beta]$, 
we denote  
        \[T_{\vQ, i} := \left\langle \left\{x \vQ : x\in \F^{n_\gamma}, x \vQ \left(v_1 + W_\gamma, \dots, v_{m_\gamma} + W_\gamma\right)^T \in T_i / W_\gamma\right\}\right\rangle\]
for any $i \in [h_\gamma - 1]$. 
In addition, 
the $(\beta + 1)$-quotient matrix tuple $\vQ$ of $\vA$ is always $\vA$ with formatting vector $(e_1, \dots, e_m)$ and $T_{\vQ, i} = T_i$ for every $i \in [\zeta]$.

In the rest of this subsection, we prove some useful properties for quotient matrix tuples. 

\begin{fact}\label{fact:basic_quotient_matrix_tuple_1}
For a matrix type $\vA$ and a depth-$\beta$ hierarchical row tuple decomposition $T_1, \dots, T_\zeta$ of $\vA$ with parameters $h_0, \dots, h_\beta$ for some integer $\beta$. 
For any $\gamma \in [\beta]$, let $\va_1 + W_\gamma, \dots, \va_{n_\gamma} + W_\gamma$ be a linear basis of $\rowtuplespace(\vA) / W_\gamma$, and $v_1 + W_\gamma, \dots, v_{m_\gamma} + W_\gamma$ be a linear basis of $\rowvectorspace(\vA) / W_\gamma$, then there is a unique matrix tuple $\vQ$ such that 
\begin{equation}\label{equ:quotient_matrix_tuple}\vQ \cdot (v_1 + W_\gamma, \dots, v_{m_\gamma} + W_\gamma)^T = \begin{bmatrix}
    \va_1 + W_\gamma \\ \vdots \\ \va_{n_\gamma} + W_\gamma
\end{bmatrix},\end{equation}
and $\vQ$ is a $\gamma$-quotient matrix tuple for $\vA$ with $(v_1 + W_\gamma, \dots, v_{m_\gamma} + W_\gamma)$ as the formatting vector. 
\end{fact}
\begin{proof}
    For any $i \in [n_\gamma], k \in [\ell]$, $(\va_i + W_\gamma)^{(k)}$ is a linear combination of $v_1 + W_\gamma, \dots, v_{m_\gamma} + W_\gamma$.
    Let $q_{k, i, j} \in \F$ such that \[(\va_i + W_\gamma)^{(k)} = \sum_{j = 1}^{m_\gamma} q_{k, i, j} \cdot v_{j} + W_\gamma.\]
    Let $\vQ = (Q_1, \dots, Q_\ell)$ be the matrix tuple such that the $i$-th column $j$-th row of $Q_k$ is $q_{k, i, j}$. 
    Then $\vQ$ is the matrix tuple satisfying Equation (\ref{equ:quotient_matrix_tuple}). 
    By Definition~\ref{def:quotient_matrix_tuple} is a $\gamma$-quotient matrix tuple for $\vA$ with $(v_1 + W_\gamma, \dots, v_{m_\gamma} + W_\gamma)$ as the formatting vector. 
\end{proof}

\begin{claim}\label{claim:ibctuple_compatible_quotient_basic}
Let $\vA$ be a matrix tuple and $W$ be a characteristic block-compatible row vector subspace of $\vA$. Let $\vD = (\vD_1, \dots, \vD_d)$ be a block-diagonalization of $\vA$ and 
$L, R$ be invertible matrices such that $\vD = L\vA R^{-1}$. 
Then we have 
\begin{equation}\label{equ:quotient_ibctuple_diag_1}\rowvectorspace(\vA) / W = (\rowvectorspace(\vA_{1, \vD, L}) / W) \oplus \dots \oplus (\rowvectorspace(\vA_{d, \vD, L}) / W),\end{equation}
and \begin{equation}\label{equ:quotient_ibctuple_diag_2}\rowtuplespace(\vA) / W = (\rowtuplespace(\vA_{1, \vD, L}) / W) \oplus \dots \oplus (\rowtuplespace(\vA_{d, \vD, L}) / W).\end{equation}
\end{claim}
\begin{proof}
Since 
\begin{equation}\label{equ:quotient_ibctuple_diag}\rowvectorspace(\vA) = \rowvectorspace(\vA_{1, \vD, L}) \oplus \dots \oplus  \rowvectorspace(\vA_{d, \vD, L}),\end{equation} we have 
\[\rowvectorspace(\vA) / W = \langle \{ (\rowvectorspace(\vA_{1, \vD, L}) / W) \cup \dots \cup  \rowvectorspace(\vA_{d, \vD, L})/W\} \rangle.\]
So, we only need to show for every $i \in [d]$
\[\rowvectorspace(\vA_{i, \vD, L})/W\} \cap \langle\rowvectorspace(\vA_{1, \vD, L})/W\}\cup \dots, \cup \rowvectorspace(\vA_{i-1, \vD, L})/W\} \rangle = \emptyset.\]
Suppose this is not true for some $i \in [d]$, then there exists a vector $v \in \rowvectorspace(\vA_{1, \vD, L}) \oplus \dots \oplus  \rowvectorspace(\vA_{i-1, \vD, L})$ and a vector $u \in W$ such that $v + u \in \rowvectorspace(\vA_{i, \vD, L}) $. 
Since $u \in \rowvectorspace(\vA_{1, \vD, L}) \oplus \dots \oplus \rowvectorspace(\vA_{i, \vD, L})$, 
let $u_1, u_2$ be the row vectors in $\rowvectorspace(\vA_{1, \vD, L}) \oplus \dots \oplus \rowvectorspace(\vA_{i - 1, \vD, L})$ and $\rowvectorspace(\vA_{i, \vD, L})$ such that $u = u_1 + u_2$. 
Hence, $v + u_1 = v + u - u_2$ is in $\rowvectorspace(\vA_{i, \vD, L})$
and also in $\rowvectorspace(\vA_{1, \vD, L}) \oplus \dots \oplus \rowvectorspace(\vA_{i - 1, \vD, L})$. This contradicts Equation (\ref{equ:quotient_ibctuple_diag}). 
Hence, Equation (\ref{equ:quotient_ibctuple_diag_1}) holds. 

Now we prove Equation (\ref{equ:quotient_ibctuple_diag_2}). 
Since 
\begin{equation}\label{equ:quotient_ibctuple_diag_10}\rowtuplespace(\vA) = \rowtuplespace(\vA_{1, \vD, L}) \oplus \dots \oplus  \rowtuplespace(\vA_{d, \vD, L}),\end{equation}
 we have 
\[\rowtuplespace(\vA) / W = \langle \{ (\rowtuplespace(\vA_{1, \vD, L}) / W) \cup \dots \cup  \rowtuplespace(\vA_{d, \vD, L})/W\} \rangle.\]
Let $S = \langle \{\va \in \rowtuplespace(\vA): \forall i\in [\ell], \va^{(i)} \in W\}\rangle$. 
By Equation (\ref{equ:quotient_ibctuple_diag_10}) and Fact~\ref{fact:ibctuple_invariant_basic}, $S$ is a characteristic block-compatible row tuple space, and thus by Fact~\ref{fact:ibctuple_invariant_fund}, \[S = (S \cap \rowtuplespace(\vA_{1, \vD, L})) \oplus \dots \oplus (S \cap \rowtuplespace(\vA_{d, \vD, L})).\]

We only need to show for every $i \in [d]$
\[\rowtuplespace(\vA_{i, \vD, L})/W\} \cap \langle\rowtuplespace(\vA_{1, \vD, L})/W\}\cup \dots, \cup \rowtuplespace(\vA_{i-1, \vD, L})/W\} \rangle = \emptyset.\]
Suppose this is not true for some $i \in [d]$, then there exists a row tuple $\va \in \rowtuplespace(\vA_{1, \vD, L}) \oplus \dots \oplus  \rowtuplespace(\vA_{i-1, \vD, L})$ and a row tuple $\vb \in S$ such that $\va + \vb \in \rowtuplespace(\vA_{i, \vD, L}) $. 
Since $\vb \in \rowtuplespace(\vA_{1, \vD, L}) \oplus \dots \oplus \rowtuplespace(\vA_{i, \vD, L})$, 
let $\vb_1, \vb_2$ be the row tuples in $\rowtuplespace(\vA_{1, \vD, L}) \oplus \dots \oplus \rowtuplespace(\vA_{i - 1, \vD, L})$ and $\rowtuplespace(\vA_{i, \vD, L})$ such that $\vb = \vb_1 + \vb_2$. 
Hence, $\va + \vb_1 = \va + \vb - \vb_2$ is in $\rowtuplespace(\vA_{i, \vD, L})$
and also in $\rowtuplespace(\vA_{1, \vD, L}) \oplus \dots \oplus \rowtuplespace(\vA_{i - 1, \vD, L})$. This contradicts Equation (\ref{equ:quotient_ibctuple_diag_10}). 
Hence, Equation (\ref{equ:quotient_ibctuple_diag_2}) holds. 
\end{proof}

\begin{claim}\label{claim:correspondance_ibctuple_diag_quotient}
For a matrix type $\vA$ and a depth-$\beta$ hierarchical row tuple decomposition $T_1, \dots, T_\zeta$ of $\vA$ with parameters $h_0, \dots, h_\beta$ for some integer $\beta$. 
Let $\vD = \diag(\vD_1, \dots, \vD_d)$ be a block-diagonalization of $\vA$, and $L$ and $R$ be arbitrary invertible matrices such that $\vD = L \vA R^{-1}$. 
Fix an integer $\gamma \in [\beta]. $
For each $i \in [d]$, let $\va_{i, 1} + W_\gamma, \dots, \va_{i, \sigma_i} + W_\gamma$ be a linear basis of $\rowtuplespace(\vA_{i, \vD, L}) / W_\gamma$ , $v_{i, 1} +W_\gamma, \dots, v_{i, \delta_i} + W_\gamma$ be a linear basis of $\rowvectorspace(\vA_{i, \vD, L}) / W_\gamma$, and $\vQ_{i}$ be the $\sigma_i \times \delta_i$ matrix tuple such that 
\[\vQ_i (v_{i, 1} +W_\gamma, \dots, v_{i, \delta_i} + W_\gamma)^T = \begin{bmatrix}
    \va_{i, 1} + W_\gamma \\ \vdots \\ \va_{i, \sigma_i} + W_\gamma
\end{bmatrix}, \]
then $\vQ_{\vA, \vD, L} := \diag(\vQ_1, \dots, \vQ_d)$ is a $\gamma$-quotient matrix tuple of $\vA$ with 
\[
(v_{1, 1} + W_\gamma, \dots, v_{1, \sigma_1} + W_\gamma, v_{2, 1} + W_\gamma,  \dots v_{d, \sigma_d}  + W_\gamma)
\]
as the formatting matrix tuple. 
\end{claim}
\begin{proof}
    By the definition of $\vQ_i$ for each $i\in[d]$, 
    we have 
    \[\vQ_{\vA, \vD, L} (v_{1, 1} + W_\gamma, \dots, v_{1, \sigma_1} + W_\gamma, v_{2, 1} + W_\gamma,  \dots v_{d, \sigma_d}  + W_\gamma)^T = \begin{bmatrix}
    \va_{1, 1} + W_\gamma \\ \vdots \\ \va_{1, \sigma_1} + W_\gamma \\ \va_{2, 1} + W_\gamma  \\ \vdots \\ \va_{d, \sigma_d} + W_\gamma
\end{bmatrix}.\]
By Claim~\ref{claim:ibctuple_compatible_quotient_basic},
\[v_{1, 1} + W_\gamma, \dots, v_{1, \sigma_1} + W_\gamma, v_{2, 1} + W_\gamma,  \dots v_{d, \sigma_d}  + W_\gamma\]
is a linear basis of $\rowvectorspace(\vA) / W_\gamma$ and \[\va_{1, 1} + W_\gamma, \dots, \va_{1, \sigma_1} + W_\gamma, \va_{2, 1} + W_\gamma, \dots, \va_{d, \sigma_d} + W_\gamma\] is a linear basis of  $\rowtuplespace(\vA) / W_\gamma$. 
By Fact~\ref{fact:basic_quotient_matrix_tuple_1}, $\vQ_{\vA, \vD, L}$ is a $\gamma$-quotient matrix tuple for $\vA$. 
\end{proof}

    \begin{claim}\label{claim:quotient_marix_tuple_basic}
    Let $\vA \in \M(n\times m, \F_q)^{\ell}$ be a matrix tuple, and $T_1, \dots, T_\zeta$ be a hierarchical row tuple decomposition with parameters $h_0, \dots, h_\beta$ for $\vA$. 
    Let $\vQ$ with formatting vector $(v_1 + W_\gamma, \dots, v_{m_\gamma} + W_\gamma)$ be a $\gamma$-quotient matrix tuple with respect to $T_1, \dots, T_\zeta$ for  $\vA$ for some $\gamma \in [\beta]$. 
    We have the following observations:
    \begin{enumerate}
        \item 
        The sequence $T_{\vQ, 1}, \dots, T_{\vQ, h_\gamma - 1}$ with parameters $h_0, \dots, h_{\gamma - 1}$ is a row tuple decomposition of $\vQ$.

        \item For any $1 < \gamma' < \gamma$, let $\vQ'$ be a $\gamma'$-quotient matrix tuple of $\vA$  with respect to $T_1, \dots, T_\zeta$ and $\vQ$ be a $\gamma$-quotient matrix tuple of $\vA$  with respect to $T_1, \dots, T_\zeta$. 
        $\vQ'$ is a $\gamma'$-quotient matrix tuple of $\vQ$ with respect to $T_{\vQ, 1}, \dots, T_{\vQ, h_\gamma - 1}$ with parameters $h_0, \dots, h_{\gamma - 1}$.

        \item Let $\vQ'$ be another $\gamma$-quotient matrix tuples for $\vA$ with respect to $T_1, \dots, T_\zeta$. There exist invertible matrices $L'$ and $R'$ such that $\vQ' = L' \vQ (R')^{-1}$, and $T_{\vQ, i}' = T_{\vQ, i} (R')^{-1}$ for every $i \in [h_\gamma - 1]$.
    
    \item Let $\vA' \in \M(n\times m, \F_q)^\ell$ be another matrix tuples,  
    $T_1', \dots, T_\zeta'$ be a hierarchical row tuple decomposition for $\vA'$ with parameters $h_0, \dots, h_\gamma$ same as $T_1, \dots, T_\zeta$, and 
    $W_\gamma' = \langle \rowvectorspace(T_{h_i}') \cup \dots, \cup \rowvectorspace(T_\zeta')\rangle$.
    Let $\vQ'$ with formatting vector $(v_1' + W'_\gamma, \dots, v_{m_\gamma}' + W'_\gamma)$ be a $\gamma$-quotient matrix tuple for $\vA'$ with respect to $T_1', \dots, T_\zeta'$. 
    
    If there are invertible matrices $L$ and $R$ such that $\vA' = L \vA R^{-1}$ and $T_i' = T_i R^{-1}$ for every $i \in [\zeta]$, then  
    there exist invertible matrices $L'$ and $R'$ satisfying the following conditions:
    \begin{enumerate}
        \item $\vQ' = L' \vQ (R')^{-1}$,
        \item $T_{\vQ', i} = T_{\vQ, i} (R')^{-1}$ for any $i \in [h_\gamma - 1]$,  
    \item $v_i' + W_\gamma' = (\sum_{j = 1}^{m_\gamma} r_{i, j}' v_j) R^{-1} + W_\gamma R^{-1}$, where $r_{i, j}'$ is the $i$-th row $j$-th column of $R'$.    
    \end{enumerate}
    \end{enumerate}
\end{claim}

\begin{proof}
The first property is by Definition~\ref{def:row_tuple_decomposition}. 
The second and third properties are by Definition~\ref{def:quotient_matrix_tuple}.

For the fourth property, by the definition of quotient matrix tuple, for every $1 \leq i \leq m_\gamma$, $v_i'  + W_\gamma' $ is a linear combination of $v_1 R^{-1} + W_\gamma R^{-1}, \dots, v_{m_\gamma} R^{-1} + W_\gamma R^{-1}$. 
Let $r_{i, j}'\in \F_q$ be the coefficients such that   
\[v_i' + W_\gamma' = \left(\sum_{j = 1}^{m_\gamma} r_{i, j}' v_j\right) R^{-1} + W_\gamma R^{-1}\] for any $1 \in [m_\gamma]$. 
Let $R'$ be the matrix such that the $i$-th row $j$-th column of $R'$ equals $r_{i, j}'$.
Because every $v_i + W_\gamma$ can also be represented as a linear combination of $v_1' R + W_\gamma' R, \dots, v_{m_\gamma}' R + W_\gamma' R$,
the $m_\gamma \times m_\gamma$ matrix with $r_{i, j}'$ as the $i$-th row $j$-th column is an invertible matrix.

Also by the definition of quotient matrix tuple, for any $v\in \F_q^n$, there is a unique vector $w \in \F_q^{n_\gamma}$ such that \begin{equation}\label{equ:basic_quotient_matrix_tuple_1}(v L) \vA + W_\gamma =  w \vQ (v_1 + W_\gamma, \dots, v_{m_\gamma} + W_\gamma)^T,\end{equation} 
and there is a unique vector $w' \in \F_q^{n_\gamma}$ such that \[v \vA' + W_\gamma' =w' \vQ' (v_1' + W_\gamma', \dots, v_{m'}' + W_\gamma')^T.\]
Since $\vA' = L \vA R^{-1}$, we have 
\begin{equation}\label{equ:basic_quotient_matrix_tuple_2}\begin{split}v \vA' + W_\gamma'= &(vL) \vA R^{-1} + W_\gamma R^{-1} \\
= & w \vQ \left(v_1 R^{-1}+ W_\gamma R^{-1}, \dots, v_{m_\gamma} R^{-1} + W_\gamma R^{-1}\right)^T \\
= & w \vQ (R')^{-1} R' \left(v_1 R^{-1}+ W_\gamma R^{-1}, \dots, v_{m_\gamma} R^{-1} + W_\gamma R^{-1}\right)^T \\
= & w \vQ (R')^{-1}\left(\left(\sum_{j = 1}^{m_\gamma} r_{1, j}' v_j\right)R^{-1} + W_\gamma R^{-1}, \dots, \left(\sum_{j = 1}^{m_\gamma} r_{m_\gamma, j}' v_j\right)R^{-1} + W_\gamma R^{-1}\right)^T \\
= & w \vQ (R')^{-1}(v_1' + W_\gamma', \dots, v_{m_\gamma}' + W_\gamma')^T.
\end{split}\end{equation}
Because Equation~(\ref{equ:basic_quotient_matrix_tuple_1}) and (\ref{equ:basic_quotient_matrix_tuple_2}) hold for all the $v \in \F^n$, and it is additive, there is an invertible matrix  $L'$ such that $\vQ' = L' \vQ (R')^{-1}$. 
Finally, by the definition of $T_{\vQ, i}$ and $T_{\vQ', i}$, 
we have 
\[\begin{split} T_{\vQ', i} = & \left\langle \left\{x \vQ' : x\in \F_q^{n_\gamma}, x \vQ' (v_1' + W_\gamma', \dots, v_{m_\gamma}' + W_\gamma')^T \in T_i' / W_\gamma'\right\}\right\rangle \\
= & \left\langle \left\{x L' \vQ (R')^{-1} : x\in \F_q^{n_\gamma}, x L' \vQ (R')^{-1} (v_1' + W_\gamma', \dots, v_{m_\gamma}' + W_\gamma')^T \in T_i' / W_\gamma'\right\}\right\rangle \\
= & \left\langle \left\{x \vQ (R')^{-1} : x\in \F_q^{n_\gamma}, x\vQ (R')^{-1} (v_1' + W_\gamma', \dots, v_{m_\gamma}' + W_\gamma')^T \in T_i' / W_\gamma'\right\}\right\rangle \\
= & \left\langle \left\{x \vQ (R')^{-1} : x\in \F_q^{n_\gamma}, x  \vQ  (v_1 R^{-1} + W_\gamma R^{-1}, \dots, v_{m_\gamma} R^{-1}  + W_\gamma R^{-1})^T \in T_i' / W_\gamma'\right\}\right\rangle \\
= & \left\langle \left\{x \vQ (R')^{-1} : x\in \F^{n_\gamma}, x  \vQ  (v_1 + W_\gamma, \dots, v_{m_\gamma}  + W_\gamma)^T \in T_i / W_\gamma\right\}\right\rangle \\
= & T_{\vQ, i} (R')^{-1}.
\end{split}\]
Thus, the fourth property holds. 
\end{proof}

\subsection{\IBCtuple extensions}\label{sec:ibc_extension}

To understand the \IBCtuples of $\vA$, we first investigate the \IBCtuples of a quotient matrix tuple $\vQ$. As a result, suppose we obtain a representative \IBCtuple sequence of $\vQ$. Next, we investigate the relations between \IBCtuples of $\vQ$ and \IBCtuples of $\vA$ by extending the \IBCtuples of $\vQ$ to extensions of $\vA$, defined as follows.

\begin{definition}[Row Tuple Extension and \IBCtuple Extension]\label{def:extension}
Let $\vA \in \M(n\times m, \F_q)^\ell$ be a matrix tuple with a depth-$\beta$ hierarchical row tuple decomposition $T_1, \dots, T_\zeta$ and parameters $h_0, \dots, h_\beta$.
For a $\gamma \in [\beta]$,
Let $\vQ \in \M(n_\gamma \times m_\gamma, \F_q)^\ell$ be a $\gamma$-quotient matrix tuple of $\vA$ with formatting vector $(v_1 + W_\gamma, \dots, v_{m_\gamma} + W_\gamma)$. 
For a row tuple $\vq \in T_{\vQ, i}$ for some $i \in [h_\gamma - 1]$, 
a row tuple $\va \in \rowtuplespace(\vA)$ is a \emph{row tuple extension} of $\vq$ in $\vA$ if $\va \in T_{i}$ and $\va + W_\gamma = \vq (v_1 + W_\gamma, \dots, v_{m_\gamma} + W_\gamma)^T$.


For an \IBCtuple $\vB$ of $\vQ$ with $\sigma$ rows satisfying the row tuple space property with respect to the hierarchical row tuple decomposition $T_{\vQ, 1}, \dots, T_{\vQ, h_\gamma - 1}$ and row tuple space parameters $t_{\vB, 1}, \dots, t_{\vB, \sigma}$, 
a matrix tuple $\vE$ is an \emph{\IBCtuple extension} of $\vB' \in \ibcspace_\vB$ in $\vA$ 
if $e_i \vE$ is a row tuple extension of $e_i \vB' \in T_{\vQ, t_{\vB, i}}$ for every $i \in [\sigma]$.



\end{definition}

Similar to the compatible block-compatible property for \IBCtuples as defined in Definition~\ref{def:characteristic_subspace} and Definition~\ref{def:ibctuple_invariant_subspace}, we extend the characteristic and block-compatible properties to linear spaces of \IBCtuple extensions. 

\begin{definition}[Characteristic Block-Compatible Extension Space]\label{def:ibctuple_invariant_extension}
Let $\vA$ be a matrix tuple, $T_1, \dots, T_\zeta$ be a hierarchical row tuple decomposition of $\vA$ with parameters $h_0, \dots, h_\beta$, and $\vQ$ be a $\beta$-quotient matrix tuple of $\vA$, and $\vB$ be an \IBCtuple of $\vQ$ satisfying the space property with \IBCtuple space $\ibcspace_\vB$ and \IBCtuple space kernel $\ibcspacekernel_{\vB}$. 
Let $\ibcextensionspace$ be a linear space of \IBCtuple extensions of matrix tuples in $\ibcspace_\vB$. 

$E$ is \emph{characteristic} if for any invertible matrices $L$ and $R$ such that $\vA = L \vA R^{-1}$, 
\[E = \langle\{ \vE R^{-1} : \vE \in E\} \rangle.\]

$E$ is \emph{block-compatible} if for any block-diagonalization $\vD = \diag(\vD_1, \dots, \vD_d)$ of $\vA$, any invertible matrices $L$ and $R$ such that $\vD = L\vA R^{-1}$, $\proj_{\vD, L}(\vE, i)$ is in $\ibcextensionspace$  
for 
every $\vE \in \ibcextensionspace$ and $i\in [d]$. 
\end{definition}

We prove some basic methods to obtain \IBCtuple extension spaces satisfying the characteristic and block-compatible properties in Claim~\ref{claim:ibctuple_compatible_exntesion_full} and Claim~\ref{claim:ibctuple_invariant_extension_refinement}.

\begin{claim}\label{claim:ibctuple_compatible_exntesion_full}
Let $\vA$ be a matrix tuple, $T_1, \dots, T_\zeta$ be a hierarchical row tuple decomposition of $\vA$ with parameters $h_0, \dots, h_\beta$, $\vQ$ be a $\gamma$-quotient matrix tuple of $\vA$ with formatting vector $(v_1 + W_\gamma, \dots, v_{m_\gamma} + W_\gamma)$, and $\vB$ be an \IBCtuple of $\vQ$ satisfying the space property (with \IBCtuple space $\ibcspace_\vB$ and \IBCtuple space kernel $\ibcspacekernel_{\vB}$) and block-compatible property. 
Then \[E = \langle \{\vE: \vE \text{ is an extension of } \vB' \in \ibcspace_\vB\}\rangle \text{ and  } F = \langle \{\vE: \vE \text{ is an extension of } \vB' \in \ibcspacekernel_\vB\}\rangle\] are characteristic block-compatible extension spaces. 
\end{claim}
\begin{proof}
We prove that $E$ is characteristic. 
Let $L$ and $R$ be arbitrary invertible matrices such that $\vA = L \vA R^{-1}$. 
We only need to show for any $\vE \in E$, 
$\vE R^{-1}$ is also in $E$.
Let $\sigma$ be the number of rows of $\vB$, and $t_{\vB, 1}, \dots, t_{\vB, \sigma}$ be the parameters of $\vB$. 
Since $T_1, \dots, T_\zeta$ are characteristic by the  hierarchical row tuple decomposition definition, $T_i = T_i R^{-1}$ for any $i \in [\zeta]$. 
Hence, $\vE R^{-1}$ is a matrix tuple such that $e_j \vE R^{-1}$ is a row tuple in $T_{t_{\vB, j}}$. 
So to prove $\vE R^{-1}$ is in $E$, we only need to show that there is a $\vB' \in \ibcspace_{\vB}$ such that $\vE R^{-1} + W_\gamma = \vB' (v_1 + W_\gamma, \dots, v_{m_\gamma} + W_\gamma)^T$. 

By the fourth property of Claim~\ref{claim:quotient_marix_tuple_basic},
there exists invertible matrices $L'$ and $R'$ such that $\vQ = L' \vQ R'^{-1}$ and 
\[\begin{bmatrix}
    v_1 + W_\gamma \\ \vdots \\ v_{m_\gamma} + W_\gamma
\end{bmatrix}  = R' \begin{bmatrix}
    v_1 R^{-1} + W_\gamma R^{-1} \\ \vdots \\ v_{m_\gamma} R^{-1} + W_\gamma R^{-1}
\end{bmatrix}.
\]
Let $\vB_\vE$ be the \IBCtuple of $\vQ$ such that $\vE$ is an extension of $\vB_\vE$. 
Using the fact that $W_\gamma$ is characteristic, 
we have 
\begin{align*}\vE R^{-1} + W_\gamma = & (\vE + W_\gamma) R^{-1} \\
= & \vB_{\vE} \begin{bmatrix}
    v_1 + W_\gamma \\ \vdots \\ v_{m_\gamma} + W_\gamma
\end{bmatrix} R^{-1}\\
= & \vB_\vE \begin{bmatrix}
    v_1 R^{-1} + W_\gamma \\ \vdots \\ v_{m_\gamma}R^{-1} + W_\gamma
\end{bmatrix} \\
= & \vB_{\vE} (R')^{-1} \begin{bmatrix}
    v_1 + W_\gamma \\ \vdots \\ v_{m_\gamma} + W_\gamma
\end{bmatrix}.\end{align*}
Note that for any \IBCtuple $\vB''$ of $\vQ$, $\vB'' (R')^{-1}$ is also an \IBCtuple for $\vQ$ by $\vQ = L' \vQ (R')^{-1}$, and $\vB'' (R')^{-1}$ is right-equivalent to $\vB''$. 
Hence, $\vE R^{-1}$ is in $E$ by letting $\vB' = \vB_\vE (R')^{-1}$.

there is a $\vB' \in \ibcspace_{\vB}$ such that $\vE R^{-1} + W_\gamma = \vB' (v_1 + W_\gamma, \dots, v_{m_\gamma} + W_\gamma)^T$. 
Hence, $E$ is characteristic.

Now we prove $E$ is block-compatible.
Let $\vD = \diag(\vD_1, \dots, \vD_d)$ be a block-diagonalization of $\vA$, and $L$ and $R$ be arbitrary invertible matrices such that $\vD = L \vA R^{-1}$. 
    We show that for every $\vE \in E$, $\proj_{\vD, L} (\vE, i)$ is in $E$ for every $i\in[d]$. 

    Let $\vC$ be the matrix tuple such that 
    \[\vC \cdot (v_1 + W_\gamma, \dots, v_{m_\gamma} + W_\gamma)^T = \vE + W_\gamma.\]
    By Definition~\ref{def:ibctuple}, $\vC$ is a matrix tuple in $\ibcspace_\vB$. 
    Let $\vC_{i}$ be the matrix tuple such that 
    \[\vC_i \cdot (v_1 + W_\gamma, \dots, v_{m_\gamma} + W_\gamma)^T = \proj_{\vD, L}(\vE, i) + W_\gamma\]
    for any $i\in[d]$. 
    By the block-compatible property, $\vC_i$ is $\proj_{\vQ_{\vA, \vD, L}, L_\vQ}(\vC, i)$ for any $i\in [d]$, where $\vQ_{\vA, \vD, L}$ is the matrix tuple defined as Claim~\ref{claim:correspondance_ibctuple_diag_quotient} and $L_\vQ$ is an arbitrary invertible matrix such that $\vQ_{\vA, \vD, L} = L_\vQ \vQ R_{\vQ}^{-1}$ for some invertible matrix $R_{\vQ}$. 
    Hence,
    $\vC_i$ is a matrix tuple in $\ibcspace_\vB$. By the correspondence between $\vD$ and $\vQ_{\vA, \vD, L}$, we have
    \[\vC_i \cdot (v_1 + W_\gamma, \dots, v_{m_\gamma} + W_\gamma)^T = \proj_{\vD, L}(\vE, i) + W_\gamma.\]
    Hence, $\proj_{\vD, L}(\vE, i)$ is in $E$ for every $i \in [d]$. 
    
    The proof for the characteristic and block-compatible properties of $F$ is similar. 
\end{proof}

\begin{claim}\label{claim:ibctuple_invariant_extension_refinement}
Let $\vA \in \M(n\times m, \F_q)^\ell$ be a matrix tuple, $T_1, \dots, T_\zeta$ be a hierarchical row tuple decomposition of $\vA$ with parameters $h_0, \dots, h_\beta$, $\vQ$ be a $\gamma$-quotient matrix tuple of $\vA$ with formatting vector $(v_1 + W_\gamma, \dots, v_{m_\gamma} + W_\gamma)$, and $\vB$ be an \IBCtuple of $\vQ$ satisfying the space property (with \IBCtuple space $\ibcspace_\vB$ and \IBCtuple space kernel $\ibcspacekernel_{\vB}$) and block-compatible property.

    Let $\ibcextensionspace$ be a characteristic block-compatible linear space of \IBCtuple extensions for matrix tuples in $\ibcspace_\vB$ such that every matrix tuple in $\ibcextensionspace$ has $\sigma$ rows. 
    We have the following properties:
    \begin{enumerate}
        \item For any row vector $u \in \F_q^{\sigma \cdot \ell}$, $\langle \{ u \cdot \mtupletomatrix(\vE) : \vE \in \ibcextensionspace \} \rangle$ is a characteristic block-compatible row vector space of $\vA$. 
        \item Let $u$ be a row vector in  $\F_q^{\sigma \cdot \ell}$ and $W_1', \dots, W_w'$ be a sequence of characteristic block-compatible row vector subspaces such that \[\langle \{ u \cdot \mtupletomatrix(\vE) : \vE \in \ibcextensionspace \} \rangle = W_1'\oplus \dots \oplus W_w'.\] 
        Then \[\langle \{ \vE \in \ibcextensionspace: \rowvecproj_{W_1', \dots, W_w'} (u \cdot \mtupletomatrix(\vE), W_k') = 0\} \rangle\] is a characteristic block-compatible linear space of \IBCtuple extensions for any $k \in [w]$. 
    \end{enumerate}
\end{claim}
\begin{proof}
Let $\vD = \diag(\vD_1, \dots, \vD_d)$ be an arbitrary block-diagonalization of $\vA$, and $L, R$ be arbitrary invertible matrices such that $\vD = L \vA R^{-1}$. Let $L', R'$ be invertible matrices such that $\vA = L' \vA (R')^{-1}$.

For the first property, 
let $W$ denote $\langle \{ u \cdot \mtupletomatrix(\vE) : \vE \in \ibcextensionspace \} \rangle$. 
By the linearity of $\ibcextensionspace$, for any $v \in W$, there is a $\vE \in \ibcextensionspace$ such that $v = u \cdot \mtupletomatrix(\vE)$. 

Since $E = \langle \{\vE' (R')^{-1} : \vE' \in E\}\rangle$, 
we have $v (R')^{-1} = u \cdot \mtupletomatrix(\vE) (R')^{-1} = u \cdot \mtupletomatrix(\vE (R')^{-1}) \in W$. Hence, $W$ is characteristic.

For the block-compatible property of $W$, 
we have $\vE = \sum_{i = 1}^d \proj_{\vD, L}(\vE, i)$, and $\proj_{\vD, L}(\vE, i) \in \ibcextensionspace$ by Definition~\ref{def:ibctuple_invariant_extension}. 
Hence, 
\[u \cdot \mtupletomatrix(\vE) = u\cdot \left( \sum_{i = 1}^d \mtupletomatrix(\proj_{\vD, L}(\vE, i))\right) = \sum_{i = 1}^d  \left(u\cdot \mtupletomatrix(\proj_{\vD, L}(\vE, i)) \right).\]
Since $\rowvectorspace(\proj_{\vD, L}(\vE, i))$ is a subspace of $\rowvectorspace(\vA_{i, \vD, L})$ for any $i\in[d]$, $u\cdot \mtupletomatrix(\proj_{\vD, L}(\vE, i))$ is in $\rowvectorspace(\vA_{i, \vD, L})$ for any $i\in[d]$. 
On the other hand, 
\[\rowvecproj_{\rowvectorspace(\vA_{1, \vD, L}), \dots, \rowvectorspace(\vA_{d, \vD, L})}(u \cdot \mtupletomatrix(\vE), \rowvectorspace(\vA_{i, \vD, L}))\]
is a vector in $\rowvectorspace(\vA_{i, \vD, L})$. 
By Fact~\ref{fact:ibctuple_diagonalzation_most_basic}, we have 
\[\rowvecproj_{\rowvectorspace(\vA_{1, \vD, L}), \dots, \rowvectorspace(\vA_{d, \vD, L})}(u \cdot \mtupletomatrix(\vE), \rowvectorspace(\vA_{i, \vD, L})) = u \cdot \mtupletomatrix(\proj_{\vD, L}(\vE, i))\]
for every $i \in [d]$. 
Since $u \cdot \mtupletomatrix(\proj_{\vD, L}(\vE, i))$ is in $W$ for every $i\in[d]$, $W$ is a block-compatible row vector space. Thus, we have the first property. 

For the second property, let $\ibcextensionspace'$ denote \[\langle \{ \vE \in \ibcextensionspace: \rowvecproj_{W_1', \dots, W_w'} (u \cdot \mtupletomatrix(\vE), W_k') = 0\} \rangle.\]

For the characteristic property of $E'$, we only need to show that $\vE R^{-1}$ is in $E'$ for any $\vE \in E'$. Let $\vE$ be an arbitrary matrix tuple in $E'$. Since $\vE R^{-1}$ is in $E$ by the characteristic property of $E$, we only need to show 
$\rowvecproj_{W_1', \dots, W_w'}(u \cdot \mtupletomatrix(\vE (R')^{-1}), W_k') = 0$.

Let $u_i = \rowtuplespace_{W_1', \dots, W_w'} (u\cdot \mtupletomatrix(\vE), W_i') \in W_i'$ for any $i \in [w]$.  
Since $W_1', \dots, W_w'$ are characteristic, 
$u_i R^{-1}$ is also in $W_i'$ for any $i\in[w]$.
Hence, \[\rowvecproj_{W_1', \dots, W_w'}(u\cdot\mtupletomatrix(\vE (R')^{-1}), W_i') = u_i R^{-1}.\] 
Since $u_k = 0$, we have $u_k (R')^{-1} = 0$. 
Thus, $\vE (R')^{-1}$ is in $E$. And $E'$ is characteristic.

For the block-compatible property, 
we only need to show that $\proj_{\vD, L}(\vE, i)$ is also in $\ibcextensionspace'$ for any $\vE \in \ibcextensionspace'$ and $i \in [d]$. 
Since $\ibcextensionspace$ is block-compatible, $\proj_{\vD, L}(\vE, i)$ is in $\ibcextensionspace$. 
By the definition of $\ibcextensionspace'$, we only need to show that 
\begin{equation}\label{equ:ibctuple_comp_extension_basic_operation} \rowvecproj_{W_1', \dots, W_w'}(u \cdot \mtupletomatrix(\proj_{\vD, L}(\vE, i)), W_k') = 0\end{equation}
for every $i \in [d]$. 
Since
\[\vE = \sum_{i = 1}^d \proj_{\vD, L}(\vE, i),\]
we have 
\[\mtupletomatrix(\vE) =\mtupletomatrix\left(\sum_{i = 1}^d \proj_{\vD, L}(\vE, i)\right) = \sum_{i = 1}^d \mtupletomatrix( \proj_{\vD, L}(\vE, i) ),\]
and thus 
\[\rowvecproj_{W_1', \dots, W_w'}(u \cdot \mtupletomatrix(\vE), W_k') = \sum_{i = 1}^d \rowvecproj_{W_1', \dots, W_w'}(u \cdot \mtupletomatrix(\proj_{\vD, L}(\vE, i)), W_k'). \] 
Note that $\rowvecproj_{W_1', \dots, W_w'}(u \cdot \mtupletomatrix(\vE), W_k')$ is a zero vector, and 
$\rowvecproj_{W_1', \dots, W_w'}(u \cdot \mtupletomatrix(\proj_{\vD, L}(\vE, i)), W_k')$ is a vector in $\rowvectorspace(\vA_{j, \vD, L}) \cap W_k'$. 
By Fact~\ref{fact:ibctuple_invariant_fund}, we have 
\[W_k' =( \rowvectorspace(\vA_{1, \vD, L}) \cap W_k') \oplus \dots \oplus (\rowvectorspace(\vA_{d, \vD, L}) \cap W_k'),\]
and thus Equation (\ref{equ:ibctuple_comp_extension_basic_operation}) holds for every $i\in[d]$.
\end{proof}

The major goal of this section is to construct a linear space spanned by essential extensions, which are defined as follows.

\begin{definition}\label{def:essential_extension}
Let $\vA \in \M(n\times m, \F_q)^\ell$ be a matrix tuple with a depth-$\beta$ hierarchical row tuple decomposition $T_1, \dots, T_\zeta$ and parameters $h_0, \dots, h_\beta$, and $\vQ \in \M(n_\beta \times m_\beta, \F_q)^\ell$ be a $\beta$-quotient matrix tuple of $\vA$. 
Let $\vB$ be an \IBCtuple of $\vQ$, and $\vC$ be an \IBCtuple extension of $\vB'$ for some $\vB' \in \ibcspace_{\vB}$. 

$\vC$ is \emph{an essential extension} if 
for any block-diagonalization $\vD = \diag(\vD_1, \dots, \vD_d)$ of $\vA$ and any invertible matrices $L$ and $R$ satisfying $\vD = L \vA R^{-1}$ such that there is an \IBCtuple extension $\vF$ of $\vB'$ in $\vA$ satisfying  $\rowtuplespace(\vF) \leq \rowtuplespace(\vA_{i, \vD, L})$ for some $i \in [d]$, 
then $(\rowvectorspace(\vC) \cap W_\beta) \leq (\rowvectorspace(\vA_{i, \vD, L}) \cap W_\beta)$. 
\end{definition}

We use the following algorithm to compute a characteristic block-compatible extension space such that for every extension in the linear space, if the extension is an extension of an \IBCtuplenospace, then the extension is essential.

\begin{framed}
\noindent \textbf{Essential Extension Algorithm}

\noindent \textbf{Input:} 
\begin{enumerate}
    \item Matrix tuple $\vA \in \M(n \times m, \F_q)^\ell$ with hierarchical row tuple decomposition $T_1, \dots, T_\zeta$ and parameters $h_0, \dots, h_\beta$ of $\vA$; 
    \item $\beta$-quotient matrix tuple $\vQ \in \M(n_\beta \times m_\beta, \F_q)^\ell$ of $\vA$ with formatting vector $(v_1 + W_\beta, \dots, v_{m_\beta} + W_\beta)$; 
    \item \IBCtuple $\vB \in \M(\sigma \times m_\beta, \F_q)^\ell$ of $\vQ$ satisfying the four properties of Definition~\ref{def:ibctuple_rel} with \IBCtuple space $\ibcspace_\vB$, \IBCtuple space kernel $\ibcspacekernel_\vB$, and row tuple space parameters $t_{\vB, 1}, \dots, t_{\vB, \sigma}$;
    \item Characteristic block-compatible row vector subspaces $W_{\beta, 1}, \dots, W_{\beta, w}$ as the output of the Row Vector Space Direct Sum Algorithm for $T_{h_\beta}$
    such that $W_\beta = W_{\beta, 1} \oplus \dots \oplus W_{\beta, w}$. 
\end{enumerate}

\noindent \textbf{Output:} $S < T_k$ for some $k \in [\zeta]$, or $W < W_{\beta, j}$ for some $j\in[w]$, 
or $\ibcextensionspace_{\vB, \vA}$ and $\ibcextensionspacekernel_{\vB, \vA}$.

\noindent \textbf{Algorithm:}

\begin{enumerate}
    \item Let $Y_\vB, Z_\vB$ be the output of the Row Vector Extraction Algorithm for $\vB$, and \[\hspace{-.2cm}\ibcextensionspace_{\vB, \vA} = \langle \{\vE :  \vE \text{ is an extension of } \vC \in \ibcspace_\vB\}\rangle, \ibcextensionspacekernel_{\vB, \vA} = \langle \{\vE : \vE \text{ is an extension of } \vC \in \ibcspacekernel_\vB\}\rangle. \]
    
    \item For $i = 1, \dots, \sigma \cdot \ell$, and $j = 1, \dots, w$, 
    let 
    \[\ibcextensionspace_{\vB, \vA}^\dagger = \langle \{\vC \in \ibcextensionspace_{\vB, \vA} : \rowvecproj_{W_{\beta, 1}, \dots, W_{\beta, w}}(e_i (I - Z_\vB Y_\vB ) \mtupletomatrix(\vC), W_{\beta, j}) = 0\}\rangle,\]
    $\ibcextensionspacekernel_{\vB, \vA}^\dagger = \ibcextensionspacekernel_{\vB, \vA} \cap \ibcextensionspace_{\vB, \vA}^\dagger$, $d = \dim(\ibcextensionspace_{\vB, \vA}^\dagger / W_\beta) - \dim(\ibcextensionspacekernel_{\vB, \vA}^\dagger / W_\beta)$, and run the following steps:
    \begin{enumerate}
        \item If $d = \dim(\ibcspace_\vB) - \dim(\ibcspacekernel_\vB)$, then use $\ibcextensionspace_{\vB, \vA}^\dagger$ to replace $\ibcextensionspace_{\vB, \vA}$, $\ibcextensionspacekernel_{\vB, \vA}^\dagger$ to replace $\ibcextensionspacekernel_{\vB, \vA}$. 
        \item If $0 < d  < \dim(\ibcspace_\vB) - \dim(\ibcspacekernel_\vB)$, then return 
        $\langle \{e_1 \vC : \vC \in \ibcextensionspace_{\vB, \vA}^\dagger\}\rangle$ for $T_{t_{\vB, 1}}$.
        \item If $d = 0$ and        
        $0 < \dim(\ibcextensionspace_{\vB, \vA} / W_\beta) - \dim(\ibcextensionspace_{\vB, \vA}^\dagger / W_\beta) < \dim(W_{\beta, j})$, then return 
        \[\langle \{ \rowvecproj_{W_{\beta,1},\dots,W_{\beta, w}}(e_i (I - Z_\vB Y_\vB ) \mtupletomatrix(\vC), W_{\beta, j}) : \vC \in \ibcextensionspace_{\vB, \vA}\}\rangle.\] 

        

    \end{enumerate}
    \item For each $i \in [\sigma]$, if $\langle \{ e_i \vC : \vC \in \ibcextensionspace_{\vB, \vA}\}\rangle \neq T_{t_{\vB, i}}$, then return $\langle \{ e_i \vC : \vC \in \ibcextensionspace_{\vB, \vA}\}\rangle$ for  $T_{t_{\vB, i}}$. 
    \item Return $\ibcextensionspace_{\vB, \vA}$ and $\ibcextensionspacekernel_{\vB, \vA}$.
\end{enumerate}
\vspace{-.3cm}
\end{framed}

In the rest of this subsection, we show the Essential Extension Algorithm outputs a canonical linear space spanned by essential extension space (Lemma~\ref{lemma:essential_extension_ibctuple_compatible} and Lemma~\ref{lemma:extension_canonical}).

\begin{claim}\label{claim:properties_ibctuple_extension}
If for the input $\vA, T_1, \dots, T_\zeta, \vQ, \vB, W_{\beta, 1}, \dots, W_{\beta, w}$ of the Essential Extension Algorithm, 
the output is $\ibcextensionspace_{\vB, \vA}$ and $\ibcextensionspacekernel_{\vB, \vA}$, then $\ibcextensionspace_{\vB, \vA}$ and $\ibcextensionspacekernel_{\vB, \vA}$ satisfy the following conditions:
    \begin{enumerate}
        \item Let $S$ be the set of pairs $(i, j)$ such that 
        \begin{equation}\label{equ: extension_property_basic}\rowvecproj_{W_{\beta, 1}, \dots, W_{\beta, w}}(e_i (I - Z_\vB Y_\vB)\mtupletomatrix(\vC), W_{\beta, j}) = \vecz\end{equation} for every $\vC \in \ibcextensionspace_{\vB, \vA}$.
        An extension $\vC'$ for an \IBCtuple of $\ibcspace_\vB$ in $\vA$ is in $\ibcextensionspace_{\vB, \vA}$ iff Equation (\ref{equ: extension_property_basic}) holds for every $(i, j) \in S$. 
        \item For any $\vC \in \ibcextensionspace_{\vB, \vA}$, $(I - Z_\vB Y_\vB)\mtupletomatrix(\vC)$ is a zero matrix iff $\vC \in \ibcextensionspacekernel_{\vB, \vA}$, and for any $\vC, \vC' \in \ibcextensionspace_{\vB, \vA}$, 
        \[(I - Z_\vB Y_\vB)\mtupletomatrix(\vC) = (I - Z_\vB Y_\vB)\mtupletomatrix(\vC')\] iff $\vC - \vC' \in \ibcextensionspacekernel_{\vB, \vA}$. 
        \item For every \IBCtuple $\vB' \in \ibcspace_{\vB}$, there is an \IBCtuple $\vB''$ strongly correlated to $\vB'$ such that $\vB''$ has an extension $\vC \in \ibcextensionspace_{\vB, \vA}$.  
        \item  (Space property) 
        For any $\vC \in \ibcextensionspace_{\vB, \vA}$, $\vC$ is in $\ibcextensionspacekernel_{\vB, \vA}$ iff Equation (\ref{equ: extension_property_basic}) holds  for any $(i, j) \in S$. 
        \item (Row tuple space property) For every $i \in [\sigma]$, $\langle \{e_i \vC : \vC \in \ibcextensionspace_{\vB, \vA}\}\rangle = T_{t_{\vB, i}}$, and any $\vC \in \ibcextensionspace_{\vB, \vA}$ is in $\ibcextensionspacekernel_{\vB, \vA}$ iff $e_i \vC \in \ibcspacekernel_{t_{\vB, i}}$.
        \item (Dimension property) For any $\vC, \vC'$ in $\ibcextensionspace_{\vB, \vA}$ but not in $\ibcextensionspacekernel_{\vB, \vA}$ and any $\vF \in \ibcextensionspacekernel_{\vB, \vA}$, 
        \[\dim(\rowvectorspace(\vF)) < \dim(\rowvectorspace(\vC)) = \dim(\rowtuplespace(\vC')).\]         

    \end{enumerate}
\end{claim}
\begin{proof}
We first prove the following properties:
\begin{enumerate}
        \item[(a)] Every matrix tuple in $\ibcextensionspace_{\vB, \vA}$ is an extension of some matrix tuple in $\ibcspace_\vB$ in $\vA$.
        A matrix tuple $\vC \in \ibcextensionspace_{\vB, \vA}$ is an extension of some matrix tuple of $\ibcspacekernel_\vB$ in $\vA$ iff $\vC$ is in $\ibcextensionspacekernel_{\vB, \vA}$.
        \item[(b)] $\dim(\ibcextensionspace_{\vB, \vA}) - \dim(\ibcextensionspacekernel_{\vB, \vA}) = \dim(\ibcspace_\vB) - \dim(\ibcspacekernel_\vB)$.
        \item[(c)] For any $\vC \in \ibcextensionspacekernel_{\vB, \vA}$, 
        $(I - Z_\vB Y_\vB ) \mtupletomatrix(\vC)$ is a zero matrix tuple. 
        \item[(d)] Let $\vC$ and  $\vC'$ be two matrix tuples in $\ibcextensionspace_{\vB, \vA}$ but not in $\ibcextensionspacekernel_{\vB, \vA}$, then for any $i \in [\sigma], j \in [w]$
        \[\rowvecproj(e_i(I - Z_\vB Y_\vB ) \mtupletomatrix(\vC), W_{\beta, j}) \neq 0\] iff
        \[\rowvecproj(e_i(I - Z_\vB Y_\vB ) \mtupletomatrix(\vC'), W_{\beta, j}) \neq 0.\] 
\end{enumerate}
    For the property (a), let $\vB', \vB''$ be two arbitrary matrix tuples in $\ibcspace_{\vB}$, and $\vC', \vC''$ be extensions of $\vB'$ and $\vB''$ respectively. By the linearity of $\ibcspace_{\vB}$, $\vB + \vB'$ is also in $\ibcspace_{\vB}$. By Definition~\ref{def:extension}, $\vC' + \vC''$ is an extension of $\vB' + \vB''$.
    By the fact that $\ibcspacekernel_\vB$ is a subspace of $\ibcspace_\vB$, the property (a) holds for $\ibcextensionspace_{\vB, \vA}$ and $\ibcextensionspacekernel_{\vB, \vA}$ after Step 1 of he Essential Extension Algorithm. 
    By induction on Step 2 of the Essential Extension Algorithm, 
    the property (a) also holds for the algorithm's output. 

    For the property (b), let $U$ be any of $\ibcextensionspace_{\vB, \vA}, \ibcextensionspacekernel_{\vB, \vA}, \ibcextensionspace_{\vB,  \vA}^\dagger, \ibcextensionspacekernel_{\vB, \vA}^\dagger$ at any step of the algorithm, 
    if $\vC, \vC' \in U$ are the extensions of the same $\vB' \in \ibcspace_\vB$, 
    then for any $\vC'' \in U$, $\vC'' + (\vC - \vC')$ is also in $U$ by the linearity. 
    Since $\ibcspace_{\vB}$ is a subspace of $\ibcspacekernel_{\vB}$, after Step 1 of the algorithm, the property (b) holds. 
    For Step 2, if the property (b) holds before the iteration of Step 2 for some $i$ and $j$, then the property (b) also holds after the iteration of Step 2 for $i$ and $j$. 
    By induction, the property (b) holds after Step 2, and thus the property (b) holds for the output of the algorithm. 

    For the property (c), we analyze Step 2 of the Essential Extension Algorithm for some $i$ and $j$.  
    Because the algorithm returns linear spaces of matrix tuples, 
    it can be neither the case of 2(b) nor 2(c). 
    If it is the case of 2(a) for some $i$ and $j$,     
    then the final $\ibcextensionspacekernel_{\vB, \vA}$ as the output of the algorithm is a subspace of $\ibcextensionspacekernel_{\vB, \vA}^\dagger$ constructed for the $i$ and $j$, and thus 
    \begin{equation}\label{equ: extension_property_basic_0}\rowvecproj(e_i (I - Z_\vB Y_\vB) \mtupletomatrix(\vC), W_{\beta, j}) = 0\end{equation} for every $\vC \in \ibcextensionspacekernel_{\vB, \vA}$ after Step 2 for the $i$ and $j$ by the definition of $\ibcextensionspacekernel_{\vB, \vA}^\dagger$. 
    
    If none of the three cases of Step 2 for $i$ and $j$ applies, then either 
    \begin{equation}\label{equ: extension_property_basic_1}\dim(\ibcextensionspace_{\vB, \vA} / W_\beta) = \dim(\ibcextensionspace_{\vB, \vA}^\dagger / W_\beta)\end{equation}
    or \begin{equation}\label{equ: extension_property_basic_2}\dim(\ibcextensionspace_{\vB, \vA} / W_\beta) - \dim(\ibcextensionspace_{\vB, \vA}^\dagger / W_\beta) = \dim(W_{\beta, j})\end{equation}
    holds for $\ibcextensionspace_{\vB, \vA}, \ibcextensionspacekernel_{\vB, \vA}, \ibcextensionspace_{\vB,  \vA}^\dagger, \ibcextensionspacekernel_{\vB, \vA}^\dagger$ after Step 2 for the $i$ and $j$, because we always have \[0 \leq \dim(\ibcextensionspace_{\vB, \vA} / W_\beta) - \dim(\ibcextensionspace_{\vB, \vA}^\dagger / W_\beta) \leq \dim(W_{\beta, j})\] by the definition of $\ibcextensionspace_{\vB, \vA}^\dagger$. On the other hand, 
    Equation (\ref{equ: extension_property_basic_1}) is impossible because otherwise we have \[\dim(\ibcextensionspace_{\vB, \vA} / W_\beta) = \dim(\ibcextensionspacekernel_{\vB, \vA}^\dagger / W_\beta),\]
    which contradicts the property (b) at the end of the algorithm. Hence, Equation (\ref{equ: extension_property_basic_2}) holds if none of the three cases of Step 2 for $i$ and $j$ applies. 
    Consequently, $\ibcextensionspacekernel_{\vB, \vA} = \ibcextensionspacekernel_{\vB, \vA}^\dagger$ because otherwise $\dim(\ibcextensionspace'_{\vB, \vA} / W_\beta)$ cannot be equal to $\dim(\ibcextensionspacekernel'_{\vB, \vA} / W_\beta)$. Thus, 
    Equation (\ref{equ: extension_property_basic_0}) holds for every $\vC \in \ibcextensionspacekernel_{\vB, \vA}$ after Step 2 for $i$ and $j$. 
    By induction on $i$ and $j$, the property (c) holds.

    For the property (d), for each $i$ and $j$, if after Step 2 for $i$ and $j$, $\ibcextensionspace_{\vB, \vA}$ has some matrix tuple $\vC$ such that the projection of $e_i (I - Z_\vB Y_\vB) \mtupletomatrix(\vC)$ on $W_{\beta, j}$ with respect to $W_{\beta, 1}, \dots, W_{\beta, w}$ is non-zero, then it must the case that none of Step 2(a), 2(b), and 2(c) applies. 
    For this case, because $\ibcextensionspacekernel_{\vB, \vA} = \ibcextensionspacekernel_{\vB, \vA}'$ and $\dim(\ibcextensionspace_{\vB, \vA}' / W_\beta) = \dim(\ibcextensionspacekernel_{\vB, \vA} / W_\beta)$ after Step 2 for the $i$ and $j$, 
    any $\vC \in \ibcextensionspace_{\vB, \vA}$ such that the projection of $e_i (I - Z_\vB Y_\vB) \mtupletomatrix(\vC)$ on $W_{\beta, j}$ with respect to $W_{\beta, 1}, \dots, W_{\beta, w}$ is non-zero if and only if $\vC$ is in $\ibcextensionspace_{\vB, \vA}$ but not in $\ibcextensionspacekernel_{\vB, \vA}$. 
    Thus, the property (d) holds.


    Now, we prove the properties of the claim.
    The first and second property of the claim are obtained by the property (c) and (d). The third property is obtained by the property (b). 
    The fourth property is obtained by the property (c) and (d). 
    The fifth property is obtained by Step 3 of the algorithm and the fact that after Step 1 of the algorithm, any $\vC \in \ibcextensionspace_{\vB, \vA}$ is in $\ibcextensionspacekernel_{\vB, \vA}$ iff $e_i \vC$ is in $K_{t_{\vB, i}}$. 
    The last property is obtained by the first and the second properties and the fact that $\dim(\rowvectorspace(\vB')) < \dim(\rowvectorspace(\vB''))$ for any $\vB' \in \ibcspacekernel_\vB$ and any $\vB'' \in \ibcspace_\vB$ but not in $\ibcspacekernel_\vB$ by the dimension property of $\vB$. 
\end{proof}

\begin{claim}\label{claim:essential_property_basic}
If for the input $\vA, T_1, \dots, T_\zeta, \vQ, \vB, W_{\beta, 1}, \dots, W_{\beta, w}$ of the Essential Extension Algorithm with $(v_1 + W_\beta, \dots, v_{m_\beta} + W_\beta)$ as the formatting matrix tuple of $\vQ$, 
the output is $\ibcextensionspace_{\vB, \vA}$ and $\ibcextensionspacekernel_{\vB, \vA}$, then $\ibcextensionspace_{\vB, \vA}$ and $\ibcextensionspacekernel_{\vB, \vA}$ has the following properties for any block-diagonalization $\vD  = \diag(\vD_1, \dots, \vD_d)$ of $\vA$, and arbitrary invertible marices $L$ and $R$ such that $\vD = L \vA R^{-1}$:
\begin{enumerate}
    \item Let $\vB_0 \in \ibcspace_\vB$ be an \IBCtuple of $\vQ$ with $\sigma$ rows such that there is an integer $i\in[d]$ satisfying \[e_j \vB_0 (v_1 + W_\beta, \dots, v_{m_\beta} + W_\beta)^T \in \langle \{ \va + W_\beta: \va \in \rowtuplespace(\vA_{i, \vD, L}) \} \rangle\]
    for every $j \in [\sigma]$. 
    Then there is an \IBCtuple extension $\vC  \in \ibcextensionspace_{\vB, \vA}$ for an \IBCtuple $\vB_1 \in \ibcspace_{\vB}$ strongly correlated to $\vB_0$ such that $\rowtuplespace(\vC) \leq \rowtuplespace(\vA_{i, \vD, L})$. 
    \item 
    Let $\vB_0 \in \ibcspace_\vB$ be an \IBCtuple of $\vQ$, and $\vC_0$ be an \IBCtuple extension of $\vB_0$ in $\ibcextensionspace_{\vB, \vA}$ but not in $\ibcextensionspacekernel_{\vB, \vA}$ such that $\rowtuplespace(\vC_0) \leq \rowtuplespace(\vA_{i, \vD, L})$ for some $i \in [d]$. 
    Then for any \IBCtuple $\vB_1 \in \ibcspace_{\vB}$ that is correlated to $\vB_0$, there is an \IBCtuple $\vB_2 \in \ibcspace_{\vB}$ strongly correlated to $\vB_1$ and a $\vC_2 \in \ibcextensionspace_{\vB, \vA}$ as an \IBCtuple extension of $\vB_2$ such that $\rowvectorspace(\vC_2) \leq \rowtuplespace(\vA_{i, \vD, L})$. 
\end{enumerate}

\end{claim}
\begin{proof}
For the first property, by the third property Claim~\ref{claim:properties_ibctuple_extension}, there exists an \IBCtuple $\vB'$ strongly correlated to $\vB_0$ such that $\vB'$ has an extension $\vC'$ in $\ibcextensionspace_{\vB, \vA}$ but not in $\ibcextensionspacekernel_{\vB, \vA}$. 
Let $\vC_k'$ be $\proj_{\vD, L}(\vC', k)$ for all the $k \in [d]$. 
By Claim~\ref{claim:ibctuple_compatible_exntesion_full} and Claim~\ref{claim:ibctuple_invariant_extension_refinement}, $\vC_k' \in \ibcextensionspace_{\vB, \vA}$ for all the $k \in [d]$. 
By the algorithm, there are $\vB_1',\dots,\vB_d' \in \ibcspace_{\vB}$ such that $\vC_k'$ is an extension of $\vB_k'$ for all the $k \in [d]$. 

Let $\vQ_{\vA, \vD, L}$ be the matrix tuple defined in Claim~\ref{claim:correspondance_ibctuple_diag_quotient}. 
By Claim~\ref{claim:quotient_marix_tuple_basic}, there exist invertible matrices $L_\vQ$ and $R_\vQ$ such that $\vQ_{\vA, \vD, L} = L_\vQ \vQ R_\vQ$. 
By the construction of $\vQ_{\vA, \vD, L}$, 
$\rowtuplespace(\vB_0)$ is a subspace of $\rowtuplespace(\vQ_{i, \vQ_{\vA, \vD, L}, L_\vQ})$, and 
$\rowtuplespace(\vB_k')$ is a subspace of $\rowtuplespace(\vQ_{k, \vQ_{\vA, \vD, L}, L_\vQ})$ for every $k \in [d]$. 
Since $\vB_0 - \vB' = \vB_0 - (\vB_1' + \dots + \vB_d')$ is in $\ibcspacekernel_{\vB}$, 
$\vB_k'$ is in $\ibcspacekernel_{\vB'}$ for all the $k\in[d]$ except $k = i$.
Hence, $\vC_i'$ is an extension of an \IBCtuple strongly correlated to $\vB_0$ with all row tuples in $\rowtuplespace(\vA_{i, \vD, L})$.
Thus the first property holds.

Now we prove the second property. 
We observe that if a matrix tuple $\vC^\dagger$ that is an extension of an \IBCtuple $\vB^\dagger$ of $\vQ$ in $\vA$ has non-zero projection only on $k'$-th block of $\vD$, then $\vB^\dagger (v_1 + W_\beta, \dots, v_{m_\beta}+ W_\beta)^T$ has every row tuple in $\rowtuplespace(\vA_{k', \vD, L}) / W_\beta$. 

Because $\ibcextensionspacekernel_{\vB, \vA}$ contains all the matrix tuples in $\ibcextensionspace_{\vB, \vA}$ that are extensions of $\ibcspacekernel_{\vB}$, and $\dim(\ibcextensionspace_{\vB, \vA}) - \dim(\ibcextensionspacekernel_{\vB, \vA}) = \dim(\ibcspace_\vB) - \dim(\ibcspacekernel_\vB)$,
there is always a $\vB'' \in \ibcspace_{\vB}$ strongly correlated to $\vB_1$ (and thus correlated to $\vB_0$) such that there is matrix tuple $\vC'' \in \ibcextensionspace_{\vB, \vA}$ that is an extension of $\vB''$. Since 
$\vB''$ and $\vB_0$ are correlated, by Definition~\ref{def:ibctuple}, there exist an invertible matrix $L''$ such that $L''\vB_0$ is in $\ibcspace_\vB$ and $L'' \vB_0 - \vB''$ is in $\ibcspacekernel_\vB$.

Because $\rowtuplespace(\vC_0) \leq \rowtuplespace(\vA_{i, \vD, L})$, $\vB_0 (v_1 + W_\beta, \dots, v_{m_\beta}+ W_\beta)^T$ has every row tuple in $\rowtuplespace(\vA_{i, \vD, L}) / W_\beta$, and thus $L'' \vB_0 (v_1 + W_\beta, \dots, v_{m'}+ W_\beta)^T$ has every row tuple in $\rowtuplespace(\vA_{i, \vD, L}) / W_\beta$.
By the first property, $\proj_{\vD, L}(\vC'', i)$ is an extension of an \IBCtuple correlated to $\vB_0$ with all row tuples in $\rowtuplespace(\vA_{i, \vD, L})$. 
\end{proof}

\begin{lemma}\label{lemma:essential_extension_ibctuple_compatible}
    For the input $\vA \in \M(n \times m, \F_q)^\ell, T_1, \dots, T_\zeta, \vQ, \vB, W_{\beta, 1}, \dots, W_{\beta, w}$ of the Essential Extension Algorithm, 
    the running time of the algorithm is $\poly(n, m, \ell, \log q)$, and 
the output satisfies the following conditions:
\begin{enumerate}
    \item If the output is $S < T_k$ for some $k\in[\zeta]$, then 
    $S$ is a characteristic block-compatible row tuple subspace such that $\dim(S / W_\gamma) > 0$, where $\gamma$ is the integer such that $h_{\gamma-1} \leq k < h_\gamma$.
    \item If the output is $W < W_{\beta, j}$ for some $j\in[w]$, then 
    $W$ is a characteristic block-compatible row vector space such that $\dim(W) > 0$. 
    \item If the output is $\ibcextensionspace_{\vB, \vA}$ and $\ibcextensionspacekernel_{\vB, \vA}$, then $\ibcextensionspace_{\vB, \vA}$ and $\ibcextensionspacekernel_{\vB, \vA}$ are characteristic block-compatible extension space of $\ibcspace_\vB$ and $\ibcspacekernel_\vB$ respectively such that for every $\vC \in \ibcextensionspace_{\vB, \vA}$ such that $\vC$ is an extension of an \IBCtuple in $\ibcspace_\vB$, then $\vC$ is an essential extension. 
\end{enumerate}
\end{lemma}
\begin{proof}
    The running time of the algorithm is obtained by Fact~\ref{fact:ibctuple_char_matrix}, the fact that $\zeta \leq n$ by Definition~\ref{def:row_tuple_decomposition}, $w \leq m$ as $W_\beta = W{\beta, 1} \oplus \dots \oplus W_{\beta, w}$, and the observation of $\dim(\ibcspace_{\vB}), \dim(\ibcspacekernel_\vB) < \sigma \cdot n$ by Definition~\ref{def:ibctuple}. 

    In the rest of this proof, we prove the properties for the output of the algorithm. 
    We first show that $\ibcextensionspace_{\vB, \vA}$, $\ibcextensionspacekernel_{\vB, \vA}$, $\ibcextensionspace_{\vB, \vA}^\dagger$, and $\ibcextensionspacekernel_{\vB, \vA}^\dagger$ obtained in every step of the algorithm are characteristic and block-compatible. 
    By Claim~\ref{claim:ibctuple_compatible_exntesion_full},  $\ibcextensionspace_{\vB, \vA}$ and $\ibcextensionspacekernel_{\vB, \vA}$ obtained in Step 1 of the algorithm are characteristic block-compatible. 
        
    Since every $\vC \in \ibcextensionspace_{\vB, \vA}$ is an extension of $\vB' \in \ibcspace_{\vB}$, and $(I - Z_\vB Y_\vB) \mtupletomatrix(\vB)$ is a zero matrix by Fact~\ref{fact:ibctuple_char_matrix}, 
    $e_i (I - Z_\vb Y_\vB)\mtupletomatrix(\vC)$ is a vector in $W_\beta$ for every $i \in [\sigma \cdot \ell]$ and $\vC \in \ibcextensionspace_{\vB, \vA}$. 
    On the other hand,  $W_{\beta,1}, \dots, W_{\beta, w}$ are characteristic block-compatible row vector subspaces such that $W_\beta = W_{\beta, 1} \oplus \dots \oplus W_{\beta, w}$. 
    By Claim~\ref{claim:ibctuple_invariant_extension_refinement}, $\ibcextensionspace_{\vB, \vA}^\dagger$ and $\ibcextensionspacekernel_{\vB, \vA}^\dagger$ for every $i$ and $j$ obtained in Step 2 of the algorithm are characteristic block-compatible. Hence, if the output is $\ibcextensionspace_{\vB, \vA}$ and $\ibcextensionspacekernel_{\vB, \vA}$, then they are characteristic block-compatible.  

    Now we prove that for the output $\ibcextensionspace_{\vB, \vA}$, 
    every $\vC \in \ibcextensionspace_{\vB, \vA}$ that is an extension of an \IBCtuple $\vB'$ in $\ibcspace_{\vB}$, $\vC$ is essential. 
    Let $\vD = \diag(\vD_1, \dots, \vD_d)$ be an arbitrary block-diagonalziation of $\vA$, and $L, R$ be arbitrary invertible matrices such that $\vD = L \vA R^{-1}$. Suppose $\vF$ is an extension of $\vB'$ in $\vA$ such that $\rowtuplespace(\vF) \leq \rowtuplespace(\vA_{i, vD, L})$ for some $i\in [d]$. 
    By Claim~\ref{claim:essential_property_basic}, 
    there is a $\vC'$ such that $\rowtuplespace(\vC') \leq \rowtuplespace(\vA_{i, \vD, L})$ and $\vC'$ is an extension of an \IBCtuple strongly correlated to $\vB'$. 
    By Claim~\ref{claim:properties_ibctuple_extension}, $\rowvectorspace(\vC) \cap W_\beta = \rowvectorspace(\vC') \cap W_\beta$. Hence, $\vC$ is essential. 
    Thus, the third property of the claim holds.

    If the output is a row vector space, then the output is obtained in Step 2(c) for some $i$ and $j$. The condition of \[0 < \dim(\ibcextensionspace_{\vB, \vA} / W_\beta) - \dim(\ibcextensionspace_{\vB, \vA}^\dagger / W_\beta) < \dim(W_{\beta, j})\] implies that 
    $W$ is a nontrivial strict subspace of $W_{\beta, j}$. By Claim~\ref{claim:ibctuple_invariant_extension_refinement}, $W$ is characteristic block-compatible. 

    In the rest of this proof, we consider the case that the output is a row tuple space $S < T_k$ for some $k \in [\zeta]$. 
    If the output is obtained in Step 2(b) for some $i$ and $j$, then by the construction of $\ibcextensionspace_{\vB, \vA}$, $\ibcextensionspacekernel_{\vB, \vA}$, $\ibcextensionspace_{\vB, \vA}^\dagger$, and $\ibcextensionspacekernel_{\vB, \vA}^\dagger$,
    it means that there exists some $\vB' \in \ibcspace_{\vB}$ but not in $\ibcspacekernel_{\vB}$ such that every \IBCtuple  in $\ibcspace_{\vB}$ strongly correlated to $\vB'$ has no \IBCtuple extension in $\ibcextensionspace_{\vB, \vA}^\dagger$, and there exists another $\vB'' \in \ibcspace_{\vB}$ but not in $\ibcspacekernel_{\vB}$ such that $\vB''$ has an \IBCtuple extension in $\ibcextensionspace_{\vB, \vA}^\dagger$. By the row tuple space property of $\vB$, the output row tuple subspace satisfies $0 < \dim(S / W_\gamma) < \dim(T_{t_{\vB, 1}} / W_\gamma)$. Furthermore, the output is characteristic block-compatible by the fact that $\ibcextensionspace_{\vB, \vA}$ is characteristic block-compatible and Claim~\ref{claim:ibctuple_invariant_extension_refinement}. 
    
    If the output is obtained in Step 3 for some $i$, 
    then by the row tuple space property of $\vB$, $\langle \{ e_i \vC : \vC \in \ibcextensionspace_{\vB, \vA}\}\rangle$ is a subset of $T_{t_{\vB, i}}$.  
    On the other hand, by Step 2 of the algorithm, 
    we have \[\dim(\ibcextensionspace_{\vB, \vA}) - \dim(\ibcextensionspace_{\vB, \vA}) = \dim(\ibcspace_\vB) - \dim(\ibcspacekernel_\vB)\]
    for $\ibcextensionspace_{\vB, \vA}$ and $\ibcextensionspacekernel_{\vB, \vA}$ after Step 2. 
    By the row tuple space property of $\vB$, this means that every \IBCtuple $\vB' \in \ibcspace_\vB$ has a strongly correlated \IBCtuple $\vB'' \in \ibcspace_\vB$ such that there is a $\vC \in \ibcextensionspace_{\vB, \vA}$ such that $\vC$ is an extension of $\vB''$. 
    So we have $\langle \{ e_i \vC : \vC \in \ibcextensionspace_{\vB, \vA}\}\rangle / W_\gamma = T_{t_{\vB, i}} / W_\gamma$ for $\gamma$ such that $h_{\gamma - 1} \leq t_{\vB, i} < h_\gamma$, and thus $\dim(\langle \{ e_i \vC : \vC \in \ibcextensionspace_{\vB, \vA}\}\rangle / W_\gamma) > 0$. 
    Finally, $\langle \{ e_i \vC : \vC \in \ibcextensionspace_{\vB, \vA}\}\rangle $ is characteristic block-compatible by the fact that $\ibcextensionspace_{\vB, \vA}$ is characteristic block-compatible. 
\end{proof}

\begin{lemma}\label{lemma:extension_canonical}

Suppose $\vA, T_1, \dots, T_\zeta, \vQ, \vB, W_{\beta, 1}, \dots, W_{\beta, w}$ and  $\vA', T_1', \dots, T_\zeta', \vQ', \vB', W_{\beta, 1}', \dots, W_{\beta, w}'$ are two inputs of the Essential Extension Algorithm such that there exist invertible matrices $L, R$ satisfying the following conditions:
\begin{enumerate}
    \item[(a)] $\vA' = L \vA R^{-1}$.
    \item[(b)] $T_i' = T_i R^{-1}$ for any $i \in [\zeta]$, and the parameters for $T_1, \dots, T_\zeta$ and  $T_1', \dots, T_\zeta'$ are the same.  
    \item[(c)] $\vB$ and $\vB'$ are right-equivalent with the parameters $t_{\vB, i} = t_{\vB', i}$ for any $i \in \sigma$, where $\sigma$ is the number of rows in $\vB$ and $\vB'$. 
    \item[(d)] $W_{\beta, j}' = W_{\beta, j} R^{-1}$. 
\end{enumerate}
    The outputs of the Essential Extension Algorithms for the two inputs satisfy the following conditions: (Let $L, R$ be arbitrary invertible matrices $L, R$ such that conditions (a) - (e) are satisfied.)
    \begin{enumerate}
        \item If the algorithm for $\vA, T_1, \dots, T_\zeta, \vQ, \vB, W_{\beta, 1}, \dots, W_{\beta, w}$ outputs a row vector subspace $W$, then the algorithm for $\vA', T_1', \dots, T_\zeta', \vQ', \vB', W_{\beta, 1}', \dots, W_{\beta, w}'$ outputs $W R^{-1}$.  
        \item If the algorithm for $\vA, T_1, \dots, T_\zeta, \vQ, \vB, W_{\beta, 1}, \dots, W_{\beta, w}$ outputs a row tuple subspace $S$, then the algorithm for $\vA', T_1', \dots, T_\zeta', \vQ', \vB', W_{\beta, 1}', \dots, W_{\beta, w}'$ outputs $S R^{-1}$.  
        \item If the algorithm for $\vA, T_1, \dots, T_\zeta, \vQ, \vB, W_{\beta, 1}, \dots, W_{\beta, w}$ outputs \IBCtuple extension spaces $\ibcextensionspace_{\vB, \vA}$ and $\ibcspacekernel_{\vB, \vA}$, then the algorithm for $\vA', T_1', \dots, T_\zeta', \vQ', \vB', W_{\beta, 1}', \dots, W_{\beta, w}'$ outputs $\{\vE R^{-1} : \vE \in \ibcextensionspace_{\vB, \vA}\}$ and $\{\vE R^{-1} : \vE \in \ibcextensionspacekernel_{\vB, \vA}\}$.
    \end{enumerate}
\end{lemma}
\begin{proof}
Let $L$ and $R$ be arbitrary invertible matrices such that the conditions (a) - (e) hold.  
    By Claim~\ref{claim:quotient_marix_tuple_basic}, there exist invertible matrices $L'$ and $R'$ such that the following conditions hold:
    \begin{enumerate}
        \item $\vQ' = L' \vQ (R')^{-1}$.
        \item $T_{\vQ', i} = T_{\vQ, i} (R')^{-1}$ for every $1 \leq i \leq h_\beta - 1$.
        \item $v_i'R^{-1} + W_\beta' R^{-1} = (\sum_{j = 1}^{m'} v_j r_{i, j}') + W_\beta$, where $r_{i, j}'$ is the $i$-th row $j$-th column of $(R')^{-1}$. 
    \end{enumerate}
    Consequently, we have $\ibcspace_\vB = \ibcspace_{\vB'} (R')^{-1}$ and $\ibcspacekernel_\vB = \ibcextensionspacekernel_{\vB'} (R')^{-1}$, and if $\vC$ is an \IBCtuple extension of $\vB^\dagger \in \ibcspace_\vB$ in $\vA$, then $\vC R^{-1}$ is an \IBCtuple extension of $\vB^\dagger (R')^{-1} \in \ibcspace_{\vB'}$ in $\vA'$. 
    Hence, 
    after the first step, we have $\ibcextensionspace_{\vB, \vA} = \ibcextensionspace_{\vB', \vA'}R^{-1}$ and 
    $\ibcextensionspacekernel_{\vB, \vA} = \ibcextensionspace_{\vB', \vA'}R^{-1}$. 
    In addition, since $\vB$ and $\vB'$ are right-equivalent, we have $Y_\vB  = Y_{\vB'}$ and $Z_\vB = Z_{\vB'}$ by Fact~\ref{fact:ibctuple_char_matrix}. 
    Thus, for any $\vC \in \ibcextensionspace_{\vB, \vA}$, 
    \begin{equation}\label{equ:extension_1}(I - Z_\vB Y_\vB )\mtupletomatrix(\vC) R^{-1} = (I - Z_{\vB'} Y_{\vB'} )\mtupletomatrix(\vC R^{-1}),\end{equation}
    and for any $1 \leq i \leq \sigma \cdot \ell, 1 \leq j \leq r$,
    \begin{equation}\label{equ:extension_2}\begin{split} & \rowvecproj(e_i (I - Z_\vB Y_\vB)\mtupletomatrix(\vC), W_{\beta, j}) R^{-1} \\ = & \rowvecproj(e_i (I - Z_{\vB'}Y_{\vB'} )\mtupletomatrix(\vC R^{-1}), W_{\beta, j}').\end{split} \end{equation}
    Similarly, if $\vC$ is an \IBCtuple extension of $\vB^\dagger \in \ibcspacekernel_\vB$ in $\vA$, then $\vC R^{-1}$ is an \IBCtuple extension of $\vB^\dagger (R')^{-1} \in \ibcspacekernel_{\vB'}$ in $\vA'$, and 
    Equation (\ref{equ:extension_1}) and (\ref{equ:extension_2}) hold for $\vC$. 

    Now consider Step 2 of the algorithm. 
    If for some $i$ and $j$ of Step 2, $\ibcextensionspace_{\vB', \vA'} = \ibcextensionspace_{\vB, \vA}R^{-1}$ and 
    $\ibcextensionspacekernel_{\vB', \vA'} = \ibcextensionspace_{\vB, \vA}R^{-1}$, 
    then $\ibcextensionspace_{\vB', \vA'}^\dagger = \ibcextensionspace_{\vB, \vA}^\dagger R^{-1}$ and 
    $\ibcextensionspacekernel_{\vB, \vA}' = \ibcextensionspace_{\vB', \vA'}' R^{-1}$ hold. 
    Thus, if one of 2(a), 2(b), or 2(c) applies for $\vA$, then the same case applies for $\vA'$. 
    By induction on $i$ and $j$ for Step 2, for every $i$ and $j$, we have $\ibcextensionspace_{\vB', \vA'} = \ibcextensionspace_{\vB, \vA}R^{-1}$ and 
    $\ibcextensionspacekernel_{\vB', \vA'} = \ibcextensionspace_{\vB, \vA}R^{-1}$. 
    And if the case of 2(b) or 2(c) applies for $\vA$ with some $i$ and $j$, then the same case applies for $\vA'$ with the same $i$ and $j$, and thus the first condition or the second condition of the claim is satisfied. 

    If the algorithm for  $\vA$ passes Step 2, then the algorithm for  $\vA'$ also passes Step 2, and we have $\ibcextensionspace_{\vB', \vA'} = \ibcextensionspace_{\vB, \vA}R^{-1}$ and 
    $\ibcextensionspacekernel_{\vB', \vA'} = \ibcextensionspace_{\vB, \vA}R^{-1}$ before Step 3.
    Based on the condition of $T_{i}' = T_i R^{-1}$ for every $1 \leq i \leq \zeta$ and $t_{\vB, j} = t_{\vB', j}$ for any $1 \leq j \leq \sigma$, by Step 3 and Step 4 of the algorithm, the claim holds.
\end{proof}

\subsection{\IBCtuples irrelevant to $W_\beta$}\label{sec:irrelevant_ibctuples}
In this subsection and the next subsection, we illustrate how to use \IBCtuple extensions for \IBCtuples in the quotient matrix tuple $\vQ$ to obtain \IBCtuples in the original matrix tuple $\vA$.
Our analysis is divided into two cases based on whether the \IBCtuples of $\vQ$ are relevant to $W_\beta$ defined as follows. 

\begin{definition}
Given an \IBCtuple $\vB$ of $\vQ$ which is a $\beta$-quotient matrix tuple of $\vA$ with respect to a hierarchical row tuple decomposition $T_1, \dots, T_\zeta$ and parameters $h_0, \dots, h_\beta$, 
we say $\vB$ is \emph{irrelevant} to $W_\beta$ if for every $\vC \in \ibcextensionspace_{\vB, \vA}$, \[(I - Z_\vB Y_\vB) \mtupletomatrix(\vC) = 0,\] where $Y_\vB$ and $Z_\vB$ are the output of the Row Vector Extraction Algorithm for $\vB$, and $ \ibcextensionspace_{\vB, \vA}$ is the output of the Essential Extension Algorithm for $\vB$. 
Otherwise we say $\vB$ is \emph{relevant} to $W_\beta$. 

A sequence of \IBCtuples $\vB_1, \dots, \vB_k$ of $\vQ$
is a \emph{representative irrelevant \IBCtuple sequence} for $\vQ$ with respect to $W_\beta$ if the following two conditions hold:

\begin{enumerate}
    \item All the \IBCtuples $\vB_1, \dots, \vB_k$ are irrelevant to $W_\beta$. 
    \item Any two different \IBCtuples in the sequence are not equivalent. 
    \item Any \IBCtuple of $\vQ$ that is irrelevant to $W_\beta$ is equivalent to one of $\vB_1, \dots, \vB_k$. 
\end{enumerate}

A sequence of \IBCtuples $\vB_1, \dots, \vB_k$ of $\vQ$
is a \emph{representative relevant \IBCtuple sequence} for $\vQ$ with respect to $W_\beta$ if the following two conditions hold:

\begin{enumerate}
    \item All the \IBCtuples $\vB_1, \dots, \vB_k$ are relevant to $W_\beta$. 
    \item Any two different \IBCtuples in the sequence are not equivalent. 
    \item Any \IBCtuple of $\vQ$ that is relevant to $W_\beta$ is equivalent to one of $\vB_1, \dots, \vB_k$. 
\end{enumerate}

\end{definition}
Section~\ref{sec:irrelevant_ibctuples} focuses on the \IBCtuples of $\vQ$ that are irrelevant to $W_\beta$, and Section~\ref{sec:relevant_ibctuples} focuses on the \IBCtuples of $\vQ$ that are relevant to $W_\beta$. 

For an \IBCtuple $\vB$ of $\vQ$ that is irrelevant to $W_\beta$, 
we show that each matrix tuple in $\ibcextensionspace_{\vB, \vA}$ but not in $\ibcextensionspacekernel_{\vB, \vA}$ is an \IBCtuple of $\vA$, where $\ibcextensionspace_{\vB, \vA}$ and  $\ibcextensionspacekernel_{\vB, \vA}$ are the output of the Essential Extension Algorithm for $\vB$.

\begin{lemma}\label{lem:irrelevant_ibctuple}
    Let $\vA\in \M(n\times m, \F_q)^\ell$ be a matrix tuple and $T_1, \dots, T_\zeta$ be a hierarchical row tuple decomposition of $\vA$ with parameters $h_0, \dots, h_\beta$. 
    Let $\vQ$ be a $\beta$-quotient matrix tuple for $\vA$ with formatting vector $(v_1 + W_\beta, \dots, v_{m_\beta} + W_\beta)$. Let $\vB$ be an \IBCtuple of $\vQ$ with $\sigma$ rows satisfying the space property with \IBCtuple space $\ibcspace_\vB$ and \IBCtuple space kernel $\ibcspacekernel_\vB$, the row tuple space property with parameters $t_{\vB, 1}, \dots, t_{\vB, \sigma}$, the dimension property, and the block-compatible property. Let $\ibcextensionspace_{\vB, \vA}$ and $\ibcextensionspacekernel_{\vB, \vA}$ be the output of the Essential Extension Algorithm for $\vB$.
    
    If $\vB$ is an \IBCtuple of $\vQ$ that is irrelevant to $W_\beta$, then every matrix tuple in $\ibcextensionspace_{\vB, \vA}$ but not in $\ibcextensionspacekernel_{\vB, \vA}$ is an \IBCtuple of $\vA$ that satisfies the space property with \IBCtuple space $\ibcextensionspace_{\vB, \vA}$ and \IBCtuple space kernel $\ibcextensionspacekernel_{\vB, \vA}$, the row tuple space property with parameters $t_{\vB, 1}, \dots, t_{\vB, \sigma}$, the dimension property, and the block-compatible property.  
\end{lemma}
\begin{proof}
    By Definition~\ref{def:quotient_matrix_tuple}, 
    $\vQ \in \M(n_\beta \times m_\beta, \F_q)^\ell$ and $\vB \in \M(\sigma \times m_\beta, \F_q)^\ell$.
    Let $\vC$ be a matrix tuple in $\ibcextensionspace_{\vB, \vA}$ but not in $\ibcextensionspacekernel_{\vB, \vA}$, and $\vB'$ be a matrix tuple in $\ibcspace_\vB$ such that $\vC$ is an extension of $\vB'$.
    By the Essential Extension Algorithm, $\vB'$ is not in $\ibcspacekernel_\vB$, and thus is an \IBCtuple of $\vQ$ by the space property of $\vB$.
    By Definition~\ref{def:ibctuple}, there is a $\vQ'$ with invertible matrices $L'$ and $R'$ such that $\vQ' = L' \vQ (R')^{-1}$, $\vQ' = \diag(\vQ_1', \vQ_2')$, and $\vB' = \vQ_{1, \vQ', L'}$. By Definition~\ref{def:quotient_matrix_tuple}, $\vQ'$ is also a $\beta$-quotient matrix tuple for $\vA$ with respect to $T_1, \dots, T_\zeta$. 

    We construct a new matrix tuple $\vA'$ as follows:
    \begin{enumerate}
        \item For each $i\in[\sigma]$, The $i$-th row of $\vA'$ equals the $i$-th row of $\vC$. 
        \item For each $i \in \{\sigma+1, \dots, n_\beta\}$, the $i$-th row of $\vA'$ is an arbitrary row tuple extension of $i$-th row of $\vQ'$ in $\vA$. 
        \item The $e_{n_\beta + 1} \vA', \dots, e_n \vA'$ form an arbitrary linear basis of $\langle T_{h_\beta},\dots, T_\zeta \rangle$. 
    \end{enumerate}
    There is an invertible matrix of $L$ such that $\vA' = L \vA$, because by the construction of $\vA'$, $e_i \vA'$ is a linear combination of row tuples of $\vA$ for every $1 \leq i \leq n$, and $e_1 \vA', \dots, e_n \vA'$ are linearly independent. 

    Next, we show that there is an invertible matrix $R$ such that $\vA' = \diag(\vA_1, \vA_2) R^{-1}$, in which  $\vA_1$ has $\sigma$ rows, and $\vA_2$ has $n - \sigma$ rows. 
    We only need to show that \begin{equation}\label{equ:irrelevant_ibctuple}\dim(\rowvectorspace(\vC)) + \dim(\rowvectorspace(\vC')) = m,\end{equation} where \[\vC' = \begin{bmatrix}
        e_{\sigma + 1} \\
        \vdots \\
        e_n
    \end{bmatrix} \vA'.\] 
    Since $(I - Z_\vB Y_\vB) \mtupletomatrix(\vC)$ is a zero matrix, we have \[\dim(\rowvectorspace(\vC)) = \dim(\rowvectorspace(\vB)) = \dim(\rowvectorspace(\vB')) = \dim(\rowvectorspace(\vQ_1')).\]
    On the other hand, \[\begin{split}\dim(\rowvectorspace(\vC')) = & \dim(\rowvectorspace(\vQ_2' / W_\beta)) + \dim(W_\beta) \\ = & (m_{\beta} - \dim(\rowvectorspace(\vQ_1')) + (m - m_\beta) \\ = & m - \dim(\rowvectorspace(\vQ_1').\end{split}\] 
    Hence, Equation (\ref{equ:irrelevant_ibctuple}) holds.
    Thus, there is a matrix tuple $\vA'' = \diag(\vA_1, \vA_2)$ and invertible matrices $L$ and $R$ such that $\vA'' = L \vA R^{-1}$ and $\vC = \vA_{1, \vA'', L}$.
        In addition, $\vA_1$ must be indecomposable because otherwise, $\vB'$ is not an \IBCtuple of $\vQ$. Hence, $\vC$ is an \IBCtuple by Definition~\ref{def:ibctuple}. 

    The four properties of \IBCtuples in $\ibcextensionspace_{\vB, \vA}$ but not in $\ibcextensionspacekernel_{\vB, \vA}$ are obtained by Claim~\ref{claim:properties_ibctuple_extension} and  Lemma~\ref{lemma:essential_extension_ibctuple_compatible}.
\end{proof}

\subsection{\IBCtuples relevant to $W_\beta$}\label{sec:relevant_ibctuples}
In this section, we give an algorithm to compute a relevant representative \IBCtuple sequence for the input matrix tuple together with \IBCtuple space and \IBCtuple space kernel for each \IBCtuple in the sequence. 

\subsubsection{Compression Matrix Tuple}\label{sec:compression_matrix_tuple}

We give an algorithm to ``compress'' an extension for an \IBCtuple $\vB$ of a $\beta$-quotient matrix tuple of $\vA$ to a row tuple with all coordinates in $W_\beta$. 

\begin{framed}
\noindent \textbf{Extension Compression Algorithm}

\noindent \textbf{Input:} 
\begin{enumerate}
    \item Matrix tuple $\vA \in \M(n \times m, \F_q)^\ell$ with hierarchical row tuple decomposition $T_1, \dots, T_\zeta$ and parameters $h_0, \dots, h_\beta$ of $\vA$; 
    \item \IBCtuple $\vB \in \M(\sigma \times m_\beta, \F_q)^\ell$ of a $\beta$-quotient matrix tuple of $\vA$ satisfying the four properties of Definition~\ref{def:ibctuple_rel} with row tuple space parameters $t_{\vB, 1}, \dots, t_{\vB, \sigma}$;
    \item Extension space $\ibcextensionspace_{\vB, \vA}$ and matrix tuple $\vC \in \ibcextensionspace_{\vB, \vA}$.
\end{enumerate}

\noindent \textbf{Output:} A row tuple $\EC(\vC)$.

\begin{enumerate}
    \item Let
    $Y_\vB, Z_\vB$ be the output of the Row Vector Extraction Algorithm for $\vB$, 
    $s = \min_{i \in [\sigma]} t_{\vB,i}$, $u = \min (\{i \in [\sigma]: t_{\vB, i} = s\})$, and 
    $g = 0$. 
    \item For $i = 1$ to $\sigma$, if $t_{\vB, i} = s$, then $g = g + 1$ and let $\vC_g$ be an arbitrary matrix tuple in $\ibcextensionspace_{\vB, \vA}$ such that $e_u \vC_g = e_i \vC$. 
    \item Return a row tuple $\EC(\vC)$ of length $g \cdot \sigma \cdot \ell$ as follows: for any $1 \leq g' \leq g$, $1 \leq i \leq \sigma$ and $1 \leq \ell' \leq \ell$, 
    let $((g' - 1) \cdot \sigma \cdot \ell + (i - 1)\cdot \ell + \ell')$-th coordinate of $\EC(\vC)$ equals $((i - 1)\cdot \ell + \ell')$-th row of $(I - Z_\vB Y_\vB) \mtupletomatrix(\vC_{g'})$. 
\end{enumerate}
\vspace{-.3cm}
\end{framed}

\begin{claim}\label{claim:basic_compression}
    Let $\vA\in \M(n\times m, \F_q)^\ell$ be a matrix tuple and $T_1, \dots, T_\zeta$ be a hierarchical row tuple decomposition of $\vA$ with parameters $h_0, \dots, h_\beta$. 
    Let $\vB$ be an \IBCtuple of a $\beta$-quotient matrix tuple for $\vA$ with $\sigma$ rows satisfying the space property with \IBCtuple space $\ibcspace_\vB$ and \IBCtuple space kernel $\ibcspacekernel_\vB$, the row tuple space property with parameters $t_{\vB, 1}, \dots, t_{\vB, \sigma}$, the dimension property, and the block-compatible property. Let $\ibcextensionspace_{\vB, \vA}$ and $\ibcextensionspacekernel_{\vB, \vA}$ be the output of the Essential Extension Algorithm for $\vB$.
    We have the following observations for the Extension Compression Algorithm: 
    \begin{enumerate}
        \item The running time of the algorithm is $\poly(n, m, \ell, \log q)$.
        \item The outputs for different matrix tuples in $\ibcextensionspace_{\vB, \vA}$ have the same length. 
        \item For every $\vC \in \ibcextensionspace_{\vB, \vA}$, $\EC(\vC)$ is the zero row tuple iff $\vC \in \ibcextensionspacekernel_{\vB, \vA}$.
        \item Let $\vC_1$ and $\vC_2$ be two matrix tuples in $\ibcextensionspace_{\vB, \vA}$ such that $\vC_1$ is an extension of $\vB_1\in \ibcspace_\vB$ and $\vC_2$ is an extension of $\vB_2 \in \ibcextensionspace_\vB$. $\vB_1$ is strongly correlated to $\vB_2$ if and only if $\EC(\vC_1) = \EC(\vC_2)$.
        \item For any $\vC_1, \vC_2\in \ibcextensionspace_{\vB, \vA}$, the $\EC(\vC_1 + \vC_2)$ equals $\EC(\vC_1) + \EC(\vC_2)$.  
        \item Let $\EC(\ibcextensionspace_{\vB, \vA})$ denote the space spanned by the outputs of the algorithm for all the matrix tuples in $\ibcextensionspace_{\vB, \vA}$. 
        For every $\va \in \EC(\ibcextensionspace_{\vB, \vA})$, there is a $\vC \in \ibcextensionspace_{\vB, \vA}$ such that $\EC(\vC) = \va$.  
        \item   
        The algorithm is canonical in the following sense: Let $\vA', T_1', \dots, T_\zeta', \vB', \ibcextensionspace_{\vB', \vA'}, \vC'$ be another input of the Extension Compression Algorithm such that there are invertible matrices $L$ and $R$ satisfying the following conditions:
        \begin{enumerate}
            \item $\vA' = L \vA R^{-1}$.
            \item $T_i' = T_i R^{-1}$ for any $i \in [\zeta]$, and the parameters for $T_1, \dots, T_\zeta$ and  $T_1', \dots, T_\zeta'$ are the same.  
            \item $\vB$ and $\vB'$ are right-equivalent with the parameters $t_{\vB, i} = t_{\vB', i}$ for any $i \in [\sigma]$. 
            \item $\vC' = \vC R^{-1}$. 
        \end{enumerate}
        Then for any invertible matrices $L$ and $R$ satisfying conditions (a) - (d), $\EC(\vC') = \EC(\vC) R^{-1}$. 
    \end{enumerate}
\end{claim}
\begin{proof}
The first property is obtained by the definition of the algorithm and Fact~\ref{fact:ibctuple_char_matrix}.

For the second property, we observe that all the matrix tuples in $\ibcextensionspace_{\vB, \vA}$ have the same row tuple space parameters $t_{\vB, 1}, \dots, t_{\vB, \sigma}$. So, $s$ and $t$ computed in the first step of the algorithm are the same for all the matrix tuples in $\ibcextensionspace_{\vB, \vA}$, and thus $g$ computed after the second step of the algorithm is the same for all the matrix tuples in $\ibcextensionspace_{\vB, \vA}$. Hence, the output row tuples for all the matrix tuple $\ibcextensionspace_{\vB, \vA}$ have the same length. 

The third property is obtained by the second property of Claim~\ref{claim:properties_ibctuple_extension} and the observation that $\vC_1, \dots, \vC_g$ obtained in Step 2 of the algorithm are in $\ibcextensionspacekernel_{\vB, \vA}$ iff $\vC$ is in $\ibcextensionspacekernel_{\vB, \vA}$ by the row tuple space property of $\vB$ as Definition~\ref{def:four_prop}.


For the fourth property, 
we first show that if $\vB_1$ is strongly correlated to $\vB_2$, then the outputs of the algorithm for $\vC_1$ and $\vC_2$ are the same.
Let $i_1, \dots, i_g$ denote all the integers such that $t_{\vB, i_j} = s$ for any $j \in [g]$.
Let $\vC_{1, 1}, \dots, \vC_{1, g}$ and $\vG_{2, 1}, \dots, \vC_{2, g}$ denote the matrix tuples obtained for $\vC_1$ and $\vC_2$ respectively in Step 2 of the algorithm. 
By the row tuple space property of Claim~\ref{claim:properties_ibctuple_extension},
for any $j\in[g]$, 
$e_{i_j} \vC_1 - e_{i_j} \vC_2$ is in $K_{t_{\vB, i_j}}$, and thus $e_{u} \vC_{1, j} - e_u \vC_{2, j}$ is in  $K_{t_{\vB, i_j}}$ for any $j\in [g]$. 
Let $\vB_{1, j}$ and $\vB_{2, j}$ be the matrix tuples in $\ibcspace_\vB$ such that $\vC_{1, j}$ and $\vC_{2, j}$ are extensions of $\vB_{1, j}$ and $\vB_{2, j}$ respectively for any $j\in[g]$.
By the row tuple space property for $\vB$, $\vB_{1, j}$ is strongly correlated to $\vB_{2, j}$ for any $j\in[g]$.  Hence, by the second property of Claim~\ref{claim:properties_ibctuple_extension}, the outputs for $\vC_1$ and $\vC_2$ are the same. 

Now we show that if $\vB_1$ is not strongly correlated to $\vB_2$, then the outputs of the algorithm for $\vC_1$ and $\vC_2$ are different.
By Definition~\ref{def:ibctuple}, $\vB_1 - \vB_2$ is not an \IBCtuplenospace, and thus not in $\ibcspacekernel_{\vB}$ by the space property of $\vB$. 
Hence, $\vC_1 - \vC_2$ is not in $\ibcextensionspacekernel_{\vB, \vA}$ by the Essential Extension Algorithm. By the second property of Claim~\ref{claim:properties_ibctuple_extension}, $(I - Z_\vB Y_\vB) \mtupletomatrix(\vC_1) \neq (I - Z_\vB Y_\vB) \mtupletomatrix(\vC_2)$, and thus the outputs of the algorithm for $\vC_1$ and $\vC_2$ are different.

For the fifth property, denote $\vC_1$ and $\vC_2$ as 
\[ \vC_1 = \begin{bmatrix}
    \vc_{1,1} \\ \vdots \\ \vc_{1, \sigma}
\end{bmatrix}, \vC_2 = \begin{bmatrix}
    \vc_{2, 1} \\ \vdots \\ \vc_{2, \sigma}
\end{bmatrix}. \]
Let $\vC_{1, 1}, \dots, \vC_{1, g}$ and $\vC_{2, 1}, \dots, \vC_{2, g}$ be the matrix tuples obtained for $\vC_1$ and $\vC_2$ respectively for Step 2 of the algorithm. 
For any $j \in [g]$, $\vC_{1, j} + \vC_{2, j}$ satisfies $e_{i_j} (\vC_{1, j} + \vC_{2, j}) = \vc_{1, i_{j}} + \vc_{2, i_{j}}$. 

On the other hand, 
let $\vC^\dagger_1,\dots,\vC^\dagger_g$ denote the matrix tuples selected in Step 2 for $\vC_1 + \vC_2$ as the input of the algorithm.
$e_u \vC^\dagger_{j}$ is also $\vc_{1, i_{j}} + \vc_{2, i_{j}}$ for any $j \in [g]$. 
By the third property of this claim, for any $j \in [g]$, $\vC^\dagger_{j} = \vC_{1, j} + \vC_{2, j} + \vC^\diamond_j$ for some $\vC^\diamond_j \in \ibcextensionspacekernel_{\vB, \vA}$.  
By Claim~\ref{claim:properties_ibctuple_extension}, the output of the algorithm for $\vC_1 + \vC_2$ equals the output of $\vC_1$ plus the output of $\vC_2$. 

The sixth property is obtained by the third and fourth properties. 

For the last property, 
since $t_{\vB, k} = t_{\vB', k}$ for any $k \in [\sigma]$. So, $s$ and $t$ computed in the first step of the algorithm are the same for $\vC$ and $\vC'$.
Let $g'$ be the integer $g$ at the end of the Step 2 for $\vC'$, and $i_1', \dots, i_g'$ denote all the integers such that $t_{\vB', i_j'} = s$ for any $j \in [g']$.
We have $g = g'$ and $i_j = i_j'$ for any $j \in [g]$ by $t_{\vB, k} = t_{\vB', k}$ for any $k \in [\sigma]$, and thus $e_{i_j'} \vC' = e_{i_j} \vC R^{-1}$ for every $j \in [g]$. 

On the other hand, by Lemma~\ref{lemma:extension_canonical}, we have $\ibcextensionspace_{\vB', \vA'} = \ibcextensionspace_{\vB, \vA} R^{-1}$, and $\ibcextensionspacekernel_{\vB', \vA'} = \ibcextensionspacekernel_{\vB, \vA} R^{-1}$. 
Let $\vC_1', \dots, \vC_g'$ denote the matrix tuples obtained in Step 2 of the algorithm for $\vC'$. 
For every $j \in [g]$, we have $\vC_j' = (\vC_j + \vC_j^\diamond) R^{-1}$ for some $\vC_j^{\diamond} \in \ibcextensionspacekernel_{\vB, \vA}$, 
and thus $\mtupletomatrix(\vC_j') = \mtupletomatrix(\vC_j) R^{-1}$ by Claim~\ref{claim:properties_ibctuple_extension}. 
Hence, $\EC(\vC') = \EC(\vC) R^{-1}$.
\end{proof}

Given a representative relevant \IBCtuple sequence for a $\beta$-quotient matrix tuple $\vQ$ of $\vA$, we construct a new matrix tuple $\vE$ by compressing essential extensions of \IBCtuples right-equivalent to $\vB_1, \dots, \vB_k$ to row tuples.

\begin{framed}
\noindent \textbf{Compression Matrix Tuple Algorithm}

\noindent \textbf{Input:} 
\begin{enumerate}
    \item Matrix tuple $\vA \in \M(n \times m, \F_q)^\ell$ with hierarchical row tuple decomposition $T_1, \dots, T_\zeta$ and parameters $h_0, \dots, h_\beta$ of $\vA$, and $\beta$-quotient matrix tuple $\vQ$ of $\vA$ with respect to $T_1, \dots, T_\zeta$. 
    \item Representative relevant \IBCtuple sequence $\vB_1, \dots, \vB_k$ with row tuple parameters $t_{\vB_i, 1}, \dots, t_{\vB_i, \sigma_i}$, essential extension space $\ibcextensionspace_{\vB_i, \vA}$, and essential extension space  kernel $\ibcextensionspacekernel_{\vB_i, \vA}$ for each $i \in [k]$. 
\end{enumerate}

\noindent \textbf{Output:} 
Matrix tuples $\vE \in \M(n_\vE \times \dim(W_\beta), \F_q)^{\ell_\vE}, \vG \in \M(n_\vE \times m, \F_q)^{\ell_\vE}, \vE_1, \dots, \vE_\rho$ for some integers $n_\vE, \ell_\vE$; vector $(v_1, \dots, v_{\dim(W_\beta)}) \in (\F_q^\ell)^{\dim(W_\beta)}$,
depth-$0$ hierarchical row tuple decomposition $T_{\vE, 1}, \dots, T_{\vE, \rho}$;integers $u_{\vE, 1}, \dots, u_{\vE, \rho}$.
\noindent \textbf{Algorithm:}

\begin{enumerate}
\item Let $\rho = 0$. 
\item For $i = 1, \dots, k,$
\begin{enumerate}
    \item Let $\rho  =\rho + 1$, $r_{\rho} := \dim(\ibcextensionspacekernel_{\vB, \vA})$, and $s_\rho := \dim(\ibcextensionspace_{\vB, \vA})$. Compute an arbitrary linear basis of $\ibcextensionspacekernel_{\vB_i, \vA}$ denoted as $\vC_1, \dots, \vC_{r_\rho}$, and arbitrary matrix tuples $\vC_{r_\rho + 1}, \dots, \vC_{s_\rho}$ such that $\vC_1, \dots, \vC_{s_\rho}$ is a linear basis of $\ibcextensionspace_{\vB_i, \vA}$. 
    \item Compute $\EC(\vC_{r_\rho + 1}), \dots, \EC(\vC_{s_\rho})$ by running the Extension Compression Algorithm for $\vC_{r_\rho + 1}, \dots, \vC_{s_\rho}$. 
    \item 
    Let $n_{\vE, \rho} := s_\rho - r_\rho$, $\ell_\rho$ be the length of $\EC(\vC_{\dim(\ibcextensionspacekernel_{\vB_i, \vA}) + 1})$, $u_{\vE, \rho} := 0$, and 
    $\vE_\rho \in \M(n_{\vE, \rho} \times m, \F_q)^{\ell_\rho}$ be the matrix tuple such that $e_j \vE_\rho = \EC(\vC_{r_\rho + j})$ for any $j \in [n_{\vE, \rho}]$.

\end{enumerate}

\item For $i = h_\beta, \dots, \zeta$, let $\rho = \rho +1$, $\ve_{\rho, 1}, \dots, \ve_{\rho, \dim(T_i)}$ be an arbitrary linear basis of $T_i$, $n_{\vE, \rho} := \dim(T_i)$, $\ell_{\vE, \rho} := \ell$, $u_{\vE, \rho} := i$, 
and $\vE_\rho \in \M(n_{\vE, \rho} \times m, \F_q)^\ell$ such that $e_j \vE_{\rho} = \ve_{\rho, j}$ for any $j\in[\dim(T_i)]$. 

\item Let $\ell_\vE := \sum_{i = 1}^\rho \ell_{\vE, i}$.  
Let $\vG_{i} = (G_{i, 1}, \dots, G_{i, \ell_{\vE}}) \in \M(n_{\vE, i} \times m, \F_q)^{\ell_{\vE}}$  for each $i\in [\rho]$ be the matrix tuple such that $(G_{i, (\sum_{j = 1}^{i - 1}n_{\vE, j})+ 1}, \dots,G_{i, \sum_{j = 1}^{i}n_{\vE, j}}) = \vG_i$, and all the other matrices in $\vG_{i}$ are zero matrices.

\item Let $n_\vE := \sum_{i = 1}^\rho n_{\vE, i}$, 
$v_1, \dots, v_{\dim(W_\beta)}$ be an arbitrary basis of $W_\beta$, 
\[\vG = \begin{bmatrix}
    \vG_{1} \\ \vdots \\ \vG_{\rho}
\end{bmatrix},\]
and $\vE$ be the matrix tuple such that $\vE \cdot (v_1, \dots, v_{\dim(W_\beta)})^T = \vG$.
For every $i\in [\rho]$, let $T_{\vE, i} = \langle \{e_{(\sum_{k = 1}^{i - 1} n_{\vE, k}) + j} \vE : j \in [n_{\vE, i}]\}\rangle$. 
\item Return $\vE$, $\vG$, $\vE_1, \dots, \vE_\rho$, $(v_1, \dots, v_{\dim(W_\beta)})$, $T_{\vE, 1}, \dots, T_{\vE, \rho}$, and $u_{\vE, 1}, \dots, u_{\vE, \rho}$. 
\end{enumerate}
\vspace{-.3cm}
\end{framed}

\begin{claim}\label{claim:vE_direct_sum}
    For a matrix tuple $\vA$ with hierarchical row tuple decomposition $T_1, \dots, T_\zeta$ and parameters $h_0, \dots, h_\beta$ of $\vA$, a $\beta$-quotient matrix tuple $\vQ$ of $\vA$ with respect to $T_1, \dots, T_\zeta$, and a representative relevant \IBCtuple sequence $\vB_1, \dots, \vB_k$ for $\vQ$ with with essential extension space $\ibcextensionspace_{\vB_i, \vA}$ and essential extension space kernel $\ibcextensionspacekernel_{\vB_i, \vA}$ for each $i\in[k]$, 
    let $\vE, \vG, (v_1, \dots, v_{\dim(W_\beta)})$ be output of the Compression Matrix Tuple Algorithm for $\vA, \vB_1, \dots, \vB_k$, $\vF = \diag(\vF_1, \dots, \vF_f)$ be a block-diagonalization of $\vE$, and $L, R$ be invertible matrices such that $\vF = L \vE R^{-1}$. 
    Let \[\vG_{j, \vF, L} = \vE_{j, \vF, L} \cdot (v_1, \dots, v_{\dim(W_\beta)})^T\] for every $j \in [k]$. The following properties hold:
    \begin{enumerate}
    \item $W_\beta = \rowvectorspace(\vG_{1, \vF, L}) \oplus \dots \oplus \rowvectorspace(\vG_{k, \vF, L})$.

    \item For each $\vB_i$ with $i\in [k]$, let 
    \[\begin{split}  \ibcspace_{\vB_i, j} := & \langle \{\vB \in \ibcspace_{\vB_i}: \exists \vC \in \ibcextensionspace_{\vB_i, \vA} \text{ s.t. } \vC \text{ is an extension of an \IBCtuple right-equivalent to } \vB, \\ & \ \ \ \ \ \ \ \ \ \ \ \ \ \ \ \rowvectorspace(\EC(\vC)) \leq \rowvectorspace(\vG_{j, \vF, L}) \} \rangle\end{split}\]
    for any $j\in[f]$. 
    $\ibcspace_{\vB_i} / \ibcspacekernel_{\vB_i} = (\ibcspace_{\vB_i, 1}  / \ibcspacekernel_{\vB_i}) \oplus \dots \oplus (\ibcspace_{\vB_i, k}  / \ibcspacekernel_{\vB_i})$. 

    \item For any $i \in [k], j \in [f]$, and any \IBCtuple $\vB\in\ibcspace_{\vB_i, j}$, every \IBCtuple $\vB'$ correlated to $\vB$ is in $\ibcspace_{\vB_i, j}$. 
    \end{enumerate}
\end{claim}
\begin{proof}
The first property is obtained by the definition of block-diagonalization and the relation between $\vE$ and $\vG$. 

For the second property, let $i$ be an arbitrary integer in $[k]$. By the correspondence between $\vG$ and $\vE_1, \dots, \vE_\rho$, the block-diagonalization $\vF$ of $\vE$ implies that \[\rowtuplespace(\vE_i) = \rowtuplespace(\vE_{i, 1}) \oplus \dots \oplus \rowtuplespace(\vE_{i, f}),\] where \[\vE_{i, j} =\left\langle \left\{ u \vE_i : u \in \F_q^{n_{\vE, i}}, \rowvectorspace(u \vE_i) \leq \rowvectorspace(\vG_{j, \vF, L}) \right\}\right\rangle.\] 

On the other hand, by the linearality of $\ibcspace_{\vB_i}$ and the fifth property of Claim~\ref{claim:basic_compression}, there is a bijection \[f: \ibcspace_{\vB_i} / \ibcspacekernel_{\vB_i} \rightarrow \langle \{\EC(\vC) : \vC \in \ibcextensionspace_{\vB_i, \vA}\} \rangle \] such that for any $\vB' \in \ibcspace_{\vB_i}$, $f(\vB' + \ibcspacekernel_{\vB_i}) = \EC(\vC)$ for some $\vC \in \ibcextensionspace_{\vB_i, \vA}$ that is an extension of a matrix tuple in $\vB' + \ibcspacekernel_{\vB_i}$. 
Note that \[\rowtuplespace(\vE_{i}) = \langle \{\EC(\vC) : \vC \in \ibcextensionspace_{\vB_i, \vA}\} \rangle .\] 
By the definition of $\ibcspace_{\vB_i, 1}, \dots, \ibcspace_{\vB_i, f}$, 
for any $1 \leq j \leq f$, 
$\ibcspace_{\vB_i, j}$ contains all the $\vB' \in \ibcspace_{\vB_i}$ such that $f(\vB' + \ibcspacekernel_{\vB_i}) \in \rowtuplespace(\vE_{i, j})$.
Hence, the second property holds. 

Now we prove the last property. By the third property of Claim~\ref{claim:properties_ibctuple_extension}, there is a $\vC^\dagger \in \ibcextensionspace_{\vB_i, \vA}$ such that $\vC^\dagger$ is an extension of an \IBCtuple $\vB^\dagger$ strongly correlated to $\vB$ satisfying $\EC(\vC^\dagger) \in \rowtuplespace(\vE_{i, j})$. 
Let $\vC^\dagger_1, \dots, \vC^\dagger_g$ be the matrix tuples selected in Step 2 of the Extension Compression Algorithm for $\vC^\dagger$. 
By the Extension Compression Algorithm, \[\rowvectorspace\left((I - Z_{\vB_i}Y_{\vB_i})\mtupletomatrix\left(\vC^{\dagger}_{g'}\right)\right)\] is a subspace of $\rowvectorspace(\vG_{j,\vF, L})$ for every $g' \in [g]$. 
By the second property of Claim~\ref{claim:properties_ibctuple_extension} and the first property of the current claim, 
$\EC(\vC^\dagger_{g'})$ is in $\rowtuplespace(\vE_{i, j})$ for every $ g' \in [g]$.

On the other hand, by Claim~\ref{claim:properties_ibctuple_extension}, there is an \IBCtuple $\vB'' \in \ibcspace_{\vB_i}$ strongly correlated to $\vB'$ such that there is a $\vC'' \in \ibcextensionspace_{\vB_i, \vA}$ as an extension of $\vB''$. 
Since $\vB''$ is strongly correlated to $\vB'$ and $\vB'$ is correlated to $\vB$, $\vB''$ is also correlated to $\vB$. By Fact~\ref{fact:ibctuple_correlated_necessary_condition},  there is an inverbile matrix $L$ such that $\vB'' = L \vB + \vK$ for some $\vK \in \ibcspacekernel_{\vB_i}$. 
Hence, $\vC''$ equals a linear combination of $\vC^\dagger_1, \dots, \vC^\dagger_1$ plus a matrix tuple in $\ibcextensionspacekernel_{\vB_i, \vA}$. 
By Claim~\ref{claim:basic_compression}, $\EC(\vC'')$ is also a row tuple of $\rowtuplespace(\vG_{j, \vF, L})$.
By the definition of $\ibcspace_{\vB_i, j}$, $\vB'$ is in $\ibcspace_{\vB_i, j}$. 
\end{proof}

\begin{claim}\label{claim:submatrix_tuples_correlated}

    For a matrix tuple $\vA$ with hierarchical row tuple decomposition $T_1, \dots, T_\zeta$ and parameters $h_0, \dots, h_\beta$ of $\vA$, a $\beta$-quotient matrix tuple $\vQ$ of $\vA$ with respect to $T_1, \dots, T_\zeta$, and a representative relevant \IBCtuple sequence $\vB_1, \dots, \vB_k$ for $\vQ$, 
    let $\vE, \vG, (v_1, \dots, v_{\dim(W_\beta)})$ be the output of the Compression Matrix Tuple Algorithm for $\vA, \vB_1, \dots, \vB_k$. 
    Let $\vD = L\vA R^{-1} = \diag(\vD_1, \dots, \vD_d)$ be a block-diagonalization of $\vA$.  
    There is a block-diagonalization $\vF = \diag(\vF_1, \dots, \vF_\xi)$ of $\vE$ with \[\xi = | \{i \in [d] : \rowvectorspace( \vA_{i, \vD, L}) \cap W_\beta \neq  \emptyset\}|.\] 
\end{claim}

\begin{proof}
We describe a method to construct a desired block-diagonalization of $\vE$. 
Without loss of generality, we assume 
\[\{i \in [d] : \rowvectorspace( \vA_{i, \vD, L}) \cap W_\beta \neq  \emptyset\} = \{1, \dots, \xi\}.\]
For each $1 \leq i \leq \xi$, let $W_{\beta, i}$ denote $\rowtuplespace(\vA_{i, \vD, L}) \cap W_\beta$, and
$v_{i, 1}, \dots, v_{i, \dim(W_{\beta, i})}$ be a linear basis of $W_{\beta, i}$. 
Let 
\[H_{i} = \langle \{\vg \in \rowtuplespace(\vG) : \rowvectorspace(\vg) \subset W_{\beta, i}\}\rangle,\]
$t_i = \dim(H_i)$, and
$\vg_{i, 1}, \dots, \vg_{i, t_i}$ be a linear basis of $H_i$. 
Let \[\vG^\dagger =\begin{bmatrix}
    \vg_{1, 1} \\ \vdots\\ \vg_{1, t_1} \\ \vg_{2, 1} \\ \\ \vdots \\ \vg_{\xi, t_\xi}
\end{bmatrix}. \]
Let $\vE^\dagger$ be the matrix tuple such that \[\vE^\dagger (v_{1, 1}, \dots, v_{1, \dim(W_{\beta, 1})}, v_{2, 1}, \dots, v_{\xi, \dim(W_{\beta, \xi})})^T = \vG^\dagger.\] 
$\vE^\dagger$ is a block-diagonalization $\xi$ blocks by the choice of $\vg_{i, 1}, \dots, \vg_{i, t_i}$ and $v_{i,1}, \dots, v_{i, \dim(W_{\beta, i})}$ for all the $i \in [\xi]$ and Fact~\ref{fact:ibctuple_diagonalzation_most_basic}.

In the rest of the proof, we show that there exists invertible matrices $L_\vE$ and $R_\vE$ such that $\vE = L_\vE \vE^\dagger (R_\vE)^{-1}$ by showing that there is a matrix $L_{\vG}$ such that $\vG = L_{\vG} \vG^\dagger$. 
By the construction of $\vG$, it is sufficient to show that for any $1 \leq j \leq \rho$, 
there is a matrix $L_{\vE_j}$ such that $\vE_j = L_{\vE_j} \vE^\diamond_j$, 
where $\vE^\diamond_j$ is constructed as follows: For each $1 \leq i \leq \xi$, take a linear basis $\ve_{j, i, 1}, \dots, \ve_{j, i, t_{j, i}}$ of 
$\langle\{u \vE_j : u \in \F_q^{n_{\vE, j}}, \rowvectorspace(u \vE_j) \leq W_{\beta, i}\}\rangle$. 
Let 
\[\vE_j^\diamond =\begin{bmatrix}
    \ve_{j, 1, 1} \\ \vdots\\ \ve_{j, 1, t_{j, 1}} \\ \ve_{j, 2, 1} \\ \\ \vdots \\ \ve_{j, \xi, t_{j, \xi}}
\end{bmatrix}.  \]
To achieve this goal, by the construction of $\vE_1, \dots, \vE_\rho$, it is sufficient to show that for any $1 \leq j \leq k$, for any $\vC \in \ibcextensionspace_{\vB_j, \vA}$, there is a row tuple $\ve'\in \rowtuplespace(\vE_{j}^\diamond)$ such that $\ve' = \EC(\vC)$. 
By the characteristic and block-compatible properties of essential extensions of \IBCtuples in $\ibcspace_{\vB_{j}}$ (Lemma~\ref{lemma:essential_extension_ibctuple_compatible}), 
$\proj_{\vD, L}(\vC, i)$ is in $\ibcextensionspace_{\vB_j, \vA}$ for any $i \in [d]$. 
We prove the following two observations, and then the claim follows by the fifth property of Claim~\ref{claim:basic_compression}.
\begin{enumerate}
    \item $\EC(\proj_{\vD, L}(\vC, i))$ is a zero row tuple for any $\xi< i \leq d$. 
    \item For every $1 \leq i \leq \xi$, $\rowvectorspace(\EC(\proj_{\vD, L}(\vC, i)) \leq  W_{\beta, i}$, and thus $\EC(\proj_{\vD, L}(\vC, i))$ is a linear combination of $\ve_{j, i, 1}, \dots, \ve_{j, i, t_{j, i}}$. 
\end{enumerate}

For the first property, $\proj_{\vD, L}(\vC, i)$ is in $\ibcextensionspacekernel_{\vB_j, \vA}$ for any $i > \xi$ by the space property of \IBCtuple extensions in $\ibcextensionspace_{\vB_j, \vA}$ (the fourth property of Claim~\ref{claim:properties_ibctuple_extension}) and the observation that $\rowvectorspace(\vA_{i, \vD, L})$ does not contain any non-zero vector in $W_\beta$. Then the first property is by the third property of Claim~\ref{claim:basic_compression}.

For the second property, 
we show that for any $\vC' \in \ibcextensionspace_{\vB_j, \vA}$ such that $\rowtuplespace(\vC')$ is a subspace of $\rowtuplespace(\vA_{i, \vD, L})$, then \[\rowvectorspace(\EC(\vC')) \leq W_{\beta, i}.\]

First, let $\vC''$ be an arbitrary matrix tuple in $\ibcextensionspace_{\vB_j, \vA}$ such that $\rowtuplespace(\vC'')$ is a subspace of $\rowtuplespace(\vA_{i, \vD, L})$. 
Every row vector of $(I - Z_{\vB_j} Y_{\vB_j}) \mtupletomatrix(\vC'')$ is in $W_\beta$ by Fact~\ref{fact:ibctuple_char_matrix}. 
Because $\rowtuplespace(\vC'')$ is a subspace of $\rowtuplespace(\vA_{i, \vD, L})$,   $\rowvectorspace((I - Z_{\vB_j} Y_{\vB_j}) \mtupletomatrix(\vC''))$ is a subspace of $\rowvectorspace(\vA_{i, \vD, L})$. 
Hence, every row vector of $(I - Z_{\vB_j} Y_{\vB_j}) \mtupletomatrix(\vC'')$ is in $W_{\beta, i}$. 


Second, let $\vC_1', \dots, \vC_k'$ be the matrix tuples selected in Step 2 of the Extension Compression Algorithm for $\vC'$. 
By the row tuple space property of essential extensions in $\ibcextensionspace_{\vB_j, \vA}$, $\proj_{\vD, L}(\vC_{k'}', i')$ is in $\ibcextensionspacekernel_{\vB_j, \vA}$ for any $1 \leq k' \leq k$ and $i' \neq i$, because 
$e_u \vC_{k'}'$ equals a row tuple in $\rowtuplespace(\vA_{i, \vD, L})$, where $u$ is the parameter selected in the first step of the Extension Compression Algorithm for $\vC'$. 
By the third property of Claim~\ref{claim:basic_compression}, we have
\[\EC(\vC_{k'}) = \EC(\proj_{\vD, L}(\vC_{k'}', i))\] for any $1 \leq k' \leq k$. Hence, every row vector of $(I - Z_{\vB_j} Y_{\vB_j}) \mtupletomatrix(\vC_{k'}')$ is in $W_{\beta, i}$ for any $1 \leq k' \leq k$.

Thus, $\rowvectorspace(\EC(\vC')) \leq W_{\beta, i}$. And the claim follows.
\end{proof}

\begin{claim}\label{claim:running_time_compression}
For an input $\vA \in \M(n\times m, \F_q)^{\ell}, T_1, \dots, T_\zeta, \vQ, \vB_1, \dots, \vB_k$ of the Compression Matrix Tuple Algorithm, 
we have the following observations for the Compression Matrix Tuple Algorithm:
\begin{enumerate}
    \item The running time of the algorithm is $\poly(n, m, \ell, \log q)$.
    \item The algorithm is canonical in the following sense: Let $\vA' \in \M(n\times m, \F_q)^{\ell}, T_1', \dots, T_\zeta', \vQ'$, $\vB_1', \dots, \vB_k'$ be another input to the algorithm such that there exist invertible matrices $L$ and $R$ satisfying the following conditions:
    \begin{enumerate}
        \item $\vA' = L \vA R^{-1}$.
        \item $T_i' = T_i R^{-1}$ for any $i \in [\zeta]$, and the parameters for $T_1, \dots, T_\zeta$ and  $T_1', \dots, T_\zeta'$ are the same.  
        \item $\vB_i$ and $\vB_i'$ are right-equivalent with the parameters $t_{\vB, i} = t_{\vB', i}$ for any $i \in [\sigma]$. 
        \item $\vE_{\vB_i', \vA'} = \vE_{\vB_i, \vA} R^{-1}$ for any $i \in [\sigma]$. 
    \end{enumerate}
    For any invertible matrices $L$ and $R$ satisfying (a)-(d), there are invertible matrices $L_\vE$ and $R_\vE$ such that the outputs of the two inputs satisfy the following conditions:
    \begin{enumerate}
    \item $\vE' = L_\vE \vE R_\vE^{-1}$. 
    \item $\vG' = L_\vE \vG R^{-1}$.
    \item $\vE_i' = \vE_i R^{-1}$ for every $i \in [k]$.
    \item $\rho' = \rho$, and $u_{\vE, i} = u_{\vE', i}$ for every $i \in [\rho]$. 
    \item $T_{\vE', i} = T_{\vE, i} R_{\vE}^{-1}$ for every $i \in [\rho]$.
    \end{enumerate}
\end{enumerate}
\end{claim}
\begin{proof}
    The first property is obtained by the algorithm and the fact that $\zeta \leq n$ by Definition~\ref{def:row_tuple_decomposition}, and $k \leq n$. 

    For the second property, since the parameters for $T_1, \dots, T_\zeta$ and $T_1', \dots, T_\zeta'$ are same, $\rho' = \rho$. 
    By Claim~\ref{claim:basic_compression}, for every $i\in [\rho]$, there is an invertible matrix $L_i$ such that $\vE_i' = L_i \vE_i R^{-1}$. 
    By the construction of $\vG$ and $\vG'$, 
    $\vG' = L_\vE \vG R^{-1}$ by letting $L_\vE = \diag(L_1, \dots, L_\rho)$. 
    Let \[(v_1, \dots, v_{\dim(W_\beta)}) \text{ and } (v_1', \dots, v_{\dim(W_\beta')}')\] be the vectors such that \[\vE\cdot (v_1, \dots, v_{\dim(W_\beta)})^T = \vG \text{ and } \vE' \cdot    (v_1', \dots, v_{\dim(W_\beta')}')^T = \vG'\] hold.
    Since $T_i' = T_i R^{-1}$ for every $i \in [\zeta]$, we have $W_\beta' = W_\beta R^{-1}$. 
    Hence $\dim(W_\beta) = \dim(W_\beta')$.
    Since $v_1, \dots, v_{\dim(W_\beta)}$ is a linear basis of $W_\beta$, and $v_1', \dots, v_{\dim(W_\beta')'}$ is a linear basis of $W_\beta'$, 
    there is an invertible matrix $R_\vE$ such that \[R_\vE (v_1', \dots, v_{\dim(W_\beta')}')^T =(v_1 R^{-1}, \dots, v_{\dim(W_\beta)}R^{-1})^T.\]
    Hence, we have 
    \[\begin{split}\vE' \cdot (v_1', \dots, v_{\dim(W_\beta')}')^T = & \vG' \\ = & L_\vE \vG R^{-1} \\ = & L_\vE \vE \cdot (v_1, \dots, v_{\dim(W_\beta)})^T R^{-1} \\ = & L_\vE \vE \cdot (v_1 R^{-1}, \dots, v_{\dim(W_\beta)}R^{-1})^T. \end{split}\]
    Since $R_\vE$ is invertible, and $v_1', \dots, v_{\dim(W_\beta')'}$ are linearly independent, $v_1 R^{-1}, \dots, v_{\dim(W_\beta)R^{-1}}$ are also linearly independent. 
    Consequently, $\vE' = L_\vE \vE' R_\vE^{-1}$.
\end{proof}

\subsubsection{A single \IBCtuple from Compression Matrix Tuple}

Now we give an algorithm to construct an \IBCtuple of the matrix tuple $\vA$ based on a given \IBCtuple of the compression matrix tuple $\vE$. 

Roughly speaking, our algorithm retrieves the extensions in $\ibcextensionspace_{\vB_i, \vA}$ for \IBCtuples $\vB_1, \dots, \vB_k$ in the representative relevant \IBCtuple sequence such that compression row tuples of the extensions by the Extension Compression Algorithm correspond to the row tuples in the given \IBCtuple of the compression matrix tuple. 
Then we use these extensions to construct an \IBCtuple for the input matrix tuple. 

\begin{framed}
\noindent \textbf{Single \IBCtuple Construction Algorithm}

\noindent \textbf{Input:} 

\begin{enumerate}
    \item Hierarchical row tuple decomposition $T_1, \dots, T_\zeta$ with parameters $h_0, \dots, h_\beta$ for a matrix tuple $\vA$. 
    \item Representative relevant \IBCtuple sequence $\vB_1, \dots, \vB_k$ for a $\beta$-quotient matrix tuple of $\vA$ with respect to $T_1, \dots, T_\zeta$ with parameters $t_{\vB_i, 1}, \dots, t_{\vB_i, \sigma_i}$, essential extension space $\ibcextensionspace_{\vB_i, \vA}$, essential extension space kernel $\ibcextensionspacekernel_{\vB_i, \vA}$ for every $1 \leq i \leq k$. 
    \item $\vE, \vG, \vE_1,\dots, \vE_\rho, (v_1, \dots, v_{\dim(W_\beta)}), T_{\vE, 1}, \dots, T_{\vE, \rho}$ as the output of the Compression Matrix Tuple Algorithm for $\vB_1, \dots, \vB_k$. 
    \item \IBCtuple $\vF$ of $\vE$ with row tuple space parameters $t_{\vF, 1}, \dots, t_{\vF, \eta}$, where $\eta$ is the number of rows of $\vF$. 
\end{enumerate}

\noindent \textbf{Output:} Matrix tuple $\vH$ of $\sigma$ rows with parameters $t_{\vH, 1}, \dots, t_{\vH, \sigma}$.

\noindent \textbf{Algorithm:}

\begin{enumerate}
    \item Let $h = 0, r= 0$. For every $1 \leq k' \leq k$, let $s_{k'} = \min_{i \in [\sigma_{k'}]} t_{\vB_{k'},i}$ and $u_{k'} = \min (\{i \in [\sigma_{k'}]: t_{\vB_{k'}, i} = s_{k'}\})$, and $H_{k'} = T_{s_{k'}} \cap \langle T_{h_\beta} \cup \dots \cup T_{\zeta}    \rangle$. 
    \item For $k' = 1, \dots, k$ and for $i = 1, \dots, \eta$, if $t_{\vF, i} = k'$, 
    
    \begin{enumerate}
        \item Let $\vG$ be an arbitrary matrix tuple in $\ibcextensionspace_{\vB_{k'}, \vA}$ such that $\EC(\vG)$ corresponds to $e_i \vF (v_1', \dots, v_{\dim(W_\beta)}')^T$. 
        \item If $e_{u_{k'}}\vG$ is not in $H$, then let 
        $h = h + 1$, $\vH_h = \vG$, $H_{k'}$ be $\langle H_{k'} \cup \{e_{j} \vG : j\in[\sigma_{k'}] \text{ s.t. } t_{\vB_{k'}, j} = s_{k'}\}\rangle$, and $t_{\vH_h, j} = t_{\vB_{k'}, j}$ for every $j \in [\sigma_{k'}]$.
    \end{enumerate}
    \item Let $h = h + 1$. For $i = 1, \dots, \eta$, if $t_{\vF, i} > k$, then let $r = r + 1$, 
    $r$-th row of $\vH_h$ be the row tuple in $\rowtuplespace(\vE_{t_{\vF, i}})$ corresponding to $e_i \vF (v_1, \dots, v_{\dim(W_\beta)})^T$, and $t_{\vH_h, r} = t_{\vF, i} - k + h_\beta - 1$. 
    
    \item Let $\vH = \begin{bmatrix}
        \vH_1 \\ \vdots \\ \vH_h
    \end{bmatrix}$, and $\sigma$ denote the number of rows in $\vH$. 
    For each $1 \leq i \leq \sigma$, let $b$ and $c$ be the integer such that $e_i \vH$ is $b$-th row of $\vH_c$, and set $t_{\vH, i} = t_{\vH_c, b}$.
    \item Return $\vH$ with $t_{\vH, 1}, \dots, t_{\vH, \sigma}$. 
\end{enumerate}
\vspace{-.3cm}
\end{framed}

We show that if an \IBCtuple of the compression matrix tuple is provided, then the Single \IBCtuple Construction Algorithm generates an \IBCtuple of the input matrix corresponding to the given \IBCtuple of the compression matrix tuple.

Lemma~\ref{lemma:single_ibctuple_relevant} and Lemma~\ref{lemma:single_ibctuple_relevant_properties} establish that the algorithm's output is an \IBCtuple of the input matrix tuple, and it satisfies the four properties defined in Definition~\ref{def:four_prop}. Lemma~\ref{lemma:canonical_single_ibctuple_construction} further illustrates that the output is canonical in the sense that while it may vary for the same input, it remains right-equivalent.

\begin{claim}\label{claim:property_ibctuple_reconstruction}
    Let $\vQ$ be the $\beta$-quotient matrix tuple of $\vA$ from which a representative relevant \IBCtuple sequence $\vB_1, \dots, \vB_k$ is constructed. 
    Suppose all of $\vB_1, \dots, \vB_k$ satisfy the space property, row tuple space property, dimension property, and block-compatible property. 
    Let $\vE$ and $\vG$ be the output of the Compression Matrix Tuple Algorithm for $\vB_1, \dots, \vB_k$, and $\vH$ be the output of the Single \IBCtuple Construction Algorithm for an \IBCtuple $\vF$ of $\vE$. The following properties hold:
    \begin{enumerate}
        \item Let $\vH_1, \dots, \vH_h$ be the matrix tuples obtained in Steps 2 and 3 of the algorithm so that \[\vH = \begin{bmatrix}
            \vH_1 \\ \vdots \\ \vH_h
        \end{bmatrix},\] and let $\vB_1', \dots, \vB_{h-1}'$ be the \IBCtuples of $\vQ$ such that $\vH_j$ is an extension of $\vB_j'$ for each $j \in [h-1]$. 
        $\vB_j'$ is not in the linear space spanned by $\vB_1', \dots, \vB_{j-1}'$ and their correlated \IBCtuples for each $j \in [h-1]$.
        \item For each $i \in [k]$, let $B_i$ be the subset of $\{\vB_1', \dots, \vB_{h-1}'\}$ containing all the \IBCtuples are right-equivalent to $\vB_i$. The linear space spanned by \IBCtuples in $B_i$ and their correlated \IBCtuples equals $\ibcspace_{\vB_i, \rowvectorspace(\vF)}$, where  
        \[\begin{split}\ibcspace_{\vB_i, \rowvectorspace(\vF)} = \langle \{\vB \in \ibcspace_{\vB_i}: & \exists \vC \in \ibcextensionspace_{\vB_i, \vA} \text{ s.t. } \EC(\vC) \leq \rowvectorspace(\vF)\cdot (v_1, \dots, v_{\dim(W_\beta)})^T, \\ &  \text{ and } \vC \text{ is an extension of } \vB\},\rangle\end{split}\]
        where $(v_1, \dots, v_{\dim(W_\beta)})$ is the vector such that $\vE \cdot (v_1, \dots, v_{\dim(W_\beta)})^T = \vG$. 
        \item For any non-zero $u \in \F^\sigma$, $u \vH$ is non-zero row tuple, where $\sigma$ is the number of rows of $\vH$. 
    \end{enumerate}
  
\end{claim}
\begin{proof}
    The first property is obtained by the row tuple space property of $\vB_1, \dots, \vB_k$ and Step 2 of the algorithm. 
    
    For the second property, 
    by Claim~\ref{claim:vE_direct_sum}, the linear space spanned by \IBCtuples in $B_i$ and their correlated books is a subspace of $\ibcspace_{\vB_i, \rowvectorspace(\vF)}$. 
    On the other hand, 
    by Claim~\ref{claim:properties_ibctuple_extension}, for every \IBCtuple $\vB' \in \ibcspace_{\vB_i, \rowvectorspace(\vF)}$, there is an \IBCtuple $\vB''$ strongly correlated to $\vB'$ such that $\vB''$ has an extension $\vC'' \in \ibcextensionspace_{\vB_i, \vA}$ such that $\rowvectorspace(\EC(\vC))$ is a subspace of $\rowvectorspace(\vF)\cdot (v_1, \dots, v_{\dim(W_\beta)})^T$. 
    By the correspondence between $\EC(\vC)$ for $\vC \in \ibcextensionspace_{\vB_i, \vA}$ and row tuples of $\vG$, there is a row tuple in $\rowtuplespace(\vF \cdot (v_1, \dots, v_{\dim(W_\beta)}))^T$ corresponding to $\EC(\vC'')$.
    By the Single \IBCtuple Construction Algorithm, $\vB'$ is in the linear space spanned by \IBCtuples in $B_i$. 
    Hence, the second property holds.

    For the last property, let $\sigma'$ denote the number of rows of \[\begin{bmatrix} \vH_1 \\ \vdots \\ \vH_{h-1} \end{bmatrix}.\]
    By Claim~\ref{claim:sequential_ibctuple_simultaneous} and Step 2 of the algorithm, for any $j \in [\sigma']$, $e_j \vH + W_\beta$ is not a linear combination of $e_1 \vH + W_\beta, \dots, e_{j - 1}\vH + W_\beta$, and thus $e_j \vH$ is independent of $e_1 \vH, \dots, e_{j-1}\vH$. 
    Also, if a linear combination of $e_1 \vH, \dots, e_{\sigma'}\vH$ has all row vectors in $W_\beta$, then the coefficients of the linear combination are all zero.

    On the other hand, 
    each of $e_{\sigma' + 1}\vH, \dots, e_{\sigma}\vH$ has all row vectors in $W_\beta$. 
    If $e_j\vH$ is a linear combination of $e_1\vH, \dots, e_{j-1}\vH$ for some $\sigma' + 1 \leq j \leq \sigma$, then the coefficients of the linear combination for $e_1\vH, \dots, e_{\sigma'}\vH$ are all zero. 
    In addition, by the construction of $\vH$, since each of $e_{\sigma' + 1}\vH, \dots, e_{\sigma}\vH$ corresponds to a row tuple of $\vF$ which produces $\vH$, by the definition of \IBCtuplenospace, $e_{\sigma' + 1}\vH, \dots, e_{\sigma}\vH$ are linearly independent. 
\end{proof}

\begin{lemma}\label{lemma:canonical_single_ibctuple_construction}
For a matrix tuple $\vA \in \M(n \times m, \F_q)^\ell$ with a depth-$\beta$ hierarchical row tuple decomposition $T_1, \dots, T_\zeta$, a representative relevant \IBCtuple sequence $\vB_1, \dots, \vB_k$ for a $\beta$-quotient matrix tuple $\vQ$ of $\vA$, a compression matrix tuple $\vE$ with respect to $\vB_1, \dots, \vB_k$, and an \IBCtuple $\vF$ of $\vE$, 
we have the following observations for the Single \IBCtuple Construction Algorithm:
\begin{enumerate}
    \item The running time of the Single \IBCtuple Construction Algorithm is $\poly(n, m, \ell, \log q)$. 
    \item The algorithm is canonical in the following sense: For another input of the algorithm, denoted as $\vA', T_1', \dots, T_\zeta', \vB_1', \dots, \vB_k', \vE', \vF'$, such that there are invertible matrices $L, R, L_\vE, R_\vE$ satisfying
    \begin{enumerate}
        \item $\vA' = L \vA R^{-1}$. 
        \item $T_1, \dots, T_\zeta$ and $T_1', \dots T_\zeta'$ have the same parameters and $T_i' = T_i R^{-1}$ for each $i \in [\zeta]$.
        \item $\vB_i$ and $\vB_i'$ are right-equivalent for every $i \in [k]$. 
        \item $\vE' = L_\vE \vE R_\vE^{-1}$.
        \item $\vF$ and $\vF'$ are right-equivalent.
    \end{enumerate}
    the outputs of the two inputs to the Single \IBCtuple Construction Algorithm are right-equivalent and have the same parameters. 
\end{enumerate}

\end{lemma}
\begin{proof}
The running time of the definition of the algorithm and the observation that all the matrix tuples have at most $\poly(n, m, \ell)$ rows, $m$ columns, and $\poly(n, m, \ell)$ length, and all the linear spaces are of dimension at most $\poly(n, m, \ell)$.

Since $\vF$ and $\vF'$ are right-equivalent, 
$\vF$ and $\vF'$ have the same number of rows, and $t_{\vF, i} = t_{\vF', i}$ for each $i \in \eta$, where $\eta$ is the number of rows of $\vF$. 
And thus $s_{k'}, u_{k'}$, and $H_{k'}$ for every $k' \in [k]$ are the same for $\vF$ and $\vF'$. 
By induction for Step 2 of the algorithm, Step 2 for any $i$ and $k'$ on $\vF$ produces a new matrix tuple if and only if Step 2 for the same $i$ and $k'$ on $\vF'$ produces a new matrix tuple.
Hence, $h$ for both $\vF$ and $\vF'$ are the same, and for any $\vH_j$ and $\vH_j'$ are essential extensions of right-equivalent \IBCtuples for any $j \in [h - 1]$. 
For Step 3, since $\vF$ and $\vF'$ are right-equivalent, by induction, 
$\vH_h$ and $\vH'_h$ are right-equivalent. 

On the other hand, 
for any $j \in [h - 1]$, \[\rowvectorspace((I - Z_{\vB_j} Y_{\vB_j}) \mtupletomatrix(\vH_j))\] is a subspace of $W_\beta$, and \[\rowvectorspace((I - Z_{\vB_j'} Y_{\vB_j'}) \mtupletomatrix(\vH_j'))\] is a subspace of $W_\beta'$. 
By Claim~\ref{claim:property_ibctuple_reconstruction} and the algorithm, we have \begin{equation}\label{equ:reconstruction_1}\begin{split} \rowvectorspace(\vH)  = &  \rowvectorspace(Y_{\vB_1}\mtupletomatrix(\vH_1))\oplus \dots \oplus \rowvectorspace( Y_{\vB_{h-1}} \mtupletomatrix(\vH_{h-1})) \\ &  \oplus \rowvectorspace(\vF\cdot (v_1, \dots, v_{\dim(W_\beta)})^T)\end{split}\end{equation}
and 
\begin{equation}\label{equ:reconstruction_2}\begin{split} \rowvectorspace(\vH')  = & \rowvectorspace( Y_{\vB_1'}\mtupletomatrix(\vH_1'))\oplus \dots \oplus \rowvectorspace(Y_{\vB_{h-1}'} \mtupletomatrix(\vH_{h-1}'))  \\  & \oplus \rowvectorspace(\vF'\cdot (v_1', \dots, v'_{\dim(W_\beta)'})^T),\end{split}\end{equation}
where $(v_1, \dots, v_{\dim(W_\beta)})$ and $(v_1', \dots, v_{\dim(W_\beta')}')$ are the vectors such that \[\vE \cdot (v_1, \dots, v_{\dim(W_\beta)})^T = \vG\] and \[\vE' \cdot (v_1', \dots, v_{\dim(W_\beta')}')^T = \vG'\] by the Compression Matrix Tuple Algorithm.
By Lemma~\ref{lemma:extension_canonical}, there exist invertible matrices $R_1, \dots, R_{h-1}$ such that \[\vH_j' = \vH_j R_j^{-1}\] for any $j \in [h-1]$. 
By the correspondence between $\EC(\vH_1), \dots, \EC(\vH_{h-1})$ and row tuples of $\vF$, there is an invertible matrix $R_h$ such that \[\vF'\cdot (v_1', \dots, v'_{\dim(W_\beta')})^T = \vF\cdot (v_1, \dots, v_{\dim(W_\beta)})^T
R_h^{-1},\]
and for every $j \in [h - 1]$,
\[(I - Z_{\vB_j'} Y_{\vB_j'}) \mtupletomatrix(\vH_j') = (I - Z_{\vB_j} Y_{\vB_j}) \mtupletomatrix(\vH_j) R_h^{-1}.\]
By Equation (\ref{equ:reconstruction_1}) and (\ref{equ:reconstruction_2}),
$\vH$ and $\vH'$ are right-equivalent.
\end{proof}

\begin{lemma}\label{lemma:single_ibctuple_relevant}
    Let $\vQ$ be the $\beta$-quotient matrix tuple of $\vA$ from which a representative relevant \IBCtuple sequence $\vB_1, \dots, \vB_k$ is constructed. 
    Suppose all of $\vB_1, \dots, \vB_k$ satisfy the space property, row tuple space property, dimension property, and block-compatible property. 
    Let $\vE$ and $\vG$ be the output of the Compression Matrix Tuple Algorithm for $\vB_1, \dots, \vB_k$.
    For a matrix tuple $\vH$ which is the output of the Single \IBCtuple Construction Algorithm for an \IBCtuple of $\vE$, $\vH$ is an \IBCtuple of $\vA$. 
\end{lemma}
\begin{proof}
We first show that there is a block-diagonalization $\vD = \diag(\vD_1, \dots, \vD_d)$ of $\vA$ and invertible matrices $L, R$ satisfying $\vD = L \vA R^{-1}$ such that $\vH = \vA_{1, \vD, L}$.

Let $\vF$ be the \IBCtuple of $\vE$ from which $\vH$ is obtained by the Single \IBCtuple Construction Algorithm. 
By Definition~\ref{def:ibctuple}, there is a minimum block-diagonalization $\vE' = \diag(\vE'_1, \dots, \vE'_{d'})$ of $\vE$ 
and invertible matrices $L'$ and $R'$ satisfying $\vE' = L' \vE (R')^{-1}$
such that $\vE_{1, \vE', L'} = \vF$. 
We describe a method to construct the desirable $\vD$ based on $\vE$ and $\vE'$ as follows:
\begin{enumerate}
    \item Let $d = d'$. Let $\vC_1 = \vH$, and for every $2 \leq i \leq d$, let $\vC_i$ be the output of the Single \IBCtuple Construction Algorithm for $\vE_{i, \vE', L'}$. 
    \item Let $\vB_1', \dots, \vB_{k_0}'$ be a representative irrelevant \IBCtuple sequence of $\vQ$. 
    For $j = 1, \dots, k_0$, 
    \begin{enumerate}
        \item Select a maximal sequence $\vC_{d+1}, \dots, \vC_{d + t} \in \ibcextensionspace_{\vB_{j}', \vA}$ satisfying the following two conditions:
        \begin{enumerate}
            \item For any $1 \leq i' \leq t$, $\vC_{d + i'}$ is not an extension of an \IBCtuple in the linear space spanned by $\vB_{j, d + 1}, \dots, \vB_{j, d + i' - 1}$ and their correlated \IBCtuples in $\ibcspace_{\vB_j'}$, where $\vB_{j, d + 1}, \dots, \vB_{j, d + i' - 1}$ are \IBCtuples in $\ibcspace_{\vB_j'}$ such that $\vC_{d + 1}, \dots, \vC_{d + i' - 1}$ are the extensions of $\vB_{j, d + 1}, \dots, \vB_{j, d + i' - 1}$.  
            \item The linear space spanned by $\vB_{j, \xi + 1}, \dots, \vB_{j, \xi + t}$ and their correlated \IBCtuples equals $\ibcspace_{\vB_j'}$.
        \end{enumerate}
        \item Set $d = d + t$. 
    \end{enumerate} 
    \item Let \[\vA' = \begin{bmatrix}
        \vC_1 \\ \vdots \\ \vC_d
    \end{bmatrix},\] and $R$ be an arbitrary matrix such that the $(\sum_{j = 1}^{i-1} \delta_j ) + 1$ to $(\sum_{j = 1}^{i} \delta_j)$-th row of $R$ form a linear basis of $\rowvectorspace(\vC_i)$ for any $i \in d$, where $\delta_j = \dim(\rowvectorspace(\vC_j))$. 
    
    \item Return $\vA' R^{-1}$. 
\end{enumerate}
In the rest of this proof, we show that $R$ is a square invertible matrix, and $\vA' R^{-1}$ is the desirable $\vD$. 
We first show that there is an invertible matrix $L$ such that $\vA' = L \vA$. 

By the construction of $\vA'$, the sequence of \IBCtuples $\vB_1, \dots, \vB_k, \vB'_1, \dots, \vB_{k_0}'$ is a representative \IBCtuple sequence of $\vQ$. 
For each $i \in [d]$, if $\vC_i$ is obtained by the Single \IBCtuple Construction Algorithm, then let $\vC_{i, 1}, \dots, \vC_{i, h_i}$ denote the matrix tuples selected by the Single \IBCtuple Construction Algorithm to construct $\vC_i$, otherwise, let $h_i = 1$ and $\vC_{i, 1} = \vC_i$. 
If $\vC_{i, j}$ is an extension of an \IBCtuple for $\vQ$, we use $\vB_{i, j}$ to denote the \IBCtuple for $\vQ$ such that $\vC_{i, j}$ is an extension of $\vB_{i, j}$. 
Let $B$ be the set of all \IBCtuples obtained, and 
for each of $\vB \in \{\vB_1, \dots, \vB_k, \vB'_1, \dots, \vB_{k_0}'\}$,
let $B_\vB$ be all 
the \IBCtuples in $B$ right-equivalent to $\vB$. 
By Claim~\ref{claim:vE_direct_sum}, Claim~\ref{claim:property_ibctuple_reconstruction} and the method to construct $\vA'$ above, for each of $\vB \in \{\vB_1, \dots, \vB_k, \vB'_1, \dots, \vB_{k_0}'\}$,
the linear space spanned by \IBCtuples in $B_\vB$ and their correlated \IBCtuples is $\ibcspace_{\vB}$, and for each $\vB' \in B_{\vB}$, $\vB'$ is not in the linear space spanned by \IBCtuples in $B_\vB \setminus \{\vB'\}$ and their correlated \IBCtuplesnospace. 
By Claim~\ref{claim:sequential_ibctuple_simultaneous}, $\rowtuplespace(\vA') / W_\beta = \rowtuplespace(\vA) / W_\beta$. 

On the other hand, by the definition of matrix tuple $\vE$ and $\vC_1, \dots, \vC_d$, we have
\[\langle T_{h_\beta} \cup \dots \cup T_\ell \rangle = (\rowtuplespace(\vC_1) \cap \langle T_{h_\beta} \cup \dots \cup T_\ell \rangle) \oplus \dots \oplus (\rowtuplespace(\vC_\xi) \cap \langle T_{h_\beta} \cup \dots \cup T_\ell \rangle).\] 
Hence $\rowtuplespace(\vA) = \rowtuplespace(\vA')$, and thus there is an invertible matrix $L$ such that $\vA' = L \vA$. 

Next, we show that the matrix $R$ constructed above is an invertible matrix such that $\vA' R^{-1}$ is the desirable block-diagonal matrix tuple. By Claim~\ref{claim:property_ibctuple_reconstruction},
for $\vC_i$ obtained in Step 1 for the construction of $\vA'$, we have 
\[\begin{split}& \rowvectorspace(\vC_i) / W_\beta \\ = & (\rowvectorspace(\vC_{i, 1}) / W_\beta) \oplus \dots \oplus (\rowvectorspace(\vC_{i, h_i - 1}) / W_\beta) \\
= & \rowtuplespace(\vB_{i, 1}\cdot (v_1 + W_\beta, \dots, v_{m_\beta} + W_\beta)^T)  \oplus \dots  \\ & \oplus \rowtuplespace(\vB_{i, h_i - 1}\cdot (v_1 + W_\beta, \dots, v_{m_\beta} + W_\beta)^T),   \end{split}\]
where $(v_1 + W_\beta, \dots, v_{m_\beta} + W_\beta)$ is the formatting vector of $\vQ$. 
For $\vC_i$ obtained in Step 2 for the construction of $\vA'$, we have 
\[\rowvectorspace(\vC_i) / W_\beta = \rowtuplespace(\vB_{i, 1}\cdot (v_1 + W_\beta, \dots, v_{m_\beta} + W_\beta)^T).\] 
By 
Claim~\ref{claim:sequential_ibctuple_simultaneous}, Claim~\ref{claim:vE_direct_sum}, 
\[\rowvectorspace(\vA) / W_\beta = (\rowvectorspace(\vC_1) / W_\beta) \oplus \dots \oplus (\rowvectorspace(\vC_d) / W_\beta).\]
In addition, by $\vE$ and the Single \IBCtuple Construction Algorithm, 
\[\rowvectorspace(\vA') \cap W_\beta = (\rowvectorspace(\vC_1) \cap W_\beta) \oplus \dots \oplus (\rowvectorspace(\vC_d) \cap W_\beta).\]
By Fact~\ref{fact:ibctuple_diagonalzation_most_basic}, $R$ is an invertible matrix such that $\vA' R^{-1}$ is a block-diagonal matrix tuple $\vD$ such that $\vC = \vA_{1, \vD, L}$.
In addition, $\vD_1$ must be indecomposable, because otherwise $\vE'$ is not a minimum block-diagonalization of $\vE$ by Claim~\ref{claim:submatrix_tuples_correlated}. Thus, $\vC$ is an \IBCtuple of $\vA$. 
\end{proof}

\begin{lemma}\label{lemma:single_ibctuple_relevant_properties}
For the input $\vB_1, \dots, \vB_k, \vE, \vG, \vF$ of the Single \IBCtuple Construction Algorithm satisfying the following conditions:
\begin{enumerate}
    \item Each of $\vB_1, \dots, \vB_k$ satisfies the space property, row tuple space property, dimension property, and the block-compatible property. 
    \item $\vF$ is an \IBCtuple of $\vE$ such that the \IBCtuple space of $\vE$ is obtained by the Direct Sum \IBCtuple Algorithm, and $\vF$ is an \IBCtuple in the \IBCtuple space. 
\end{enumerate}
    the output $\vH$ of the Single \IBCtuple Construction Algorithm satisfies the space property, row tuple space property with parameters $t_{\vH, 1}, \dots, t_{\vH, \sigma}$, dimension property, and block-compatible property. 
\end{lemma}
\begin{proof}
    Let $T_1, \dots, T_\zeta$ be the hierarchical row tuple decomposition with parameters $h_0, \dots, h_\beta$ for a matrix tuple $\vA$,  
    $\vB_1, \dots, \vB_k$ be the representative relevant \IBCtuple sequence for a $\beta$-quotient matrix tuple of $\vA$ with respect to $T_1, \dots, T_\zeta$ with parameters $t_{\vB_i, 1}, \dots, t_{\vB_i, \sigma_i}$, essential extension space $\ibcextensionspace_{\vB_i, \vA}$, essential extension space kernel $\ibcextensionspacekernel_{\vB_i, \vA}$ for every $1 \leq i \leq k$.
    
    Let 
    $\vE, \vG, \vE_1,\dots, \vE_\rho, (v_1, \dots, v_{\dim(W_\beta)}), T_{\vE, 1}, \dots, T_{\vE, \rho}$ be the output of the Compression Matrix Tuple Algorithm for $\vB_1, \dots, \vB_k$, and \IBCtuple $\vF$ be the \IBCtuple of $\vE$ such that $\vH$ is obtained by running the Single \IBCtuple Construction From Compression Algorithm for $\vF$. 
     
    We first prove the space property. 
    Denote \[\vH = \begin{bmatrix}
        \vH_1 \\ \vdots \\ \vH_h
    \end{bmatrix}.\] 
    Let $\vB_i'$ denote the \IBCtuple of $\vQ$ such that $\vH_i$ is an \IBCtuple extension of $\vB_i'$ for each $i \in [h-1]$. 
    By Claim~\ref{claim:properties_ibctuple_extension}, for any $\vH'$ in 
    \begin{equation}\label{equ:ibctuple_space_kernel_relevant}\ibcspacekernel = \left\langle \left\{\begin{bmatrix}
            \vH_{1}' \\ \vdots \\ \vH_{h-1}' \\ 0
        \end{bmatrix} :\forall i \in [h - 1], \vH_{i} \in \ibcextensionspacekernel_{\vB_{i}', \vA}\right\}\right\rangle,
    \end{equation}
    $\vH + \vH'$ is also a valid output for the Single \IBCtuple Construction Algorithm for $\vF$. 
    Furthermore, if $\vH + \vH''$ is a valid output for $\vF$, then $\vH''$ is in $K$. 

    On the other hand, for another \IBCtuple $\vF'$ of $\vE$ that is right-equivalent to $\vF$, 
    let $\vH'$ be an output of the Single \IBCtuple Construction Algorithm for $\vF'$. 
    By Lemma~\ref{lem:direct_sum_ibctuple_algo}, $\vF + \vF'$ is also an \IBCtuple for $\vE$, and $\vH + \vH'$ is a valid output for the Single \IBCtuple Construction Algorithm for $\vF + \vF'$. 

    Hence, let $\ibcspace$ be the linear space spanned by all the possible outputs by the \IBCtuples right-equivalent to $\vF$ for the Single \IBCtuple Construction Algorithm. 
    Then, every matrix tuple in $\ibcspace$ but not in $\ibcspacekernel$ is a valid output for some \IBCtuple right-equivalent to $\vF$.
    Hence, let $\ibcspace_\vH = \ibcspace$ and $\ibcspacekernel_\vH = \ibcspacekernel$. We obtain the space property of $\vH$.

    Note that the essential extensions of $\vB_1, \dots, \vB_k$ satisfy the row tuple space property, dimension property, and block-compatible property by Claim~\ref{claim:properties_ibctuple_extension} and Lemma~\ref{lemma:essential_extension_ibctuple_compatible},
    and the \IBCtuple $\vF$ satisfies the row tuple space property, dimension property, and block-compatible property by Lemma~\ref{lem:direct_sum_ibctuple_algo}. 
    Together with the correspondence between essential extensions $\vB_1, \dots, \vB_k$ and $\vE$, we obtain the row tuple space property, dimension property, and block-compatible property for $\vH$. 
\end{proof}

\subsubsection{Relevant representative \IBCtuple sequence}

Finally, we give an algorithm to canonically compute an \IBCtuple of $\vA$, which contains row tuples in $\rowtuplespace(\langle T_{h_\beta},\dots, T_\zeta \rangle)$, together with its \IBCtuple space and \IBCtuple space kernel. 
Actually, the relevant representative \IBCtuple sequence contains only a single \IBCtuple all the \IBCtuples in the compression matrix tuple are equivalent. 
\begin{framed}
\noindent \textbf{Relevant \IBCtuple Space Algorithm}

\noindent \textbf{Input:} 
\begin{enumerate}
    \item Matrix tuple $\vA$ with hierarchical row tuple decomposition $T_1, \dots, T_\zeta$ and parameters $h_0, \dots, h_\beta$.
    \item Representative relevant \IBCtuple sequence $\vB_1, \dots, \vB_k$ for a $\beta$-quotient matrix tuple of $\vA$ with respect to $T_1, \dots, T_\zeta$ with parameters $t_{\vB_i, 1}, \dots, t_{\vB_i, \sigma_i}$, essential extension space $\ibcextensionspace_{\vB_i, \vA}$, essential extension space kernel $\ibcextensionspacekernel_{\vB_i, \vA}$ for every $1 \leq i \leq k$. 
\end{enumerate}

\noindent \textbf{Output:}  
\IBCtuple $\vH$ of $\vA$ with $\ibcspace_\vH, \ibcspacekernel_{\vH}, t_{\vH, 1}, \dots, t_{\vH, \sigma}$, or a nontrivial row tuple subspace $S < T_{h_\beta}$. 

\begin{enumerate}
    \item Run the Compression Matrix Tuple  Algorithm for $\vB_{1}, \dots, \vB_{k}$ to get $\vE, \vG, T_{\vE, 1}, \dots, T_{\vE, \rho}$. 
    \item Run the Direct Sum Decomposition Algorithm with $\vE$ with direct sum decomposition $T_{\vE, k + 1}, \dots, T_{\vE, \rho}, T_{\vE, 1}, \dots, T_{\vE, k}$. 
    If the output is $S < T_{\vE, k+1}$, then return the corresponding subspace of $S'$ in $T_{h_\beta}$, otherwise, let $\vF$, $\ibcspace_{\vF}, \ibcspacekernel_{\vF}, t_{\vF, 1}, \dots, t_{\vF, \eta}$ be the output of the algorithm with $\vF$ having $\eta$ rows. 
    
        \item Run Single \IBCtuple Construction Algorithm on $\vF$, and let 
        \[\vH = \begin{bmatrix}
            \vH_{1} \\ \vdots \\ \vH_{h}
        \end{bmatrix}\] and $t_{\vH, 1}, \dots, t_{\vH, \sigma}$ be the output, and 
        $\vB_1', \dots, \vB_{h-1}' \in \{\vB_1, \dots, \vB_k\}$ such that $H_j$ is an extension of an \IBCtuple right-equivalent to $\vB_j'$ for any $j \in [h-1]$.
        
        \item Compute an arbitrary linear basis $\vF_1, \dots, \vF_{\dim(\ibcspace_{\vF})}$ of $\ibcspace_{\vF}$, and compute $\vH_{1}, \dots, \vH_{\dim(\ibcspace_{\vF})}$ as the outputs of the Single \IBCtuple Construction Algorithm for $\vF_1, \dots, \vF_{\dim(\ibcspace_{\vF})}$.
        \item Return $\vH, \ibcspace_\vH, \ibcspacekernel_\vH, t_{\vH, 1}, \dots, t_{\vH, \sigma}$, where  $\ibcspacekernel_{\vH}$ is obtained by Equation (\ref{equ:ibctuple_space_kernel_relevant}), and $\ibcspace_{\vH} = \langle \{ \vH_{1}, \dots, \vH_{\dim(\vF)} \} \cup \ibcspacekernel_{\vH} \rangle$. 
\end{enumerate}
\vspace{-.3cm}
\end{framed}

\begin{claim}\label{claim:relevant_ibctuple_space}

    The Relevant \IBCtuple Space Algorithm satisfies the following properties:
    \begin{enumerate}
        \item 
            Given a matrix tuple $\vA = \M(n \times m, \F_q)^\ell$ with a hierarchical row tuple decomposition $T_1, \dots, T_\zeta$ and parameters $h_0, \dots, h_\beta$, 
            and a representative relevant \IBCtuple sequence $\vB_1, \dots, \vB_k$ for a $\beta$-quotient matrix tuple of $\vA$ with essential extension space $\ibcextensionspace_{\vB_i, \vA}$ and essential extension space kernel $\ibcextensionspacekernel_{\vB_i, \vA}$ for every $i \in [k]$, in $\poly(n, m, \ell, \log q)$ running time, 
            the output of the algorithm satisfies the following conditions:
            
        \begin{enumerate}
            \item The output is a characteristic block-compatible nontrivial subspace $S < T_{h_\beta}$. 
            \item The output is an \IBCtuple $\vH$ of $\vA$ satisfying the space property with $\ibcspace_\vH, \ibcspacekernel_\vH$, row tuple space property with parameters $t_{\vH, 1}, \dots, t_{\vH, \sigma}$, dimension property, and the block-compatible property. 
        \end{enumerate}
        \item The algorithm is canonical in the following sense: For two inputs $\vA, T_1, \dots, T_\zeta, \vB_1, \dots, \vB_k$ with $\ibcextensionspace_{\vB_i, \vA}$ and $\ibcextensionspacekernel_{\vB_i, \vA}$ for each $i\in[k]$ and $\vA', T_1', \dots, T_\zeta', \vB_1', \dots, \vB_k'$ with $\ibcextensionspace_{\vB_i', \vA'}$ and $\ibcextensionspacekernel_{\vB_i', \vA'}$ for each $i\in[k]$ satisfying the following conditions:
        \begin{enumerate}
        \item[i.] $\vA' = L \vA R^{-1}$. 
        \item[ii.] $T_1, \dots, T_\zeta$ and $T_1', \dots T_\zeta'$ have the same parameters and $T_i' = T_i R^{-1}$ for each $i \in [\zeta]$.
        \item[iii.] $\vB_i$ and $\vB_i'$ are right-equivalent for every $i \in [k]$. 
        \item[iv.] $\ibcextensionspace_{\vB_i', \vA'} = \{\vC R^{-1} :\vC \in \ibcextensionspace_{\vB, \vA}\}$ and $\ibcextensionspacekernel_{\vB_i', \vA'} = \{\vC R^{-1} :\vC \in \ibcextensionspacekernel_{\vB, \vA}\}$ for every $i \in [d]$.
        \end{enumerate}
        the outputs $\vH, \ibcspace_\vH, \ibcspacekernel_\vH, t_{\vH, 1}, \dots, t_{\vH, \sigma}$ and $\vH', \ibcspace_{\vH'}, \ibcspacekernel_{\vH'}, t_{\vH', 1}, \dots, t_{\vH', \sigma'}$ satisfying the following conditions for arbitrary $L$ and $R$ satisfying the condition (i) - (iv)
        \begin{enumerate}
            \item $\vH$ and $\vH'$ are right-equivalent. 
            \item $\ibcspace_{\vH'} = \{\vH'' R^{-1} :\vH'' \in \ibcspace_{\vH}\}$ and $\ibcspacekernel_{\vH'} = \{\vH'' R^{-1} :\vH'' \in \ibcspacekernel_{\vH}\}$.
            \item $\sigma = \sigma'$ and $t_{\vH, i} = t_{\vH', i}$ for any $i\in [\sigma]$. 
        \end{enumerate}

    \end{enumerate}
\end{claim}

\begin{proof}
The first property is obtained by the definition of the algorithm, Lemma~\ref{lem:direct_sum_ibctuple_algo}, Claim~\ref{claim:running_time_compression}, Lemma~\ref{lemma:single_ibctuple_relevant}, and 
Lemma~\ref{lemma:single_ibctuple_relevant_properties}. 
The second property is obtained by Lemma~\ref{lem:direct_sum_ibctuple_algo}, Claim~\ref{claim:running_time_compression}, and Lemma~\ref{lemma:canonical_single_ibctuple_construction}. 
\end{proof}

\subsection{\IBCtuples from quotient matrix tuple}\label{sec:quotient_matrix_tuple_ibc_final}

We put all the pieces from the previous subsections together to give an algorithm to compute a representative \IBCtuple sequence with \IBCtuple space and \IBCtuple space kernel for each \IBCtuple in the sequence.

In our algorithm, we first decompose $W_\beta$ as a direct sum of row vector subspaces for the purpose of computing the essential extensions by the Essential Extension Algorithm.  
Then, we run the Essential Extension Algorithm for each \IBCtuple in the given representative \IBCtuple sequence of the quotient matrix tuple. 
Based on these essential extensions, we partition the representative \IBCtuple sequence of the quotient matrix tuple into a representative relevant \IBCtuple sequence and a representative irrelevant \IBCtuple sequence, depending on whether the \IBCtuples are relevant to $W_\beta$. 
If an \IBCtuple is irrelevant to $W_\beta$, then by Lemma~\ref{lem:irrelevant_ibctuple}, the essential extensions obtained are \IBCtuples of the input matrix tuple. 
Otherwise, we run the Relevant \IBCtuple Space Algorithm to obtain to compute the \IBCtuples of the input matrix tuple that contains row vectors in $W_\beta$.

\begin{framed}
\noindent \textbf{\IBCtuple Space Algorithm}

\noindent \textbf{Input:}

\begin{enumerate}
    \item Matrix tuple $\vA$ with hierarchical row tuple decomposition $T_1, \dots, T_\zeta$ and parameters $h_0, \dots, h_\beta$.
    \item $\beta$-quotient matrix tuple $\vQ$ of $\vA$ with formatting vector $(v_1 + W_\beta, \dots, v_{m_\beta} + W_\beta)$.
    \item Representative \IBCtuple sequence $\vB_1, \dots, \vB_k$ for $\vQ$ with parameters $t_{\vB_i, 1}, \dots,$ $ t_{\vB_i, \sigma_i}$, \IBCtuple space $\ibcspace_{\vB_i}$, \IBCtuple space kernel $\ibcspacekernel_{\vB_i}$ for every $i \in [k]$. 
\end{enumerate}

\noindent \textbf{Output:} A representative \IBCtuple sequence $\vH_1, \dots, \vH_r$ of $\vA$ with parameters $t_{\vH_i, 1}, \dots,$ $ t_{\vH_i, \delta_i}$, \IBCtuple space $\ibcspace_{\vH_i}$, \IBCtuple space kernel $\ibcspacekernel_{\vH_i}$ for every $i \in [r]$, or a row tuple subspace $S< T_i$ for some $1 \leq i \leq \zeta$.

\noindent \textbf{Algorithm:}

\begin{enumerate}
    \item Run the Row Vector Space Direct Sum Algorithm with $T_{h_\beta}$. If the output is a subspace $S < T_{h_\beta}$, then return $S$. Otherwise let $W_{\beta, 1}, \dots, W_{\beta, w}$ denote the output.
    \item For $i = 1, \dots, k$, run the Essential Extension Algorithm for $\vB_i$, 
    \begin{enumerate}
        \item If the output is a subspace $S < T_i$ for some $i \in [\zeta]$, then return $S$. 
        \item If the output is $W_{\beta, j}'$ for some $j\in [w]$, then 
        let $k$ be the smallest integer such that $T_{h_\beta}^{(k)} = W_{\beta, j}$, and return $\{\va \in T_{h_\beta}: \va^{(k)} \in W_{\beta, j}\}$. 
        \item Otherwise, let  $\ibcextensionspace_{\vB_i, \vA}$ and $\ibcextensionspacekernel_{\vB_i, \vA}$ be the output.
    \end{enumerate}
    \item Partition $\vB_1, \dots, \vB_k$ into a representative relevant \IBCtuple sequence $\vB_{i_1}, \dots, \vB_{i_{k_0}}$ and a representative irrelevant \IBCtuple sequence $\vB_{j_1}, \dots, \vB_{j_{k_1}}$ such that $i_1 < \dots < i_{k_0}$ and $j_1 < \dots < j_{k_1}$.
    \item Run the Relevant \IBCtuple Space Algorithm with $\vB_{i_1}, \dots, \vB_{i_{k_0}}$. If the output is a subspace of $T_{h_\beta}$, then return this subspace. Otherwise, let $\vH, \ibcspace_\vH, \ibcspacekernel_\vH, t_{\vH,1}, \dots, t_{\vH, \sigma_\vH}$ be the output. 
    \item Return an \IBCtuple sequence $\vH_1, \dots, \vH_{k_1 + 1}$ as follows:
    \begin{enumerate}
        \item For each $k'\in [k_1]$, $\vH_{k'}$ is an arbitrary matrix tuple in $\ibcextensionspace_{\vB_{j_{k'}}, \vA} \setminus \ibcextensionspacekernel_{\vB_{j_{k'}}, \vA}$, $\ibcspace_{\vH_{k'}} = \ibcextensionspace_{\vB_{j_{k'}}, \vA}$, $\ibcspacekernel_{\vH_{k'}} = \ibcextensionspacekernel_{\vB_{j_{k'}}, \vA}$, $t_{\vH_{k'}, i} = t_{\vB_{j_{k'}}, i}$. 
        \item $\vH_{k_1 + 1} = \vH$, $\ibcspace_{\vH_{k_1 + 1}} = \ibcspace_\vH$, $\ibcspacekernel_{\vH_{k_1 + 1}} = \ibcspacekernel_\vH$, $t_{\vH_{k_1 + 1, j}} = t_{\vH, j}$. 
    \end{enumerate}
\end{enumerate}
\vspace{-.3cm}
\end{framed}

\begin{proof}[Proof of Lemma~\ref{lemma:extension_main}]
The lemma is obtained by the observation that for every $j \in [w]$, there is a $k \in [\ell]$ such that $T^{[k]} = W_{\beta, j}$ according to the Row Vector Space Direct Sum Algorithm,
and Fact~\ref{fact:ibctuple_invariant_basic}, Claim~\ref{claim:row_vector_space_direct_sum}, Lemma~\ref{lemma:essential_extension_ibctuple_compatible}, Lemma~\ref{lem:irrelevant_ibctuple}, 
Claim~\ref{claim:relevant_ibctuple_space}. 
\end{proof}

\begin{proof}[Proof of Lemma~\ref{lemma:extension_main_canonical}]
The lemma is obtained by an induction of each step of the algorithm and Claim~\ref{claim:row_vector_space_direct_sum}, Lemma~\ref{lemma:extension_canonical}, and  Claim~\ref{claim:relevant_ibctuple_space}. 
\end{proof}

\section{Canonical form algorithm}\label{sec:overall_algorithm}
In this section, we first present an algorithm to initialize the hierarchical row tuple decomposition and another algorithm to maintain the hierarchical row tuple decomposition as defined in Definition~\ref{def:row_tuple_decomposition} if any row tuple subspace in the decomposition is updated in Section~\ref{sec:overall_algorithm}. 
Then we give our canonical form algorithm in Section~\ref{sec:overall_final}
\subsection{Hierarchical row tuple decomposition initialization and update}

Before presenting our algorithm for initializing and maintaining a hierarchical row tuple decomposition, we first introduce a few subroutines used by these algorithms.

The following subroutine decomposes a given row tuple subspace into a direct sum of several row tuple subspaces based on the projection of the row vectors in the row tuples onto a given decomposition of the row vector subspace.


\begin{framed}
\noindent \textbf{Direct Sum Row Tuple Decomposition With Row Vector Decomposition Algorithm}

\noindent \textbf{Input:} Characteristic block-compatible row tuple subspace $T$ and characteristic block-compatible row vector subspaces $W_1, \dots, W_w$ such that $\rowvectorspace(T) = W_1 \oplus \dots \oplus W_w$.

\noindent \textbf{Output:} Characteristic block-compatible row tuple subspaces $T_1, \dots, T_\zeta$ such that $T = T_1 \oplus \dots \oplus T_\zeta$, or nontrivial characteristic block-compatible row vector subspace $W < W_i$ for some $1 \leq i \leq w$.

\noindent \textbf{Algorithm:}
\begin{enumerate}
    \item Let $T_1 = T$, $\zeta = 1$, and $\ell$ denote the row tuple length for row tuples in $T$. 
    \item For $j=1, \dots, w$, and for $k = 1, \dots, \ell$, if there exists a row tuple $\va$ in $T_\zeta$ such that 
    $\rowvecproj_{W_1, \dots, W_w}(\va^{(k)}, W_j)$ is non-zero, 
    \begin{enumerate}
        \item Let $W = \langle\{\rowvecproj_{W_1, \dots, W_w}(\va^{(k)}, W_j) : \va \in T_\zeta \}\rangle)$. If $\dim(W) < \dim(W_j)$, then return $W$. 
        
        \item Otherwise, let $k_\zeta = k$, $j_\zeta = j$, $\zeta = \zeta + 1$, and 
        and \[T_\zeta = \langle \{\va \in T_{\zeta - 1} : \rowvecproj_{W_1, \dots, W_w}(\va^{(k)}, W_j) = 0\}\rangle.\] 
        
    \end{enumerate}
    \item $\zeta = \zeta - 1$. 
    \item For $i = \zeta, \dots, 1$, and for $i' = i - 1, \dots, 1$, let \[T = \langle \{\va \in T_{i'} : \rowvecproj_{W_1, \dots, W_w}(\va^{(k_i)}, W_{j_i}) = 0\}\rangle,\] 
    and $T_{i'} \leftarrow T$. 
    
    \item For $i = 1, \dots, \zeta$, for $j = 1, \dots, d$, and for $k = 1, \dots, \ell$, if there is a row tuple $\va \in T_i$ such that $\rowvecproj_{W_1, \dots, W_w}(\va^{(k)}, W_j)$ is non-zero, then 
    \begin{enumerate}
        \item If there exist a non-zero tuple $\va \in T_i$ such that $\rowvecproj_{W_1, \dots, W_w}(\va^{(k)}, W_j)$ is zero, then
        return \[\langle \{\rowvecproj_{W_1, \dots, W_w}(\va^{(k_i)}, W_{j_i}) : \va \in T_i, \rowvecproj_{W_1, \dots, W_w}(\va^{(k)}, W_{j}) = 0\} \rangle.\]
        
        
        \item If $\dim(T_i) < \dim(W_j)$, return $\langle \{\rowvecproj_{W_1, \dots, W_w}(\va^{(k)}, W_{j}) : \va \in T_i\} \rangle$.
    \end{enumerate}
    \item Return $T_1, \dots, T_\zeta$.
\end{enumerate}    
\vspace{-.3cm}
\end{framed}

\begin{claim}\label{claim:tuple_partition_by_vector}
The Direct Sum Row Tuple Decomposition With Row Vector Decomposition Algorithm has the following properties:
\begin{enumerate}
\item
    Given a characteristic block-compatible row tuple subspace $T$ such that every row tuple is in $(\F_q^m)^\ell$ and a sequence of characteristic block-compatible row vector subspaces $W_1, \dots, W_w$ such that $\rowvectorspace(T) = W_1 \oplus \dots \oplus W_w$, in $\poly(m, \ell, \dim(T), \log q)$ time, the Direct Sum Row Tuple Decomposition Without Pivot Algorithm gives one of the following outputs:
    \begin{enumerate}
        \item A sequence of characteristic block-compatible row tuple space $T_1, \dots, T_\zeta$ satisfying the following conditions:
        \begin{enumerate}
            \item $T = T_1 \oplus \dots \oplus T_\zeta$.
            \item For every $1 \leq i \leq \zeta$, every $1 \leq j \leq w$ and every $1 \leq k \leq \ell$,
            the projection of $(T_i)^{(k)}$ on $W_j$ with respect to $W_1, \dots, W_w$ is either a zero projection or a matching projection. 
        \end{enumerate}
        \item A nontrivial characteristic block-compatible subspace $S < W_j$ for some $1 \leq j \leq w$. 
    \end{enumerate}
    \item The algorithm is canonical in the following sense: 
    Let $T'$ be another row tuple space with every row tuple in $(\F_q^m)^\ell$, and $W_1', \dots, W_w'$ be another sequence of row vector spaces such that $\rowvectorspace(T') = W_1' \oplus \dots \oplus W_w'$ and there exists an invertible matrix $R$ satisfying $T' = T R^{-1}$ and $W_j' = W_j R^{-1}$ for every $ 1\leq j \leq w$. The outputs for $T, W_1, \dots, W_w$ and $T', W_1', \dots, W_w'$ satisfy the following conditions:
    (Let $R$ be an arbitrary matrix such that such that $T' = T R^{-1}$ and $W_j' = W_j R^{-1}$ for every $ 1\leq j \leq w$.)
    \begin{enumerate}
        \item If the output for $T$ is a sequence of row tuple spaces $T_1, \dots, T_\zeta$, then the output for $T'$ is $T_1 R^{-1}, \dots, T_\zeta R^{-1}$. 
        \item If the output for $T$ is a subspace $S < W_j$ for some $1 \leq j \leq \zeta$, then the output for $T'$ is $S R^{-1}$. 
    \end{enumerate}
\end{enumerate}

\end{claim}

\begin{proof} 
By the algorithm, after Step 2, $T_\zeta \subsetneq T_{\zeta - 1} \subsetneq \dots \subsetneq T_1$. $T_\zeta$ contains only the zero row tuple because every row vector in any row tuple of $T_\zeta$ has a zero projection on each of $W_1, \dots, W_w$. 
Hence, after Step 4, 
for each $1 \leq i \leq \zeta$,
the $T_i^{(k_i)}$ has a matching projection on $W_{j_i}$,
and the $T_{i'}^{(k_i)}$ for any $i' \neq i$ has a zero projection on $W_{j_i}$. 
This implies that $T$ is the direct sum of $T_1, \dots, T_\zeta$.
Furthermore, Step 5 ensures $(T_i)^{(k)}$ has either a zero projection or a matching projection on $W_j$ for every $1 \leq i \leq \zeta, 1 \leq j \leq w, 1 \leq k \leq \ell$. 
By Fact~\ref{fact:ibctuple_invariant_basic}, the output is always characteristic and block-compatible. 
Hence, the first property holds.

Now, we prove the algorithm is canonical. 
By the algorithm and by induction, after Step 2, we have $T_{i}' = T_{i} R^{-1}$ for each $i$, where $T_i'$ denote the sets constructed by the algorithm for $T'$, and $j_i$ and $k_i$ are the same for the executions for $T$ and $T'$. 
Hence, if the algorithm for $T$ returns a row vector subspace $S$ in Step 2(a), the output for $T'$ is $S R^{-1}$.  
Otherwise, the algorithm for $T$ and $T'$ has the same $\zeta$ after Step 3. 
By induction for Step 4, 
either the algorithm for $T$ and $T'$ terminates during Step 4 with outputs satisfies the second property of the claim, 
or $T_i' = T_i R^{-1}$ for every $1 \leq i \leq \zeta$. 
Similarly, by induction for Step 5, the second property holds for the outputs of the algorithm for $T$ and $T'$. 
\end{proof}

Next, we give an algorithm to decompose a given row tuple space into a direct sum of a few row tuple subspaces canonically with a row tuple subspace already distinguished. 
\begin{framed}
\noindent \textbf{Direct Sum Row Tuple Decomposition Algorithm}

\noindent \textbf{Input:} Characteristic block-compatible row tuple subspace $V$, and $T < V$ also characteristic block-compatible such that $\rowvectorspace(V) = \rowvectorspace(T)$.

\noindent \textbf{Output:} Characteristic block-compatible row tuple subspaces $T_1, \dots, T_\zeta$ such that $V = T \oplus T_1 \oplus \dots \oplus T_\zeta$, or a nontrivial characteristic block-compatible subspace $S < T$.

\noindent \textbf{Algorithm:}

\begin{enumerate}
\item Run the Row Vector Space Direct Sum Algorithm with $T$. Return the output if the output is a subspace of $T$. Otherwise, denote the output as $W_1, \dots, W_w$. 
\item Let $b$ be the smallest integer such that $\dim(T^{(b)}) > 0$, and $c$ be the smallest integer such that the projection of $T^{(b)}$ on $W_c$ is a matching projection. 
\item Run the Direct Sum Row Tuple Decomposition With Row Vector Decomposition Algorithm for $\langle \{\va \in V : \rowvecproj_{W_1, \dots, W_w} (\va^{(b)},  W_c) = 0\}\rangle$.
If the output is a sequence of row tuple spaces $T_1, \dots, T_{\zeta}$, then return $T_1, \dots, T_{\zeta}$. 
\item Otherwise, the output of Step 3 is $W < W_i$ for some $1 \leq i \leq w$. 
Let $j$ be the smallest integer such that the projection of $T^{(j)}$ on $W_i$ is non-zero. Return $\langle \{\va \in T: \rowvecproj_{W_1, \dots, W_w}(\va^{(j)}, W_i) \in W\}\rangle$.
\end{enumerate}
\vspace{-.3cm}
\end{framed}

\begin{claim}\label{claim:tuple_decomposition_direct_sum}
The Direct Sum Row Tuple Decomposition Algorithm has the following properties:
\begin{enumerate}
\item
    Given a characteristic block-compatible row tuple subspace $V$ with every row tuple of $V$ in $(\F_q^m)^\ell$ and a characteristic block-compatible subspace $T$ of $V$ with $\rowvectorspace(T) = \rowvectorspace(V)$, in $\text{poly}(m, \ell, \dim(V), \log q)$ time, the algorithm gives one of the following outputs:
    \begin{enumerate}
        \item A sequence of characteristic block-compatible row tuple subspaces $T_1, \dots, T_\zeta$ satisfying the following conditions: (i.) $V = T \oplus T_1 \oplus \dots \oplus T_\zeta$;
        and (ii.)
            Let $W_1, \dots, W_w$ be the output of the Row Vector Space Direct Sum Algorithm for $T$. 
            For every $1 \leq i \leq \zeta$, $1 \leq j \leq w$, and $1 \leq k \leq \ell$, the projection of $(T_i)^{(k)}$ on $W_j$ with respect to $W_1, \dots, W_w$ is either a zero projection or a matching projection. 
        \item A characteristic block-compatible nontrivial subspace $S < T$.
    \end{enumerate}
    \item The algorithm is canonical in the following sense: 
    Let $V'$ and $T'$ be another input for the algorithm such that there exists an invertible matrix $R$ satisfying 
    $V' = V R^{-1}$ and 
    $T' = T R^{-1}$. The outputs for $V, T$ and $V', T'$ satisfy the following conditions:
    (Let $R$ be an arbitrary invertible matrix such that such that $V' = V R^{-1}$ and  $T' = T R^{-1}$.)
    \begin{enumerate}
        \item If the output for $V$ and $T$ is a sequence of row tuple spaces $T_1, \dots, T_\zeta$, then the output for $V'$ and $T'$ is $T_1 R^{-1}, \dots, T_\zeta R^{-1}$. 
        \item If the output for $V$ and $T$ is a subspace $S < W_j$ for some $1 \leq j \leq \zeta$, then the output for $V'$ and $T'$ is $S R^{-1}$. 
    \end{enumerate}
\end{enumerate}

\end{claim}
\begin{proof}
    For the first property, by Claim~\ref{claim:row_vector_space_direct_sum}, if the output is a subspace of $T$, then the condition (b) of the first property satisfies. Otherwise,  $\rowvectorspace(T) = W_1 \oplus \dots \oplus W_w$ and for every $1 \leq \ell' \leq \ell$, the projection $T^{(\ell')}$ on $W_j$ with respect to $W_1, \dots, W_w$ is either a zero projection or a matching projection.
    By Claim~\ref{claim:tuple_partition_by_vector}, if the output is a sequence of row tuple subspaces $T_1, \dots, T_\zeta$, then the condition (a) of the first property satisfies. Otherwise, the condition (b) of the first property satisfies using the fact that if the projection of $T^{(j)}$ on $W_i$ with respect to $W_1, \dots, W_w$  is a non-zero projection, then it is a matching projection. 
    By Fact~\ref{fact:ibctuple_invariant_basic},  for all the cases, the first property holds. 

    Now, we prove the algorithm is canonical. by Claim~\ref{claim:row_vector_space_direct_sum}, if the output of the Row Vector Space Direct Sum Algorithm is a subspace of $T$, then the condition (b) of the second property holds. 
    Otherwise, let $W_1, \dots, W_w$ denote the row vector spaces obtained in Step 1 for $V, T$, and let $W_1', \dots, W_w'$ denote the row vector spaces obtained in Step 1 for $T'$, $W_i' = W_i R^{-1}$ for every $1 \leq i \leq w$. Then $b$ and $c$ obtained in Step 2 for $V, T$, and $V', T'$ are the same, respectively. 
    Hence, $\langle \{\va \in V : \text{the projection of } \va^{(b)} \text{ on } W_c \text{ is zero}\}\rangle$ equals $\langle \{\va \in V' : \text{the projection of } \va^{(b)} \text{ on } W_c' \text{ is zero}\}\rangle R^{-1}$. By Claim~\ref{claim:tuple_partition_by_vector}, either the condition (a) or the condition (b) of the second property holds. 
\end{proof}

\begin{framed}
\noindent \textbf{Initial Hierarchical Row Tuple Decomposition Algorithm}

\noindent \textbf{Input:} Matrix Tuple $\vA \in \M(n \times m, \F_q)^\ell$.

\noindent \textbf{Output:} Hierarchical row tuple decomposition $T_1, \dots, T_\zeta$ with parameters $h_0, \dots, h_\beta$.

\noindent \textbf{Algorithm:}

\begin{enumerate}
\item Let $T = \rowtuplespace(\vA)$. 
\item Let $W = \rowvectorspace(T)$,  $V = \langle \{\va \in \rowtuplespace(\vA) : \forall i \in [\ell], \va^{(i)} \in W\} \rangle$. 
Run the Direct Sum Row Tuple Decomposition Algorithm with $T$ and $V$. 
If the output is a subspace $S < T$, then let $T = S$ and go to Step 2. 
Otherwise, denote the output as $U_1, \dots, U_\xi$. 
\item Let $n_0 = \dim(\rowtuplespace(\vA) / W)$, $m_0 = \dim(\rowvectorspace(\vA) / W)$, $\ve_1, \dots, \ve_{n_0} \in \rowtuplespace(\vA)$ be row tuples such that $\ve_1 + W, \dots, \ve_{n_0} + W$ form a linear basis of $\rowtuplespace(\vA) / W$, and $v_1, \dots, v_{m_0} \in \rowvectorspace(\vA)$ such that $v_1 + W_, \dots, v_{m_0} + W$ form a linear basis of $\rowvectorspace(\vA) / W$. 
Let $\vC \in \M(n_0 \times m_0, \F_q)^\ell$ be the matrix tuple such that 
$e_i \vC_i \cdot (v_1 + W, \dots, v_{m_0} + W)^T = \ve_i + W$ for each $i\in[n_0]$.
\item Run Initial Hierarchical Row Tuple Decomposition Algorithm for $\vC$ and let $T_{\vC, 1}, \dots, T_{\vC, \zeta_{\vC}}$ with parameters $h_{\vC, 1}, \dots, h_{\vC, \beta_\vC}$ be the output. 
\item Return $T_1, \dots, T_\zeta$ with $\zeta = \zeta_\vC + 1 + \xi$ and $h_0, \dots, h_\beta$ with $\beta = \beta_{\vC}  +1$ such that \[T_i = \langle \{ \va \in \rowtuplespace(\vA) : \va + W \in \{\vc \cdot (v_1 + W, \dots, v_{m_0} + W)^T  : \vc \in T_{\vC, i}\}\} \rangle\]
for each $i \in [\zeta_\vC]$, $T_{\zeta_\vC + 1} = T$, $T_{\zeta_\vC + 1 + i} = U_i$ for each $i \in [\zeta]$,
$h_i = h_{\vC, i}$ for each $i \in [\beta - 1]$, and $h_\beta = \zeta_{\vC} + 1$. 
\end{enumerate}
\vspace{-.4cm}
\end{framed}

\begin{lemma}\label{lemma:initialization}
The Initial Hierarchical Row Tuple Decomposition Algorithm has the following properties:
\begin{enumerate}
\item
    Given a matrix tuple $\vA \in \M(n\times m, \F_q)^\ell$, in $\text{poly}(n, m, \ell, \log q)$ time, the algorithm outputs a hierarchical row tuple decomposition $T_1, \dots, T_\zeta$ of $\vA$ with parameters $h_0, \dots, h_\beta$ such that for every $i \in [\beta + 1]$, 
    let $W_{i, 1}, \dots, W_{i, w_i} < W_{i-1} / W_i$ be the output of the Row Vector Space Direct Sum Algorithm for $T_{h_{i - 1}}/ W_i$, 
    $(T_{i'} / W_i)^{(k)}$ has either a matching projection or a zero projection on $W_{i, j}$ with respect to $W_{i, 1}, \dots, W_{i, w_i}$ for any $h_{i-1} \leq i' < h_i, j\in[w_i], k\in[\ell]$. 
    \item The algorithm is canonical in the following sense: 
    Let $\vA'$ be another input for the algorithm such that there exists an invertible matrix $L$ and $R$ satisfying $\vA' = L\vA R^{-1}$. The outputs for $\vA$ and $\vA'$ satisfy the following conditions:
    $\zeta$, $\beta$, and parameters $h_0, \dots, h_\beta$ are the same for $\vA$ and $\vA'$, and 
    $T_i' = T_i R^{-1}$ for every $i\in[\zeta]$, where $R$ is an arbitrary invertible matrices such that $\vA' = L \vA R^{-1}$ for some invertible matrix $L$.

\end{enumerate}

\end{lemma}

\begin{proof}
    We observe that Step 2 of the algorithm is executed for at most $n$ times because every execution of Step 2, the dimension of $T$ reduces by at least one. 
    On other hand hand, since 
    at Step 3 of the algorithm, 
    $\dim(T) > 0$ by Claim~\ref{claim:tuple_decomposition_direct_sum} and thus $\dim(W) > 0$, the recursion depth of the algorithm is at most $m$. 
    By Claim~\ref{claim:tuple_decomposition_direct_sum}, the running time algorithm is $\poly(n, m, \ell, \log q)$. 
    The first property is then obtained by Claim~\ref{claim:tuple_decomposition_direct_sum}. 

    The second property is obtained by  Claim~\ref{claim:tuple_decomposition_direct_sum} and an induction on each step of the algorithm.
\end{proof}

\begin{framed}
\noindent \textbf{Hierarchical Row Tuple Decomposition Refinement Algorithm}

\noindent \textbf{Input:} 
 Hierarchical row tuple decomposition $T_1, \dots, T_\zeta$ with parameters $h_0, \dots, h_\beta$ for a matrix tuple $\vA$, and 
characteristic block-compatible row tuple subspace $S < T_{r}$ for some $r \in [\zeta]$ 
such that $\dim(S / W_{\gamma+1}) > 0$, where $\gamma$ is the integer such that $h_{\gamma} \leq r < h_{\gamma + 1}$. 


\noindent \textbf{Output:} Hierarchical row tuple decomposition $U_1, \dots, U_\eta$ and parameters $g_1, \dots, g_\delta$.

\noindent \textbf{Algorithm:}
\begin{enumerate}
    \item If $S/W_{\gamma+1} = T_{r} / W_{\gamma+1}$, then return $T_1, \dots, T_{\xi - 1}, S, T_{\xi + 1}, \dots, T_\zeta$ with $h_0, \dots, h_\beta$. 
    \item Let $V = \langle \{\va \in \rowtuplespace(\vA) : \forall i \in [\ell], \va^{(i)} \in \langle \rowvectorspace(S) \cup W_{\gamma + 1}\rangle\} \rangle$. 
Run Direct Sum Row Tuple Decomposition Algorithm with $S / W_{\gamma + 1}$ and $V / W_{\gamma + 1}$.     
    \begin{enumerate}
    \item If the output is a subspace $S^\dagger$ of $S / W_{\gamma + 1}$, then go back to Step 2 with the new linear space $S$ as $\langle \{\va \in S : \va + W_{\gamma + 1} \in S^\dagger \} \rangle$. 
    \item Otherwise, the output is $S^\dagger_1, \dots, S^\dagger_\mu < V/W_{\gamma + 1}$. 
    \end{enumerate}
    \item If $\dim(\rowvectorspace(V)) < \dim(W_\gamma)$, then 
    \begin{enumerate}
    \item Let $P = \langle T_{h_\gamma} \cup \dots \cup T_\zeta\rangle$, $W = \rowvectorspace(V)$,
    $n_0 = \dim(P / W)$, $m_0 = \dim(\rowvectorspace(P) / W)$, $\ve_1, \dots, \ve_{n_0} \in P$ be row tuples such that $\ve_1 + W, \dots, \ve_{n_0} + W$ form a linear basis of $P / W$, and $v_1, \dots, v_{m_0} \in \rowvectorspace(P)$ such that $v_1 + W_, \dots, v_{m_0} + W$ form a linear basis of $\rowvectorspace(P) / W$. 
    Let $\vC \in \M(n_0 \times m_0, \F_q)^\ell$ be the matrix tuple such that $e_i \vC_i \cdot (v_1 + W, \dots, v_{m_0} + W)^T = \ve_i + W$ for each $i\in[n_0]$.
    \item Run Initial Hierarchical Row Tuple Decomposition Algorithm for $\vC$ and let $T_{\vC, 1}, \dots, T_{\vC, \zeta_{\vC}}$ with parameters $h_{\vC, 1}, \dots, h_{\vC, \beta_\vC}$ be the output. 
    \end{enumerate}
    \item Construct a sequence of row tuple subspaces as follows:
    \begin{enumerate}
    \item  Let $\eta = h_{\gamma} - 1, \delta = \gamma - 1$. For $i = 1, \dots, \eta$, let $U_i = T_i$. 
    For $j = 1, \dots, \delta$, $g_j = h_j$. 
    \item If $\dim(\rowvectorspace(V)) < \dim(W_\gamma)$, then
    \begin{enumerate}
        \item For $j = 0, \dots, \beta_\vC$, $\delta= \delta + 1$, $g_\delta = \eta + h_{\vC, j}$.
        \item For $i = 1, \dots, \zeta_{\vC}$, 
        $\eta = \eta + 1$ and $U_\eta = \langle \{\va \in P : \va + W \in T_{\vC, i}\} \rangle$.
    \end{enumerate}

    \item $\eta = \eta + 1$, $U_\eta = S$, $\delta = \delta + 1$, $g_\delta = \eta$. 
    \item For $i = 1, \dots, \mu$, $\eta = \eta + 1$, $U_\eta = \langle \{\va \in V: \va + W_{\gamma+1} \in S_i^\dagger\}\rangle$.
    \item For $i = h_{\gamma + 1}, \dots, \zeta$, let $\eta = \eta + 1$ and $U_{\eta} = T_i$.
    For $j = \gamma + 1, \dots, \beta$, $\delta = \delta + 1$ and $g_\delta = h_j$. 
    \end{enumerate}
    \item Return $U_1, \dots, U_\eta$ and $g_1, \dots, g_\delta$. 
\end{enumerate}
\vspace{-.3cm}
\end{framed}

\begin{lemma}\label{lemma:row_tuple_decomposition_refinement}
The Hierarchical Row Tuple Decomposition Refinement Algorithm has the following properties:
\begin{enumerate}
\item
    Given a depth-$\beta$ hierarchical row tuple decomposition $T_1, \dots, T_\zeta$ with parameters $h_0, \dots, h_\beta$ for a matrix tuple in $\M(n \times m, \F_q)^\ell$ such that for every $i \in [\beta + 1]$, 
    let $W_{i, 1}, \dots, W_{i, w_i} < W_{i-1} / W_i$ be the output of the Row Vector Space Direct Sum Algorithm for $T_{h_{i - 1}}/ W_i$, 
    $(T_{i'} / W_i)^{(k)}$ has either a matching projection or a zero projection on $W_{i, j}$ with respect to $W_{i, 1}, \dots, W_{i, w_i}$ for any $h_{i-1} \leq i' < h_i, j\in[w_i], k\in[\ell]$ and a subspace $S$ of $T_r$ for some $1 \leq r \leq \zeta$ 
    such that $\dim(S / W_{\gamma + 1}) > 0$, where $\gamma$ is the integer such that $h_\gamma \leq r < h_{\gamma + 1}$,
    in $\poly(n, m, \ell, \log q)$ time, 
    the algorithm outputs a new depth-$\gamma$ hierarchical row tuple decomposition $U_1, \dots, U_\eta$ with parameters $g_1, \dots, g_\delta$ satisfying the following properties
    \begin{enumerate}
    \item 
    Let $X_i = \langle \rowvectorspace(U_{g_i}) \cup \dots\cup \rowvectorspace(U_\eta) \rangle$ for every $i \in \{0, \delta + 1\}$. 
    For every $i \in [\delta + 1]$, 
    let $X_{i, 1}, \dots, X_{i, x_i} < X_{i-1} / X_i$ be the output of the Row Vector Space Direct Sum Algorithm for $U_{g_{i - 1}}/ X_i$, 
    $(U_{i'} / X_i)^{(k)}$ has either a matching projection or a zero projection on $X_{i, j}$ with respect to $X_{i, 1}, \dots, X_{i, x_i}$ for any $g_{i-1} \leq i' < g_i, j\in[x_i], k\in[\ell]$. 
    \item The output satisfies one of the following conditions:
    \begin{enumerate}
        \item $\delta > \beta$.
        \item $\delta = \beta$, and $\eta > \zeta$. 
        \item $\delta = \beta, \eta = \zeta$, $g_i = h_i$ for every $1 \leq i \leq \beta$,
        $U_r < T_r$, and $U_j = T_j$ for every $j \in \{1, \dots, r - 1\} \cup \{r + 1, \dots, \eta\}$. 
    \end{enumerate}
    \end{enumerate}
    \item The algorithm is canonical in the following sense: 
    Let $T_1', \dots, T_\zeta'$ be another depth-$\beta$ hierarchical row tuple decomposition  with parameters the same to $T_1, \dots, T_\zeta$ and a subspace $S'$ of $T_r'$ such that there exists an invertible matrix $R$ satisfying $T_i' = T R^{-1}$ for every $1 \leq i \leq \zeta$ and $S' = S R^{-1}$, then the outputs for $T_1, \dots, T_\zeta, S$ and $T_1', \dots, T_\zeta', S'$ satisfy the following condition:
    Let $R$ be an arbitrary invertible matrix such that such that $T_i' = T_i R^{-1}$ for any $1 \leq i \leq \zeta$ and  $S' = S R^{-1}$. 
    If the output for $T_1, \dots, T_\zeta, S$ is $U_1, \dots, U_\eta$ with parameters $g_1, \dots, g_\delta$, then the output for 
    $T_1', \dots, T_\zeta', S'$ is $U_1 R^{-1}, \dots, U_\eta R^{-1}$ with the same parameters $g_1, \dots, g_\delta$. 
\end{enumerate}
\end{lemma}

\begin{proof}
    We prove the first property. The running time of the algorithm is obtained by Claim~\ref{claim:tuple_decomposition_direct_sum} and the observation that Step 2(a) is executed for at most $n$ times because $\dim(S)$ for the input $S$ is at most $n$.

    By Fact~\ref{fact:ibctuple_invariant_basic}, Claim~\ref{claim:tuple_decomposition_direct_sum}, Lemma~\ref{lemma:initialization}, all the row tuple subspaces are characteristic block-compatible. 
    So, to prove the first property of the lemma, we only need to show that the output is a hierarchical row tuple decomposition satisfying the conditions 1(a) and 1(b).

    If the input satisfies $S / W_{\gamma + 1} = T_r / W_{\gamma + 1}$ and $\dim(S / W_{\gamma + 1}) > 0$, then by Definition~\ref{def:row_tuple_decomposition} and Step 1 of the algorithm, the output is also a hierarchical row tuple decomposition of the matrix tuple satisfying condition 1(a) and  1(b)(iii) of the first property. 
    
    Now, we assume the input satisfies $S / W_{\gamma + 1} < T_r / W_{\gamma + 1}$. 
    By Claim~\ref{claim:tuple_decomposition_direct_sum}, 
    if the algorithm goes back to the start of Step 2 when during the execution of Step 2, then the new row tuple space $S$ is a strict subspace of old $S$, and is not a subspace of $\langle \{T_{h_{\gamma+1}} \cup \dots \cup T_{\zeta}\}\rangle$. 
    Hence, the new $S$ also satisfies satisfies $S / W_{\gamma + 1} = T_r / W_{\gamma + 1}$ and $\dim(S / W_{\gamma + 1}) > 0$, and the Step 2 restarts. 
    
    If the algorithm passes Step 2 after several restarts of Step 2, by Claim~\ref{claim:tuple_decomposition_direct_sum}, we have \[V/W_{\gamma + 1} = (S / W_{\gamma + 1}) \oplus S^\dagger_1 \oplus \dots \oplus S^\dagger_\mu.\]
    Hence, if $\dim(\rowvectorspace(V)) = \dim(W_\gamma)$, by the algorithm, the conditions 1(a) and 1(b)(ii) of the first property are satisfied. 
    Otherwise, $\dim(\rowvectorspace(V)) < \dim(W_\gamma)$, and Step 3 of the algorithm is executed. 
    Since $\dim(P / W) > 0$, 
    by the construction of $\vC$ and Lemma~\ref{lemma:initialization}, the conditions 1(a) and 1(b)(i) are satisfied. 
    

    The second property of the claim is obtained by the correspondence between $T_1, \dots, T_\zeta, S$ and $T_1', \dots, T_\zeta', S'$, Claim~\ref{claim:tuple_decomposition_direct_sum}, Lemma~\ref{lemma:initialization}, and an induction on all the steps of the algorithm. 
\end{proof}

\subsection{Canonical form algorithm}\label{sec:overall_final}

In this section, we give our canonical form algorithms for Theorem~\ref{thm:main} and Corollary~\ref{cor:main}. 

\begin{framed}
\noindent \textbf{Matrix Tuple Equivalence Canonical Form Algorithm}

\noindent \textbf{Input:} Matrix tuple $\vA \in \M(n \times m, \F_q)^\ell$.

\noindent \textbf{Output:} 
Matrix tuple in $\M(n \times m, \F_q)^\ell$.

\begin{enumerate}
\item 
Run the Initial Matrix Tuple Decomposition Algorithm to get a hierarchical row tuple decomposition $T_1, \dots, T_\zeta$ with parameters $h_0, \dots, h_\beta$.

\item 
Let $\vQ_{\beta + 1} = \vA$ with formatting vector $v_i = (e_1, \dots, e_m)$ and depth-$\beta$ hierarchical row tuple decomposition $T_{\vQ_{\beta + 1}, i} = T_i$ for each $i\in [\zeta]$ with parameters $h_0, \dots, h_\beta$. 
\item 
For $i = \beta, \dots, 1$, construct the $i$-quotient matrix tuple $\vQ_i$ of $\vA$ from $\vQ_{i + 1}$ with formatting vector $(v_{i, 1} + W_i, \dots, v_{i, m_i} + W_i)$ and depth-$(i-1)$ hierarchical row tuple decomposition $T_{\vQ_i, 1}, \dots, T_{\vQ_i, h_{i} - 1}$ with parameters $h_0, \dots, h_{i - 1}$. 
\item For $i = 1$ to $\beta + 1$, 
\begin{enumerate}
\item If $i = 1$, then run the Direct Sum Decomposition Algorithm for $\vQ_i$ with $T_{\vQ_i, 1}, \dots, T_{\vQ_{i}, h_i - 1}$, otherwise, run the \IBCtuple Space Algorithm for $\vQ_i$ with $T_{\vQ_i, 1}, \dots, T_{\vQ_{i}, h_i - 1}$ and representative \IBCtuple sequence $\vB_{i - 1, 1}, \dots, \vB_{i - 1, k_{i-1}}$ for $\vQ_{i-1}$. \item If the output is a row tuple subspace $S \leq T_{\vQ_i, j}$ for some $j \in [h_i - 1]$, then run the Hierarchical Row Tuple Decomposition Refinement Algorithm with $T_1, \dots, T_\zeta$ and $\langle \va \in \rowtuplespace(T_j) : \va + W_i \in \vQ_i \cdot (v_{i, 1} + W_i, \dots, v_{i, m_i} + W_i)\}\rangle$, use the output to replace the hierarchical row tuple decomposition $T_1, \dots, T_\zeta$, and go to Step 2 with the new hierarchical row tuple decomposition. 
\item Otherwise, let $\vB_{i, 1}, \dots, \vB_{i, k_i}$ be the output of Step (a). 

\end{enumerate}
\item Return the output of the \IBCtuple Selection Algorithm for $\vA$ with the representative \IBCtuple sequence $\vB_{\beta + 1, 1}, \dots, \vB_{\beta+ 1, k_{\beta + 1}}$. 
\end{enumerate}
\vspace{-.4cm}
\end{framed}

\begin{proof}[Proof of Theorem~\ref{thm:main}]
    For two input matrix tuples $\vA$ and $\vA'$ such that there exist invertible matrices $L$ and $R$ such that $\vA' = L \vA R^{-1}$, 
    the outputs for $\vA$ and $\vA'$ are the same by 
    Lemma~\ref{lem:canonical_with_ibctuples}, 
    Lemma~\ref{lem:direct_sum_ibctuple_algo}, Lemma~\ref{lemma:extension_main_canonical},
    Lemma~\ref{lemma:initialization}, Lemma~\ref{lemma:row_tuple_decomposition_refinement}, and induction on each step of the algorithm. 

    Now we prove the algorithm's running time. By Definition~\ref{def:row_tuple_decomposition}, any hierarchical row tuple decomposition $T_1, \dots, T_\zeta$ with parameters $h_0, \dots, h_\beta$ of a matrix tuple in $\M(n\times m, \F_q)^\ell$ satisfies $\zeta \leq n$, $\beta \leq \zeta$, and $\dim(T_i) \leq n$ for any $i \in [\zeta]$. 
    By Lemma~\ref{lemma:row_tuple_decomposition_refinement}, it takes at most $\poly(n)$ times to jump from Step 4(b) to Step 2. 
    By Lemma~\ref{lem:canonical_with_ibctuples}, 
    Lemma~\ref{lem:direct_sum_ibctuple_algo}, Lemma~\ref{lemma:extension_main},
    Lemma~\ref{lemma:initialization}, and Lemma~\ref{lemma:row_tuple_decomposition_refinement}, the running time of the algorithm is $\poly(n, m, \ell, \log q)$. 
\end{proof}

\begin{framed}
\noindent \textbf{Matrix Tuple Conjugation Canonical Form Algorithm}

\noindent \textbf{Input:} Matrix tuple $\vA \in \M(n, \F_q)^\ell$.

\noindent \textbf{Output:} 
Matrix tuple in $\M(n, \F_q)^\ell$.

\begin{enumerate}
\item Let $\vC = (C_1, \dots, C_{\ell + 1}) \in \M(n, \F_q)^{\ell + 1}$ be the matrix tuple such that 
$C_i = A_i$ for every $i \in [\ell]$, and $C_{\ell + 1} = I_n$. 
\item Run the Matrix Tuple Conjugation Canonical Form Algorithm for $\vC$, and denote the output as $\vE = (E_1', \dots, E_{\ell + 1})$. 
\item Return $(E_1 E_{\ell + 1}^{-1}, \dots, E_\ell E_{\ell + 1}^{-1})$. 
\end{enumerate}
\vspace{-.4cm}
\end{framed}

\begin{proof}[Proof of Corollary~\ref{cor:main}]
We first prove the output is obtained by a conjugation action for the input matrix tuple. 
By Theorem~\ref{thm:main}, there exist invertible matrices $L$ and $R$ such that $\vE = L \vC R^{-1}$.

Let $L$ and $R$ be arbitrary invertible matrices such that $\vE = L \vC R^{-1}$.
Since $C_{\ell + 1} = I$, we have $E_{\ell + 1} = L R^{-1}$. 
Hence, $(E_1 E_{\ell + 1}^{-1}, \dots, E_\ell E_{\ell + 1}^{-1}) = (L E_1 R^{-1} R L^{-1}, \dots, L E_\ell R^{-1} R L^{-1}) =(L E_1  L^{-1}, \dots, L E_\ell L^{-1})$. So the output is a conjugation of the input matrix tuple. 

Next, we show that the outputs for conjugate inputs are identical. Let $\vA=(A_1, \dots, A_\ell), \vA'=(A_1', \dots, A_\ell')\in\M(n, \F)^\ell$ be two conjugate matrix tuples, i.e., there is an invertible matrix $L$ such that $\vA' = L \vA L^{-1}$. 
Let $\vC$ and $\vC'$ be the matrix tuples constructed in Step 1 of the algorithm. 
We have $\vC' = L \vC L^{-1}$. 
Let $\vE$ and $\vE'$ be the matrix tuples obtained in Step 2 of the algorithm.
By Theorem~\ref{thm:main}, we have $\vE = \vE'$. Thus, the algorithm outputs for $\vA$ and $\vA'$ are the same. 
\end{proof}

\begin{proof}[Proof of 
Theorem~\ref{thm:main_space_structure}]
    By the Matrix Tuple Equivalence Canonical Form Algorithm and Theorem~\ref{thm:main}, 
    in Step 5 of the Matrix Tuple Equivalence Canonical Form Algorithm, the sequence $\vB_{\beta + 1, 1}, \dots, \vB_{\beta+ 1, k_{\beta + 1}}$ is a representative \IBCtuple sequence for $\vA$. 
    
    By Lemma~\ref{lemma:extension_main}, every $\vB_{\beta + 1, 1}, \dots, \vB_{\beta+ 1, k_{\beta + 1}}$ satisfies the space property for \IBCtuplesnospace. 
    Hence, by Definition~\ref{def:representative_ibctuple_sequence}, every \IBCtuple is equivalent to one of $\vB_{\beta + 1, 1}, \dots, \vB_{\beta+ 1, k_{\beta + 1}}$, and thus satisfies the space property. 
    By Definition~\ref{def:four_prop}, the block space property is equivalent to the condition of Theorem~\ref{thm:main_space_structure}. 
    Hence, Theorem~\ref{thm:main_space_structure} holds.
\end{proof}

\bibliographystyle{alphaurl}
\bibliography{references}

\end{document}